\RequirePackage{rotating}
\documentclass[acmsmall,dvipsnames]{acmart}
\usepackage{gdefs}
\usepackage{refs}
\newcommand*\xor{\oplus}
\usepackage{kbordermatrix}
\usepackage{rotating}
\DeclareMathOperator{\tensor}{\otimes}     
\DeclareMathOperator{\bigO}{\mathcal{O}}   
\DeclareMathOperator{\bigOmega}{\Omega}    
\usepackage{listings}
\usepackage{url}
\usepackage[export]{adjustbox}
\usepackage{multirow}
\usepackage{float}
\usepackage{subcaption}
\usepackage{amsmath}
\usepackage{amsthm}
\usepackage{autobreak}
\usepackage{enumitem}
\usepackage{setspace}
\usepackage{hyperref}
\usepackage{wrapfig}
\usepackage{graphbox}
\usepackage[framemethod=default]{mdframed}
\usepackage[ruled,linesnumbered]{algorithm2e}
\usepackage{nicematrix}
\SetKwInput{Input}{Input}
\SetKwInput{Output}{Output}
\SetKwBlock{Begin}{begin}{end}

\newcommand{\MPP}{\mbox{Matched-Path Principle\/}}
\newcommand{\CIP}{\mbox{Contextual-Interpretation Principle\/}}
\newcommand{\SIP}{\mbox{Sequential-Invocation Principle\/}}
\newcommand{\GIP}{\mbox{Growth in Paths\/}}
\newcommand{\Matched}{{\it matched\/}}
\newcommand{\NumConnections}{{\rm NumConnections\/}}
\newtheorem{Construction}{\indent\,Construction}

\usepackage{stackengine,graphicx}
\def\rdots{{\rotatebox[origin=l]{29}{$\scriptscriptstyle\ldots\mathstrut$}}}

\makeatletter
\def\thmhead@plain#1#2#3{%
  \thmname{#1}\thmnumber{\@ifnotempty{#1}{ }\@upn{#2}}%
  \thmnote{ {\the\thm@notefont#3}}}
\let\thmhead\thmhead@plain
\makeatother
\theoremstyle{definition}
\newtheorem*{theorem*}{Theorem}

\newenvironment{Constr}{\par %
                \begin{Construction}}{$\QED$\end{Construction}}
\newcommand{\Unfold}{{\it Unfold\/}}
\newcommand{\Fold}{{\it Fold\/}}

\newcounter{LineNumber}

\newenvironment{BaseCase}{\par\noindent %
                {\em Base case\/}:}{}
\newenvironment{InductionStep}{\par\noindent %
                {\em Induction step\/}:}{}
                
\newcommand{\ITE}{{\rm ITE\/}}

\newcommand{\BoolOp}[6]{{%
  \begin{array}{cccc}%
       &   & \multicolumn{2}{c}{#2}%
   \\%
       &   & F & T%
   \\%
     #1 & \begin{array}{@{\hspace{0in}}c@{\hspace{0in}}}%
            F \\ T%
          \end{array}%
        & \multicolumn{2}{@{\hspace{0in}}c@{\hspace{0in}}}{%
            \begin{array}{|cc|}%
              \hline%
              #3 & #4 \\%
              #5 & #6 \\%
              \hline%
            \end{array}%
          }%
  \end{array}}}
  
\usepackage[color]{changebar}
\cbcolor{red}

\newcommand{\sem}[1]{\llbracket {#1} \rrbracket}

\newcommand{\ID}{{\textit{ID}}}
\newcommand{\ADD}{{\textit{ADD}}}
\newcommand{\EQ}{{\textit{EQ}}}
\newcommand{\XOR}{{\textit{XOR}}}
\newcommand{\Op}{{\textit{Op}}}
\newcommand{\QFT}{{\textit{QFT}}}
\newcommand{\GHZ}{{\textit{GHZ}}}

\newcommand{\CNOT}{{\textit{CNOT}}}
\newcommand{\ev}{{\textit{ev}}}
\newcommand{\EV}{{\textit{EV}}}
\newcommand{\bp}{{\textit{bp}}}
\newcommand{\BP}{{\textit{BP}}}
\newcommand{\ZeroBP}{{0_{\BP}}}
\newcommand{\rt}{{\textit{rt}}}
\newcommand{\red}{{\textit{red}}}
\newcommand{\vt}{{\textit{vt}}}

\newcommand{\Cost}{{\textit{Cost}}}

\newcommand{\PairProduct}{{\textit{PairProduct}}}

\newcommand{\HadamardFamily}{{\mathcal{H}}}

\newcommand{\Omit}[1]{}

\newcommand{\twr}[1]{{\color{magenta}{T: #1}}}
\newcommand{\twrchanged}[1]{{\color{cyan}{#1}}}

\newcommand{\changed}[1]{{#1}}

\definecolor{lightgray}{rgb}{0.55,0.52,0.54}

\newcommand{\Supplemental}[2]{#1}      
\newcommand{\OnlySupplemental}[1]{}    
\date{}

\title[CFLOBDDs]{CFLOBDDs: \\ Context-Free-Language Ordered Binary Decision Diagrams}
\Supplemental{
\author{Meghana Aparna Sistla}
\affiliation{%
\institution{The University of Texas at Austin}
\country{USA}
}
\email{mesistla@utexas.edu}
\author{Swarat Chaudhuri}
\affiliation{%
\institution{The University of Texas at Austin}
\country{USA}
}
\email{swarat@cs.utexas.edu}
\author{Thomas Reps}
\affiliation{%
\institution{University of Wisconsin-Madison}
\country{USA}
}
\email{reps@cs.wisc.edu}
}{\author{Double Blind}}

\begin{abstract}

This paper presents a new compressed representation of Boolean functions, called CFLOBDDs (for Context-Free-Language Ordered Binary Decision Diagrams).
They are essentially a plug-compatible alternative to BDDs (Binary Decision Diagrams), and hence useful for representing certain classes of functions, matrices, graphs, relations, etc.\ in a highly compressed fashion.
CFLOBDDs share many of the good properties of BDDs, but---in the best case---the CFLOBDD for a Boolean function can be \emph{exponentially smaller than any BDD for that function}.
Compared with the size of the decision tree for a function, a CFLOBDD---again, in the best case---can give a \emph{double-exponential reduction in size}.
They have the potential to permit applications to (i) execute much faster, and (ii) handle much larger problem instances than has been possible heretofore.

We applied CFLOBDDs in quantum-circuit simulation, and found that for several standard problems the improvement in scalability, compared to BDDs, is quite dramatic.
With a 15-minute timeout, the number of qubits that CFLOBDDs can handle are 65,536 for GHZ, 524,288 for BV; 4,194,304 for DJ; and 4,096 for Grover’s Algorithm, besting BDDs by factors of
128$\times$, 1,024$\times$, 8,192$\times$, and 128$\times$, respectively.
\end{abstract}

\begin{document}
\NiceMatrixOptions{code-for-first-col=\scriptstyle,code-for-first-row=\scriptstyle}


\maketitle{}


\section{Introduction}
\label{Se:Introduction}





Many areas of computer science---such as hardware and software verification, logic synthesis, and equivalence checking of combinatorial circuits---require a space-efficient representation of data, as well as space- and time-efficient operations on data stored in such a representation.
Many of the tasks in the aforementioned areas involve operations on either (i) Boolean functions, or (ii) non-Boolean-valued
functions over Boolean arguments.
In some cases, a level of encoding is involved:
the data of interest could be decision trees, graphs, relations, matrices, circuits, signals, etc., which are encoded as functions of type (i) or (ii).
Binary Decision Diagrams (BDDs) \cite{toc:Bryant86} are one data structure that is widely used for such purposes.
A Boolean function in $B_n = \{ 0,1\}^n \rightarrow \{0,1\}$ is represented in a compressed form as an ROBDD (Reduced Ordered BDD) data structure.
All manipulations of these Boolean functions are carried out using algorithms that operate on ROBDDs. 
ROBDDs are BDDs in which the same variable ordering is imposed on the Boolean variables (``Ordered''), and so-called \emph{don't-care} nodes are removed (``Reduced'').
ROBDDs with non-binary-valued terminals are called Multi-Terminal BDDs (MTBDDs) \cite{dac:CMZFY93,CMU:CS-95-160} or Algebraic Decision Diagrams (ADDs) \cite{iccad:BFGHMPS93}.
We will refer to ROBDDs/MTBDDs/ADDs generically as BDDs from hereon.

\Omit{
In our size comparisons of BDDs versus decision trees and CFLOBDDs, we will consider the version of BDDs in which variable ordering is respected, but don't-care nodes are \emph{not} removed (i.e., plies are not skipped).
In other words, all paths from a root to the sink have length $n$.
This variant is sometimes called \emph{quasi-reduced OBDDs}.
The size of a quasi-reduced OBDD is at most a factor of $n+1$ larger than the size of the corresponding reduced OBDD \cite[Theorem 3.2.3]{Book:Wegener00}.
Throughout the paper, in all theoretical size comparisons, ``BDD'' should be understood to mean ``quasi-reduced OBDD.''\,\footnote{
  The experiments presented in \sectref{evaluation} compare the performance of CFLOBDDs against ROBDDs---specifically, ROBDDs as implemented in CUDD \cite{somenzi2012cudd}.
}
In the best case, a BDD obtains \emph{exponential compression in space}:
the BDD for a Boolean function is exponentially smaller than the decision tree for the function.
}

In the programming-languages community, BDDs are widely used for program analysis and have been used in Datalog interpreters.
\begin{itemize}
  \item
    The SLAM system (later called Static Driver Verifier) was a Microsoft tool for checking temporal properties of device drivers (e.g., that drivers correctly follow API-usage rules) \cite{DBLP:conf/spin/BallR01}.
    BDDs were used in SLAM to represent the abstract transformers of Boolean programs that were abstractions of a driver's source code.
    BDDs allowed the SLAM developers to increase the capabilities of the IFDS framework for interprocedural dataflow analysis \cite{POPL:RHS95} to handle relations over valuations over a Boolean program's Boolean variables \cite{DBLP:conf/paste/BallR01}.
  \item
    The Datalog solver \texttt{bddbddb}, which uses BDDs as the backing representation of relations, was developed by Whaley and Lam to support a variety of program analyses \cite{PLDI:WL04,APLAS:WACL05}.
  \item
  Lhot{\'{a}}k used BDDs in interprocedural program analyses to represent and manipulate collections of large sets, allowing him to use larger programs than previous studies of the factors that affect analysis precision \cite{Thesis:Lhotak06}.
\end{itemize}

In some applications of BDDs, the initial and final BDD structures are of a reasonable size, but there is an ``intermediate swell'' during the computation\Omit{:
temporary BDD structures, which represent intermediate results needed during the computation, grow much larger, and in some cases blow up exponentially.
That is, some intermediate result---the BDD for some function $f$, say---approaches the size of the decision tree for $f$}.
Such a blow-up can cause operations to take a long time, or cause an application to run out of memory.
The size-explosion issue generally limits the use of BDDs to problems involving at most a few hundred Boolean variables.



In this paper, we introduce a new data structure, called \emph{Context-Free-Language Ordered Binary Decision Diagrams} (CFLOBDDs), which are essentially a plug-compatible replacement for BDDs.
CFLOBDDs share many of the good properties of BDDs, but---in the best case---the CFLOBDD for a Boolean function can be \emph{exponentially smaller than any BDD for that function}.
Compared with the size of the decision tree for a function, a CFLOBDD---again, in the best case---can give a \emph{double-exponential reduction in size}.
\Omit{Compared to a decision-tree representation $T_g$ of a Boolean function $g$, a CFLOBDD can provide \emph{double-exponential compression} in space compared to $T_g$.
In contrast, when $g$ is represented using a BDD, it can provide only an \emph{exponential compression} in space compared to $T_g$.
Thus, CFLOBDDs can sometimes provide an \emph{additional exponential degree of compression} than is possible with BDDs.
}
Obviously, not every Boolean function has such a highly compressed representation, but for the ones that do, CFLOBDDs offer much better compression than BDDs, and thus have the potential to permit applications to (ii) execute much faster, and (ii) handle much larger problem instances than has been possible heretofore.

Similar to BDDs, CFLOBDDs can represent functions, matrices, graphs, relations, etc.\ (using either binary-valued or multi-valued terminals, as appropriate), again with the possible advantage of an exponential degree of compression over BDDs.
Even for objects that do not fall into the best-case scenario of perfect double-exponential compression,  CFLOBDDs may provide better compression than BDDs.
Like BDDs, CFLOBDDs are \emph{canonical} (\sectref{Canonicalness}), and operations are performed on them directly (\sectref{cflobdd-algos}):
they are never unfolded to the full decision tree.
Moreover, an implementation can ensure that only a single representative is ever constructed for a given function;
consequently, the test of whether two CFLOBDDs represent equal functions can be performed merely by comparing the values of two pointers.

%
CFLOBDDs are based on the following key insight:

\begin{mdframed}[innerleftmargin = 3pt, innerrightmargin = 3pt, skipbelow=-0.0em]
  A BDD can be considered to be a special form of bounded-size, branching, but non-looping program.
  From that viewpoint, a CFLOBDD can be considered to be a bounded-size, branching, but non-looping program in which a certain form of \emph{procedure call} is permitted.
\end{mdframed}

\smallskip
\noindent
The advantages of this idea are two-fold.
First, whereas a BDD of size $n$ can have at most $2^n$ paths, the ``procedure-call'' mechanism in CFLOBDDs allows a CFLOBDD of size $n$ to have $2^{2^n}$ paths (\sectref{GrowthOfNumberOfPathsWithLevel}).
This difference is what lies behind the potential compression advantage of CFLOBDDs.
Second, even when best-case compression is not possible, such ``procedure calls'' allow there to be additional sharing of structure beyond what is possible in BDDs:
    a BDD can share sub-DAGs, whereas a procedure call in a CFLOBDD shares the ``middle of a DAG.''
    (See \figrefs{walsh1}{walshK}.)

We evaluated CFLOBDDs and BDDs on synthetic benchmarks and for quantum simulation.
We compared the performance in terms of size and execution time:
on problem sizes for which both approaches ran successfully, CFLOBDDs were generally smaller and had lower execution times, particularly at the upper end of the capabilities of BDDs.
Moreover, the improvement that CFLOBDDs bring in scalability is quite dramatic, both for the synthetic benchmarks (\sectref{ResearchQuestionOne}) and for quantum simulation (\sectref{ResearchQuestionTwo}).


Our work makes the following contributions:
\begin{itemize}
  \item
    We introduce a new data structure, called CFLOBDDs, for representing functions, matrices, graphs, relations, and other discrete structures in a highly compressed fashion (\sectref{Overview}).
    In the best case, a CFLOBDD obtains \emph{double-exponential compression in space}:
    i.e., the CFLOBDD for a Boolean function $f$ is double-exponentially smaller than the decision tree for $f$.
  \item
    We present \emph{algorithms} for creating CFLOBDDs and performing operations on them (\sectref{cflobdd-algos}).
    Most operations have low cost:
    For many of the functions of $2^k$ variables for which the CFLOBDD representation is double-exponentially smaller than a decision tree of size $2^{2^k}$, the CFLOBDD can be constructed in time $O(k)$ and space $O(k)$.
    Most unary operations on CFLOBDDs are either constant-time or linear in the size of the argument CFLOBDD.
    The cost of most binary operations is bounded by the product of
    (i) the sizes of the two argument CFLOBDDs, and (ii) the size of the answer CFLOBDD. 
  \item
    We show an \emph{exponential gap} between CFLOBDDs and BDDs (\sectref{efficient-relations}).
  \item
    We give efficient CFLOBDD representations for matrices used in quantum algorithms (\sectref{quantum-algos}).
  \item
    We measured the performance of CFLOBDDs and BDDs on synthetic and quantum-simulation benchmarks (\sectref{evaluation}).
    For several problems, the improvement in scalability enabled by CFLOBDDs is quite dramatic.
    In particular, in the quantum-simulation benchmarks, the number of qubits that could be handled using CFLOBDDs was larger, compared to BDDs, by a factor of
    128$\times$ for GHZ; 1,024$\times$ for BV; 8,192$\times$ for DJ; and 128$\times$ for Grover's algorithm.
\end{itemize}

\paragraph{Organization.}
\sectref{Preliminaries} reviews decision trees and BDDs.
\sectref{Overview} introduces the basic principles underlying CFLOBDDs.
\sectref{Canonicalness} introduces some additional structural invariants that allow us to establish that each Boolean function has a unique, canonical representation as a CFLOBDD.
\sectref{Pragmatics} discusses how some standard techniques---hash-consing 
\cite{Tokyo-TR-74-03:Goto74}, function-caching 
{(or \emph{memo functions} \cite{Edinburgh-MIP-R-29:Michie67}), and reference counting---apply to CFLOBDDs.
\sectref{ADenotationalSemantics} gives a denotational definition of the function that a CFLOBDD represents.
\sectref{cflobdd-algos} presents algorithms for a variety of CFLOBDD operations.
\sectref{efficient-relations} demonstrates an exponential gap between CFLOBDDs and BDDs:
the CFLOBDD for a function $f$ can be exponentially smaller than
any BDD for $f$.
\sectref{quantum-algos} turns to the application of CFLOBDDs for quantum simulation, and shows how CFLOBDDs can represent efficiently some special matrices used in quantum algorithms.
\sectref{evaluation} poses two experimental questions and presents the results of experiments on synthetic and quantum-simulation benchmarks.
\sectref{RelatedWork} discusses related work.
\sectref{conclusion} concludes.
Additional details are provided in
Appendices \sectseqref{additional-notation}{CostOfReduce}.



\section{Preliminaries: A Family of Examples, Decision Trees, and BDDs}
\label{Se:Preliminaries}

The theme of this paper is compressibility of Boolean functions, for which we need a family of examples that are indexed by some parameter.
This section presents the set of \emph{Hadamard matrices} $\HadamardFamily$, which recur in \sectref{Overview} and \sectseqref{efficient-relations}{evaluation}.
It also reviews Boolean functions, decision trees, and BDDs, and shows how decision trees and BDDs can encode the members of $\HadamardFamily$.

\begin{figure}[tb!]
    \centering
    $\begin{array}{c@{\hspace{6ex}}c}
      H_2 = 
        \begin{array}{c@{\hspace{0.25ex}}c}
              & \mbox{\scriptsize{$\quad y_0$}} \\
          \mbox{\scriptsize{$x_0$}}
              & \begin{bNiceArray}{cc}[first-col,first-row]
                       & 0 & 1\\
                     0 & 1 & 1\\
                     1 & 1 & -1\\
                \end{bNiceArray}
        \end{array}
      &
      H_4 = H_2 \tensor H_2
          = \begin{array}{c@{\hspace{0.25ex}}c}
                       & \mbox{\scriptsize{$\quad y_0$}} \\
              \mbox{\scriptsize{$x_0$}}
                       &  \begin{bNiceArray}{cc}[first-col,first-row]
                              & 0   & 1\\
                            0 & H_2 & H_2\\
                            1 & H_2 & -H_2\\
                          \end{bNiceArray}
            \end{array}
         =
         \begin{array}{c@{\hspace{0.25ex}}c}
                       & \mbox{\scriptsize{$\quad y_0 y_1$}} \\
              \mbox{\scriptsize{$x_0 x_1$}}
                       & \begin{bNiceArray}{cccc}[first-col,first-row]
                              & 00 & 01 & 10 & 11\\
                           00 & 1 & 1 & 1 & 1\\
                           01 & 1 & -1 & 1 & -1\\
                           10 & 1 & 1 & -1 & -1\\
                           11 & 1 & -1 & -1 & 1\\
                         \end{bNiceArray}
        \end{array}
    \end{array}$
    \caption{$H_2$ and $H_4$, the first two members of the family of Hadamard matrices $\HadamardFamily = \{ H_{2^i} \mid i \geq 1 \}$.}
    \label{Fi:HadamardMatrices}
\end{figure}

\paragraph{Hadamard Matrices}

The family of Hadamard matrices, $\HadamardFamily = \{ H_{2^i} \mid i \geq 1 \}$, can be defined recursively:
for $i \geq 1$, $H_{2^{i+1}} = H_{2^i} \tensor H_{2^i}$,
with $H_2$ from \figref{HadamardMatrices} as the base case.
where $\tensor$ denotes Kronecker product.\footnote{
  \label{Footnote:KroneckerProduct}
  Others use a different indexing scheme:
  $H_2$ is the same as with our scheme (as is $H_4$), but the recursive definition is $H_{2^{i+1}} = H_2 \tensor H_{2^i}$, for $i \geq 1$.
  Thus, for $i \geq 0$, $H_{2^i}$ is a $2^i \times 2^i$ matrix (and thus has $2^{2i}$ entries).
  In contrast, with our indexing scheme, the matrix we call $H_{2^i}$ is a $2^{2^{i-1}} \times 2^{2^{i-1}}$ matrix, for $i \geq 1$ (and thus has $2^{2^i}$ entries).

  \hspace{1.5ex}
  Put another way, what we call $H_{2^i}$ would conventionally be known as $H_{2^{2^{i-1}}}$.
  Not only do we avoid having to write a doubly superscripted subscript, we will see in \sectref{EncodingHFour} that the recursive rule ``$H_{2^{i+1}} = H_{2^i} \tensor H_{2^i}$'' fits particularly well with the internal structure of CFLOBDDs.
}
\figref{HadamardMatrices} shows $H_2$ and $H_4$, the first two matrices in $\HadamardFamily$.
The \emph{Kronecker product} of two matrices is defined as
\begin{equation*}
  A \tensor B = \left[ \begin{array}{ccc}
                              a_{1,1} & \cdots & a_{1,m} \\
                              \vdots  & \ddots & \vdots  \\
                               a_{n,1} & \cdots & a_{n,m}
                             \end{array}
                    \right] \otimes B
              = \left[ \begin{array}{ccc}
                              a_{1,1}B & \cdots & a_{1,m}B \\
                               \vdots  & \ddots & \vdots   \\
                              a_{n,1}B & \cdots & a_{n,m}B
                             \end{array}
                    \right]
\end{equation*}
Equivalently, $(A \tensor B)_{ii',jj'} = A_{i,j} \times B_{i',j'}$.
If $A$ is $n \times m$ and $B$ is $n' \times m'$, then $A \otimes B$ is $nn' \times mm'$.
\Omit{It is easy to see that the Kronecker product is associative, i.e.,
\begin{center}
  $(A \otimes B) \otimes C = A \otimes (B \otimes C)$.
\end{center}
}

For $i \geq 1$, $H_{2^i}$ is a square matrix of size $2^{2^{i-1}} \times 2^{2^{i-1}}$.
Thus, the number of rows/columns/entries in $H_{2^{i+1}}$ is the \emph{square} of the number of rows/columns/entries in $H_{2^i}$.
For example, $H_4$ is $4 \times 4$ (16 entries);
$H_8$ is $16 \times 16$ (256 entries).
An indexing scheme for $H_{2^i}$ can be defined that uses $2^{i-1} + 2^{i-1}$ $= 2 * 2^{i-1}$ $= 2^i$ Boolean variables.
As shown in \figref{HadamardMatrices}, $H_2$ requires 2 variables---$x_0$ for the row index and $y_0$ for the column index---whereas $H_4$ requires 4 variables---$x_0$ and $x_1$ for the row index, and $y_0$ and $y_1$ for the column index.
In general, $H_{2^i}$ can be treated as a Boolean function of type $\{0,1\}^{2^{i-1}}\times \{ 0,1 \}^{2^{i-1}} \rightarrow \{-1, 1\}$.
Our convention is that $x_0$ and $y_0$ are the most-significant bits of the row and column indexes, respectively;
$x_1$ and $y_1$ are the next-most-significant bits, respectively, etc.

\begin{figure}[tb!]
  \centering
  \begin{subfigure}[t]{0.45\linewidth}
    \centering
    \includegraphics[scale=0.5]{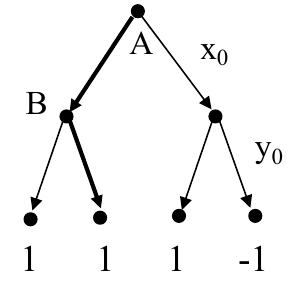}
    \caption{Decision tree for $H_2$}
    \label{Fi:walsh1_dd}
  \end{subfigure}
  \begin{subfigure}[t]{0.25\linewidth}
    \centering
    \includegraphics[scale=0.45]{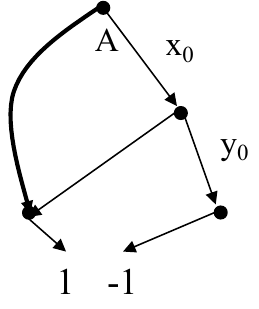}
    \caption{BDD for $H_2$}
    \label{Fi:walsh1_bdd}
  \end{subfigure}
  \begin{subfigure}[t]{0.47\linewidth}
    \includegraphics[width=\linewidth]{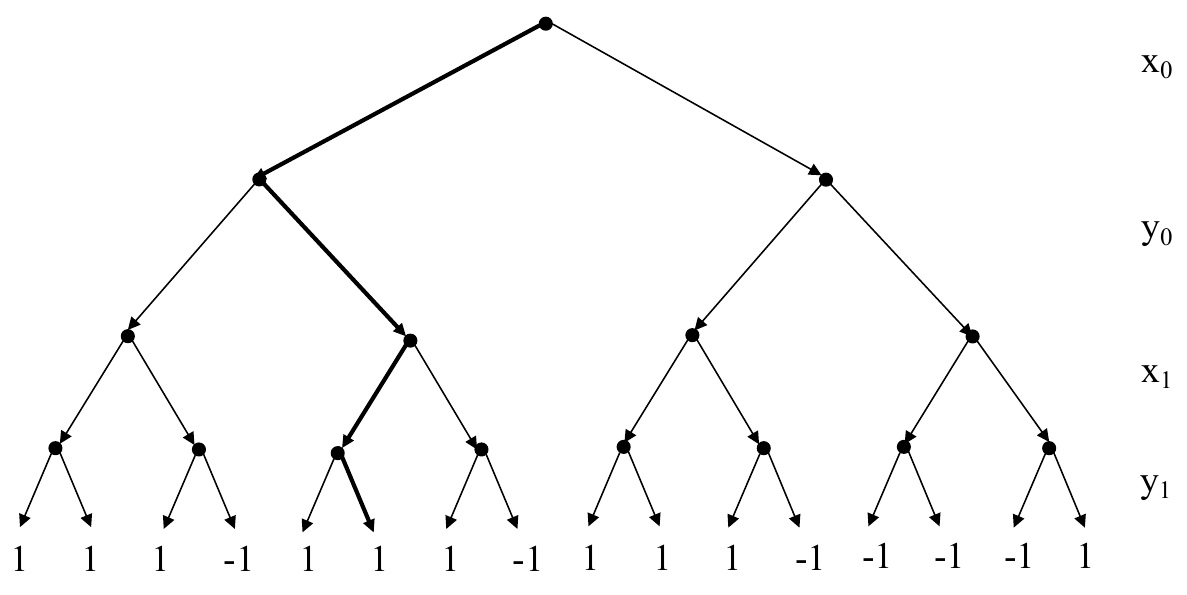}
    \caption{Decision tree for $H_4$}
    \label{Fi:walsh2_dd}
  \end{subfigure}
  \begin{subfigure}[t]{0.25\linewidth}
    \centering
    \includegraphics[scale=0.4]{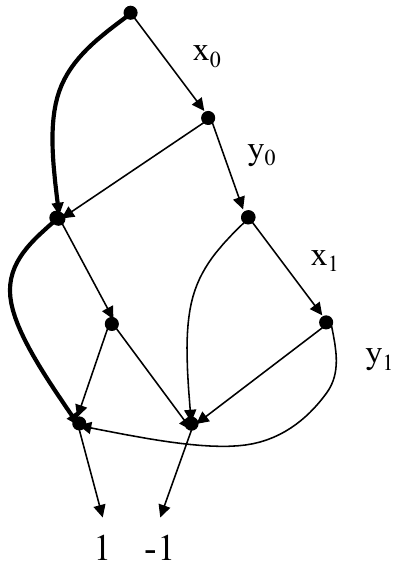}
    \caption{BDD for $H_4$}
    \label{Fi:walsh2_bdd}
  \end{subfigure}
  \begin{subfigure}[t]{0.25\linewidth}
    \centering
    \includegraphics[scale=0.4]{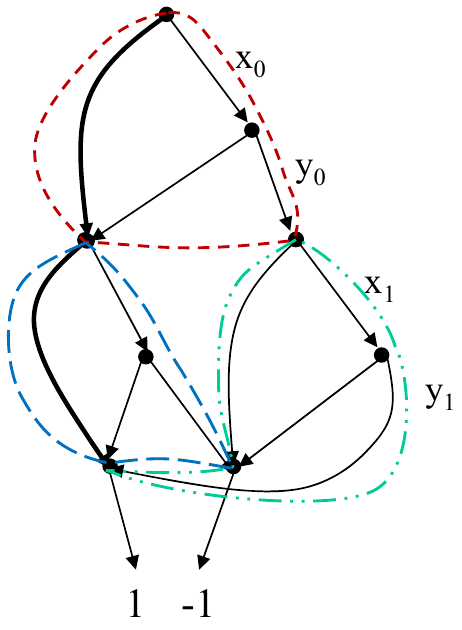}
    \caption{Occurrences of $H_2$ in the BDD for $H_4$}
    \label{Fi:walsh2_bdd_tensor_product}
  \end{subfigure} 
  \caption{\protect \raggedright 
    Decision trees and BDDs for $H_2$ and $H_4$, with plies in interleaved most-significant-bit order---$\langle x_0, y_0 \rangle$ and $\langle x_0, y_0, x_1, y_1 \rangle$, respectively.
    The bold paths show the assignments $[{x_0} \mapsto F, {y_0} \mapsto T]$ (for $H_2[0,1]$) and $[{x_0} \mapsto F, {y_0} \mapsto T, {x_1} \mapsto F, {y_1} \mapsto T]$ (for $H_4[0,3]$), respectively.
  }
  \label{Fi:HadamardDecisionTreesAndBDDs}
\end{figure}

\paragraph{Boolean Functions}
A \emph{Boolean function} over $n$ variables is a function in $\{F,T\}^n \rightarrow \{F,T\}$.
This paper is also concerned with pseudo-Boolean functions: a \emph{pseudo-Boolean} function over $n$ variables and value domain $W$ is a function in $\{F,T\}^n \rightarrow W$.
Because there is little chance of confusion, for brevity, we typically refer to such a function as a ``Boolean function.''
We also use $0$ and $1$ as synonyms for $F$ and $T$, respectively. 

Hadamard matrix $H_{2^i}$ can be considered to be a (pseudo-)Boolean function in $\{0,1\}^{2^i} \rightarrow \{-1, 1\}$, with some convention about how the $2^i$ input variables correspond to bits of the row-index and the column-index of the matrix.

\paragraph{Decision Trees}

A \emph{decision tree} is a tree representation of a Boolean function.
For a Boolean function $B$ in $\{F,T\}^n \rightarrow W$, the decision tree $T_B$ for $B$ is a complete binary tree with $n$ plies and a value from $W$ at each leaf.
$T_B$ comes with a specific ordering on the $n$ Boolean inputs of $B$:
each ply of $T_B$ corresponds to some specific Boolean variable $v$ among $B$'s $n$ Boolean input variables.
$T_B$---and hence $B$---can be evaluated with respect to an input assignment $[v_1 \mapsto b_1, \ldots, v_n \mapsto b_n$] (where $b_1, \ldots, b_n \in \{F,T\}$) by following a root-to-leaf path in $T_B$, returning the value that labels the leaf.
(Note that $v_1$ is not necessarily associated with the ply at the root.
The order used by $T_B$ is fixed, but can be any of the permutations of the sequence $\langle v_1, \ldots, v_n \rangle$.)

\figrefs{walsh1_dd}{walsh2_dd} show two decision trees, with the convention that will be used throughout the paper that at each interior node, the left branch is taken when the current Boolean variable in the assignment has the value $F$ (or $0$);
the right branch is taken for the value $T$ (or $1$). 
\figref{walsh1_dd} shows the decision tree for $H_2$, which has $2$ plies, $3$ interior nodes, and $4$ leaf nodes, using the variable ordering $\langle x_0, y_0 \rangle$.
In \figref{walsh1_dd}, the path highlighted in bold is for the assignment $[{x_0} \mapsto F, {y_0} \mapsto T]$, which corresponds to $H_2[0,1]$ whose value is $1$.
\figref{walsh2_dd} shows the decision tree for $H_4$, which has $4$ plies and $15$ interior nodes, using the interleaved-variable ordering $\langle x_0,y_0,x_1,y_1 \rangle$.
The path in bold is for $[{x_0} \mapsto F, {y_0} \mapsto T, {x_1} \mapsto F, {y_1} \mapsto T]$, which corresponds to $H_4[0,3]$, whose value is $1$.

In \figref{walsh2_dd}, the Kronecker product in the expression $H_4 = H_2 \tensor H_2$ corresponds to \emph{stacking decision trees}.
In essence, the $\langle x_0, y_0 \rangle$ plies correspond to the left occurrence of $H_2$ in``$H_2 \tensor H_2$.''
At each ``leaf'' (the four interior nodes after the $y_0$ ply), there is another copy of $H_2$ in the $\langle x_1, y_1 \rangle$ plies, with the terminal values labeled with the product of the left $H_2$'s value and the right $H_2$'s value.
We can construct a decision tree for each member of $\HadamardFamily$ by repeated stacking, doubling the number of plies each time in accordance with the definition $H_{2^{i+1}} = H_{2^i} \tensor H_{2^i}$.

Boolean functions and decision trees are related by the following fact:
\begin{Obs}\label{Obs:DecisionTreesRepresentBooleanFunctions}
  Consider the sets of (i) Boolean functions in $\{0,1\}^n \rightarrow W$, and (ii) $n$-ply decision trees with leaves labeled by values in $W$, using a variable ordering that is some fixed permutation of $\langle v_1, \ldots, v_n \rangle$.
  These sets can be put into one-to-one correspondence.
\end{Obs}
For each Boolean function $B: \{0,1\}^n \rightarrow W$, create the $n$-ply decision tree $T_B$ in which the value $B(b_1,\ldots,b_n)$ is placed at the end of the path in $T_B$ for the assignment $[v_1 \mapsto b_1, \ldots, v_n \mapsto b_n]$.
Conversely, for each decision tree $T_B$, let $B$ be the function in $\{0,1\}^n \rightarrow W$ for which $B(b_1,\ldots,b_n)$ equals the value $w$ at the end of the path in $T_B$ for the assignment $[v_1 \mapsto b_1, \ldots, v_n \mapsto b_n]$.
Finally, if two decision trees $T_1$ and $T_2$ represent the same Boolean function $B$, then the sequence of leaves in left-to-right order from each tree are equal, and thus $T_1$ and $T_2$ are the same tree (as a mathematical object).
Thus, the $n$-ply decision trees that use a given variable ordering represent the Boolean functions in $\{0,1\}^n \rightarrow W$ uniquely.

\paragraph{BDDs}
A BDD is a compressed representation of a decision tree.
\figref{walsh1_bdd} shows the BDD for $H_2$, using the variable ordering $\langle x_0, y_0 \rangle$.
Again, left branches are for $F$ (or $0$);
right branches are for $T$ (or $1$).
In the $H_2$ matrix, rows $0$ and $1$ are different, and hence the BDD node for $x_0$ is a $\textit{fork\_node}$, which forks to two different substructures.
In row $0$ of the matrix, columns $0$ and $1$ are identical, and hence the $y_0$ ply is skipped in the $F$ branch of $x_0$, with the $F$ branch of $x_0$ leading directly to the terminal value $1$.
Conversely, in row $1$ of the matrix, the columns differ, and hence the BDD node for $y_0$ in the $T$ branch of $x_0$ is a $\textrm{fork\_node}$.
In \figref{walsh1_bdd}, the bold path is for the assignment $[{x_0} \mapsto F, {y_0} \mapsto T]$ for $H_2[0,1]$.
(Only the edge for $x_0 \mapsto F$ is highlighted because the ply for $y_0$ is skipped when $x_0 \mapsto F$.)

\figref{walsh2_bdd} shows the BDD for $H_4$ under the interleaved-variable ordering $\langle x_0,y_0,x_1,y_1 \rangle$.
The bold path is for the assignment $[{x_0} \mapsto F, {y_0} \mapsto T, {x_1} \mapsto F, {y_1} \mapsto T]$, which corresponds to $H_4[0,3]$.
(The path in the BDD only shows $x_0 \mapsto F, x_1 \mapsto F$ because the plies for $y_0$ when $x_0 \mapsto F$, and $y_1$ when $x_0 \mapsto F$ and $x_1 \mapsto F$ are skipped.)

\figref{walsh2_bdd_tensor_product} shows that the Kronecker product $H_4 = H_2 \tensor H_2$ corresponds to \emph{stacking BDDs}---in essence, each terminal of the BDD for the left occurrence of $H_2$ in ``$H_2 \tensor H_2$'' is replaced by a copy of $H_2$.
The BDD for $H_4$ contains three occurrences of $H_2$: one in the $\langle x_0, y_0 \rangle$ plies, and two in the $\langle x_1, y_1 \rangle$ plies.
The leftmost $\langle x_1, y_1 \rangle$ occurrence (blue-dashed outline) accounts for the three occurrences of matrix $H_2$ in the $H_4$ matrix;
the rightmost occurrence (green dashed-double-dotted outline) corresponds to the negated matrix $\,{-}H_2$ in the lower-right corner of $H_4$ (cf.\ \figref{HadamardMatrices}).
Consequently, one can construct a BDD for each member of $\HadamardFamily$ by repeated stacking, doubling the number of plies each time, per $H_{2^{i+1}} = H_{2^i} \tensor H_{2^i}$, but only \emph{tripling} the size with each such stacking operation (e.g., $H_8 = H_4 \tensor H_4$ has three copies of $H_4$, etc.). 
Consequently, the size of the BDD for $H_{2^i}$ is $O(3^i)$.

\paragraph{Discussion}

The decision tree for $H_{2^i}$ has height $2^{i}$, $2^{2^{i}}$ leaves, and $2^{2^{i}}\! - \! 1$ internal nodes.
Thus, the size of the tree is double exponential in $i$.
As observed above, the size of the BDD for $H_{2^i}$ is $O(3^i)$, and hence, compared to decision trees, BDDs achieve \emph{exponential compression} on $\HadamardFamily$.

In contrast, CFLOBDDs employ a different principle than stacking to account for Kronecker product.
Looking ahead, this principle is explained in \sectref{EncodingHFour}, and as we will see when we get to \figref{walshKGeneralCase}, there is a CFLOBDD of size $O(i)$ that encodes $H_{2^i}$.
Consequently, CFLOBDDs achieve \emph{double-exponential} compression on $\HadamardFamily$.
Moreover, in \sectref{Separation:HadamardRelation}, we show that this exponential separation is inherent:
a BDD that represents $H_{2^i}$ requires $\Omega(2^i)$ nodes (\theoref{HadamardSeparation}).

\medskip
In the remainder of the paper, detailed knowledge about BDDs is not essential.
The primary purpose of the material that discusses BDDs is to show that CFLOBDDs offer something new, but that material is tangential to being able to understand the CFLOBDD algorithms that we give.
The paper gives what is essentially a complete account of CFLOBDD operations and invariants, and we hope that it could be read by someone who knows little about BDDs.
Nevertheless, additional knowledge about BDD internals could help readers
appreciate the material in the paper.
For background about how BDDs are implemented, the reader is referred to Brace et al.\ \cite{dac:BRB90}.


\section{CFLOBDDs}
\label{Se:Overview}

CFLOBDDs are a binary decision diagram inspired by BDDs, 
but the two data structures are based on different principles.
A BDD is an acyclic finite-state machine (modulo ply-skipping), whereas a CFLOBDD is a particular kind of \emph{single-entry, multi-exit, non-recursive, hierarchical finite-state machine} (HFSM) \cite{TOPLAS:ABEGRY05}.
This section describes the basic principles of CFLOBDDs, illustrating them via encodings of $H_2$ and $H_4$ with the variable orderings $\langle x_0, y_0 \rangle$ and $\langle x_0, y_0, x_1, y_1 \rangle$, respectively.

\begin{figure}[tb!]
  \centering
  \begin{tabular}{@{\hspace{0ex}}c@{\hspace{4.5ex}}c@{\hspace{0ex}}}
    \includegraphics[scale=0.45]{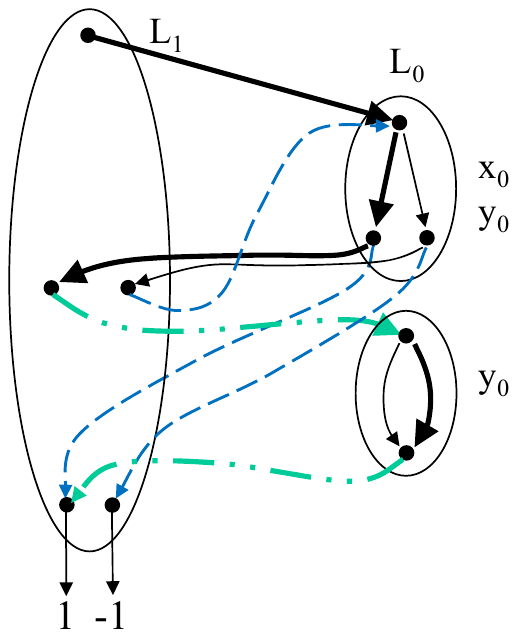}
    &
    \includegraphics[scale=0.42]{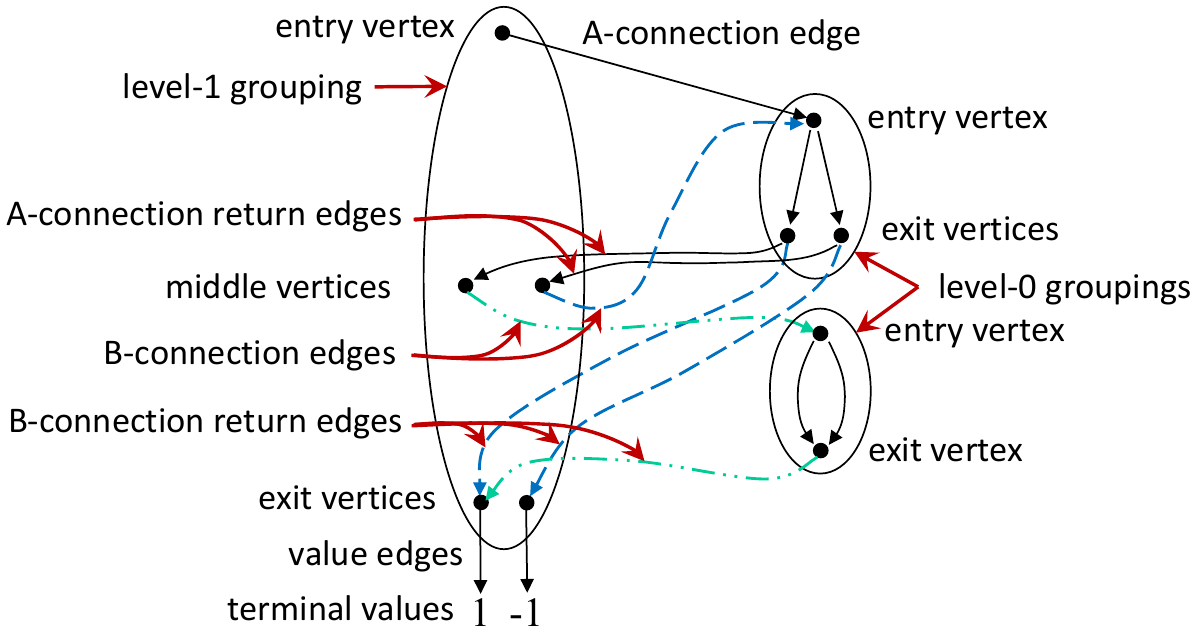}
    \\
    {\small (a)} & {\small (b)}
  \end{tabular}
  \caption{\protect \raggedright 
    (a) CFLOBDD for $H_2$ using the variable ordering $\langle x_0, y_0 \rangle$.
    The bold path is for the assignment $[{x_0} \mapsto F, {y_0} \mapsto T]$ for $H_2[0,1]$.
    (b) Guide to the terminology introduced in \defref{MockCFLOBDD}.
  }
  \label{Fi:walsh1}
\end{figure}

\paragraph{Intuition}
Before discussing the CFLOBDD data structure in detail, we give some intuition about the decomposition principle used in CFLOBDDs. 

Consider a function $f : \{0,1\}^n \rightarrow [1 \dots m]$ over variables $x_0, \dots , x_{n-1}$.
In the classical Shannon decomposition of $f$, one looks at the value of $x_0$ and then derives two co-factors $g_0 = f |_{x_0 = 0}$ and $g_1 = f |_{x_0 = 1}$, both of which are functions over variables $x_1, \ldots, x_{n-1}$.
Functions $g_0$ and $g_1$ can be combined to yield $f$ by the identity $f = \overline{x_0} \cdot g_0 + x_0 \cdot g_1$ (where $\overline{x_0}$ denotes the complement of $x_0$, ``$\cdot$'' denotes logical-and, and ``$+$'' denotes logical-or).
(See \cite[\S4.2]{RDF:4:CFZ96} for a precise definition of the generalization of the Shannon decomposition for MTBDDs.)
The same decomposition can be carried out recursively on $g_0$ and $g_1$, and OBDDs---whether reduced or not---exploit this decomposition by sharing common co-factors that arise in the different plies of the recursive decomposition.

The decomposition used in CFLOBDDs is different.
The number of variables $n$ is assumed to be a power of 2, and at each decomposition level the variables are divided into two halves: $x_0, \dots, x_{n/2-1}$ and $x_{n/2},\dots, x_{n-1}$.\footnote{
  For a Boolean function of $m$ variables that is not a power of 2, one can pad the function with dummy Boolean variables to reach the next higher power of 2.
  Depending on the function, the user may choose to interleave the dummy variables among the ``legitimate'' variables or place them all at the end (or some combination of both).
  By this device, every Boolean function can be represented as a CFLOBDD.
  (See also the discussion in \sectref{CanonicalnessHighLevel} of property (\ref{It:DecisionTreeRepresentedByDecisionTree}).)
}
Let $g_0$ be the function of the first $n/2$ variables that maps them to $[1 \ldots k]$, where $k$ is the number of equivalence classes of residual functions one has after the first $n/2$ variables of $f$ are read.
($k$ equals the number of nodes in the corresponding BDD for $f$ at ply $n/2$.)
For each $i \in [1 \dots k]$, $g_i$ is the appropriate function over the remaining $n/2$ variables, which combined with $g_0$ (based on index $i$), and an appropriate matching of returned values, yields $f$.\footnote{
  We are being deliberately vague about how $g_0, g_1, \ldots, g_k$ are combined, because the details are somewhat complicated.
  See \defrefs{MockCFLOBDD}{CFLOBDD} for the precise definition.
}
The representation allows sharing across all the functions $g_0, g_1, \ldots, g_k$.
Moreover, the divide-the-variables-in-half decomposition is carried out recursively on $g_0, g_1, \ldots, g_k$, with mutual sharing of the decomposed functions that arise at all levels.

Rather than producing a DAG-structured data structure, as one has with BDDs, the divide-the-variables-in-half decomposition leads to a structure that resembles an HFSM (or, alternatively, the interprocedural control-flow graph for a non-recursive, multi-procedure program).

\subsection{Matched Paths}
\label{Se:MatchedPaths}

The CFLOBDD representation of $H_2$ consists of three \emph{groupings}, shown as three ovals in \figref{walsh1}a.\footnote{
  \label{Footnote:NodeGroupingVertexTerminology}
  Groupings are represented in memory as a kind of node structure, but we will use ``nodes'' solely for decision trees and BDDs.
  Groupings are depicted as ovals, and the dots inside will be referred to as ``vertices.''
}
Each CFLOBDD grouping is associated with a given \emph{level}.
The two small ovals are at level $0$ (labeled $L_0$), and the large oval is at level $1$ (labeled $L_1$).
There is an implicit hierarchical structure to the levels, and level-$0$ groupings are said to be \emph{leaves} of the CFLOBDD.
There are only two possible types of level-0 groupings:
\begin{itemize}
  \item
    A level-0 grouping like the one at the upper right in \figref{walsh1}a is called a \emph{fork grouping}.
  \item
    A level-0 grouping like the one at the lower right in \figref{walsh1}a is called a \emph{don't-care grouping}.
\end{itemize}
The vertex at the top of each grouping is the grouping's \emph{entry vertex}.
The entry vertex of a level-$0$ grouping corresponds to a decision point:
left branches are for $F$ (or $0$);
right branches are for $T$ (or $1$). 
The vertices at the bottom of each grouping are called \emph{exit vertices}; those in the middle of the level-$1$ grouping are called \emph{middle vertices}.

In matrix $H_2$, each entry is either $1$ or $-1$.
Each assignment over $\langle x_0, y_0 \rangle$ corresponds to a special kind of path in \figref{walsh1}a that leads to either $1$ or $-1$.
Each such path starts from the entry vertex of the level-$1$ grouping, making ``decisions'' for the next variable in sequence each time the entry vertex of a level-$0$ grouping is encountered.

\figref{walsh1}a illustrates the key principle behind CFLOBDDs---namely, the use of a \emph{matching condition on paths}.
The bold path is for the assignment $[{x_0} \mapsto F, {y_0} \mapsto T]$, which corresponds to $H_2[0,1]$.
The path starts at the level-$1$ grouping's entry vertex and goes
to the entry vertex of the level-$0$ fork grouping via a solid edge (---);
takes the left branch of the fork grouping (corresponding to $x_0 \mapsto F$);
and leaves the fork grouping via a solid edge (---), reaching the leftmost of the
middle
vertices of the level-$1$ grouping.
The path then goes to the entry vertex of the level-$0$ don't-care grouping via a dashed-double-dotted edge (${\color{SeaGreen} -\cdot\cdot\,- }$);
takes the right branch of the don't-care grouping (corresponding to $y_0 \mapsto T$);
and leaves via a dashed-double-dotted edge (${\color{SeaGreen} -\cdot\cdot\,- }$), reaching the leftmost exit vertex of the level-$1$ grouping, which is connected to the terminal value $1$
(the value of $H_2[0,1]$).
A pair of incoming and outgoing edges of a grouping, such as the pairs of black solid edges and green dashed-double-dotted edges in the bold path in \figref{walsh1}a, are said to be \emph{matched}.
The bold path itself is called a \emph{matched path}.
This example illustrates the following principle:

\begin{mdframed}[innerleftmargin = 3pt, innerrightmargin = 3pt, skipbelow=-0.0em]
  $\textbf{\MPP}$.  {\em When a path follows an edge that returns to level $i$ from level $i-1$, it must follow an edge that matches the closest preceding edge from level $i$ to level $i-1$.\/}
\end{mdframed}

Formally, the matched-path principle can be expressed as a condition that---for a path to be \emph{matched}---the word spelled out by the labels on the edges of the path must be a word in a certain context-free language
\cite{kn:Yann90}.
(This idea is the origin of ``CFL'' in ``CFLOBDD''.)
One way to formalize the condition is to label each edge from level $i$ to level $i-1$ with an open-parenthesis symbol of the form ``$(_b$'', where $b$ is an index that distinguishes the edge from all other edges to any entry vertex of any grouping of the CFLOBDD.
(In particular, suppose that there are  $\NumConnections$ such edges, and that the value of $b$ runs from $1$ to $\NumConnections$.)
Each return edge that runs from an exit vertex of the level $i-1$ grouping back to level $i$, and corresponds to the edge labeled ``$(_b$'', is labeled ``$)_b$''.
Each path in a CFLOBDD then generates a string of parenthesis symbols formed by concatenating, in order, the labels of the edges on the path.
(Unlabeled edges in the level-$0$ groupings are ignored in forming this string.)
A path in a CFLOBDD is called a \emph{$\Matched$-path} iff the path's word is in the language $L(\Matched)$ of balanced-parenthesis strings generated by
\begin{equation}
  \label{Eq:MatchedPathGrammar}
  \Matched
  \rightarrow
  \epsilon \mid \Matched~\Matched \mid {(_b}~\Matched~{)_b} \quad 1 \leq b \leq \NumConnections 
\end{equation}
Only $\Matched$-paths that start at the entry vertex of the CFLOBDD's highest-level grouping and end at a terminal value are considered in interpreting a CFLOBDD.

In the figures in the paper, we use black solid (---), {\color{RoyalBlue} blue dashed} ({\textbf {\color{RoyalBlue} --\,--\,--}}), {\color{red} red short-dashed} ({\textbf {\color{red} -\,-\,-\,-}}), {\color{Orchid} purple dashed-dotted} (${\color{Orchid} -\cdot\cdot\,- }$), and {\color{SeaGreen} green dashed-double-dotted} (${\color{SeaGreen} -\cdot\cdot\,- }$) edges, in the indicated colors, rather than attaching explicit labels to edges.
To reduce the number of colors used, we sometimes re-use colors in a given figure;
however, it should still be clear which pairs of edges match.

The matched-path principle allows a given grouping to play multiple roles during the evaluation of a Boolean function.
In particular, the level-$0$ groupings are shared, and thus are used to interpret \emph{different variables at different places in a matched path through a CFLOBDD}.
For example, the level-$0$ fork grouping in \figref{walsh1}a is used to interpret
(i) $x_0$ (when ``called'' via the black solid edge), and
(ii) $y_0$ (when ``called'' via the blue dashed edge, which happens when
$x_0 \mapsto T$).\footnote{
  The term ``call'' is by analogy with how matched paths model the actions of procedure calls in graphs used for interprocedural dataflow analysis \cite{kn:SP81,kn:RHS95}, interprocedural slicing \cite{kn:HRB90}, and model checking hierarchical state machines \cite[\S5]{DBLP:conf/icalp/BenediktGR01}.
}
The edge-matching condition is important because the black solid return edges lead to the level-$1$ grouping's middle vertices, whereas the blue dashed return edges lead to the level-$1$ grouping's exit vertices.

In \figref{walsh1}a, the fork grouping is labeled with $x_0$ and $y_0$, and the don't-care grouping with $y_0$.
However, because the level-$0$ groupings interpret different variables at different places in a matched path, in later diagrams the level-$0$ groupings are generally not labeled with specific variables.
In general, the principle is as follows:

\begin{mdframed}[innerleftmargin = 3pt, innerrightmargin = 3pt, skipbelow=-0.0em]
  $\textbf{\CIP}$.
  {\em A level-$0$ grouping is not associated with a specific Boolean variable.
  Instead, the variable that a level-$0$ grouping refers to is determined by context:
  the $n^{\textit{th}}$ level-$0$ grouping visited along a matched path is used to interpret the $n^{\textit{th}}$ Boolean variable.}
\end{mdframed}


The reader might be worried by the fact that \figref{walsh1}a contains cycles.
That is, if one ignores the ovals in \figref{walsh1}a, as well as the distinctions among solid, dashed, and dashed-double-dotted edges, one is left with a cyclic graph:
there is a cycle that starts at the rightmost middle vertex of the level-$1$ grouping, follows the blue dashed edge ({\textbf {\color{RoyalBlue} --\,--\,--}}) to the entry vertex of the level-$0$ fork-grouping, takes the right branch, and returns along the black solid edge (---) to the rightmost middle vertex of the level-$1$ grouping.
However, that path is excluded from consideration because it is not a matched path:
the solid edge does not match with the preceding dashed edge.

\subsection{CFLOBDD Requirements}
\label{Se:Requirements}

In designing CFLOBDDs, the goal is to meet the following five requirements:

\begin{mdframed}[innerleftmargin = 2pt, innerrightmargin = 2pt, skipbelow=-0.0em]
  \begin{enumerate}[left=0pt .. 1.5\parindent]
    \item
      \label{Req:DTSoundness}
      \emph{Soundness}: Every level-$k$ CFLOBDD represents a decision tree of height $2^k$ and size $2^{2^k}$
    \item
      \label{Req:DTCompleteness}
      \emph{Completeness}: each decision tree of height $2^k$ and size $2^{2^k}$ can be encoded as a level-$k$ CFLOBDD
    \item
      \label{Req:Compression}
      \emph{Best-case double-exponential compression}: in the best case, a decision tree of height $2^k$ and size $2^{2^k}$ can be encoded as a level-$k$ CFLOBDD of size $k$
    \item
      \label{Req:Canonicity}
      \emph{Canonicity}: CFLOBDDs are a canonical representation of Boolean functions
    \item
      \label{Req:ComputationalEfficiency}
      \emph{Computational efficiency}: 
      most operations run in time polynomial in the sizes of (i) the input CFLOBDDs, or (ii) the input CFLOBDDs and the output CFLOBDD
  \end{enumerate}
\end{mdframed}

\noindent
These requirements are similar to those for BDDs, but with double-exponential parameters---rather than single-exponential parameters---in Requirements (\ref{Req:DTSoundness})--(\ref{Req:Compression}).
To satisfy these more stringent requirements, we define a data structure that is quite different from BDDs (see \sectrefs{CFLOBDDsDefined:PartI}{StructuralRestrictions}).

Requirements (\ref{Req:DTSoundness}) and (\ref{Req:Compression}) are established in \sectrefs{CFLOBDDOperationalSemantics}{NoDistinctionProtoCFLOBDDs}, respectively.
Requirements (\ref{Req:DTCompleteness}) and (\ref{Req:Canonicity}) are established in \sectref{Canonicalness} and Appendix \sectref{canonical-proof}.
Requirement (\ref{Req:ComputationalEfficiency}) is addressed in \sectrefsp{Pragmatics}{cflobdd-algos}{CostOfReduce};
in particular, \tableref{algo-list-cflobdds} at the beginning of \sectref{cflobdd-algos} lists the fourteen main operations on CFLOBDDs and the asymptotic running times of the algorithms that we give for the operations.
BDDs enjoy the more desirable property that most operations run in time polynomial in the sizes of the input BDDs, but the same property does not seem possible for CFLOBDDs.
\sectref{CostOfReduce} establishes that the time complexity of a key subroutine used in several of the CFLOBDD operations to maintain canonicity is polynomial in the sizes of the input and output CFLOBDDs.

\subsection{CFLOBDDs Defined, Part I: Basic Structure}
\label{Se:CFLOBDDsDefined:PartI}

Our formal definition of CFLOBDDs is given in two parts:
\defref{MockCFLOBDD} (below) and \defref{CFLOBDD} (\sectref{StructuralRestrictions}).
\defref{MockCFLOBDD} defines the basic structure of CFLOBDDs, whose various elements are depicted in \figref{walsh1}b.
\defref{CFLOBDD} imposes some additional structural invariants to ensure that CFLOBDDs provide a canonical representation of Boolean functions.
Much about CFLOBDDs can be understood just from \defref{MockCFLOBDD}, so we postpone introducing the structural invariants until we address canonicity in \sectref{Canonicalness}.
Where necessary, we distinguish between \emph{mock-CFLOBDDs} (\defref{MockCFLOBDD}) and \emph{CFLOBDDs} (\defref{CFLOBDD}), although we typically drop the qualifier ``mock-'' when there is little danger of confusion.
\figref{walsh1}b illustrates \defref{MockCFLOBDD} using the CFLOBDD that represents Hadamard matrix $H_2$.

\begin{definition}[Mock-CFLOBDD; see \figref{walsh1}b]\label{De:MockCFLOBDD}
    A \emph{mock-CFLOBDD} at level $k$ is a hierarchical structure made up of some number of \emph{groupings}, of which there is one grouping at level $k$, and at least one at each level $0, 1, \ldots, k-1$.
    The grouping at level $k$ is the \emph{head} of the mock-CFLOBDD.
    
    Each grouping $g_i$ at level $0 \le i \le k$ has a unique \emph{entry vertex}, which is disjoint from the set of \emph{exit vertices} of $g_i$.

    If $i = 0$, $g_i$ is either a \emph{fork grouping} or a \emph{don't-care grouping}, as depicted in the upper right and lower right of \figref{walsh1}b, respectively.
    The entry vertex of a level-$0$ grouping corresponds to a decision point:
    left branches are for $F$ (or $0$);
    right branches are for $T$ (or $1$).
    A don't-care grouping has a single exit vertex, and the edges for the left and right branches both connect the entry vertex to the exit vertex.
    A fork grouping has two exit vertices:
    the entry vertex's left and right branches connect the entry vertex to the first and second exit vertices, respectively.
    
    If $i \ge 1$, $g_i$ has a further disjoint set of \emph{middle vertices}.
    We assume that both the middle vertices and the exit vertices are associated with some fixed, known total order (i.e., the sets of middle vertices and exit vertices could each be stored in an array).
    Moreover, $g_i$ has an \emph{A-connection} edge that, from $g_i$'s entry vertex, ``calls'' a level-$i\text{-}1$ grouping $a_{i-1}$, along with a set of matching \emph{return edges};
    each return edge from $a_{i-1}$ connects one of the exit vertices of $a_{i-1}$ to one of the middle vertices of $g_i$.
    In addition, for each middle vertex $m_j$, $g_i$ has a \emph{B-connection} edge that ``calls'' a level-$i\text{-}1$ grouping $b_j$, along with a set of matching \emph{return edges};
    each return edge from $b_j$ connects one of the exit vertices of $b_j$ to one of the exit vertices of $g_i$.

    If $i = k$, $g_k$ has a set of \emph{value edges} that connect each exit vertex of $g_k$ to a \emph{terminal value}.
\end{definition}

\figref{walsh1}b shows where the concepts from \defref{MockCFLOBDD} occur in the CFLOBDD that represents Hadamard matrix $H_2$.

\begin{figure}[tb!]
\begin{tabular}{@{\hspace{0ex}}c@{\hspace{0ex}}c@{\hspace{0ex}}}
  \begin{minipage}{.55\linewidth}
    {
    \LinesNotNumbered
    \setlength{\interspacetitleruled}{0pt}%
    \setlength{\algotitleheightrule}{0pt}%
    \SetInd{0em}{0.75em}
    {\small
    \begin{algorithm*}[H]
    \DontPrintSemicolon
    abstract class Grouping \{\;
    \Indp    
        level: int\;
        numberOfExits: int\;
    \Indm \}\;
    class InternalGrouping extends Grouping \{\;
    \Indp
        AConnection: Grouping\;
        AReturnTuple: ReturnTuple\;
        numberOfBConnections: int\;
        BConnections:\;
        \Indp
          array[1..numberOfBConnections] of Grouping\;
        \Indm
        BReturnTuples:\;
        \Indp
          array[1..numberOfBConnections] of ReturnTuple\;
        \Indm
    \Indm \}\;
    class DontCareGrouping extends Grouping \{\;
    \Indp
        level = 0\;
        numberOfExits = 1\;
    \Indm \}\;
    class ForkGrouping extends Grouping \{\;
    \Indp
        level = 0\;
        numberOfExits = 2\;
    \Indm \}\;
    class CFLOBDD \{\qquad  \tcp{\textrm{Multi-terminal CFLOBDD}}
    \Indp
        grouping: Grouping\;
        valueTuple: ValueTuple\;
    \Indm\}\;
    \end{algorithm*}
    }
    }    
  \end{minipage}
  &
  \begin{tabular}{@{\hspace{0ex}}c@{\hspace{0ex}}}
    \includegraphics[scale=0.45]{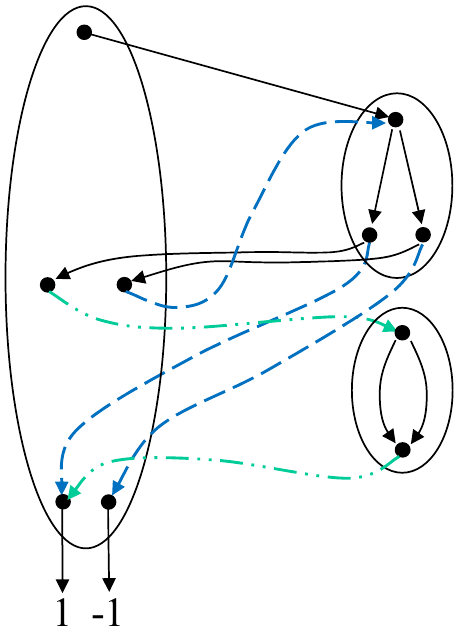}
    \\
    {\small (b)}
    \\
    \includegraphics[scale=0.45]{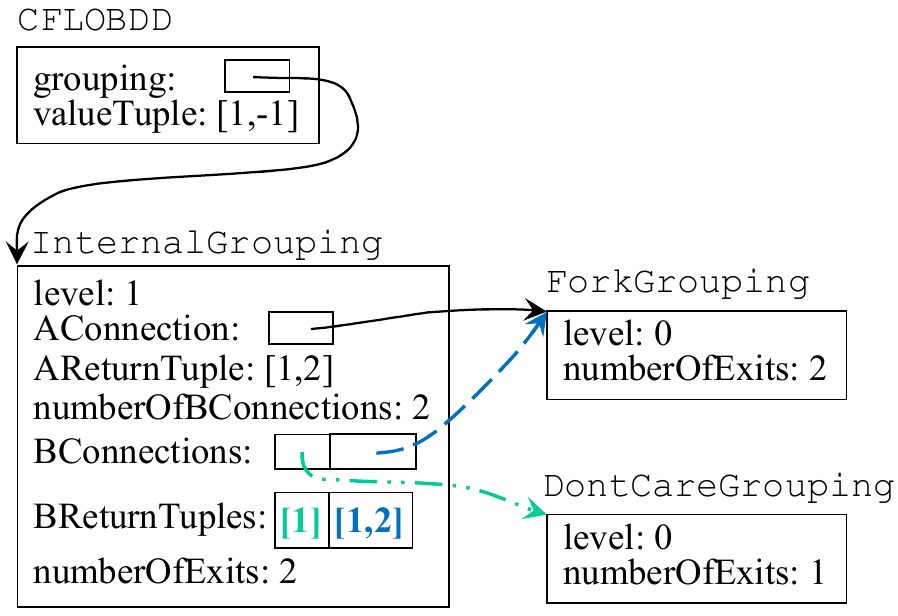}
    \\
    {\small (c)}
  \end{tabular}
  \\
  {\small (a)} & 
\end{tabular}
\caption{\protect \raggedright 
(a) Datatypes for Grouping, InternalGrouping, DontCareGrouping, ForkGrouping, and CFLOBDD.
(b) The CFLOBDD for $H_2$ (repeated from \figref{walsh1}a).
(c) An instance of class \texttt{CFLOBDD} that represents $H_2$. 
}
\label{Fi:ClassDefinitions}
\end{figure}

\subsubsection{An Object-Oriented Pseudo-Code}
In later parts of the paper,
we state algorithms using an object-oriented pseudo-code.
In accordance with the terminology introduced above, the basic classes that are used for representing multi-terminal CFLOBDDs are defined in \figref{ClassDefinitions}a: \texttt{Grouping}, \texttt{InternalGrouping}, \texttt{DontCareGrouping}, \texttt{ForkGrouping}, and \texttt{CFLOBDD}.
More details about the notation used in our pseudo-code can be found in Appendix~\sectref{additional-notation}.

\figref{ClassDefinitions}c shows how the CFLOBDD from \figref{walsh1}a is represented as an instance of class \texttt{CFLOBDD}.
There are no entry, middle, and exit vertices
as such.
Instead, a pointer to a \texttt{Grouping} object serves as the object's entry vertex.
Numbers in the range $[1..\mathtt{numberOfBConnections}]$ serve as middle vertices, and numbers in the range $[1..\mathtt{numberOfExits}]$ serve as exit vertices.
In the level-$1$ \texttt{InternalGrouping} in \figref{ClassDefinitions}c, one can see that a \texttt{ReturnTuple}---which holds a sequence of return-edge targets---is associated with each outgoing \texttt{AConnection} or \texttt{BConnection} edge.
This organization facilitates implementing the matched-path principle:
when a level-$l\text{+}1$ grouping $g_1$ ``calls'' level-$l$ grouping $g_2$, there is an associated \texttt{ReturnTuple} $\textit{rt}_1$ (stored in $g_1$);
a matched path starting at the entry of $g_2$ leads to some exit-vertex index $i$ of $g_2$; and
${\textit{rt}_1}[i]$ holds the target in $g_1$ of the matching return edge.

Similarly, there are no explicit edges in \texttt{DontCareGrouping} and \texttt{ForkGrouping} objects.
Instead, the decision taken at the level-$0$ grouping's entry vertex selects the appropriate exit-vertex index, which is used to index into a \texttt{ReturnTuple} of the ``calling'' level-$1$ \texttt{InternalGrouping}.

\subsubsection{Rationale}
\label{Se:CFLOBDDsRationale}

\defref{MockCFLOBDD}, \figref{walsh1}b, and \figref{ClassDefinitions}a introduce a substantial amount of new terminology.
However, the rationale behind it is really quite simple, and goes back to the $\MPP$.
In particular, each \texttt{InternalGrouping} object $g$ at level $i > 0$ represents a family of matched paths.
A traversal of a matched path from $g$'s entry vertex to an exit vertex of $g$ uses the fields of $g$ (\figref{ClassDefinitions}a) in the following order:
\begin{equation}
  \label{Eq:MatchedPathsInGrouping}
  \begin{array}{rcl}
    \Matched(\textrm{at level $i$)} & = &
      \textrm{AConnection} \quad
      \Matched(\textrm{at level $i\text{-}1$)} \quad
      \textrm{AReturnTuple}[\cdot]
    \\
       &  & \textrm{BConnection} \quad
            \Matched(\textrm{at level $i\text{-}1$)} \quad
            \textrm{BReturnTuples}[\cdot]
  \end{array}
\end{equation}
\Omit{and thus each field of $g$ is essential to the representation.}
Note how \eqref{MatchedPathsInGrouping} mimics the form of the grammar for matched paths from \eqref{MatchedPathGrammar}.

\subsubsection{Inductive Arguments about CFLOBDDs}
\label{Se:InductiveArguments}

To be able to make inductive arguments about CFLOBDDs, it is convenient to introduce one additional bit of terminology:

\begin{definition}[Mock-proto-CFLOBDD]\label{De:MockProtoCFLOBDD}
  A \emph{mock-proto-CFLOBDD} at level $i$ is a grouping at level $i$, together with the lower-level groupings to which it is connected (and the connecting edges).
  In other words, a mock-proto-CFLOBDD has the following recursive structure:
  \begin{itemize}
    \item
      a mock-proto-CFLOBDD at level 0 is either a fork grouping or a don't-care grouping
    \item 
      a mock-proto-CFLOBDD at level $i$ is headed by a grouping at level $i$ whose
      \begin{itemize}
        \item
          A-connection edge and associated return edges ``call'' a level-($i$-1) mock-proto-CFLOBDD, and
        \item
          B-connection edges and their associated return edges ``call'' some number of level-($i$-1) mock-proto-CFLOBDDs.
      \end{itemize}
  \end{itemize}
\end{definition}

The difference between a proto-CFLOBDD and a CFLOBDD is that the exit vertices of a proto-CFLOBDD have not been associated with specific values.
One cannot argue inductively in terms of CFLOBDDs because its constituents are proto-CFLOBDDs, not full-fledged CFLOBDDs.
Thus, to prove that some property holds for a CFLOBDD, there will typically be an inductive argument to establish a property of the proto-CFLOBDD headed by the outermost grouping of the CFLOBDD, with an additional argument about the CFLOBDD's value edges and terminal values.

One example of an inductive argument allows us to establish the number of times $D(i)$ that each matched path in a level-$i$ proto-CFLOBDD reaches a decision vertex---i.e., the entry vertex of a level-$0$ grouping.
In particular, $D(i)$ is described by the following recurrence relation:
\begin{equation}
  \label{Eq:DecisionRecurrence}
  \begin{array}{r@{\hspace{0.75ex}}c@{\hspace{0.75ex}}l  @{\hspace{12ex}}   r@{\hspace{0.75ex}}c@{\hspace{0.75ex}}l}
    D(0) & = & 1
    &
    D(i) & = & D(i-1) + D(i-1),
  \end{array}
\end{equation}
which has the solution $D(i) = 2^i$.


\subsubsection{Soundness and an Operational Semantics}
\label{Se:CFLOBDDOperationalSemantics}

\eqref{DecisionRecurrence} allows us to establish Requirement (\ref{Req:DTSoundness}) from \sectref{Requirements}.
\eqref{DecisionRecurrence} has the solution $D(i) = 2^i$, so each matched path from the entry vertex of a level-$k$ CFLOBDD passes through the entry vertex of a level-$0$ grouping exactly $2^k$ times before reaching a terminal value $v \in V$, for some value domain $V$.
Consequently, each (multi-terminal) CFLOBDD represents a function in $\{ T,F \}^{2^k} \rightarrow V$---i.e., the same set of functions that decision trees represent.

\begin{algorithm}[tb!]
\caption{An operational semantics of CFLOBDDs\label{Fi:OperationalSemantics}}
\SetKwFunction{InterpretCFLOBDD}{InterpretCFLOBDD}
\SetKwFunction{InterpretGrouping}{InterpretGrouping}
\SetKwProg{myalg}{Algorithm}{}{end}
  \myalg{\InterpretCFLOBDD{n, a}}{
  \Input{CFLOBDD n, Assignment a[1..$2^{\textrm{n.grouping.level}}$]}
  \Output{A value in the range of the function represented by n}
  \Begin{
    \Return valueTuple[InterpretGrouping(n.grouping, a)]\;
  }
  }{}
  \SetKwProg{myproc}{SubRoutine}{}{end}
  \myproc{\InterpretGrouping{g, a}}{
  \Input{Grouping g, Assignment a[1..$2^{\textrm{g.level}}$]}
  \Output{An unsigned integer}
  \Begin{
    \lIf{$\textrm{g}$ == ForkGrouping}{\Return 1 + a[1]\tcp*[f]{F$\mapsto$1; T$\mapsto$2}} \label{Li:Interpret:Fork}
    \lIf{$\textrm{g}$ == DontCareGrouping)}{\Return 1\tcp*[f]{F,T$\mapsto$1}} \label{Li:Interpret:DontCare}
    \lIf{$\textrm{g}$ == NoDistinctionProtoCFLOBDD($\textrm{g.level}$)}{\Return 1\tcp*[f]{F,T$\mapsto$1}} \label{Li:Interpret:NoDistinction}
    Assignment $\textrm{a}_\textrm{A}$ = a[1, $2^{\textrm{g.level}-1}$]\;
    Assignment $\textrm{a}_\textrm{B}$ = a[$2^{\textrm{g.level}-1}$ + 1, $2^{\textrm{g.level}}$]\;
    unsigned int i = InterpretGrouping(g.AConnection, $\textrm{a}_\textrm{A}$)\label{Li:Interpret:AConnection}\;
    unsigned int k = InterpretGrouping(g.BConnections[i], $\textrm{a}_\textrm{B}$)\label{Li:Interpret:BConnections}\;  
    \Return g.BReturnTuples[i](k)\;
  }
  }
\end{algorithm}

We can also use the $\CIP$ to obtain an operational semantics for (mock-)CFLOBDDs, given as \algref{OperationalSemantics}.
This algorithm is a divide-order-and-conquer algorithm that specifies how to interpret a given CFLOBDD \texttt{n} with respect to a given Assignment \texttt{a} to the Boolean variables.
(We assume that an Assignment is given as an array of Booleans, whose entries---starting at index-position 1---are the values of the successive variables.)

Subroutine \texttt{InterpretGrouping} performs a recursive traversal over n, following AConnections, BConnections, and return edges.
When a level-$0$ grouping is reached, the value of the current Boolean variable is consulted (\lineref{Interpret:Fork}, in the case of a ForkGrouping), or ignored (\lineref{Interpret:DontCare}, in the case of a DontCareGrouping).
(\Lineref{Interpret:NoDistinction} can be ignored for now;
it is an optimization that is discussed in \sectref{NoDistinctionProtoCFLOBDDs}.)
In \linerefs{Interpret:AConnection}{Interpret:BConnections}, Assignment \texttt{a} is split in half:
the Boolean values in the first half are interpreted during the traversal of \texttt{g}'s AConnection (\lineref{Interpret:AConnection});
the values in the second half are interpreted during the traversal of one of \texttt{g}'s BConnections (\lineref{Interpret:BConnections}),
selected according to the value \texttt{i} obtained in \lineref{Interpret:AConnection} from the call on \texttt{InterpretGrouping()} with \texttt{g}'s AConnection.

\subsubsection{Multiple Middle Vertices and Exit Vertices}
In a Boolean-valued CFLOBDD, the outermost grouping
has at most two exit vertices, and these are mapped to $\{F,T\}$.
In a multi-terminal CFLOBDD, there can be an arbitrary number of exit vertices, which are mapped to values drawn from some finite set of values $V$.
\figref{walsh1}a is a multi-terminal CFLOBDD;
the level-$1$ grouping has two exit vertices that are mapped to $1$ and $-1$.

\begin{figure}[tb!]
    \centering
    \includegraphics[width=0.565\linewidth]{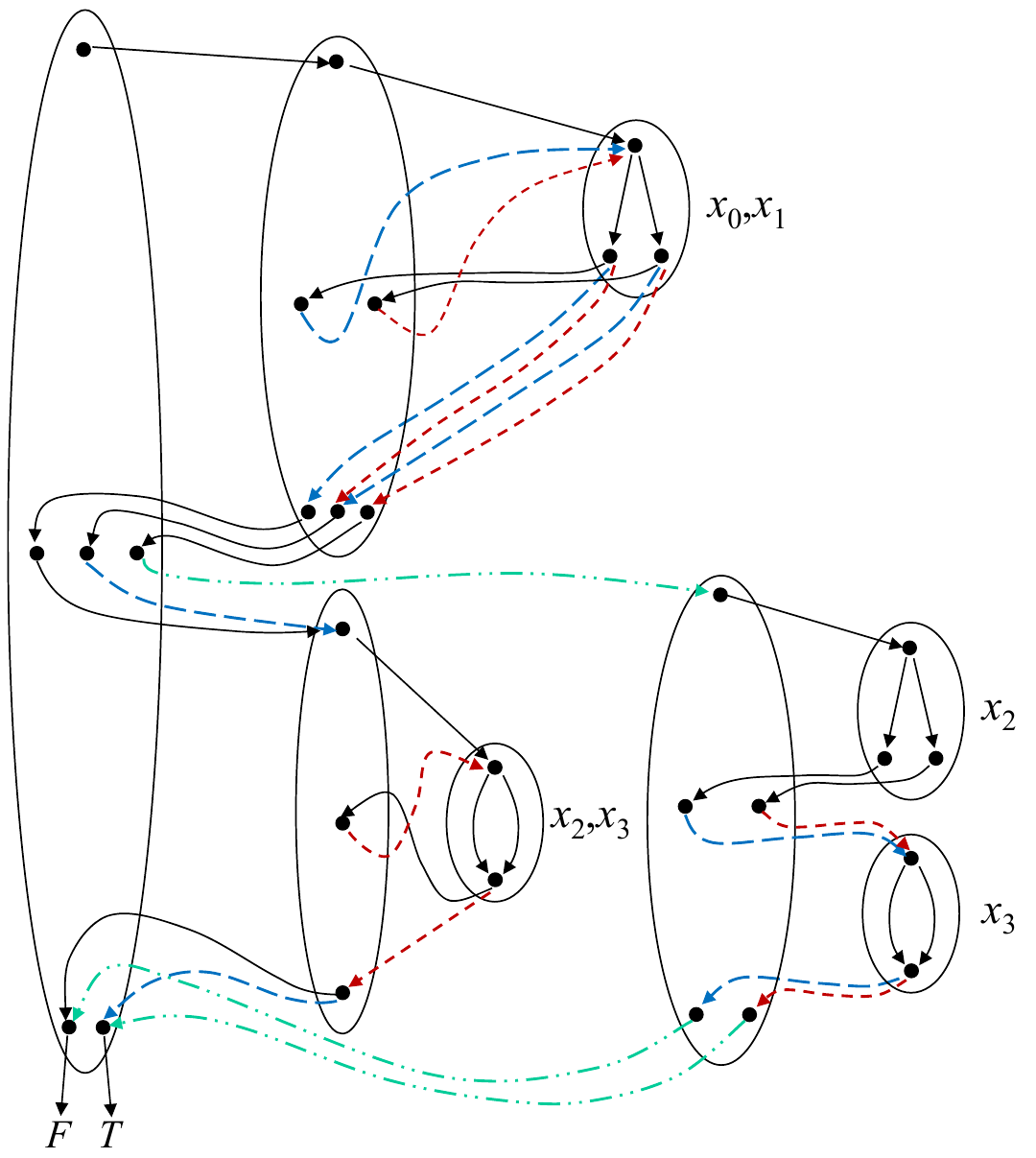}
    \caption{\protect \raggedright 
      CFLOBDD for the Boolean function $\lambda {x_0}{x_1}{x_2}{x_3} . (x_0 \xor x_1) \lor (x_0 \land x_1 \land x_2)$.
      (For clarity, some of the level-0 groupings have been duplicated.)
    }
    \label{Fi:MultipleMiddleVertices}
\end{figure}

Groupings that have more than two exit vertices naturally arise in the interior groupings of CFLOBDDs---even in Boolean-valued CFLOBDDs.
For instance, a level-$i\textrm{--}1$ grouping used as an $A$-connection can more than two exit vertices, in which case the ``calling'' level-$i$ grouping would have more than two middle vertices.
Such multi-terminal groupings can arise in both $A$-connections and $B$-connections.
\figref{MultipleMiddleVertices} shows a Boolean-valued CFLOBDD that contains a level-$1$ grouping that has three exit vertices.
The grouping is the A-connection of the outermost grouping (at level $2$), which thus has three middle vertices.

\begin{figure}
  \centering
  \begin{subfigure}[t]{0.45\linewidth}
    \centering
    \includegraphics[scale=0.5]{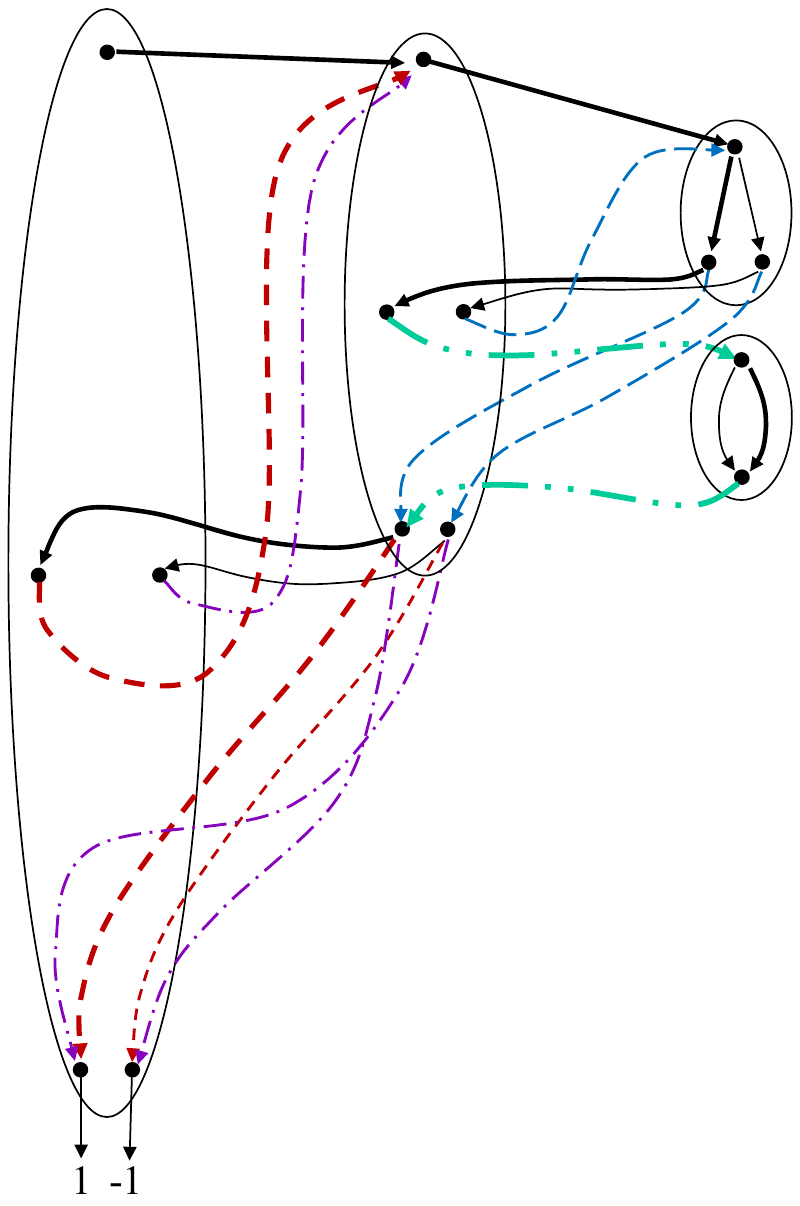}
    \caption{\protect \raggedright 
    The CFLOBDD representation of $H_4$ with the interleaved-variable ordering $\langle x_0, y_0, x_1,y_1 \rangle$.
    The matched path for $[x_0 \mapsto F, y_0 \mapsto T, x_1 \mapsto F, y_1 \mapsto T]$, which corresponds to $H_4[0,3]$, is shown in bold.}
    \label{Fi:walsh4_0_3_path}
  \end{subfigure}
  \quad
  \begin{subfigure}[t]{0.45\linewidth}
    \centering
    \includegraphics[scale=0.5]{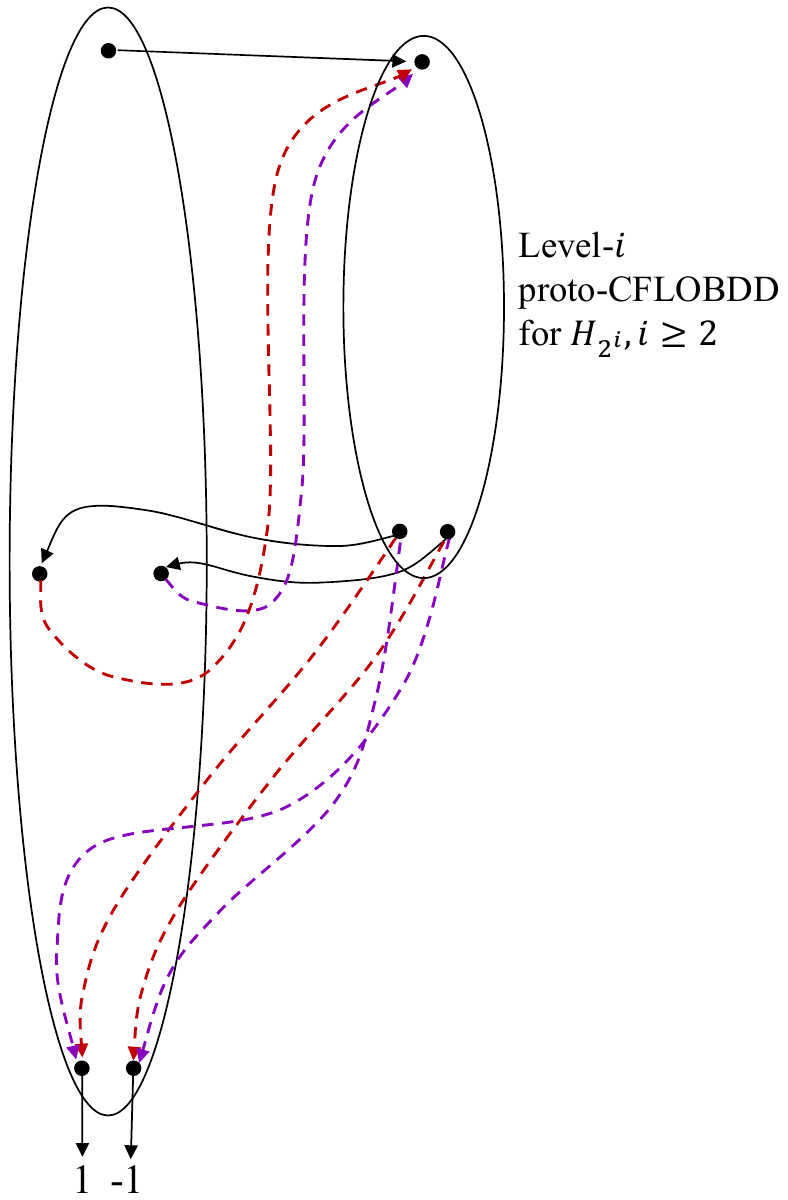}
    \caption{\protect \raggedright 
    Diagram supporting the inductive argument that, with the interleaved-variable ordering, the members of $\HadamardFamily = \{ H_{2^i} \mid i \geq 1 \}$ can be constructed by successively introducing a new outermost grouping at one greater level.
    At each step, the same pattern of ``calls'' is used for the A- and B-connections, and their return edges.
    }
    \label{Fi:walshKGeneralCase}
  \end{subfigure}
  \caption{\protect \raggedright 
  Construction of successively larger members of $\HadamardFamily = \{ H_{2^i} \mid i \geq 1 \}$.
  At level-$(i\textrm{+}1)$, each matched path makes two sequential invocations of the level-$i$ grouping (for $H_{2^i}$), thereby creating $H_{2^{i+1}} = H_{2^i} \tensor H_{2^i}$. 
  }
  \label{Fi:walshK}
\end{figure}

\subsection{Encoding $H_4$ and Other Members of $\HadamardFamily$ with a CFLOBDD}
\label{Se:EncodingHFour}

\figref{walsh4_0_3_path} shows the CFLOBDD representation of Hadamard matrix $H_4$ with the variable ordering $\langle x_0, y_0, x_1,y_1 \rangle$.
In $H_4$, the level-$1$ proto-CFLOBDD is identical to the level-$1$ proto-CFLOBDD in $H_2$ (cf.\ \figref{walsh1}a and \figref{walsh4_0_3_path}).
Moreover, in $H_4$ the A-connection call and both B-connection calls are to the level-$1$ $H_2$ lookalike.

Consider how \figref{walsh4_0_3_path} encodes $H_4[0,3] = 1$.
The value is obtained by evaluating the assignment $[{x_0} \mapsto F, {y_0} \mapsto T, {x_1} \mapsto F, {y_1} \mapsto T]$, following the matched path highlighted in bold.
The path starts from the level-$2$ grouping's entry vertex.
It goes to the level-$1$ grouping's entry vertex, where $[{x_0} \mapsto F, {y_0} \mapsto T]$ is interpreted as for $H_2$---i.e., the first occurrence of $H_2$ in ``$H_2 \tensor H_2$''---in this case, returning to the leftmost middle vertex of the level-$2$ grouping.
At this point, the path follows the red dashed edge back to the level-$1$ grouping's entry vertex, where $[{x_1} \mapsto F, {y_1} \mapsto T]$ is interpreted as for $H_2$---the second occurrence of $H_2$ in ``$H_2 \tensor H_2$.''
The path then follows the matching red dashed return edge to the leftmost exit vertex of the level-$2$ grouping, and reaches terminal value $1$.

To create $H_4$ using $H_2$, we introduced a level-$2$ grouping that makes one A-connection and two B-connection ``calls'' to the level-$1$ $H_2$ lookalike, and 
thus each matched path makes two sequential invocations of $H_2$.
This pattern produces the same effect as the \emph{stacking of plies} in decision trees and BDDs.
However, rather than \emph{tripling} the size of the data structure (as with BDDs---see \figref{walsh2_bdd_tensor_product}),
the ability of CFLOBDDs to reuse parts of a data structure via a ``call'' means that there is only a \emph{constant-size increase} in going from from $H_2$ to $H_4$: one grouping with five vertices and nine edges (one A-connection, two B-connections, and six return edges).

The continuation of this pattern gives an inductive construction of the CFLOBDDs for the other members of $\HadamardFamily$.
Given the level-$i$ CFLOBDD for $H_{2^i}$, $i \ge 2$, $H_{2^{i+1}} = H_{2^i} \tensor H_{2^i}$ is created by introducing a new outermost grouping at level $i + 1$, again with five vertices and nine edges.
(See \figref{walshKGeneralCase}.)
The same pattern of ``calls'' is used for the A- and B-connections and their return edges:
each matched path makes two sequential invocations of the level-$i$ grouping for $H_{2^i}$.
In other words,

\begin{mdframed}[innerleftmargin = 3pt, innerrightmargin = 3pt, skipbelow=-0.0em]
  $\textbf{\SIP}$.
  {\em A Kronecker product $P \tensor Q$ can be represented economically in a CFLOBDD by a grouping at level $i+1$ whose A-connection ``calls'' the level-$i$ proto-CFLOBDD for $P$ and all of whose B-connection ``calls'' are to the level-$i$ CFLOBDD for $Q$.}
\end{mdframed}

\subsection{Reuse of Groupings and Compression of Boolean Functions}
\label{Se:ReuseOfGroupings}

The reason CFLOBDDs can represent certain Boolean functions in a highly compressed fashion is the reuse of groupings that the matched-path and sequential-invocation principles enable.

\subsubsection{Growth of Number of Paths with Level}
\label{Se:GrowthOfNumberOfPathsWithLevel}

Let $P(i)$ be the number of matched paths in a CFLOBDD at level $i$.
Each level-$0$ grouping has two paths, so $P(0) = 2$.
In a grouping $g$ at level $i \ge 1$, each matched path through the A-connection's level-$(i\textrm{--}1)$ proto-CFLOBDD reaches a middle vertex of $g$, where it is routed through the level-$(i\textrm{--}1)$ proto-CFLOBDD of the vertex's $B$-connection.
Let $A_j(i-1)$ be the number of matched paths through $g$'s A-connection proto-CFLOBDD to the $j^{\textit{th}}$ middle vertex of $g$.
Thus, $P(i)$ satisfies the following recurrence equation:
\begin{equation}
  \label{Eq:NumberOfPaths}
  \begin{array}{r@{\hspace{0.75ex}}c@{\hspace{0.75ex}}l@{\hspace{12.0ex}}r@{\hspace{0.75ex}}c@{\hspace{0.75ex}}l}
    P(0) & = & 2
    &
    P(i) & = & \sum_j A_j(i-1) \cdot P(i-1).
\end{array}
\end{equation}
The
total
number of matched paths through $g$'s A-connection proto-CFLOBDD is $P(i-1)$, so $\sum_j A_j(i-1) = P(i-1)$, and hence \eqref{NumberOfPaths} can be rewritten as
\begin{equation}
  \label{Eq:PathSquaring}
  P(i) = P(i-1) \cdot P(i-1),
\end{equation}
which has the solution $P(i) = 2^{2^i}$.

\begin{mdframed}[innerleftmargin = 3pt, innerrightmargin = 3pt, skipbelow=-0.0em]
  $\textbf{\GIP}$.  {\em  The number of matched paths in a CFLOBDD is \emph{squared} with each increase in level by 1 (\eqref{PathSquaring}).
  Consequently, a CFLOBDD at level $i$ has $2^{2^i}$ matched paths.}
\end{mdframed}

\subsubsection{Best-Case Compression: No-Distinction Proto-CFLOBDDs}
\label{Se:NoDistinctionProtoCFLOBDDs}


\begin{figure}
\begin{center}
  \begin{tabular}{@{\hspace{0ex}}c@{\hspace{1.5ex}}c@{\hspace{0ex}}}
    \begin{tabular}{@{\hspace{0ex}}cc@{\hspace{0ex}}}
      \begin{tabular}{c}
        \begin{tabular}{c}
          \rotatebox{90}{\includegraphics[scale=0.45]{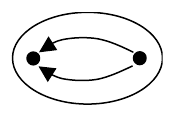}}
          \\
          \begin{minipage}{.23\textwidth}
            (a) Member at level 0
          \end{minipage}
        \end{tabular}
        \\
        \begin{tabular}{c}
          \includegraphics[scale=0.5]{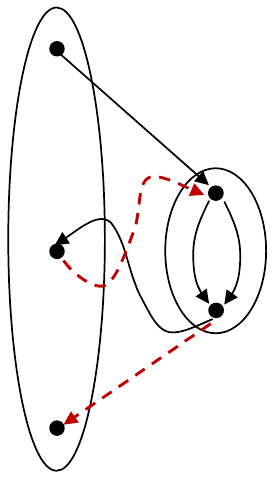}
          \\
          \begin{minipage}{.23\textwidth}
            (b) Member at level 1
          \end{minipage}
        \end{tabular}
      \end{tabular}
      &
      \begin{tabular}{c}
        \includegraphics[scale=0.45]{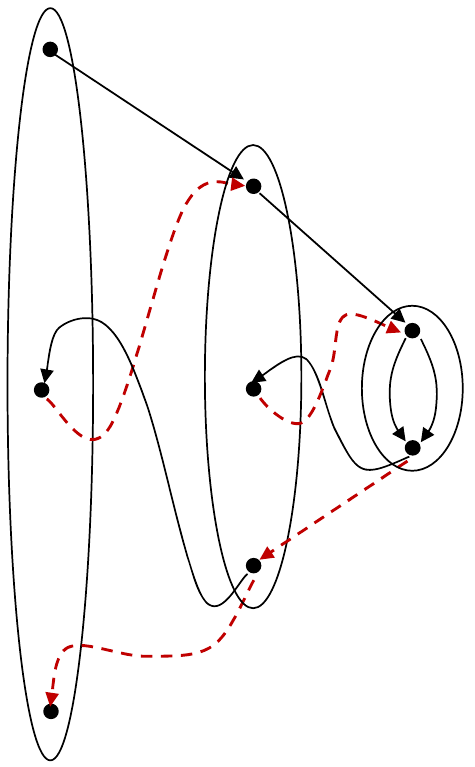}
        \\
        \begin{minipage}{.23\textwidth}
          (c) Member at level 2
        \end{minipage}
      \end{tabular}
    \end{tabular}
    &
    \begin{tabular}{c}
      \includegraphics[height=2.2in]{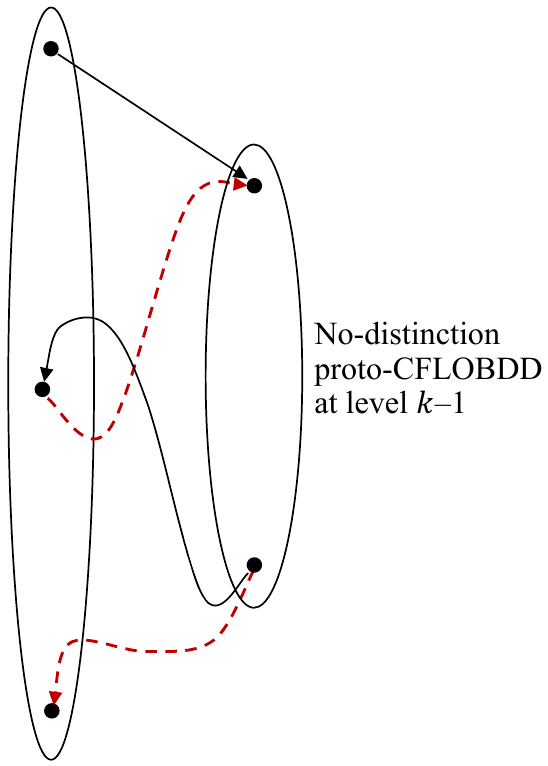}
      \\
      \begin{minipage}{.35\textwidth}
        (d) Illustration of how the no-distinction proto-CFLOBDD at level $k$ is constructed from the one at level $k-1$
      \end{minipage} 
    \end{tabular}
  \end{tabular}
\end{center}
\caption{The family of no-distinction proto-CFLOBDDs.}
\label{Fi:NoDistinctionProtoCFLOBDD}
\end{figure}

\figref{NoDistinctionProtoCFLOBDD}a, \ref{Fi:NoDistinctionProtoCFLOBDD}b, and \ref{Fi:NoDistinctionProtoCFLOBDD}c show the first three members of a family of proto-CFLOBDDs that often arise as sub-structures of CFLOBDDs:
the single-entry/single-exit proto-CFLOBDDs of levels 0, 1, and 2, respectively.
Because every matched path through each of these structures ends up at the unique exit vertex of the highest-level grouping, there is no ``decision'' to be made during each visit to a level-0 grouping.
In essence, as we work our way through such a structure during the interpretation of an assignment, the value assigned to each argument variable makes no difference.

We call this family the \emph{no-distinction proto-CFLOBDD}s.
\figref{NoDistinctionProtoCFLOBDD}d illustrates the structure of a no-distinction proto-CFLOBDD at an arbitrary level $k > 0$, which continues the pattern that one sees in the level-1 and level-2 structures:
the level-$k$ grouping has a single middle vertex, and both its $A$-connection and its one $B$-connection are to the no-distinction proto-CFLOBDD for level $k-1$.
Moreover, because the no-distinction proto-CFLOBDD at level $k$ shares all but one constant-sized grouping with the no-distinction proto-CFLOBDD at level $k-1$, each additional level costs only a constant amount of additional space.
Thus, the no-distinction proto-CFLOBDD at level $k$ is of size $O(k)$, and hence the no-distinction proto-CFLOBDDs exhibit double-exponential compression.

The Boolean-valued CFLOBDD for the constant function $\lambda x_0, x_1, \ldots, x_{2^k-1} . F$ is merely the CFLOBDD in which a value edge connects the (one) exit vertex of the no-distinction proto-CFLOBDD at level $k$ to $F$.
Likewise, in the constant function $\lambda x_0, x_1, \ldots, x_{2^k-1} . T$, the value edge connects the exit vertex of the no-distinction proto-CFLOBDD at level-$k$ to $T$.
Thus, as the number of Boolean variables increases, the best-case growth of CFLOBDDs compares with the growth of decision trees as follows:

\medskip
\noindent
\begin{minipage}{\textwidth}
\centering
  \begin{tabular}{c@{\hspace{2.0ex}}c@{\hspace{1.0ex}}|c@{\hspace{2.0ex}}c@{\hspace{2.0ex}}c@{\hspace{1.0ex}}|@{\hspace{1.0ex}}c@{\hspace{2.0ex}}c@{\hspace{2.0ex}}c@{\hspace{2.0ex}}c}
    Boolean & Number    & \multicolumn{3}{c|@{\hspace{1.0ex}}}{Decision trees} & \multicolumn{4}{c}{CFLOBDDs (best case)} \\
    \cline{3-9}
    vars.   & of paths  & height & \#nodes  & \#edges   & height\footnote{The height of a CFLOBDD is the level of the outermost grouping.} & \#groupings & \#vertices & \#edges \\
    \hline
    1       & 2         & 1      &  3     & 2         & 0      & 1 & 2 & 3 \\
    2       & 4         & 2      &  7     & 6         & 1      & 2 & 5 & 7 \\
    4       & 16        & 4      &  31    & 30        & 2      & 3 & 8 & 11 \\
    8       & 256       & 8      &  511   & 510       & 3      & 4 & 11 & 15 \\
    \vdots  & \vdots    & \vdots & \vdots & \vdots    & \vdots & \vdots & \vdots & \vdots \\
    $2^k$   & $2^{2^k}$ & $2^k$  & $2 \cdot 2^{2^k} - 1$    & $2 \cdot 2^{2^k} - 2$ & $k$  & $k+1$ & $3k + 2$ & $4k + 3$ \\
    \hline
  \end{tabular}
\end{minipage}

\medskip
\noindent
The best-case CFLOBDD size---whether measured in the number of groupings, vertices, or edges---grows linearly with the level of the outermost grouping, which is \emph{logarithmic} in the number of Boolean variables.
In contrast, decision trees grow \emph{exponentially} in the number of Boolean variables.
These observations show that Requirement (\ref{Req:Compression}) from \sectref{Requirements} is met: in the best case, a decision tree of height $2^k$ and size $2^{2^k}$ can be encoded as a level-$k$ CFLOBDD of size $k$.

\paragraph{Remark}
Because the family of no-distinction proto-CFLOBDDs is so compact, in designing CFLOBDDs we did not feel the need to mimic the ``ply-skipping transformation'' of Reduced OBDDs (ROBDDs)~\cite{toc:Bryant86,dac:BRB90}, in which ``don't-care'' nodes are removed from the representation.
In ROBDDs, in addition to reducing the size of the data structure, the chief benefit of ply-skipping is that operations can skip over levels in portions of the data structure in which no distinctions among variables are made.
Essentially the same benefit is obtained by having the algorithms that process CFLOBDDs carry out appropriate special-case processing when no-distinction proto-CFLOBDDs are encountered.
Such processing is carried out, for instance, in \lineref{Interpret:NoDistinction} of \algref{OperationalSemantics}:
in \texttt{InterpretCFLOBDD()}, when Grouping g is the head of a NoDistinctionProtoCFLOBDD, both g and the entire Assignment a can be ignored because g has only a single exit vertex.

Whereas in the best case, the CFLOBDD for a function $f$ can be double-exponentially smaller than the decision tree for $f$, ROBDDs are incapable of such a degree of compression.
\emph{Quasi-reduced BDDs} are the version of BDDs in which don't-care nodes are \emph{not} removed (i.e., plies are not skipped), and thus all paths from the root to a terminal value have length $n$, where $n$ is the number of variables.
The size of a quasi-reduced BDD is at most a factor of $n+1$ larger than the size of the corresponding ROBDD \cite[Thm.\ 3.2.3]{Book:Wegener00}.
Thus, although ROBDDs can give better-than-exponential compression compared to decision trees, what one has is not double-exponential compression:
at best, it is linear compression of exponential compression.
Moreover, in \sectref{efficient-relations} we show that the CFLOBDD for a function $g$ can be exponentially smaller than \emph{any} ROBDD for $g$.

\subsubsection{Asymptotic Best-Case Compression}

Consider a family of functions $F = \{ f_j \mid j \geq 0 \}$, where the $j^{\textit{th}}$ member has $2^j$ Boolean arguments.
The following property is a sufficient condition for the sizes of the CFLOBDDs for members of $F$ to grow linearly in the level $i$, and therefore exhibit double-exponential compression compared to decision trees:
\begin{enumerate}
  \item
    There exists a family of functions $G = \{ g_j \mid j \geq 0 \}$ that grows linearly in the level $i$.\footnote{
      The family of no-distinction proto-CFLOBDDs from \figref{NoDistinctionProtoCFLOBDD} is one such family $G$.
    }
  \item
    There exists a level $m$ such that, for all levels $i \geq m$,
    \begin{enumerate}
      \item
        \label{It:ConstantSizeGrouping}
        the number of vertices in the level-$i$ grouping of $f_i$ is a constant independent of $i$
      \item
        \label{It:ConstantNumberOfCalls}
        the level-$i$ grouping of $f_i$ makes ``procedure calls'' only to (i) the level-$(i\textrm{-}1)$ grouping used in the CFLOBDD for $f_{i\textrm{-}1}$, and (ii) level-$(i\textrm{-}1)$ groupings used in the CFLOBDD for $g_{i\textrm{-}1}$.\footnote{
          Condition~\ref{It:ConstantNumberOfCalls} can be generalized so that $f_i$ can ``call'' the $(i\textrm{-}1)$ groupings used in the CFLOBDDs for some constant number of function families $G_1$, $G_2$, $\ldots$, $G_l$ that each grow linearly in the level $i$.
        }
    \end{enumerate}
\end{enumerate}

\noindent
In such a case, the CFLOBDD for each $f_i$ is double-exponentially smaller than the decision tree for $f_i$---i.e., of size $O(i)$ rather than $O(2^{2^i})$.
As shown in \figref{walshKGeneralCase}, the family of Hadamard matrices $\HadamardFamily$ meets the above conditions.

Moreover, in all cases encountered to date, it is possible to give an explicit algorithm for constructing the $i^{\textit{th}}$ member of $F$, where the algorithm runs in time $O(i)$ and uses at most $O(i)$ space.

No information-theoretic limit is being violated here.
Not all families of functions can be represented with CFLOBDDs in which each level has a constant number of groupings, each of constant size---and thus, not every function over Boolean-valued arguments can be represented in such a compressed fashion.
However, the potential benefit of CFLOBDDs is that, just as with BDDs, there may turn out to be enough regularity in problems that arise in practice that CFLOBDDs stay of manageable size.
Moreover, double-exponential compression (or any kind of super-exponential compression) could allow problems to be completed much faster (due to the smaller-sized structures involved), or allow far larger problems to be addressed than has been possible heretofore.

\section{Canonicalness}
\label{Se:Canonicalness}

In this section, we impose some further structural restrictions on proto-CFLOBDDs and CFLOBDDs that go beyond the ideas illustrated earlier (\sectref{StructuralRestrictions}).
We then discuss how to establish that CFLOBDDs are a canonical representation of Boolean functions (\sectref{CanonicalnessHighLevel} and Appendix \sectref{canonical-proof}).

\subsection{CFLOBDDs Defined, Part II: Additional Structural Invariants}
\label{Se:StructuralRestrictions}

As described in~\sectref{Overview}, the structure of a mock-CFLOBDD consists of different groupings organized into levels, which are connected by edges in a particular fashion.
In this section, we describe additional \emph{structural invariants} that are imposed on CFLOBDDs, which go beyond the basic hierarchical structure that is provided by the
entry vertex, A-Connection, middle vertices, B-Connections, return edges, and exit vertices of a grouping.

Most of the structural invariants concern the organization of
what we call \emph{return tuples} (following the terminology introduced in \figref{ClassDefinitions}).
For a given $A$-connection edge or $B$-connection edge $c$ from grouping $g_i$ to $g_{i-1}$, the return tuple $rt_c$ associated with $c$ consists of the sequence of targets of return edges from $g_{i-1}$ to $g_i$ that correspond to $c$ (listed in the order in which the corresponding exit vertices occur in $g_{i-1}$).
Similarly, the sequence of targets of value edges that emanate from the exit vertices of the highest-level grouping $g$ (listed in the order in which the corresponding exit vertices occur in $g$) is called the CFLOBDD's \emph{value tuple}.

Return tuples represent mapping functions that map exit vertices at one level to middle vertices or exit vertices at the next greater level.
Similarly, value tuples represent mapping functions that map exit vertices of the highest-level grouping to terminal values.
In both cases, the $i^{\textit{th}}$ entry of the tuple indicates the element that the $i^{\textit{th}}$ exit vertex is mapped to.


Because the middle vertices and exit vertices of a grouping are each arranged in some fixed known order, and hence can be stored in an array, it is often convenient to assume that each element of a return tuple is simply an index into such an array.
For example, in \figref{MultipleMiddleVertices},
\begin{itemize}
  \item
    The return tuple associated with the $1^{\textit{st}}$ $B$-connection of the upper level-$1$ grouping is $[1,2]$.
  \item
    The return tuple associated with the $2^{\textit{nd}}$ $B$-connection of the upper level-$1$ grouping is $[2,3]$.
  \item
    The return tuple associated with the $A$-connection of the level-$2$ grouping is $[1,2,3]$.
  \item
    The value tuple associated with the CFLOBDD is the $2$-tuple $[F,T]$.
\end{itemize}

\subsubsection*{Rationale}
The structural invariants are designed to ensure that---for a given order on the Boolean variables---each Boolean function has a unique, canonical representation as a CFLOBDD.
In reading \defref{CFLOBDD} below, it will help to keep in mind that the goal of the invariants is to force there to be a \emph{unique} way to fold a given decision tree into a CFLOBDD that represents the same Boolean function.
The decision-tree folding method is discussed in \sectref{CanonicalnessHighLevel} and Appendix~\sectref{canonical-proof}, but the main characteristic of the folding method is that it works greedily, left to right.
This directional bias shows up in structural invariants~\ref{Inv:1}, \ref{Inv:2a}, and \ref{Inv:2b}.

We can now complete the formal definition of a CFLOBDD.

\begin{definition}[Proto-CFLOBDD and CFLOBDD]\label{De:CFLOBDD}
    A \emph{proto-CFLOBDD} $n$ is a mock-proto-CFLOBDD (\defrefs{MockCFLOBDD}{MockProtoCFLOBDD}) in which every grouping/proto-CFLOBDD in $n$ satisfies the \emph{structural invariants} given below.  
    In particular, let $c$ be an $A$-connection edge or $B$-connection edge from grouping $g_i$ to $g_{i-1}$, with associated return tuple $rt_c$.
    \begin{enumerate}
      \item
        \label{Inv:1}
        If $c$ is an $A$-connection, then $rt_c$ must map the exit vertices of $g_{i-1}$ one-to-one, and in order, onto the middle vertices of $g_i$:
        Given that $g_{i-1}$ has $k$ exit vertices, there must also be $k$ middle vertices in $g_i$, and $rt_c$ must be the $k$-tuple $[1,2,\ldots,k]$.
        (That is, when $rt_c$ is considered as a map on indices of exit vertices of $g_{i-1}$, $rt_c$ is the identity map.)
      \item
        \label{Inv:2}
        If $c$ is the $B$-connection edge whose source is middle vertex $j+1$ of $g_i$ and whose target is $g_{i-1}$, then $rt_c$ must meet two conditions:
        \begin{enumerate}
          \item
            \label{Inv:2a}
            It must map the exit vertices of $g_{i-1}$ one-to-one (but not necessarily onto) the exit vertices of $g_i$.
            (That is, there are no repetitions in $rt_c$.)
          \item
            \label{Inv:2b}
            It must ``compactly extend'' the set of exit vertices in $g_i$ defined by the return tuples for the previous $j$ $B$-connections:
            Let $rt_{c_1}$, $rt_{c_2}$, $\ldots$, $rt_{c_j}$ be the return tuples for the first $j$ $B$-connection edges out of $g_i$.
            Let $S$ be the set of indices of exit vertices of $g_i$ that occur in return tuples $rt_{c_1}$, $rt_{c_2}$, $\ldots$, $rt_{c_j}$, and let $n$ be the largest value in $S$.
            (That is, $n$ is the index of the rightmost exit vertex of $g_i$ that is a target of any of the return tuples $rt_{c_1}$, $rt_{c_2}$, $\ldots$, $rt_{c_j}$.)
            If $S$ is empty, then let $n$ be $0$.
    
            \hspace*{1.5ex}
            Now consider $rt_c$ ($= rt_{c_{j+1}}$).
            Let $R$ be the (not necessarily contiguous) sub-sequence of $rt_c$ whose values are strictly greater than $n$.
            Let $m$ be the size of $R$.
            Then $R$ must be exactly the sequence $[n+1, n+2, \ldots, n+m]$.
        \end{enumerate}
      \item
        \label{Inv:3}
        While a proto-CFLOBDD may be used as a substructure more than once (i.e., a proto-CFLOBDD may be {\em pointed to\/} multiple times), a proto-CFLOBDD never contains two separate {\em instances\/} of equal proto-CFLOBDDs.\footnote{
          \label{Footnote:CFLOBDDEquality}
          Equality on proto-CFLOBDDs is defined inductively on their hierarchical structure in the obvious manner.
          Two CFLOBDDs are equal when (i) their proto-CFLOBDDs are equal, and (ii) their value tuples are equal.
          \sectref{HashConsing} discusses how hash-consing \cite{Tokyo-TR-74-03:Goto74} can be used to enforce the invariant that only a single representative CFLOBDD/proto-CFLOBDD exists for each equivalence class of CFLOBDD/proto-CFLOBDD values.
          However, when we wish to consider the possibility that \emph{multiple} data-structure instances exist that are equal---as we do shortly in \sectref{CanonicalnessHighLevel}---we say that such structures are ``isomorphic'' or ``equal (up to isomorphism).''  

          \hspace*{1.5ex}
          To reduce clutter,
          our diagrams often show multiple instances of the two kinds of level-0 groupings; in fact, a CFLOBDD can contain at most one copy of each.
        }
      \item
        \label{Inv:4}
        For every pair of $B$-connections $c$ and $c'$ of grouping $g_i$, with associated return tuples $rt_c$ and $rt_{c'}$, if $c$ and $c'$ lead to level $i-1$ proto-CFLOBDDs, say $p_{i-1}$ and $p'_{i-1}$, such that $p_{i-1} = p'_{i-1}$, then the associated return tuples must be different (i.e., $rt_c \neq rt_{c'}$).
    \end{enumerate}

    A \emph{CFLOBDD} at level $k$ is a mock-CFLOBDD at level $k$ for which
    \begin{enumerate}[resume]
      \item
        \label{Inv:5}
        The grouping at level $k$ heads a proto-CFLOBDD.
      \item
        \label{Inv:6}
        The value tuple associated with the grouping at level $k$ maps each exit vertex to a \emph{distinct} value.
    \end{enumerate}
\end{definition}

\begin{figure*}
\centering
\begin{tabular}{c@{\hspace{.33in}}c}
  \includegraphics[height=1.62in]{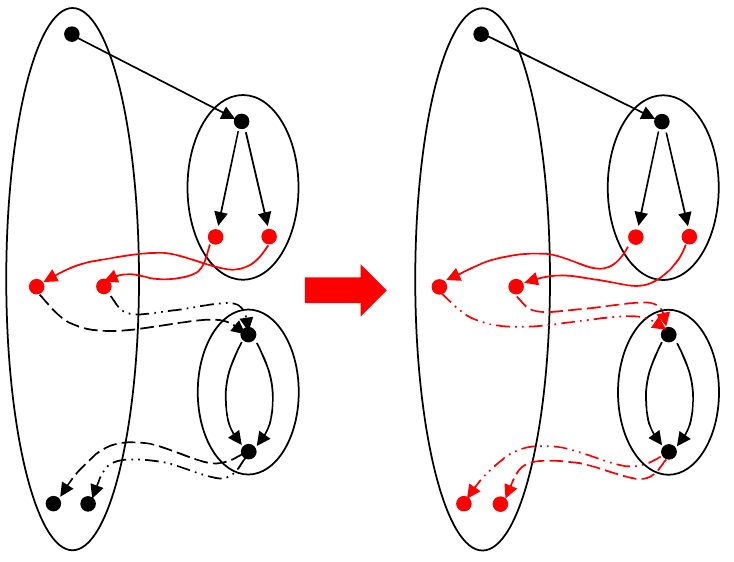}
  &
  \includegraphics[height=1.62in]{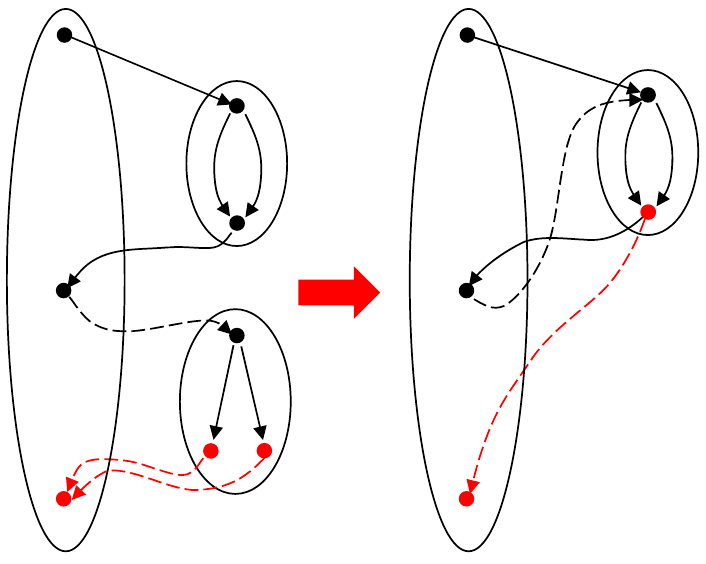}
  \\
  (a) Structural invariant~\ref{Inv:1}
  &
  (b) Structural invariant~\ref{Inv:2a}
  \\
  \includegraphics[height=1.62in]{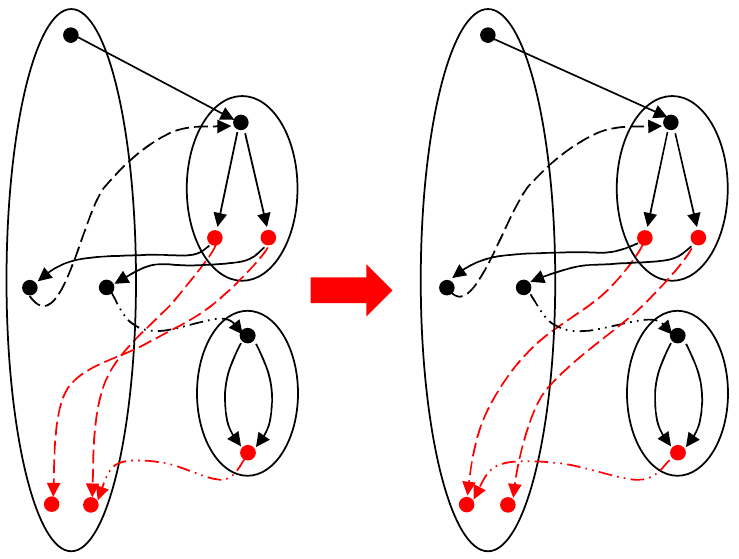}
  &
  \includegraphics[height=1.62in]{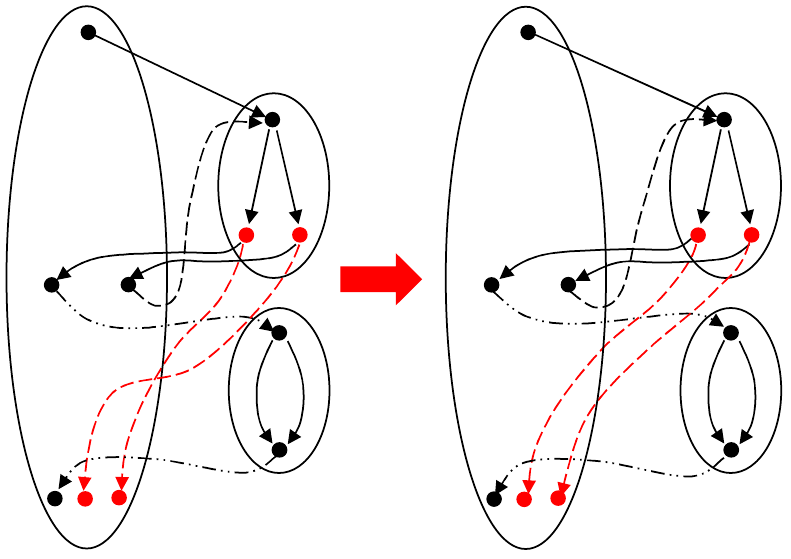}
  \\
  (c) Structural invariant~\ref{Inv:2b}
  &
  (d) Another case of Structural invariant~\ref{Inv:2b}
  \\
  \includegraphics[height=1.62in]{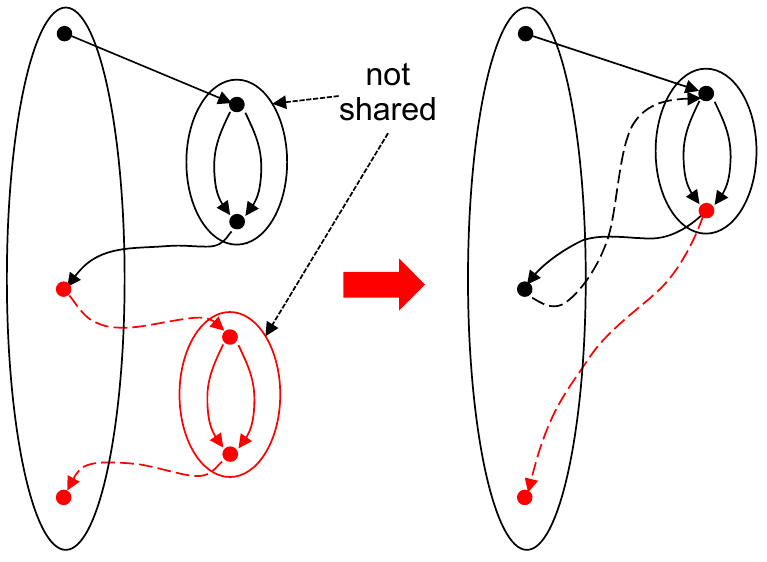}
  &
  \includegraphics[height=1.62in]{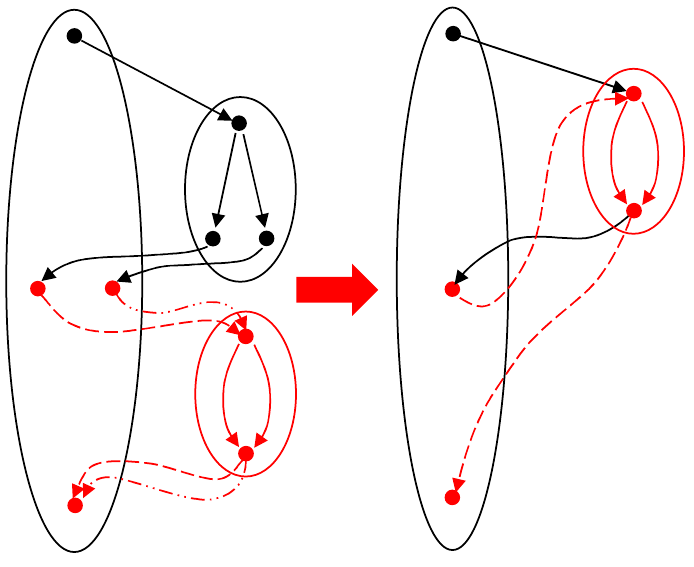}
  \\
  (e) Structural invariant~\ref{Inv:3}
  &
  (f) Structural invariant~\ref{Inv:4}
  \\
  \includegraphics[height=1.75in]{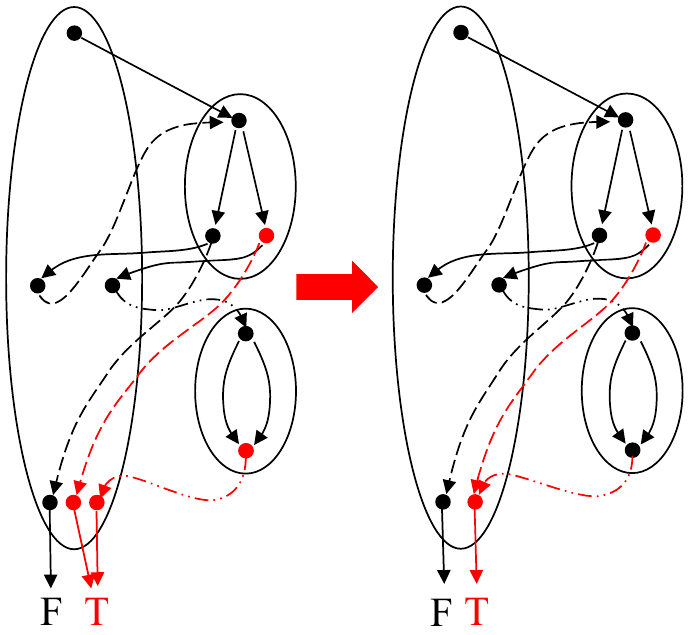}
  \\
  (g) Structural invariant~\ref{Inv:6}
\end{tabular}
\caption{\protect \raggedright 
To the left of each arrow, a mock-proto-CFLOBDD that violates the indicated structural invariant; to the right, a corrected proto-CFLOBDD.
Invariant violations and their rectifications are shown in red. 
}
\label{Fi:StructuralInvariantsIllustrated}
\end{figure*}

\figref{StructuralInvariantsIllustrated} illustrates structural invariants~\ref{Inv:1}, \ref{Inv:2a}, \ref{Inv:2b}, \ref{Inv:3}, \ref{Inv:4}, and~\ref{Inv:6}.
In each case, a mock-proto-CFLOBDD that violates one of the structural invariants is shown on the left, and an equivalent proto-CFLOBDD that satisfies the structural invariants is shown on the right.

The CFLOBDD from \figref{MultipleMiddleVertices} also illustrates the structural invariants. 
\begin{itemize}
  \item
    The level-1 grouping pointed to by the $A$-connection of the level-2 grouping has three exit vertices.
    These are the targets of two return tuples from the uppermost level-0 fork grouping.
    Note that the blue dashed lines in this proto-CFLOBDD correspond to $B$-connection 1 and $rt_1$, whereas the red short-dashed lines correspond to $B$-connection 2 and $rt_2$.

    \hspace*{1.5ex}
    In the case of $rt_1$, the set $S$ mentioned in
    structural invariant~\ref{Inv:2b} is empty;
    therefore, $n = 0$ and $rt_1$ is constrained by
    structural invariant~\ref{Inv:2b} to be $[1,2]$.

    \hspace*{1.5ex}
    In the case of $rt_2$, the set $S$ is $\{1,2\}$, and
    therefore $n = 2$.  The first entry of $rt_2$, namely 2,
    falls within the range $[1..2]$; the second entry of $rt_2$
    lies outside that range and is thus constrained to be $3$.
    Consequently, $rt_2 = [2,3]$.

    \hspace*{1.5ex}
    Also in \figref{MultipleMiddleVertices}, because the level-1
    grouping pointed to by the $A$-connection of the level-2 grouping
    has three exit vertices, these are constrained by structural
    invariant~\ref{Inv:1} to map in order over to the three middle
    vertices of the level-2 grouping; i.e., the corresponding return tuple is $[1,2,3]$.
  \item
    The $B$-connections for the
    first and second middle vertices of the level-2 grouping
    are to the same level-1 grouping; however, the two return tuples
    are different, and thus are consistent with structural
    invariant~\ref{Inv:4}.
\end{itemize}

One artifact of the greedy, left-to-right decision-tree folding method used in \sectref{CanonicalnessHighLevel} and Appendix~\sectref{canonical-proof} is that matched paths through proto-CFLOBDDs (and hence through CFLOBDDs) have a left-to-right bias in the ordering of paths with respect to Boolean-variable-to-Boolean-value
assignments.
This bias is captured in the following proposition.

\begin{Pro}[Lexicographic-Order Proposition]\label{Prop:LexicographicOrder}
Let $ex_C$ be the sequence of exit vertices of proto-CFLOBDD $C$.
Let $ex_L$ be the sequence of exit vertices reached by traversing $C$ on each possible
Boolean-variable-to-Boolean-value assignment, generated in lexicographic
order of assignments.
Let $s$ be the subsequence of $ex_L$ that retains just the leftmost occurrences
of members of $ex_L$ (arranged in order as they first appear in $ex_L$).
Then $ex_C = s$.
\end{Pro}
The proof of \propref{LexicographicOrder} is provided in Appendix~\sectref{prop-LexicographicOrder}.

Earlier in this section, the ``Rationale'' paragraph motivated the structural invariants as enforcing an implicit ``greedy left-to-right folding'' of the corresponding decision tree to create the CFLOBDD, and Figure 8 illustrates the structural invariants from a syntactic/operational viewpoint.
In contrast, \propref{LexicographicOrder} elucidates a semantic consequence of the structural invariants.\footnote{
    \sectref{CanonicalnessHighLevel} gives a high-level overview of the proof that CFLOBDDs are a canonical representation of Boolean functions.
    In the proof of canonicity in \sectref{canonical-proof}, \propref{LexicographicOrder} is used in the proof of \propref{UnfoldFoldReversability}, which establishes property (\ref{It:NoDecisionTreeRepresentedByMultipleCFLOBDDs}) from \sectref{CanonicalnessHighLevel}.
}

\begin{example}
  \propref{LexicographicOrder} can be illustrated using   \figref{MultipleMiddleVertices}.
  If we use numbers to identify exit vertices, $ex_C$ for any grouping $g$ is the sequence $[1..g.\textrm{numberOfExits}]$.
  In the upper level-$1$ grouping in \figref{MultipleMiddleVertices}, $ex_L$ is $[1,2,2,3]$, so $s$ is $[1,2,3]$.
  In the level-$1$ grouping at the lower right, $ex_L$ is $[1,1,2,2]$, so $s$ is $[1,2]$.
  In the level-$2$ grouping, $ex_L$ is $[1,1,1,1,2,2,2,2,2,2,2,2,1,1,2,2]$, so $s$ is $[1,2]$.
\end{example}


\subsection{Canonicity of CFLOBDDs}
\label{Se:CanonicalnessHighLevel}

CFLOBDDs are a canonical representation of functions over Boolean arguments, i.e., each decision tree with $2^{2^k}$ leaves is represented by exactly one isomorphism class of level-$k$ CFLOBDDs.
(The notion of isomorphism of CFLOBDDs was introduced in \footnoteref{CFLOBDDEquality}.)

\begin{theorem}[Canonicity]\label{The:Canonicity}
If $C_1$ and $C_2$ are level-$k$ CFLOBDDs for the same Boolean function over $2^k$ Boolean variables, and $C_1$ and $C_2$ use the same variable ordering, then $C_1$ and $C_2$ are isomorphic.
\end{theorem}

To prove this theorem, we make use of \obsref{DecisionTreesRepresentBooleanFunctions}, and argue not in terms of Boolean functions but in terms of \emph{representations} of Boolean functions---specifically, we relate two kinds of Boolean-function representations
\begin{itemize}
  \item
    the decision tree $T_B$ for a Boolean function $B$, using some fixed, but otherwise unspecified, variable ordering Ord, and
  \item
    the CFLOBDD for $B$, again using variable ordering Ord.
\end{itemize}
By \obsref{DecisionTreesRepresentBooleanFunctions}, we use $T_B$ as a stand-in for $B$, thereby avoiding having to talk about $B$ itself.
In particular, we must establish that three properties hold:
\begin{enumerate}
  \item
    \label{It:CFLOBDDRepresentsADecisionTree}
    Every level-$k$ CFLOBDD represents a decision tree with
    $2^{2^k}$ leaves.
  \item
    \label{It:DecisionTreeRepresentedByDecisionTree}
    Every decision tree with $2^{2^k}$ leaves is represented
    by some level-$k$ CFLOBDD.
  \item
    \label{It:NoDecisionTreeRepresentedByMultipleCFLOBDDs}
    No decision tree with $2^{2^k}$ leaves is represented by
    more than one level-$k$ CFLOBDD
    (up to isomorphism).
\end{enumerate}
The proof that CFLOBDDs are a canonical representation of Boolean functions is in Appendix~\sectref{canonical-proof}.

We already showed that Obligation~\ref{Obligation:1} is satisfied in \sectref{CFLOBDDOperationalSemantics}.

Obligation~\ref{Obligation:2} is established by showing that there is a recursive procedure for constructing a level-$k$ CFLOBDD from an arbitrary decision tree with $2^{2^k}$ leaves (i.e., of height $2^k$)---see Construction \ref{Constr:DecisionTreeToCFLOBDD} in Appendix~\sectref{canonical-proof}.
In essence, the construction shows how such a decision tree can be folded together to form a CFLOBDD
that represents the same Boolean function.
The construction ensures that the structural invariants are obeyed.

Obligation~\ref{Obligation:3} is established by showing that (i) unfolding a CFLOBDD $C$ into a decision tree $T$ and then (ii) folding $T$ back to a CFLOBDD yields a CFLOBDD that is isomorphic to $C$.
In particular, the folding-back step applies the same algorithm we use to establish Obligation~\ref{Obligation:2}, namely, Construction \ref{Constr:DecisionTreeToCFLOBDD} from Appendix~\sectref{canonical-proof}.
Construction \ref{Constr:DecisionTreeToCFLOBDD} is a \emph{deterministic algorithm}, and thus the proof establishes that $T$ can only be mapped to a CFLOBDD $C'$ that is isomorphic to $C$.
(See \propref{UnfoldFoldReversability}.)

Note that Obligation~\ref{Obligation:1} and \ref{Obligation:2} are exactly Requirements~(\ref{Req:DTSoundness}) and (\ref{Req:DTCompleteness}) from \sectref{Requirements}, respectively.
Moreover, Obligations~\ref{Obligation:1}--\ref{Obligation:3} together show that
Requirement~(\ref{Req:Canonicity}) from \sectref{Requirements} is met.




\section{Pragmatics}
\label{Se:Pragmatics}

The structure of the groupings in a CFLOBDD is acyclic:
a level-$k$ grouping has calls exclusively to groupings at level $k\text{-}1$;
conversely, a given grouping at level $k\text{-}1$ can be called from multiple groupings, but only ones at level $k$.
This property allows CFLOBDDs to be implemented in a functional style without side-effects.
Moreover, because groupings are acyclic, storage can be managed via smart-pointer-based reference counting.

The remainder of this section discusses pragmatics---namely, how some of the standard techniques for working with a functional data structure apply to CFLOBDDs.
All three of the techniques discussed contribute to an implementation being able to satisfy Requirement (\ref{Req:ComputationalEfficiency}) that operations on a CFLOBDD run in time 
polynomial in the sizes of (i) the input
CFLOBDDs, or (ii) the input CFLOBDDs and the output CFLOBDD.

\subsection{Hash-Consing of Groupings and CFLOBDDs to Create Unique Representatives}
\label{Se:HashConsing}

Hash-consing \cite{Tokyo-TR-74-03:Goto74} enforces the invariant that only a single representative exists for each value constructed from some datatype.
Hash-consing should not be confused with canonicity (\sectref{CanonicalnessHighLevel} and Appendix~\sectref{canonical-proof}).
Canonicity is a semantic property: if two CFLOBDDs $C_1$ and $C_2$ represent the same function, then $C_1$ and $C_2$ are isomorphic.
Hash-consing concerns concrete memory representations: for a given data-structure construction pattern, only a single representative exists in memory, no matter how many times that value arises in a computation.

However, because canonicity holds for CFLOBDDs, an implementation that uses hash-consing\footnote{
  It can also be useful to use hash-consing for the objects of classes {\tt ReturnTuple}, {\tt PairTuple}, and {\tt ValueTuple}.
}
satisfies an even stronger form of equivalence.
In particular, \theoref{Canonicity} can be restated to read ``$\ldots$ then $C_1$ and $C_2$ are identical.''

Because the operations that construct {\tt Grouping\/}s and {\tt CFLOBDD\/}s involve a certain amount of processing before the object being constructed is finally complete, we will assume that two operations, named {\tt RepresentativeGrouping\/} and {\tt RepresentativeCFLOBDD}, are available for explicitly maintaining the tables of representative {\tt Grouping\/}s and {\tt CFLOBDD}s, respectively.
For instance, a call {\tt RepresentativeGrouping(g)\/} checks to see whether a representative for {\tt g\/} is already in the table of representative {\tt Grouping\/}s;
if there is such a representative, say {\tt h}, then {\tt g\/} is discarded and {\tt h} is returned as the result;
if there is no such representative, then {\tt g\/} is installed in the table and returned as the result.
The operations {\tt RepresentativeForkGrouping\/} and {\tt RepresentativeDontCareGrouping\/} return the unique representatives of types {\tt ForkGrouping\/} and {\tt DontCareGrouping}, respectively.

Operations discussed in \sectref{cflobdd-algos} that create {\tt InternalGrouping}s, such as {\tt PairProduct} (\algref{PairProduct}) and {\tt Reduce} (\algref{Reduce}), have the following form:
\begin{center}
{\tt
\begin{minipage}{\columnwidth}
\begin{tabbing}
Op\=eration() \{ \+ \\
    $\ldots$ \\
    Internal\=Grouping g = new InternalGrouping(k); \\
    $\ldots$ \\
    // Operations to fill in the members of g, including g.AConnection and the \\
    // elements of array g.BConnections, with level k-1 Groupings \\
    $\ldots$ \\
    return RepresentativeGrouping(g); \- \\
\}
\end{tabbing}
\end{minipage}
}
\end{center}
The operation {\tt NoDistinctionProtoCFLOBDD} (\algref{NoDistinctionProtoCFLOBDDAlgorithm}), which constructs the members of the family of no-distinction proto-CFLOBDDs depicted in \figref{NoDistinctionProtoCFLOBDD}, also has this form.

\texttt{RepresentativeCFLOBDD} is similar to \texttt{RepresentativeGrouping}, but in addition to a \texttt{Grouping} argument, it also has a value-tuple argument.
The operation {\tt ConstantCFLOBDD} (\algref{ConstantCFLOBDD}) illustrates the use of {\tt RepresentativeCFLOBDD}:
{\tt ConstantCFLOBDD(k,v)} returns a hash-consed CFLOBDD that represents a constant function of the form $\lambda x_0, x_1, \ldots, x_{2^i-1} . v$.

In our implementation, we maintain the invariant that the Groupings that appear in the hash-consing tables are the heads of fully-fledged proto-CFLOBDDs, not mock-proto-CFLOBDDs---i.e., structural invariants (\ref{Inv:1})--(\ref{Inv:4}) of \defref{CFLOBDD} hold.
When a proto-CFLOBDD $p$ is associated with terminal values to create a CFLOBDD $c$, it is necessary to ensure that structural invariant (\ref{Inv:6}) holds.
In particular, if there are any duplicate terminal values, a ``reduction'' step is applied (see \algref{Reduce} of \sectref{BinaryOperationsOnCFLOBDDs}), which may cause smaller versions of some of the groupings in $p$ to be constructed.
The original groupings would be collected if their reference counts go to 0.
However, there is never any issue of the hash-cons tables being polluted by mock-proto-CFLOBDDs that violate the proto-CFLOBDD structural invariants. 

\subsection{Equality Testing for CFLOBDDs and proto-CFLOBDDs}
\label{Se:CFLOBDDEquality}

As discussed in \sectref{HashConsing}, the combined effect of hash-consing and canonicity is that an implementation can maintain the invariant that, at any given time, there is a unique concrete memory representation of a given Boolean function.
Consequently, it is possible to test in unit time---by comparing two pointers---whether two variables of type \texttt{CFLOBDD}
represent
the same Boolean function.
This property is important in user-level applications in which various kinds of data are implemented using class \texttt{CFLOBDD}.
For example, in applications structured as fixed-point-finding loops, this property provides a unit-cost test of whether the fixed-point has been reached.

Again, because of the use of hash-consing, it is also possible to test whether two variables of type {\tt Grouping\/} are equal via a single pointer comparision.
Because each grouping is always the highest-level grouping of some proto-CFLOBDD, the equality test on {\tt Grouping}s is really a test of whether two proto-CFLOBDDs are equal.
The property of being able to test two proto-CFLOBDDs for equality quickly is important because proto-CFLOBDD equality tests are used during the various operations on CFLOBDDs to maintain the structural invariants from \defref{CFLOBDD}.

Finally, the ability to test two proto-CFLOBDDs for equality quickly also allows some functions---typically near the beginning of the function---to identify important special-case values of parameters, which can lead to faster performance.
For instance, in \algref{OperationalSemantics}, \lineref{Interpret:NoDistinction}, we saw how testing whether the argument g is a NoDistinctionProtoCFLOBDD allows further recursive calls to \texttt{InterpretGrouping()} to be short-circuited.

\subsection{Function Caching}
\label{Se:FunctionCaching}

A function cache (or \emph{memo function} \cite{Edinburgh-MIP-R-29:Michie67}) for a function $F$ is an associative-lookup table---typically a hash table---of pairs of the form $[x,F(x)]$, keyed on the value of $x$.
The table is consulted each time $F$ is applied to some argument, and updated after a return value is computed for a never-before-seen argument.
The technique saves the cost of re-performing the computation of $F$ for an argument on which $F$ has previously been called, at the expense of performing a lookup on $F$'s argument at the beginning of each call.
Our implementation of CFLOBDDs uses function caching for a number of the operations described in the remainder of the paper, such as {\tt PairProduct} (\algref{PairProduct}) and {\tt Reduce} (\algref{Reduce}).
To reduce clutter in the pseudo-code that we give, we elide the lines for querying and updating the cache.
The full statement of such a function would have the following form:
\begin{center}
{\tt
\begin{minipage}{\columnwidth}
\begin{tabbing}
F(x)\=\ \{ \+ \\
    $\textbf{if}~\texttt{cache}_F(x) \neq \texttt{NULL}~\textbf{return}~\texttt{cache}_F(x)$; \\
    $\ldots$ \\
    $\texttt{cache}_F(x) = \texttt{retVal}$;  // Update the cache with the return value\\
    return retVal; \- \\
\}
\end{tabbing}
\end{minipage}
}
\end{center}

Function caching involves hashing, and it is necessary to perform equality tests to resolve hash collisions.
Thus, the ability to test two proto-CFLOBDDs for equality in unit time (\sectref{CFLOBDDEquality}) also improves the performance of function caching.


\section{A Denotational Semantics}
\label{Se:ADenotationalSemantics}

In \sectref{CFLOBDDOperationalSemantics}, we gave an operational semantic definition of the function that a CFLOBDD represents.
In this section, we give denotational semantics, which defines the function that a CFLOBDD denotes.
In particular, this semantics associates the $i^{\textit{th}}$ element of a CFLOBDD's valueTuple with the set of Assignments (i.e., language of bit-strings) that map to the $i^{\textit{th}}$ element.
That is, a CFLOBDD $n$ at level $k$ denotes a function
\[
  \sem{n}: [1..|n.\textit{valueTuple}|] \rightarrow \mathcal{P}({\{ 0, 1\}}^{2^k}).
\]
Moreover, the $|n.\textit{valueTuple}|$ range values form a partition of $\{0,1\}^{2^k}$.
(That is, the range values are pairwise disjoint, and $\bigcup_{i=1}^{|\textit{valueTuple}|}\sem{n}[i] = \{0,1\}^{2^k}$.)

The definition of $\sem{n}$ is given in terms of a recursive definition of the semantics of Groupings (really proto-CFLOBDDs).
Consider a Grouping $g$ at level $l$ with $m$ exit vertices.
Suppose that $g.\textit{AConnection}$ has $p$ exits, and $g.\textit{BConnections}[i]$ has $k_i$ exit vertices.
The return map $g.\textit{AReturnTuple}$ is always a 1-1 map, and hence it will not play an explicit role in defining $\sem{g}$, but $g.\textit{BReturnTuples}[i]$---the return edges from $g.\textit{BConnections}[i]$'s exit vertices to $g$'s exit vertices---does play a role.
For convenience, we define $\sem{g}$ to be a vector, of dimension $1 \times m$.
The $m$ entries of the vector form a partition of $\{0,1\}^{2^l}$.
In particular, the vectors for a ForkGrouping $g_{\textit{Fork}}$ and a DontCareGrouping $g_{\textit{DC}}$ are
\[
  \sem{g_{\textit{Fork}}} \eqdef [\{0\},\{1\}]  \qquad\qquad \sem{g_{\textit{DC}}} \eqdef [\{0,1\}].
\]

The semantics of a Grouping $g$ at level $l$ is defined recursively in terms of $g$'s A- and B-connection Groupings at level $l-1$.
To give such a definition, we need to define the meaning of $g.\textit{BReturnTuples}[i]$, the return edges from the exit vertices of a Grouping's $i^{\textit{th}}$ B-connection.
We define $\sem{g.\textit{BReturnTuples}[i]}$ to be a ``permutation matrix'' of size $k_i \times m$.
Each entry of the matrix is either $\emptyset$ or $\{\epsilon\}$, where $\epsilon$ denotes the empty string, with the properties that (i) every row must have exactly one occurrence of $\{\epsilon\}$, and (ii) every column must have at most one occurrence of $\{\epsilon\}$.
For example, if $g.\textit{BReturnTuples}[i]$ maps $g.\textit{BConnections}[i]$'s $3$ exit vertices into $g$'s $5$ exit vertices by $[1 \mapsto 2, 2 \mapsto 4, 3 \mapsto 3]$, then
\[
  \sem{g.\textit{BReturnTuples}[i]} = 
  \begin{bmatrix} 
    \emptyset & \{\epsilon\} & \emptyset     & \emptyset   & \emptyset \\
    \emptyset & \emptyset    & \emptyset   & \{\epsilon\} & \emptyset \\
    \emptyset & \emptyset    & \{\epsilon\} & \emptyset   & \emptyset \\
  \end{bmatrix}_{3\times5}
\]
We now define $\sem{g}$ recursively, where the subscripts denote the dimensions of the vector or matrix, and the two matrix-multiplication primitives are language concatenation and language union:
\[
\begin{array}{@{\hspace{0ex}}l@{\hspace{0ex}}}
  \sem{g}_{1 \times m} = \\
  \qquad \begin{cases}
           [\{0\}, \{1\}]_{1\times2} & \text{if g = \textit{ForkGrouping}} \\
           [\{0, 1\}]_{1\times1} & \text{if g = \textit{DontCareGrouping}} \\
           \begin{array}{@{\hspace{0ex}}l@{\hspace{0ex}}}
             \sem{g.\textit{AConnection}}_{1\times p} \, \times \\
             \qquad \begin{bmatrix}
                      \vdots\\
                      \sem{g.\textit{BConnections}[i]}_{1\times k_i} \times   \sem{g.\textit{BReturnTuples}[i]}_{k_i \times m}\\
                      \vdots
                    \end{bmatrix}_{\scriptsize {\begin{array}{@{\hspace{0ex}}l@{\hspace{0ex}}} {p \times m} \\ i \in \{1..p\} \end{array}}}
           \end{array} & \text{otherwise}
         \end{cases}
\end{array}
\]

Finally, for a CFLOBDD n, $\sem{n}$ is defined as follows:
\[
  \sem{n}_{1 \times |n.\textit{valueTuple}|} \eqdef \sem{n.\textit{grouping}}_{1 \times |n.\textit{grouping}.\textit{numberOfExits}|}
\]

\begin{figure}[tb!]
    \centering
    \includegraphics[scale=0.5]{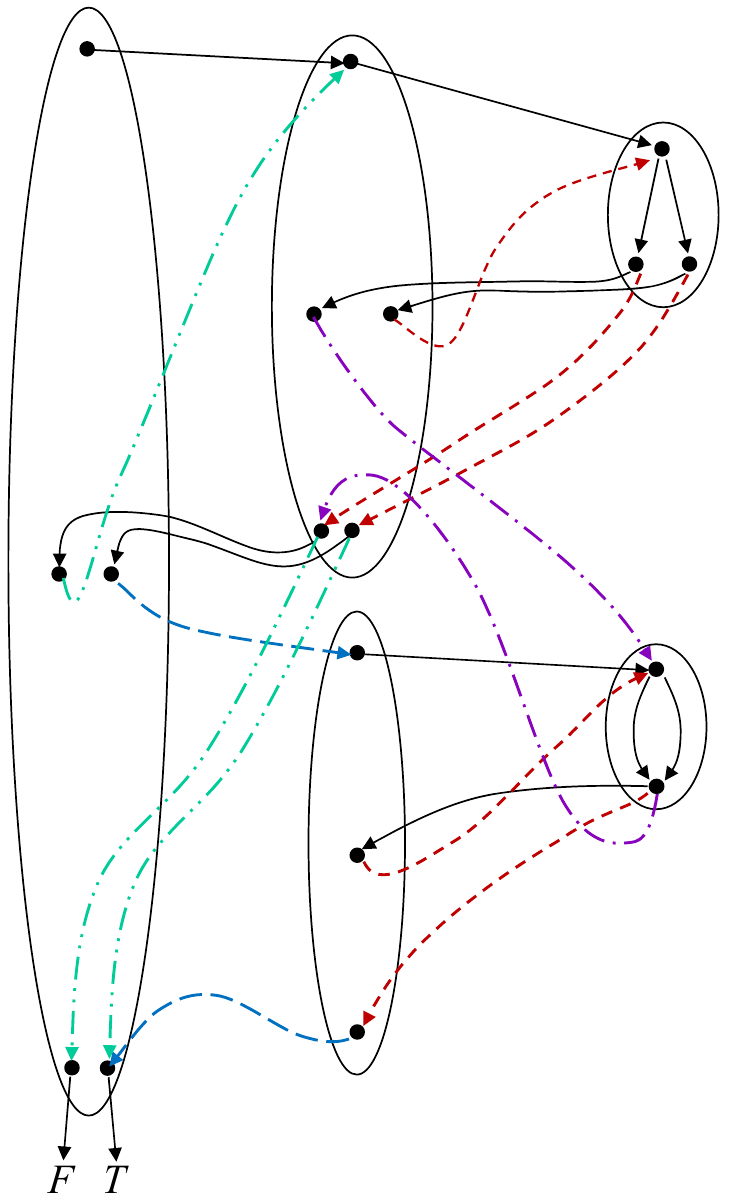}
    \caption{\protect \raggedright 
    CFLOBDD representation of the function $\lambda w,x,y,z. (w \land x) \lor (y \land z)$, with variable ordering $\langle w,x,y,z \rangle$}
    \label{Fi:product4RepFigure}
\end{figure}

\begin{example}\label{Exa:CFLOBDDLanguageSemantics}
For the five proto-CFLOBDDs depicted in \figref{product4RepFigure}, the vectors of languages are as follows (read top-to-bottom by level):
\[
  \begin{array}{lll}
    \multicolumn{1}{c}{\textrm{level 2}} & \multicolumn{1}{c}{\textrm{level 1}} & \multicolumn{1}{c}{\textrm{level 0}} \\
    \hline
    \left[\begin{array}{@{\hspace{0.15ex}}l@{\hspace{0.15ex}}}
             \{ 0000, 0001, 0010, 0100, 0101, 0110, 1000, 1001, 1010 \}, \\
             \{ 1100, 1101, 1110, 1111, 0011, 0111, 1011 \}
    \end{array}\right]\rule{0em}{4.0ex}
    &
    [\{00, 01, 10\}, \{11\}]
    &
    [\{0\}, \{1\}]
    \\
    & [\{00, 01, 10, 11\}] & [\{0, 1\}]
    \\
    \hline
  \end{array}
\]
\end{example}


\section{Algorithms on CFLOBDDs}
\label{Se:cflobdd-algos}

In this section, we describe operations to construct or combine CFLOBDDs.
To aid the reader, \tableref{algo-list-cflobdds} lists the fourteen main operations on CFLOBDDs, together with references to where the algorithm for each operation is presented (and where it is discussed), along with each operation's asymptotic running time and the asymptotic running time of the analogous BDD operation.
Readers familiar with BDDs will find that the algorithms for operations on CFLOBDDs are somewhat more complicated than their BDD counterparts, mainly due to the need to maintain the CFLOBDD structural invariants (\defref{CFLOBDD}).

\begin{sidewaystable}
    \centering
    \resizebox{\textwidth}{!}{
    \begin{tabular}{|l|c|c|c|c|}
        \hline
         \multicolumn{1}{|c|}{\multirow{2}{*}{Operation}} & \multirow{2}{*}{Type Signature} & \multirow{2}{*}{Description} & \multicolumn{2}{|c|}{Time Complexity} \\
         \cline{4-5}
                                                          &                                 &                              &   CFLOBDD & BDD \\
         \hline
         \multirow{2}{*}{Equal} & CFLOBDD $\times$ CFLOBDD  & \multirow{2}{*}{Checks if two CFLOBDDs are equal} & \multirow{2}{*}{$\mathcal{O}(1)$} & \multirow{2}{*}{$\mathcal{O}(1)$} \\
         \multicolumn{1}{|r|}{(\sectref{CFLOBDDEquality})}
         & $\rightarrow$ Boolean & & & \\
         \hline
         ConstantCFLOBDD & \multirow{2}{*}{Int($k$) $\times$ Value$(v)$ $\rightarrow$ CFLOBDD} & Creates a CFLOBDD for a & \multirow{2}{*}{$\mathcal{O}(k)$} & \multirow{2}{*}{$\mathcal{O}(2^k)$}\\
         \multicolumn{1}{|r|}{(\algref{ConstantCFLOBDD}, \sectref{ConstantFunctions})} &  & constant function $\lambda x_0 \dots x_{2^k-1}.v$ &  & \\
         \hline
         FalseCFLOBDD & \multirow{2}{*}{Int ($k$) $\rightarrow$ CFLOBDD}  & Creates a CFLOBDD for & \multirow{2}{*}{$\mathcal{O}(k)$} & \multirow{2}{*}{$\mathcal{O}(2^k)$} \\
         \multicolumn{1}{|r|}{(\algref{FalseCFLOBDD}, \sectref{ConstantFunctions})} &  & the function $\lambda x_0 \dots x_{2^k-1}.F$ & & \\
         \hline
         TrueCFLOBDD & \multirow{2}{*}{Int ($k$) $\rightarrow$ CFLOBDD} & Creates a CFLOBDD for & \multirow{2}{*}{$\mathcal{O}(k)$} & \multirow{2}{*}{$\mathcal{O}(2^k)$} \\
         \multicolumn{1}{|r|}{(\algref{TrueCFLOBDD}, \sectref{ConstantFunctions})} &  & the function $\lambda x_0 \dots x_{2^k-1}.T$ & & \\
         \hline
         NoDistinctionProtoCFLOBDD & \multirow{2}{*}{Int ($k$) $\rightarrow$ Proto-CFLOBDD} & Creates a NoDistinctionProtoCFLOBDD & \multirow{2}{*}{$\mathcal{O}(k)$} & \multirow{2}{*}{N/A} \\
         \multicolumn{1}{|r|}{(\algref{NoDistinctionProtoCFLOBDDAlgorithm}, \sectref{ConstantFunctions})} &  & for $2^k$ variables & & \\
         \hline
         ProjectionCFLOBDD & \multirow{2}{*}{Int($k$) $\times$ Int($i$) $\rightarrow$ CFLOBDD} & Creates a CFLOBDD for & \multirow{2}{*}{$\mathcal{O}(k)$} & \multirow{2}{*}{$\mathcal{O}(2^k)$}\\
         \multicolumn{1}{|r|}{(\algref{ProjectionCFLOBDDAlgorithm}, \sectref{ProjectionFunctions})} &  & the function $\lambda x_0 \dots x_{2^k-1}. x_i$  & & \\
         \hline
         FlipValueTupleCFLOBDD & \multirow{2}{*}{CFLOBDD($c$) $\rightarrow$ CFLOBDD} & Creates a CFLOBDD & \multirow{2}{*}{$\mathcal{O}(1)$} & \multirow{2}{*}{$\mathcal{O}(|c|_B)$} \\
         \multicolumn{1}{|r|}{(\algref{FlipValueTupleCFLOBDDAlgorithm}, \sectref{FlipValueTupleFunction})} &  &  such that the output values are flipped & & \\
         \hline
         ComplementCFLOBDD & \multirow{2}{*}{CFLOBDD($c$) $\rightarrow$ CFLOBDD} & Creates a CFLOBDD such that  & \multirow{2}{*}{$\mathcal{O}(1)$} & \multirow{2}{*}{$\mathcal{O}(|c|_B)$}  \\
         \multicolumn{1}{|r|}{(\algref{FlipValueTupleCFLOBDDAlgorithm}, \sectref{FlipValueTupleFunction})} &  &  the output values are complemented & & \\
         \hline
         ScalarMultiplyCFLOBDD & CFLOBDD ($c$) $\times$ Value ($v$)  & \multirow{2}{*}{\changed{Performs $c' = c \ast v$}} & \multirow{2}{*}{\changed{$\mathcal{O}(|c| \times |c'|)$}} & \multirow{2}{*}{$\mathcal{O}(|c|_B)$} \\
         \multicolumn{1}{|r|}{(\algref{ScalarMultiplyCFLOBDDAlgorithm}, \sectref{ScalarMultiplication})} & $\rightarrow$ CFLOBDD & & & \\
         \hline
         BinaryApplyAndReduce & CFLOBDD ($c_1$) $\times$ CFLOBDD ($c_2$) & \multirow{2}{*}{\changed{Performs $c' = c_1 \,\textit{op}\, c_2$}} & \multirow{2}{*}{\changed{$\mathcal{O}(|c_1| \times |c_2| \times |c'|)$}} & \multirow{2}{*}{$\mathcal{O}(|c_1|_B \times |c_2|_B)$} \\
         \multicolumn{1}{|r|}{(\algref{BinaryApplyAndReduce}, \sectref{BinaryOperationsOnCFLOBDDs})} &  $\times$ Operation $\textit{op}$ $\rightarrow$ CFLOBDD &  &  & \\
         \hline
         PathCounting & \multirow{2}{*}{CFLOBDD($c$) $\rightarrow$ CFLOBDD} & Computes the number of paths to  & \multirow{2}{*}{$\mathcal{O}(|c|)$}  & $\mathcal{O}(|c|_B)$ \\
         \multicolumn{1}{|r|}{(\algref{PathCounting}, \sectref{PathCountingInACFLOBDD})} &  & every exit vertex of every grouping &  & (See \sectref{SamplingInBDDs}.) \\
         \hline
         Sampling & \multirow{2}{*}{CFLOBDD($c$) $\rightarrow$ String} & \multirow{2}{*}{Samples a path from $c$} & \multirow{2}{*}{$\mathcal{O}(\max (\textit{vars}, |c|))$} & $\mathcal{O}(\max (\textit{vars}, |c|_B))$ \\
         \multicolumn{1}{|r|}{(\algref{Sampling}, \sectref{SamplingInACFLOBDD})} &  &  & & (See \sectref{SamplingInBDDs}.) \\
         \hline
         KroneckerProduct & CFLOBDD($c_1$) $\times$ CFLOBDD($c_2$)  & \multirow{2}{*}{Performs $c' = c_1 \tensor c_2$} & \multirow{2}{*}{$\mathcal{O}((|c_1| + |c_2| + c_1.\text{\#exits} \times c_2.\text{\#exits}) \times |c'|)$}  & \multirow{2}{*}{$\mathcal{O}(|c_1|_B)$} \\
         \multicolumn{1}{|r|}{(App.\sectref{kronecker--product} \& \algref{KP4Voc}, \sectref{KroneckerProduct})} & $\rightarrow$ CFLOBDD & & & \\
         \hline
         MatrixMultiply & CFLOBDD($c_1$) $\times$ CFLOBDD($c_2$)  & \changed{Performs $c' = c_1 \times c_2$} & $\mathcal{O}(N^3)$, plus the time for & \multirow{2}{*}{$\mathcal{O}(N^3)$} \\
         \multicolumn{1}{|r|}{(\algrefs{MatrixMult}{MatrixMultGrouping} , \sectref{matrix-mult})} & $\rightarrow$ CFLOBDD & \changed{for matrices of size $N \times N$} & \changed{a final call to \texttt{Reduce}} & \\
         \hline
    \end{tabular}
    }
    \caption{
      List of operations on CFLOBDDs;
      $\textit{vars}$ denotes the number of Boolean variables ($= 2^k$, where $k$ is the number of levels of the CFLOBDD).
      The size measure $|\cdot|$ counts the number of groupings, vertices, and edges---with no double-counting of shared groupings due to hash-consing.
      In the column for the time complexities of BDD operations, an occurrence of $c$ refers to a BDD argument of the operation, and $|c|_B$ denotes the size of BDD $c$ (the number of nodes and edges).
      For quasi-reduced BDDs, the time to construct the analog of NoDistinctionProtoCFLOBDD is $\mathcal{O}(2^k)$.
      Note that the complexity of MatrixMultiply is in terms of the sizes of matrices represented by $c_1$ and $c_2$ and not the sizes of $c_1$ and $c_2$.
    }
    \label{Ta:algo-list-cflobdds}
\end{sidewaystable}

\subsection{Primitive CFLOBDD-Creation Operations}

\subsubsection{Constant Functions}
\label{Se:ConstantFunctions}

\begin{algorithm}[tb!]
\caption{ConstantCFLOBDD\label{Fi:ConstantCFLOBDD}}
\Input{int k, Value v}
\Output{CFLOBDD representation of a function with $2^k$ variables and constant value $v$}
\Begin{
\Return RepresentativeCFLOBDD(NoDistinctionProtoCFLOBDD(k), [v])\;
}
\end{algorithm}

\begin{algorithm}[tb!]
\caption{NoDistinctionProtoCFLOBDD\label{Fi:NoDistinctionProtoCFLOBDDAlgorithm}}
\Input{int k}
\Output{Proto-CFLOBDD representation of a function with $2^k$ variables}
\Begin{
\If{k == 0}{\Return RepresentativeDontCareGrouping\;}
InternalGrouping g = new InternalGrouping(k)\;
g.AConnection = NoDistinctionProtoCFLOBDD(k-1)\;
g.AReturnTuple = [1]\;
g.numberOfBConnections = 1\;
g.BConnections[1] = g.AConnection\;
g.BReturnTuples[1] = [1]\;
g.numberOfExits = 1\;
\Return RepresentativeGrouping(g)\;
}
\end{algorithm}

\begin{algorithm}[tb!]
\caption{FalseCFLOBDD\label{Fi:FalseCFLOBDD}}
\Input{int k}
\Output{CFLOBDD representation of a function with $2^k$ variables and constant value $F$}
\Begin{
\Return ConstantCFLOBDD(k, $F$);
}
\end{algorithm}

\begin{algorithm}[tb!]
\caption{TrueCFLOBDD\label{Fi:TrueCFLOBDD}}
\Input{int k}
\Output{CFLOBDD representation of a function with $2^k$ variables and constant value $T$}
\Begin{
\Return ConstantCFLOBDD(k, $T$);
}
\end{algorithm}



The CFLOBDD-creation operation {\tt ConstantCFLOBDD}, given as \algref{ConstantCFLOBDD}, produces the family of CFLOBDDs that represent functions of the form $\lambda x_0, x_1, \ldots, x_{2^k-1} . v$, where $v$ is some constant value.
\texttt{ConstantCFLOBDD}($k,v$) uses as a subroutine {\tt NoDistinctionProtoCFLOBDD} (\algref{NoDistinctionProtoCFLOBDDAlgorithm}),
which constructs the no-distinction proto-CFLOBDD for a given level $k$ (see also \figref{NoDistinctionProtoCFLOBDD}).
{\tt ConstantCFLOBDD} can be used to construct CFLOBDDs for the constant functions $\lambda x_0, x_1, \ldots, x_{2^k-1} . F$ (\algref{FalseCFLOBDD}) and $\lambda x_0, x_1, \ldots, x_{2^k-1} . T$ (\algref{TrueCFLOBDD}).
\texttt{ConstantCFLOBDD}($k,v$) runs in time $O(k)$ and uses at most $O(k)$ space.

 \subsubsection{Projection Functions}
 \label{Se:ProjectionFunctions}
 
\begin{figure}[ht!]
\begin{center}
\begin{tabular}{c@{\hspace{.33in}}c}
     \includegraphics[height=2.3in]{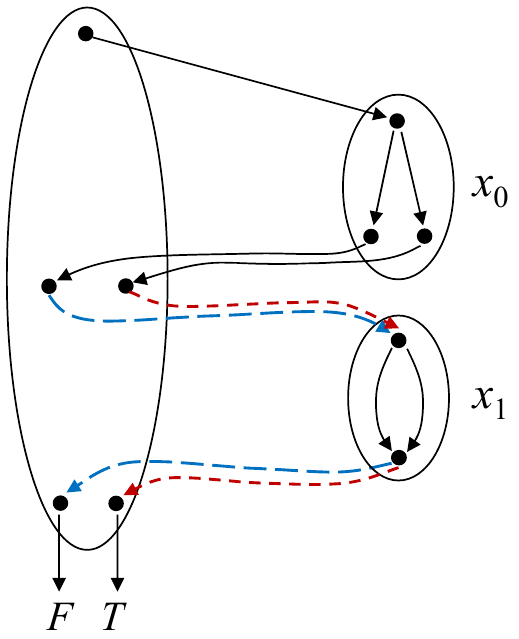}
   & \includegraphics[height=2.3in]{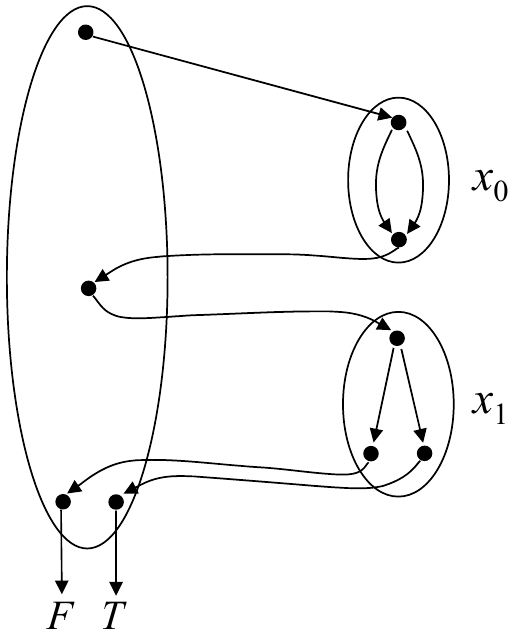}
 \\
     (a) & (b) 
 \\
 \\
     \includegraphics[height=3.5in]{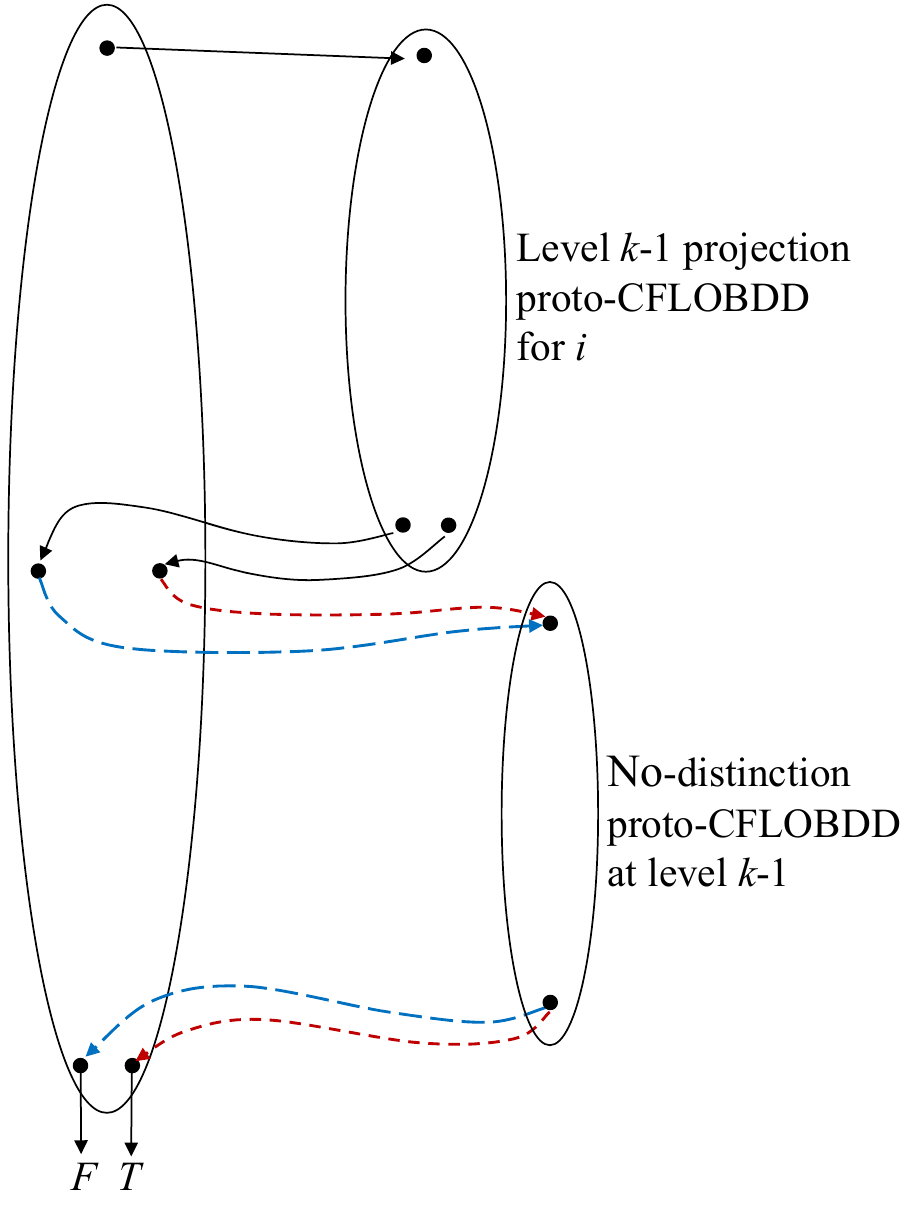}
   & \includegraphics[height=3.5in]{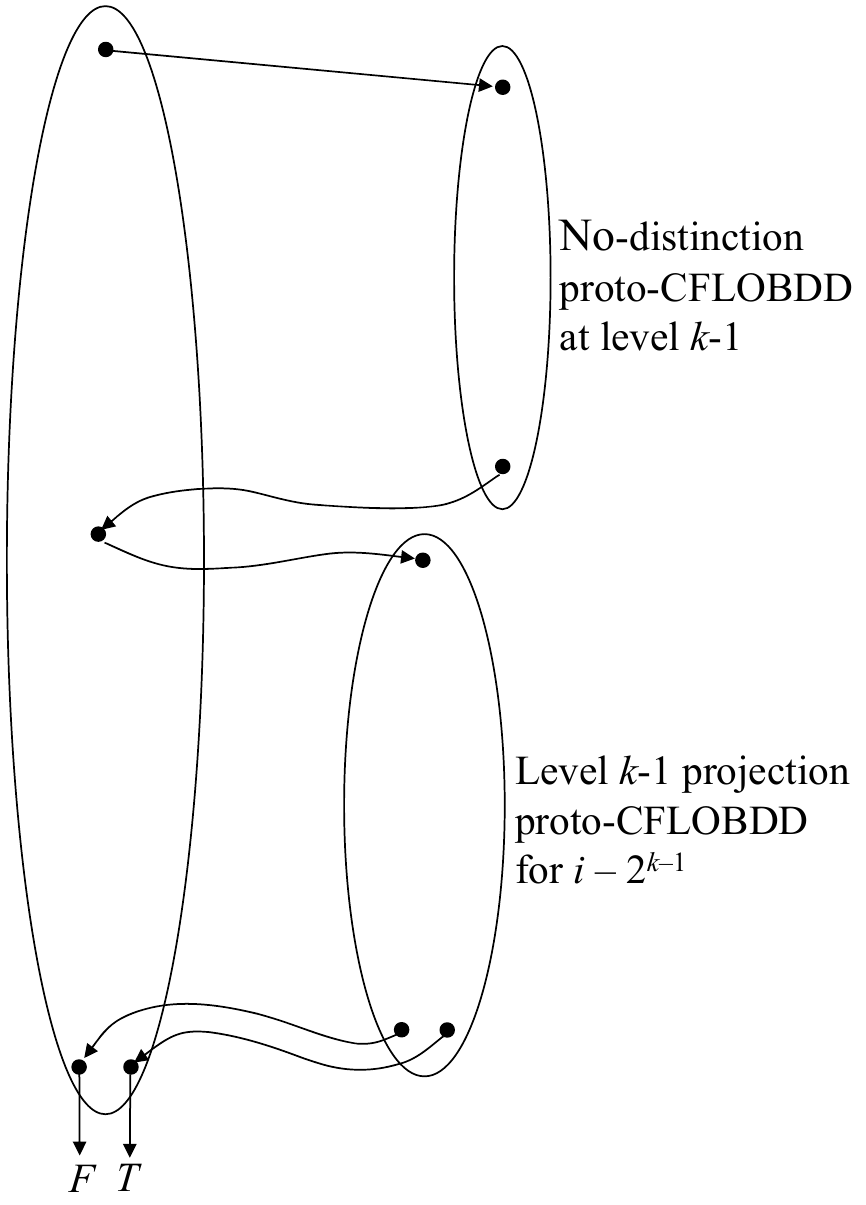}
 \\ 
    (c) & (d) 
\end{tabular}
\end{center}
\caption{
(a) CFLOBDD for $\lambda {x_0}{x_1}.{x_0}$;
(b) CFLOBDD for $\lambda {x_0}{x_1}.{x_1}$;
(c) schematic drawing of CFLOBDDs that represent projection functions of the form $\lambda x_0, x_1, \ldots, x_{2^k-1} . x_i$, when $0 \leq i < 2^{k-1}$;
(d) schematic drawing of CFLOBDDs that represent projection functions of the form $\lambda x_0, x_1, \ldots, x_{2^k-1} . x_i$, when $2^{k-1} \leq i < 2^k$.
}
\label{Fi:ProjectionCFLOBDD}
\end{figure}

\begin{algorithm}[tb!]
\SetKwFunction{ProjectionCFLOBDD}{ProjectionCFLOBDD}
\SetKwFunction{ProjectionProtoCFLOBDD}{ProjectionProtoCFLOBDD}
\SetKwProg{myalg}{Algorithm}{}{end}
  \myalg{\ProjectionCFLOBDD{k, i}}{
  \Input{int k (level), int i (index)}
  \Output{CFLOBDD representing function $\lambda x_0, x_1, \ldots, x_{2^k-1} . x_i$}
  \Begin{
    assert(0 <= i < 2**k)\;
    \Return RepresentativeCFLOBDD(ProjectionProtoCFLOBDD(k,i), [F,T])\;
  }
  }{}
  \setcounter{AlgoLine}{0}
  \SetKwProg{myproc}{SubRoutine}{}{end}
  \myproc{\ProjectionProtoCFLOBDD{k, i}}{
  \Input{int k (level), int i (index)}
  \Output{Grouping g representing function $\lambda x_0, x_1, \ldots, x_{2^k-1} . x_i$}
  \Begin{
    \eIf(\tcp*[f]{i must also be 0}){k == 0}{\Return RepresentativeForkGrouping\;}
    {
        InternalGrouping g = new InternalGrouping(k)\;
        \eIf(\tcp*[f]{i falls in AConnection range}){i < 2$\ast \ast$(k-1)}{
            g.AConnection = ProjectionProtoCFLOBDD(k-1,i)\;
            g.AReturnTuple = [1,2]\;
            g.numBConnections = 2\;
            g.BConnection[1] = NoDistinctionProtoCFLOBDD(k-1)\;
            g.BReturnTuples[2] = [1]\;
            g.BConnections[2] = g.BConnection[1]\;
            g.BReturnTuples[2] = [2]\;
            g.numberOfExits = 2\;
        }
        (\tcp*[f]{i falls in BConnection range}){
            g.AConnection = NoDistinctionProtoCFLOBDD(k-1)\;
            g.AReturnTuple = [1]\;
            g.numBConnections = 1\;
            i = i - 2$\ast \ast$(k-1)\tcp*[r]{Remove high-order bit for recursive call}
            g.BConnections[1] = ProjectionProtoCFLOBDD(k-1,i)\;
            g.BReturnTuples[1] = [1,2]\;
            g.numberOfExits = 2\;
        }
        \Return RepresentativeGrouping(g)\;
    }
  }
  }
  \caption{ProjectionProtoCFLOBDD\label{Fi:ProjectionCFLOBDDAlgorithm}}
\end{algorithm}

A second family of CFLOBDD-creation operations produces the Boolean-valued \emph{(single-variable) projection functions} of the form
$\lambda x_0, x_1, \ldots, x_{2^k-1} . x_i$, where $i$ ranges from $0$ to $2^k-1$.
\figref{ProjectionCFLOBDD} illustrates the structure of the CFLOBDDs that represent these functions.
\algref{ProjectionCFLOBDDAlgorithm} gives pseudo-code for \texttt{ProjectionCFLOBDD}($k, i$), which constructs the $i^{\textit{th}}$ such function.
\texttt{ProjectionCFLOBDD}($k, i$) runs in time $O(k)$ and uses at most $O(k)$ space.


\subsection{Unary Operations on CFLOBDDs}
This section discusses how to perform certain unary operations
on CFLOBDDs:

\begin{algorithm}[tb!]
\caption{ComplementCFLOBDD\label{Fi:FlipValueTupleCFLOBDDAlgorithm}}
\SetKwFunction{ComplementCFLOBDD}{ComplementCFLOBDD}
\SetKwFunction{FlipValueTupleCFLOBDD}{FlipValueTupleCFLOBDD}
\small{
\SetKwProg{myalg}{Algorithm}{}{end}
\myalg{\FlipValueTupleCFLOBDD{c}}{
\Input{CFLOBDD $c$}
\Output{CFLOBDD $c'$ such that the output values are flipped}
\Begin{
assert(|$c$.valueTuple| == 2)\;
return RepresentativeCFLOBDD($c$.grouping, [$c$.valueTuple[2], $c$.valueTuple[1]])\;
}
}
\SetKwProg{myproc}{Algorithm}{}{end}
\myproc{\ComplementCFLOBDD{c}}{
\Input{CFLOBDD $c$}
\Output{CFLOBDD $c'$ such that the output values are complemented}
\Begin{
\If{c == FalseCFLOBDD(c.grouping.level)}{\Return TrueCFLOBDD(\textit{c.grouping.level})\;}
\If{c == TrueCFLOBDD(c.grouping.level)}{\Return FalseCFLOBDD(\textit{c.grouping.level})\;}
\Return FlipValueTupleCFLOBDD($c$)\;
}
}{}
}
\end{algorithm}

\begin{algorithm}[tb!]
\caption{ScalarMultiplyCFLOBDD\label{Fi:ScalarMultiplyCFLOBDDAlgorithm}}
\small{
\Input{CFLOBDD $c$, Value $v$}
\Output{CFLOBDD $c' = c \ast v$}
\Begin{
\changed{
\tcp{Multiply CFLOBDD $c$ by the CFLOBDD for the constant function $\lambda x_0, x_1, \ldots, x_{2^k-1} . v$}
\Return BinaryApplyAndReduce($c$, ConstantCFLOBDD($c$.level, $v$), (op)Times)\tcp*[r]{(See \sectref{BinaryOperationsOnCFLOBDDs})}
}
}
}
\end{algorithm}

\subsubsection{FlipValueTuple Function.}
\label{Se:FlipValueTupleFunction}
    The function \texttt{FlipValueTupleCFLOBDD} applies
    in the special situation in which a CFLOBDD
    maps Boolean-variable-to-Boolean-value assignments to just two
    possible values;
    {\tt FlipValueTupleCFLOBDD\/} flips the two values in the
    CFLOBDD's {\tt valueTuple\/} field and returns the resulting CFLOBDD.
    In the case of Boolean-valued CFLOBDDs, this operation can be used to
    implement the operation {\tt ComplementCFLOBDD}, which forms the
    Boolean complement of its argument, in an efficient manner.
    The pseudocode for these functions is given in \algref{FlipValueTupleCFLOBDDAlgorithm}.
    \texttt{FlipValueTupleCFLOBDD} and \texttt{ComplementCFLOBDD} are constant-time operations.
    
\subsubsection{Scalar Multiplication.}
\label{Se:ScalarMultiplication}
    Function \texttt{ScalarMultiplyCFLOBDD} of
    \algref{ScalarMultiplyCFLOBDDAlgorithm} applies to any
    CFLOBDD that maps Boolean-variable-to-Boolean-value assignments
    to values on which multiplication by a scalar value of type {\tt Value\/}
    is defined.  
    \texttt{ScalarMultiplyCFLOBDD} constructs a CFLOBDD for the constant function $\lambda x_0, x_1, \ldots, x_{2^k-1} . v$, which is multiplied by CFLOBDD $c$ using \texttt{BinaryApplyAndReduce}---the generic operation for binary CFLOBDD operations (discussed in \sectref{BinaryOperationsOnCFLOBDDs})---with the multiplication operator $\text{Times}$ passed as the third argument.
    
\subsection{Binary Operations on CFLOBDDs}
\label{Se:BinaryOperationsOnCFLOBDDs}

\begin{figure}[ht!]
\centering
\begin{tabular}{c@{\hspace{3.5ex}}c@{\hspace{3.5ex}}c}
     \includegraphics[height=1.5in]{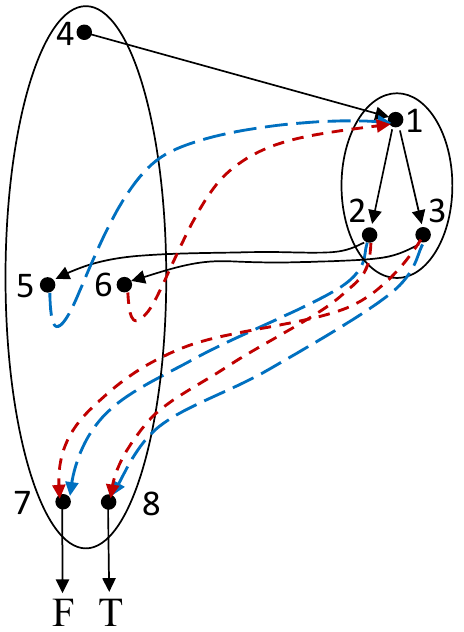}
   & \includegraphics[height=1.5in]{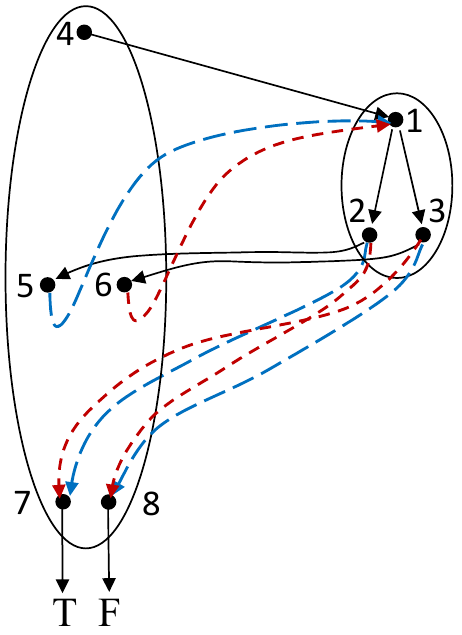}
   & \includegraphics[height=1.5in]{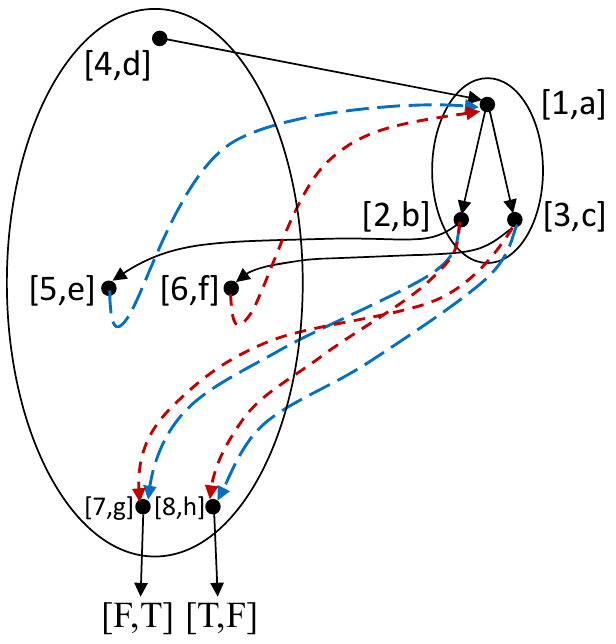}
 \\
     \begin{minipage}{.20\textwidth}
       {\small (a) $\lambda x_0, x_1 . x_0 \oplus x_1$}
     \end{minipage}     
   & \begin{minipage}{.20\textwidth}
       {\small (b) $\lambda x_0, x_1 . x_0 \Leftrightarrow x_1$}
     \end{minipage} 
   & \begin{minipage}{.54\textwidth}
       {\small (c) Result of calling \texttt{PairProduct} on (a) and (b)}
     \end{minipage} 
  \\
  \\
  \\
\end{tabular}
\begin{tabular}{c@{\hspace{3.0ex}}c}
     \includegraphics[height=1.5in]{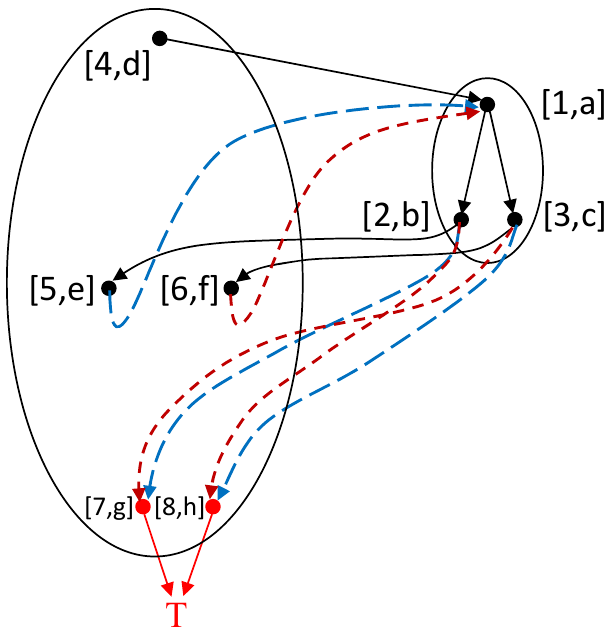}
   & \includegraphics[height=1.5in]{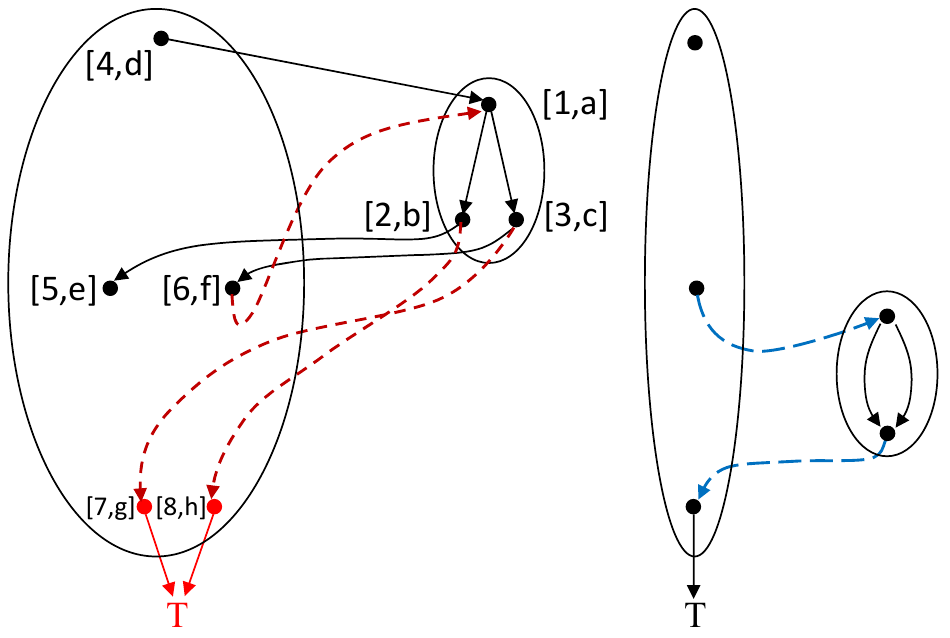}
 \\
     \begin{minipage}{.47\textwidth}
       {\small (d) Result of applying $\lor$ to the values in each of the terminal-value pairs $[F,T]$ and $[T,F]$.
       At this point, it is necessary to perform a reduction that folds together the two exit vertices.}
     \end{minipage}     
   & \begin{minipage}{.47\textwidth}
       {\small (e) Result of calling \texttt{Reduce} on the first B-connection with \texttt{reductionTuple} $[1,1]$.}
     \end{minipage} 
\\
\\
\\
     \includegraphics[height=1.5in]{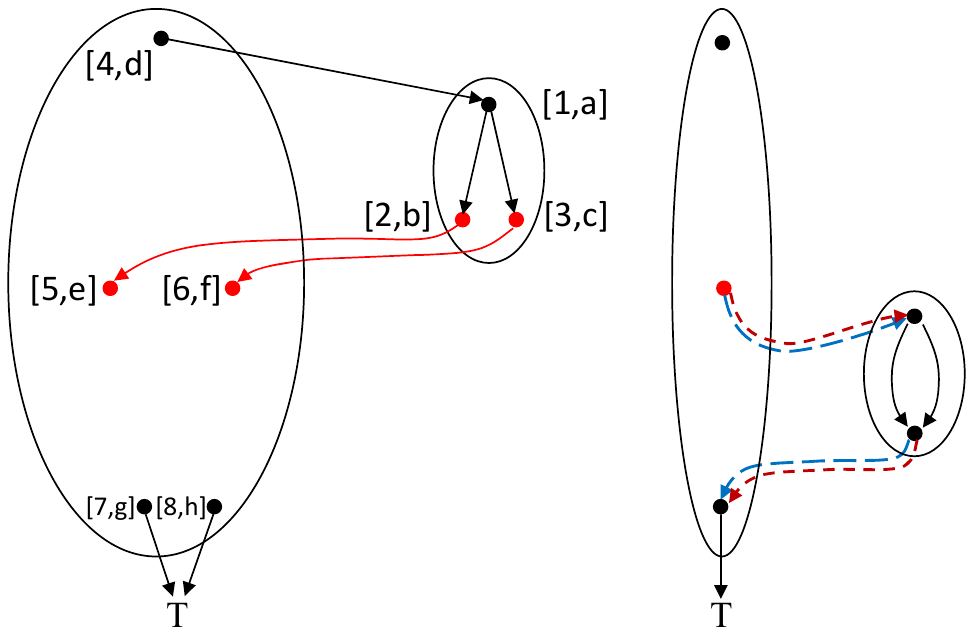}
   & \includegraphics[height=1.5in]{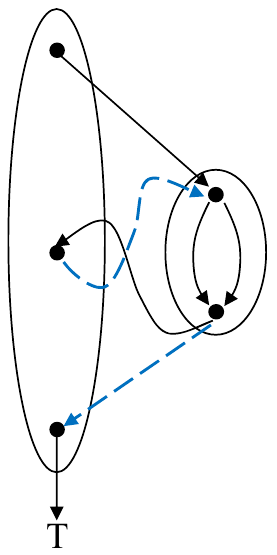}
 \\
     \begin{minipage}{.47\textwidth}
       {\small (f) Result of calling \texttt{Reduce} on the second B-connection with \texttt{reductionTuple} $[1,1]$.
       The two calls on \texttt{Reduce} produce the same B-connection proto-CFLOBDDs with identical return edges---indicated by the coincidence of the blue and red dashed edges in the structure on the right.
       At this point, it is necessary to perform a reduction that folds together the two middle vertices.}
     \end{minipage}     
   & \begin{minipage}{.47\textwidth}
       {\small (g) After calling \texttt{Reduce} on the A-connection with \texttt{reductionTuple} $[1,1]$, the final result is the CFLOBDD for $\lambda x_0, x_1 . T$}.
     \end{minipage} 
\end{tabular}
\caption{
  Illustration of how $(\lambda x_0, x_1 . x_0 \oplus x_1)$ $\lor$ $(\lambda x_0, x_1 . x_0 \Leftrightarrow x_1)$ results in $\lambda x_0, x_1 . T$.
}
\label{Fi:BinaryApplyAndReduceExample}
\end{figure}

This section presents an algorithm for performing binary operations on CFLOBDDs.
The algorithm is parameterized in terms of a binary operation \texttt{op} that is to be applied pointwise to the range values of two CFLOBDDs.
That is, given the CFLOBDDs for two functions $n_1$ and $n_2$ and  binary operation \texttt{op}, the goal of the algorithm is to create the CFLOBDD for $n_1 \,\texttt{op}\, n_2$ where, for each assignment $a$, $(n_1 \,\texttt{op}\, n_2)(a) = n_1(a) \,\texttt{op}\, n_2(a)$.
Operation \texttt{op} could be $+, -, *, /$, etc., or---if the functions are Boolean-valued---$\lor, \land, \xor$, etc.
As with BDDs, such operations on CFLOBDDs can be implemented via a two-step process\footnote{
  The two-step process is conceptual for BDDs: the two steps can be combined in an implementation (e.g., see \cite[\S3.3]{DBLP:conf/ml/FilliatreC06}).
  For CFLOBDDs, it does not appear possible to combine the two steps, at least not easily.
  For more details, see the Remark just after \exref{xor-cross-equiv}.
}
\begin{enumerate}
    \item perform a product construction
    \item perform a reduction step on the result of step one.
\end{enumerate}
Just as there can be multiple occurrences of a given node in a BDD, there can be multiple occurrences of a given grouping in a CFLOBDD.
To avoid a blow-up in costs, binary operations need to avoid making repeated calls on a given pair of groupings $g_1 \in n_1$ and $g_2 \in n_2$.
Assuming that the 
hash-table lookup and insertion
methods used for hash-consing (\sectref{HashConsing}) and function caching (\sectref{FunctionCaching}) run in
(expected) unit-cost time, the time to perform the product construction is asymptotically bounded by the product of the sizes of the two argument CFLOBDDs---i.e., $O(|n_1|\times |n_2|)$.\footnote{
  \label{Footnote:ProductConstructionCost}
  More precisely, let $n {\uparrow} \textit{gr}[k]$ denote the set of groupings at level $k {\in} [0 .. l]$ in CFLOBDD $n$.
  The time to construct the product of $n_1$ and $n_2$ is asymptotically bounded by
  $
    \sum_{k=0}^l \sum \Big\{ |g_1| {\times} |g_2| \mid g_1 {\in}\, n_1  {\uparrow} \textit{gr}[k]~\textrm{and}~g_2 {\in}\, n_2  {\uparrow} \textit{gr}[k] \Big\}.
  $
}
\sectref{CostOfReduce} shows that the time for the reduction step is $O(|n_1|\times |n_2| \times |n'|)$, where $n'$ denotes the CFLOBDD that is the result of $n_1 \,\texttt{op}\, n_2$.
Consequently, binary operations satisfy Requirement (\ref{Req:ComputationalEfficiency});
i.e., they run in (expected) time that is polynomial in the sizes of the input and output CFLOBDDs.

\figref{BinaryApplyAndReduceExample} illustrates this method by showing how the CFLOBDD for $\lambda x_0, x_1 . T$ is
obtained as the result of a Boolean-$\lor$: $(\lambda x_0, x_1 . x_0 \oplus x_1) \lor (\lambda x_0, x_1 . x_0 \Leftrightarrow x_1)$.
\figref{BinaryApplyAndReduceExample}c shows the result of the product construction (\texttt{PairProduct}, \algref{PairProduct}).
\figref{BinaryApplyAndReduceExample}d, e, f, and g illustrate some of the steps of the reduction algorithm (\texttt{Reduce}, \algref{Reduce}).
\figref{BinaryApplyAndReduceExample} is discussed in more detail in \exref{xor-cross-equiv}.

\begin{algorithm}[tb!]
\caption{CollapseClassesLeftmost\label{Fi:InducedTuples}}
\Input{Tuple equivClasses}
\Output{Tuple$\times$Tuple [projectedClasses, renumberedClasses]}
\Begin{
\tcp{Project the tuple equivClasses, preserving left-to-right order, retaining the leftmost instance of each class}
\label{Li:IT:TupleCollapseStart}
Tuple projectedClasses = [equivClasses(i) : i $\in$ [1..|equivClasses|] | i = min\{j $\in$ [1..|equivClasses|] | equivClasses(j) = equivClasses(i)\}]\;
\label{Li:IT:TupleCollapseEnd}
\tcp{Create tuple in which classes in equivClasses are renumbered according to their ordinal position in projectedClasses}
Map orderOfProjectedClasses = \{[x,i]: i $\in$ [1..|projectedClasses|] | x = projectedClasses(i)\}\;
Tuple renumberedClasses = [orderOfProjectedClasses(v) : v $\in$ equivClasses]\;
\Return [projectedClasses, renumberedClasses]\;
}
\end{algorithm}

\begin{algorithm}[tb!]
\Input{CFLOBDDs n1, n2 and Operation op}
\Output{CFLOBDD n = n1 op n2}
\Begin{
\tcp{Perform cross product}
    Grouping$\times$PairTuple [g,pt] = PairProduct(n1.grouping,n2.grouping)\;  \label{Li:BAAR:CallPairProduct}
\tcp{Create tuple of ``leaf'' values}
    ValueTuple deducedValueTuple = [ op(n1.valueTuple[i1],n2.valueTuple[i2])~:~[i1,i2] $\in$ pt ]\;  \label{Li:BAAR:LeafValues}
     \tcp{Collapse duplicate leaf values, folding to the left}
    Tuple$\times$Tuple [inducedValueTuple,inducedReductionTuple] = CollapseClassesLeftmost(deducedValueTuple) \;  \label{Li:BAAR:CollapseLeafValues}
\tcp{Perform corresponding reduction on g, folding g's exit vertices w.r.t.\ inducedReductionTuple}
    Grouping g' = Reduce(g, inducedReductionTuple) \;  \label{Li:BAAR:CallReduce}
    CFLOBDD n = RepresentativeCFLOBDD(g', inducedValueTuple) \;
    \Return n\;
}
\caption{BinaryApplyAndReduce\label{Fi:BinaryApplyAndReduce}}
\end{algorithm}

\begin{algorithm}[tb!]
\caption{PairProduct\label{Fi:PairProduct}}
\Input{Groupings g1, g2}
\Output{Grouping g: product of g1 and g2; PairTuple ptAns: tuple of pairs of exit vertices}
\Begin{
    \lIf {g1 and g2 are both no-distinction proto-CFLOBDDs  \label{Li:PPNoDistinctionStart}}  
    { \Return [ g1, [[1,1]] ]}   
    \lIf {g1 is a no-distinction proto-CFLOBDD}
    {\Return [ g2, [[1,k]~:~k $\in$ [1..g2.numberOfExits]] ]}
    \lIf {g2 is a no-distinction proto-CFLOBDD}
    {\Return [ g1, [[k,1]~:~k $\in$ [1..g1.numberOfExits]] ]}      \label{Li:PPNoDistinctionEnd}
    \lIf {g1 and g2 are both fork groupings  \label{Li:PPBothForkGroupings}} {\Return [ g1, [[1,1],[2,2]] ]} \label{Li:PPBothForkGroupingsEnd}
    \tcp{ Pair the A-connections}
    Grouping$\times$PairTuple [gA,ptA] = PairProduct(g1.AConnection, g2.AConnection)\;
    InternalGrouping g = new InternalGrouping(g1.level)\;
    g.AConnection = gA \;
    g.AReturnTuple = [1..|ptA|]\tcp*[r]{Represents the middle vertices}
    g.numberOfBConnections = |ptA| \;  \label{Li:PPNumberOfBConnections}
\tcp{Continued in \algref{PairProductContinued}}
}
\end{algorithm}

\begin{algorithm}
\caption{PairProduct (cont.)\label{Fi:PairProductContinued}}
\setcounter{AlgoLine}{18}
\SetKwBlock{Begin}{}{end}
\Begin{
\tcp{Pair the B-connections, but only for pairs in ptA }
    \tcp{Descriptor of pairings of exit vertices}
    Tuple ptAns = []\;  \label{Li:ptAnsDeclaration}
    \tcp{Create a B-connection for each middle vertex}
    \For {$j \leftarrow 1$ \KwTo $|ptA|$}{
        Grouping$\times$PairTuple [gB,ptB] = PairProduct(g1.BConnections[ptA(j)(1)], g2.BConnections[ptA(j)(2)])\;   \label{Li:BConnectionPairProduct}
        g.BConnections[j] = gB \;
        \tcp{Now create g.BReturnTuples[j], and augment ptAns as necessary}
        g.BReturnTuples[j] = [] \;        \label{Li:PPExitVertexLoopStart}
        \For{$i \leftarrow 1$ \KwTo $|ptB|$} {
            c1 = g1.BReturnTuples[ptA(j)(1)](ptB(i)(1))\tcp*[r]{an exit vertex of g1}
            c2 = g2.BReturnTuples[ptA(j)(2)](ptB(i)(2)) \tcp*[r]{an exit vertex of g2}
            \eIf(\tcp*[f]{Not a new exit vertex of g}){[c1,c2] $\in$ ptAns}
                {
                index = the k such that ptAns(k) == [c1,c2] \;
                g.BReturnTuples[j] = g.BReturnTuples[j] || index \;
                }
            (\tcp*[f]{Identified a new exit vertex of g}){
                g.numberOfExits = g.numberOfExits + 1 \;
                g.BReturnTuples[j] = g.BReturnTuples[j] || g.numberOfExits \;
                ptAns = ptAns || [c1,c2] \; \label{Li:PPConcatenateExitPair}
            }
        }  \label{Li:PPExitVertexLoopEnd}
    }
    \Return [RepresentativeGrouping(g), ptAns]\; \label{Li:PPTabulateAnswer}
}
\end{algorithm}

\begin{algorithm}
\SetKwFunction{Reduce}{Reduce}
\SetKwFunction{InsertBConnection}{InsertBConnection}
\SetKwProg{myalg}{Algorithm}{}{end}
\Input{Grouping g, ReductionTuple reductionTuple}
\Output{Grouping g' that is "reduced"}
\Begin{
\label{Li:ReduceNoProcessingStart}
\tcp{Test whether any reduction actually needs to be carried out}
\label{Li:ReduceNoProcessingEnd}
\If{reductionTuple == [1..|reductionTuple|]}{\Return g\;}
\label{Li:ReduceNoDistinctionStart}
\tcp{If only one exit vertex, then collapse to no-distinction proto-CFLOBDD}
\label{Li:ReduceNoDistinctionEnd}
\If{|\{x : x $\in$ reductionTuple\}| == 1}{ \Return NoDistinctionProtoCFLOBDD(g.level)\;}
InternalGrouping g' = new InternalGrouping(g.level)\;
g'.numberOfExits = |\{ x : x $\in$ reductionTuple \}|\;
Tuple reductionTupleA = []\;
\For{$i \leftarrow 1$ \KwTo $g.numberOfBConnections$}{
\label{Li:BConnectionCollapseStart}
Tuple deducedReturnClasses = [reductionTuple(v) : v $\in$ g.BReturnTuples[i]]\;
\label{Li:AConnectionProcessingStart}
Tuple$\times$Tuple [inducedReturnTuple, inducedReductionTuple] = CollapseClassesLeftmost(deducedReturnClasses)\;
\label{Li:BConnectionCollapseEnd}
Grouping h = Reduce(g.BConnection[i], inducedReturnTuple)\;
int position = InsertBConnection(g', h, inducedReturnTuple)\;  \label{Li:CallInsertBConnection}
reductionTupleA = reductionTupleA || position\;  \label{Li:ExtendReductionTupleA}
}
Tuple$\times$Tuple [inducedReturnTuple, inducedReductionTuple] = CollapseClassesLeftmost(reductionTupleA)\;  \label{Li:CallCollapseClassesLeftmost}
Grouping h' = Reduce(g.AConnection, inducedReductionTuple)\;
g'.AConnection = h'\;
g'.AReturnTuple = inducedReturnTuple\;  \label{Li:AConnectionProcessingEnd}
\Return RepresentativeGrouping(g')\;
}
\caption{Reduce\label{Fi:Reduce}}
\end{algorithm}

\begin{algorithm}[tb!]
\caption{InsertBConnection\label{Fi:InsertBConnection}}
\small{
\Input{InternalGrouping g, Grouping h, ReturnTuple returnTuple}
\Output{int --  Insert (h, ReturnTuple) as the next B-connection of g, if they are a new combination; otherwise return the index of the existing occurence of (h, ReturnTuple)}
\Begin{
\label{Li:InsertBConnectionStart}
\changed{\lIf{there exists $i \in [1..g.numberOfBConnections]$ such that g.BConnection[i] == h $\&\&$ g.BReturnTuples[i] == returnTuple}{\Return i}
g.numberOfBConnections = g.numberOfBConnections + 1\;
g.BConnections[g.numberOfBConnections] = h\;
g.BReturnTuples[g.numberOfBConnections] = returnTuple\;
\Return g.numberOfBConnections\;}
\label{Li:InsertBConnectionEnd}
}
}
\end{algorithm}

The algorithms involved are given as \algseqref{InducedTuples}{InsertBConnection}.
(In \algrefs{BinaryApplyAndReduce}{PairProduct}, we assume that the \texttt{CFLOBDD} or \texttt{Grouping} arguments are objects whose highest-level groupings are all at the same level.)
\begin{itemize}
  \item
    The operation \texttt{BinaryApplyAndReduce}, given as \algref{BinaryApplyAndReduce}, starts with a call on \texttt{PairProduct} (\lineref{BAAR:CallPairProduct}).
    \texttt{PairProduct}, given as \algrefs{PairProduct}{PairProductContinued}, performs a recursive traversal of the two \texttt{Grouping} arguments, \texttt{g1} and \texttt{g2}, to create a proto-CFLOBDD that represents a kind of cross product.
    \texttt{PairProduct} returns \texttt{g}, the proto-CFLOBDD formed in this way, as well as \texttt{pt}, a descriptor of the exit vertices of \texttt{g} in terms of pairs of exit vertices of the highest-level groupings of \texttt{g1} and \texttt{g2}.
    (See \algref{PairProduct},
    \lineseqref{PPNoDistinctionStart}{PPBothForkGroupingsEnd} and \algref{PairProductContinued}, \lineseqref{ptAnsDeclaration}{PPTabulateAnswer}.)
    
    \hspace*{1.5ex}
    From the semantic perspective, each exit vertex $e_1$ of \texttt{g1} represents a (non-empty) set $A_1$ of variable-to-Boolean-value
    assignments that lead to $e_1$ along a matched path in \texttt{g1};
    similarly, each exit vertex $e_2$ of \texttt{g2} represents a (non-empty) set of variable-to-Boolean-value assignments $A_2$ that lead to $e_2$ along a matched path in \texttt{g2}.
    If \texttt{pt}, the descriptor of \texttt{g}'s exit vertices returned by \texttt{PairProduct}, indicates that exit vertex $e$ of \texttt{g} corresponds to $[e_1,e_2]$, then $e$ represents the (non-empty) set of assignments $A_1 \cap A_2$.

    \hspace*{1.5ex}
    Function caching (\sectref{FunctionCaching}) is performed for \texttt{PairProduct}.
    Consequently, for a given invocation of \texttt{BinaryApplyAndReduce} on CFLOBDDs $n_1$ and $n_2$, for each level $k$, the number of calls on \texttt{PairProduct} for level $k$ is bounded by the product of the numbers of level-$k$ groupings in $n_1$ and $n_2$.
    Moreover, for each call on \texttt{PairProduct}($g_1$, $g_2$),
    the number of exit vertices in grouping $g$ is bounded by the product of the numbers of exit vertices in $g_1$ and $g_2$ (see \lineref{PPConcatenateExitPair}).
    Similarly, the number of middle vertices in $g$ is bounded by the product of the numbers of middle vertices in $g_1$ and $g_2$ (see \lineref{PPNumberOfBConnections}).
    Thus, the size of \texttt{g} is bounded by the product of the sizes of $g_1$ and $g_2$.
    Consequently, the cost of the call on \texttt{PairProduct} in \lineref{BAAR:CallPairProduct} of \algref{BinaryApplyAndReduce} is bounded by the sum over $k \in [0..l]$ of the products of the sizes of the level-$k$ groupings in $n_1$ and $n_2$, and hence polynomial in the sizes of $n_1$ and $n_2$ (see \footnoteref{ProductConstructionCost}).

    \hspace*{1.5ex}
    \Lineseqref{PPNoDistinctionStart}{PPNoDistinctionEnd} of \texttt{PairProduct} perform special-case processing when either argument to \texttt{PairProduct} is a NoDistinctionProtoCFLOBDD.
    At level $0$, these checks---along with \lineref{PPBothForkGroupingsEnd}---implement the base case of \texttt{PairProduct}.
    However, at levels greater than $0$, they allow \texttt{PairProduct} to return immediately, without making any recursive calls to traverse $g_1$ or $g_2$, potentially saving considerable work.

  \item
    \texttt{BinaryApplyAndReduce} then uses \texttt{pt}, together with \texttt{op} and the value tuples from \texttt{CFLOBDD}s \texttt{n1} and \texttt{n2}, to create the tuple \texttt{deducedValueTuple} of leaf values that should be associated with the exit vertices (see \algref{BinaryApplyAndReduce}, \lineref{BAAR:LeafValues}]).

    \hspace*{1.5ex}
    However, \texttt{deducedValueTuple} is a \emph{tentative} value tuple for the constructed \texttt{CFLOBDD};
    because of Structural Invariant~\ref{Inv:6}, this tuple needs to be collapsed if it contains duplicate values.

  \item
    \texttt{BinaryApplyAndReduce} obtains two tuples, \texttt{inducedValueTuple} and  \texttt{inducedReductionTuple}, which describe the collapsing of duplicate leaf values, by calling the subroutine \texttt{CollapseClassesLeftmost} (\algref{InducedTuples}):
    \begin{itemize}
      \item
        Tuple \texttt{inducedValueTuple} serves as the final value tuple for the \texttt{CFLOBDD} constructed by \texttt{BinaryApplyAndReduce}.
        In \texttt{inducedValueTuple}, the leftmost occurrence of a value in \texttt{deducedValueTuple} is retained as the representative for that equivalence class of values.
        For example, if \texttt{deducedValueTuple} is $[2,2,1,1,4,1,1]$, then \texttt{inducedValueTuple} is $[2,1,4]$.

        \hspace*{1.5ex}
        The use of leftward folding is dictated by Structural Invariant~\ref{Inv:2b}.
      \item
        Tuple {\tt inducedReductionTuple\/} describes the collapsing of duplicate
        values that took place in creating {\tt inducedValueTuple\/} from
        {\tt deducedValueTuple\/}:
        {\tt inducedReductionTuple\/} is the same length as {\tt deducedValueTuple},
        but each entry {\tt inducedReductionTuple(i)\/} gives the ordinal position
        of {\tt deducedValueTuple(i)\/} in {\tt inducedValueTuple}.
        For example, if {\tt deducedValueTuple\/} is $[2,2,1,1,4,1,1]$ (and thus
        {\tt inducedValueTuple\/} is $[2,1,4]$), then {\tt inducedReductionTuple\/}
        is $[1,1,2,2,3,2,2]$---meaning that positions 1 and 2 in {\tt deducedValueTuple\/}
        were folded to position 1 in {\tt inducedValueTuple}, positions 3, 4, 6, and 7
        were folded to position 2 in {\tt inducedValueTuple},
        and position 5 was folded to position 3 in {\tt inducedValueTuple}.
    \end{itemize}
    (See \algref{BinaryApplyAndReduce}, \lineref{BAAR:CollapseLeafValues}, as well as \algref{InducedTuples}.)
  \item
    Finally, {\tt BinaryApplyAndReduce\/} performs a corresponding reduction
    on {\tt Grouping g\/}, by calling the subroutine {\tt Reduce}, which
    creates a new {\tt Grouping\/} in which {\tt g\/}'s exit vertices
    are folded together with respect to tuple {\tt inducedReductionTuple} (\algref{BinaryApplyAndReduce}, \lineref{BAAR:CallReduce}).

    \hspace*{1.5ex}
    Procedure {\tt Reduce}, given as \algref{Reduce},
    recursively traverses {\tt Grouping g}, working in the backwards
    direction, first processing each of {\tt g\/}'s $B$-connections in turn,
    and then processing {\tt g\/}'s $A$-connection.
    In both cases, the processing is similar to the (leftward) collapsing
    of duplicate leaf values that is carried out by {\tt BinaryApplyAndReduce\/}:
    \begin{itemize}
      \item
        In the case of each $B$-connection, rather than collapsing with respect
        to a tuple of duplicate final values, {\tt Reduce\/}'s actions are
        controlled by its second argument, {\tt reductionTuple},
        which clients of {\tt Reduce\/}---namely, {\tt BinaryApplyAndReduce\/}
        and {\tt Reduce\/} itself---use to inform {\tt Reduce\/} how
        {\tt g\/}'s exit vertices are to be folded together.
        For instance, the value of {\tt reductionTuple} could be
        $[1,1,2,2,3,2,2]$---meaning that exit vertices 1 and 2 are to be
        folded together to form exit vertex 1, exit vertices 3, 4, 6, and 7
        are to be folded together to form exit vertex 2, and exit vertex 5
        by itself is to form exit vertex 3.

        \hspace*{1.5ex}
        In \algref{Reduce}, \lineref{AConnectionProcessingStart},
        the value of {\tt reductionTuple} is used to create a tuple that
        indicates the equivalence classes of targets of return edges for
        the $B$-connection under consideration (in terms of the new exit
        vertices in the {\tt Grouping\/} that will be created to replace
        {\tt g\/}).
        
        \hspace*{1.5ex}
        Then, by calling the subroutine {\tt CollapseClassesLeftmost},
        {\tt Reduce\/} obtains two tuples,
        \linebreak
        {\tt inducedReturnTuple\/} and 
        {\tt inducedReductionTuple}, that describe the collapsing that needs
        to be carried out on the exit vertices of the $B$-connection under
        consideration (\algref{Reduce}, \lineref{BConnectionCollapseEnd}).

        \hspace*{1.5ex}
        Tuple {\tt inducedReductionTuple} is used to make a recursive call on {\tt Reduce\/} to process the $B$-connection;
        {\tt inducedReturnTuple\/} is used as the return tuple for the {\tt Grouping\/} returned from that call.
        Note how the call on {\tt InsertBConnection\/} (\algref{InsertBConnection}) in \lineref{CallInsertBConnection} of {\tt Reduce\/} enforces structural invariant~\ref{Inv:4} of \defref{CFLOBDD}.\footnote{
          In our implementation, \texttt{InsertBConnection} performs a left-to-right search of \texttt{g.BConnection} and \texttt{g.BReturnTuples}, but it could be implemented as an (expected) unit-time operation using a hashed dictionary, keyed on (Grouping, ReturnTuple) pairs.
        }
      \item
        As the $B$-connections are processed, {\tt Reduce\/} uses the {\tt position\/}
        information returned from
        \linebreak
        {\tt InsertBConnection\/} to build up
        the tuple {\tt reductionTupleA}        
        (\algref{Reduce}, \lineref{ExtendReductionTupleA}).
        This tuple indicates how to reduce the $A$-connection of {\tt g}.
      \item
        Finally, via processing similar to what was done for each $B$-connection,
        two tuples are obtained that describe the collapsing that needs
        to be carried out on the exit vertices of the $A$-connection, and
        an additional call on {\tt Reduce\/} is carried out.
        (See \algref{Reduce}, \lineseqref{CallCollapseClassesLeftmost}{AConnectionProcessingEnd}.)
    \end{itemize}
\end{itemize}
Function caching (\sectref{FunctionCaching}) is performed for \texttt{Reduce}, with respect to both arguments \texttt{g} and \texttt{reductionTuple}.

\sectref{CostOfReduce} shows that the time for a call to $\text{Reduce}(n,\rt)$ with output CFLOBDD $n'$ is asymptotically bounded by $O(|n| \times |n'|)$.
Because the time for \texttt{PairProduct} to perform the product construction of two CFLOBDDs $n_1$ and $n_2$ is asymptotically bounded by the product of their sizes (i.e., $O(|n_1|\times |n_2|)$), the overall time to perform $\text{BinaryApplyAndReduce}(n_1, n_2)$ is $O(|n_1|\times |n_2| \times |n'|)$, which is polynomial in the sizes of the input and output CFLOBDDs.

Recall that a call on {\tt RepresentativeGrouping(g)\/} may have the side effect
of installing {\tt g\/} into the table of memoized {\tt Grouping\/}s.
We do not wish for this table to ever be polluted by non-well-formed
proto-CFLOBDDs.
Thus, there is a subtle point as to why the grouping {\tt g\/} constructed
during a call on {\tt PairProduct\/} meets
structural invariant~\ref{Inv:4}---and hence why it is permissible
to call {\tt RepresentativeGrouping(g)\/} in line~[\ref{Li:PPTabulateAnswer}]
of \algref{PairProductContinued}.
We give this proof in Appendix~\sectref{PairProductProof}.

Lastly, in the case of Boolean-valued CFLOBDDs, there are
$16$ possible binary operations, corresponding to the
$16$ possible two-argument truth tables ($2 \times 2$ matrices with
Boolean entries).
All $16$ possible binary operations are special cases of
{\tt BinaryApplyAndReduce}; these can be performed by
passing {\tt BinaryApplyAndReduce} an appropriate
value for argument {\tt op} (i.e., some $2 \times 2$ Boolean matrix).

\begin{example}\label{Exa:xor-cross-equiv}
  \figref{BinaryApplyAndReduceExample} illustrates how the CFLOBDD for $\lambda x_0, x_1 . T$ is created from the ``or'' ($\lor$) of the CFLOBDDs for $\lambda x_0, x_1 . x_0 \oplus x_1$ and $\lambda x_0, x_1 . x_0 \Leftrightarrow x_1$.
  \figref{BinaryApplyAndReduceExample}c is the result of calling \texttt{PairProduct} on the CFLOBDDs for $\lambda x_0, x_1 . x_0 \oplus x_1$ and $\lambda x_0, x_1 . x_0 \Leftrightarrow x_1$.
  After $\lor$ is applied to the values in each of the terminal-value pairs $[F,T]$ and $[T,F]$, we obtain a mock-CFLOBDD that has two exit vertices associated with terminal value $T$.
  To restore the structural invariants and create a CFLOBDD, the two exit vertices must be folded together, and a reduction performed on each of the two B-connections.
  In each case, \texttt{Reduce} is called with \texttt{reductionTuple} $[1,1]$.
  Because these reductions result in the same B-connection proto-CFLOBDDs with identical return edges (\figref{BinaryApplyAndReduceExample}e and \figref{BinaryApplyAndReduceExample}f), which would be discovered by \texttt{InsertBConnection} (\algref{InsertBConnection}), it is necessary to fold together the two middle vertices and perform a reduction on the A-connection:
  \texttt{Reduce} is called with \texttt{reductionTuple} $[1,1]$,
  This step produces the CFLOBDD for $\lambda x_0, x_1 . T$ (\figref{BinaryApplyAndReduceExample}g). 
\end{example}

For another example that illustrates \texttt{Reduce}, see \exref{ReduceTwoLevel}.

\paragraph{Remark}
For BDDs, the two-step process of ``pair-product-followed-by-reduction'' need only be conceptual.
Binary operations on BDDs can be implemented during a single recursive pass by performing the appropriate value-reduction operation on terminal values, and then, as the recursion unwinds, having the BDD-node constructor perform hash-consing (suppressing the construction of don't-care nodes) so that non-reduced structures are never created \cite[\S3.3]{DBLP:conf/ml/FilliatreC06}.

Such an approach does not seem to be possible with CFLOBDDs because reduction is not obtained as a side-effect of hash-consing.
The flow of control in \texttt{Reduce} (\algref{Reduce}) follows the sequence of elements of a matched path backwards.
\texttt{Reduce} makes recursive calls for the B-connection proto-CFLOBDDs and then a recursive call for the A-connection proto-CFLOBDD (rather than working bottom-up from level-$0$ groupings to level-$k$ groupings, which would be the analogue of the bottom-up construction performed with BDDs.)
Consequently, our CFLOBDD implementation maintains the weaker invariant that the \texttt{Grouping}s that appear in the hash-consing tables are the heads of fully-fledged proto-CFLOBDDs, not mock-proto-CFLOBDDs---i.e., structural invariants (\ref{Inv:1})--(\ref{Inv:4}) of \defref{CFLOBDD} hold.
While such \texttt{Grouping}s may have to be reduced later, there is never any issue of the hash-cons tables being polluted by mock-proto-CFLOBDDs that violate the proto-CFLOBDD structural invariants.

\smallskip
Some unary operations on CFLOBDDs may also need to apply \texttt{Reduce}.
For example, if the terminal values of a CFLOBDD are numeric values, the unary function that squares all terminal values could initially result in a mock-CFLOBDD that has duplicate terminal values.
\texttt{Reduce}, with an appropriate \texttt{ReductionTuple}, would be then applied to create the corresponding CFLOBDD.

In a manner similar to the binary operations on CFLOBDDs, we can perform ternary operations on CFLOBDDs.
More details about how to perform these operations can be found in Appendix \sectref{ternary-op}.
Other operations, such as restriction ($f|_{x_i = v}$) and existential quantification ($\exists x_i.f$) can also be performed on a CFLOBDD;
the corresponding algorithms can be found in Appendices \sectref{restrict-op} and \sectref{existential-op}, respectively.

\subsection{Representing Matrices and Vectors using CFLOBDDs}
Matrices and vectors are important data structures used in quantum simulation (\sectref{ResearchQuestionTwo}).
In this subsection and following subsections, we discuss how to represent Boolean or non-Boolean matrices and vectors using CFLOBDDs,
and how to perform operations on such representations of matrices and vectors.

\subsubsection{Matrix Representation}
\label{Se:matrix-rep}

We represent square matrices using CFLOBDDs by having the Boolean variables correspond to bit positions in the
row and column indices.
That is, suppose that $M$ is a $2^n \times 2^n$ matrix;
$M$ is represented using a CFLOBDD over $2n$ Boolean variables $\{x_0,x_1,\ldots,x_{n-1}\} \cup \{y_0,y_1,\ldots,y_{n-1}\}$, where the variables $\{x_0,x_1,\ldots,x_{n-1}\}$ represent the successive bits of $x$---the first index into $M$---and the variables $\{y_0,y_1,\ldots,y_{n-1}\}$ represent the successive bits of
$y$---the second index into $M$, with $\log{n} + 1$ levels.\footnote{
  Matrices of other sizes, including non-square matrices, can be represented by embedding them within a larger square matrix.
  For matrices with >2 dimensions, there would be a set of Boolean variables for the index-bits of each dimension.
}
The indices of elements of matrices represented in this way
start at 0;
for example, the upper-left corner element of a matrix $M$
is $M(0,0)$.
When $n = 2$, $M(0,0)$ corresponds to the value
associated with the assignment $[x_0 \mapsto 0, x_1 \mapsto 0, y_0
\mapsto 0, y_1 \mapsto 0]$.

It is often convenient to use either the {\em interleaved\/} ordering
i.e., the order of the Boolean variables is
chosen to be $x_0,y_0,x_1,y_1,\ldots,x_{n-1},y_{n-1}$---or the {\em
reverse interleaved\/} ordering---i.e., the order is
$y_{n-1},x_{n-1},y_{n-2},y_{n-2},\ldots,y_0,x_0$.

One nice property of the interleaved-variable ordering is that, as we
work through each pair of variables in an assignment,
the matrix elements that remain ``in play''
represent a sub-block of the full matrix.
For instance, suppose that we have a Boolean matrix whose entries are defined by the function $\lambda {x_0}{y_0}{x_1}{y_1}.({x_0} \land
{y_0}) \lor ({x_1} \land {y_1})$, as shown below:
\[
\begin{array}{ccrrrrr}
          &   & 0 & 0 & 1 & 1 & y_0 \\
          &   & 0 & 1 & 0 & 1 & y_1 \\
                \cline{3-6}
        0 & 0 & \multicolumn{1}{|r}{F} & \multicolumn{1}{r|}{F} & F & \multicolumn{1}{r|}{F} & \\
        0 & 1 & \multicolumn{1}{|r}{F} & \multicolumn{1}{r|}{T} & F & \multicolumn{1}{r|}{T} & \\
                \cline{3-6}
        1 & 0 & \multicolumn{1}{|r}{F} & \multicolumn{1}{r|}{F} & T & \multicolumn{1}{r|}{T} & \\
        1 & 1 & \multicolumn{1}{|r}{F} & \multicolumn{1}{r|}{T} & T & \multicolumn{1}{r|}{T} & \\
                \cline{3-6} 
      x_0 & x_1 & & & & &
\end{array}
\]
If we were to evaluate the 16 possible assignments in lexicographic
order, i.e., in the order
\begin{center}
\begin{tabular}{l}
  $[{x_0} \mapsto 0, {y_0} \mapsto 0, {x_1} \mapsto 0, {y_1} \mapsto 0]$, \\
  $[{x_0} \mapsto 0, {y_0} \mapsto 0, {x_1} \mapsto 0, {y_1} \mapsto 1]$, \\
  $[{x_0} \mapsto 0, {y_0} \mapsto 0, {x_1} \mapsto 1, {y_1} \mapsto 0]$, \\
  $[{x_0} \mapsto 0, {y_0} \mapsto 0, {x_1} \mapsto 1, {y_1} \mapsto 1]$, \\
  $[{x_0} \mapsto 0, {y_0} \mapsto 1, {x_1} \mapsto 0, {y_1} \mapsto 0]$, \\
  \multicolumn{1}{c}{$\vdots$} \\
  $[{x_0} \mapsto 1, {y_0} \mapsto 1, {x_1} \mapsto 1, {y_1} \mapsto 0]$, \\
  $[{x_0} \mapsto 1, {y_0} \mapsto 1, {x_1} \mapsto 1, {y_1} \mapsto 1]$
\end{tabular}
\end{center}
then we would step through the array elements in the order shown
below:
\[
\begin{array}{ccrrrrr}
         &   & 0 & 0 & 1 & 1 & y_0 \\
         &   & 0 & 1 & 0 & 1 & y_1 \\
               \cline{3-6}
       0 & 0 & \multicolumn{1}{|r}{1} & \multicolumn{1}{r|}{2} & 5 & \multicolumn{1}{r|}{6} & \\
       0 & 1 & \multicolumn{1}{|r}{3} & \multicolumn{1}{r|}{4} & 7 & \multicolumn{1}{r|}{8} & \\
               \cline{3-6}
       1 & 0 & \multicolumn{1}{|r}{9} & \multicolumn{1}{r|}{10} & 13 & \multicolumn{1}{r|}{14} & \\
       1 & 1 & \multicolumn{1}{|r}{11} & \multicolumn{1}{r|}{12} & 15 & \multicolumn{1}{r|}{16} & \\
               \cline{3-6}
     x_0 & x_1 & & & & &
\end{array}
\]
If the first two elements of an assignment are $[x_0 \mapsto 0, y_0 \mapsto 1]$, the elements still in play are the ones in the positions labeled 5, 6, 7, and 8 in the upper-right quadrant.

\begin{figure}
    \centering
    \includegraphics[width=.75\linewidth]{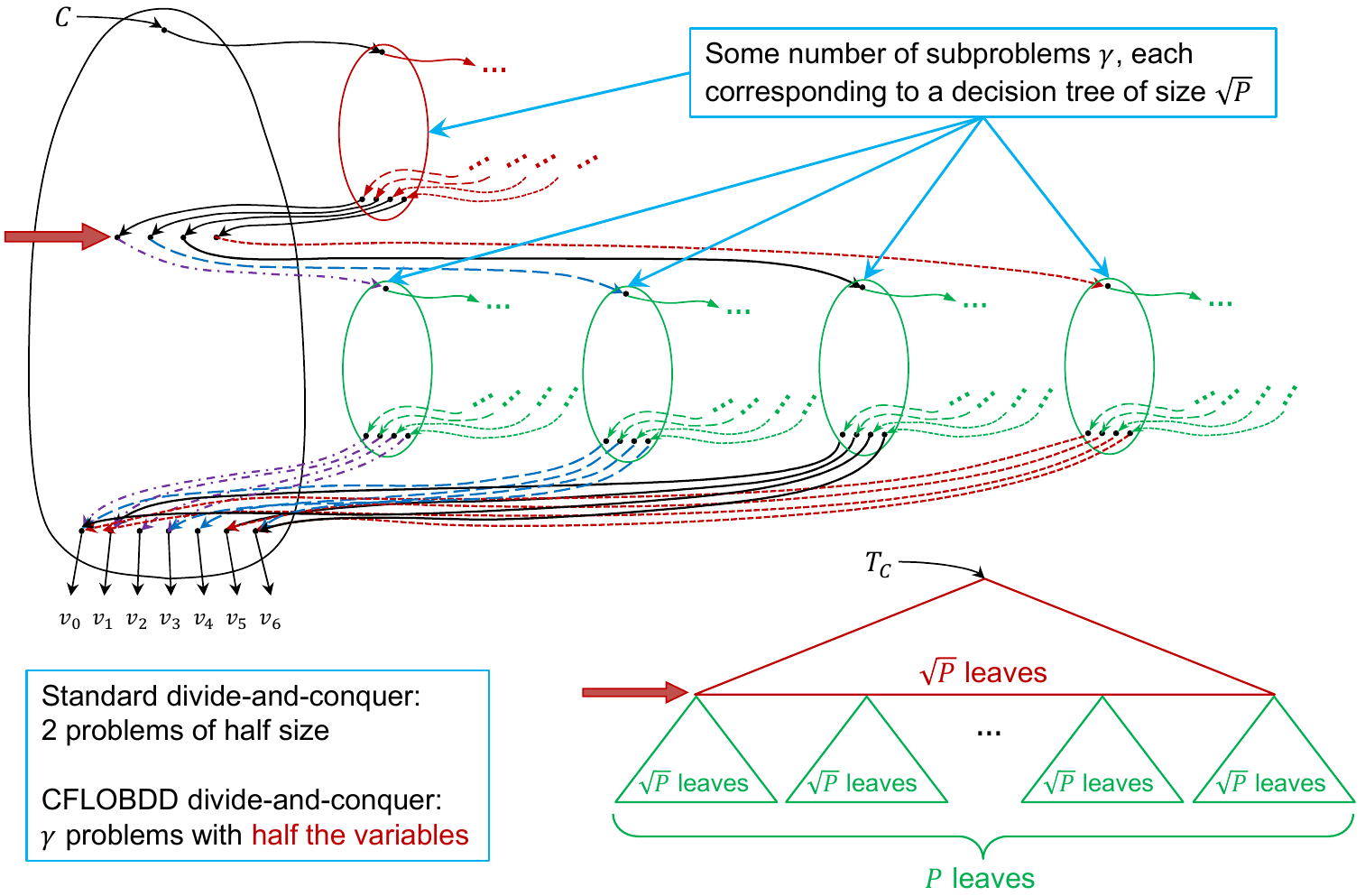}
    \caption{
      Why a $\sqrt{P} \times \sqrt{P}$-decomposition is the natural problem decomposition for divide-and-conquer algorithms on structures represented as CFLOBDDs.
    }
    \label{Fi:RootNDivideAndConquer}
\end{figure}

There is an important, non-standard consequence of using a CFLOBDD to represent a matrix that very likely is not apparent from the discussion above,
having to do with the sizes of subproblems in a divide-and-conquer algorithm.
In fact, the same issue arises in designing a divide-and-conquer algorithm over any data structure represented via a CFLOBDD, as illustrated in \figref{RootNDivideAndConquer}.
Suppose that a CFLOBDD $C$ represents a decision tree $T_C$ that has $\log_2 P$ variables, and thus $P$ leaves.
The A-connection proto-CFLOBDD accounts for half the variables, namely, $\frac{\log_2 P}{2}$, and the B-connection proto-CFLOBDDs account for the remaining half.
The natural way to divide $C$ in a divide-and-conquer algorithm is at the middle vertices of the outermost grouping: process the A-connection proto-CFLOBDD, and then the B-connection proto-CFLOBDDs (or vice versa).
In $T_C$, this division corresponds to the tree partitioning shown in the lower-right corner of \figref{RootNDivideAndConquer}: $C$'s A-connection proto-CFLOBDD corresponds to the {\color{red} red tree} rooted at the apex of $T_C$ (which has $\sqrt{P}$ leaves);
$C$'s B-connection proto-CFLOBDDs correspond to the
$\sqrt{P}$ {\color{SeaGreen} green trees} in the bottom half of $T_C$ (each of which has $\sqrt{P}$ leaves).
In contrast with standard divide-and-conquer algorithms, which often divide a problem into two subproblems of half size, this approach divides the original problem into $\sqrt{P} + 1$ subproblems, each of size $O(\sqrt{P})$.
With CFLOBDDs, in contrast to decision trees, there is the potential for subproblems to be shared among the A-connection and B-connections, so one ends up with some number of subproblems $\gamma$ ($= 1 + \textrm{the number of middle vertices}$), each of size $O(\sqrt{P})$.

Whereas with a decision tree it would be easy to take the conventional approach of dividing a problem into two problems of half \emph{size}---using the left child and right child of the apex, in essence ``peeling off'' the topmost ply, such a division is not convenient for CFLOBDDs because the decision variable for the topmost ply is associated with the level-0 grouping found by following the A-connection of the A-connection of the $\ldots$, etc.
For CFLOBDDs, the natural structure of a divide-and-conquer algorithm lies with the A-connection proto-CFLOBDD and the set of B-connection proto-CFLOBDDs---a division based on dividing \emph{the number of variables} in half.

In certain cases, including matrix multiplication (\sectref{matrix-mult}), the $\gamma \times \sqrt{P}$-decomposition structure forced us to rethink how to perform various algorithms.

Let us now consider how such a decomposition works for an $N \times N$ matrix $M$, assuming the interleaved-variable ordering, where $N = 2^n$.
Thus, $n$ is the number of bits in a row-index (respectively, column-index);
there are $2n$ Boolean variables in total;
and $P = N^2$.
$M$ would be decomposed into $\sqrt{P} = \sqrt{N^2} = N$ sub-matrices, each of size $\sqrt{N} \times \sqrt{N}$.
At top level, the A-connection of the CFLOBDD for $M$ captures commonalities in the $\sqrt{N} \times \sqrt{N}$ block structure of $M$, and the B-connections represent the blocks: sub-matrices of $M$ of size $\sqrt{N} \times \sqrt{N}$.

For instance, when a level-$3$ CFLOBDD is used to represent a matrix, there are $2n = 8 = 2^3$ index variables---i.e., $n = 4$ variables for each dimension---so the matrix is of size $16 \times 16$.
Its natural constituents are level-$2$ proto-CFLOBDDs, which each have $2^2 = 4$ index variables.
Thus, there are $2$ A-connection variables for each dimension of the block structure, and $2$ B-connection variables for each dimension of the sub-matrix for a block.
Consequently, a matrix of size $16 \times 16$ is decomposed into $16$ ($= 4 \times 4 = \sqrt{16} \times \sqrt{16}$) blocks, each of size $4 \times 4 = \sqrt{16} \times \sqrt{16}$, as indicated below:

\noindent
{\small
\begin{equation*}
\begin{array}{ccccrrrrrrrrrrrrrrrrr}
          &   &   &   & 0 & 0 & 0 & 0 & 0 & 0 & 0 & 0 & 1 & 1 & 1 & 1 & 1 & 1 & 1 & 1 & y_0 \\
          &   &   &   & 0 & 0 & 0 & 0 & 1 & 1 & 1 & 1 & 0 & 0 & 0 & 0 & 1 & 1 & 1 & 1 & y_1 \\
          &   &   &   & 0 & 0 & 1 & 1 & 0 & 0 & 1 & 1 & 0 & 0 & 1 & 1 & 0 & 0 & 1 & 1 & y_2 \\
          &   &   &   & 0 & 1 & 0 & 1 & 0 & 1 & 0 & 1 & 0 & 1 & 0 & 1 & 0 & 1 & 0 & 1 & y_3 \\
                \cline{5-20}
        0 & 0 & 0 & 0 & \multicolumn{1}{|r}{F} & F & F & \multicolumn{1}{r|}{F} & F & F & F & \multicolumn{1}{r|}{F} & F & F & F & \multicolumn{1}{r|}{F} & F & F & F & \multicolumn{1}{r|}{F} & \\
        0 & 0 & 0 & 1 & \multicolumn{1}{|r}{F} & T & F & \multicolumn{1}{r|}{T} & F & F & F & \multicolumn{1}{r|}{F} & F & F & F & \multicolumn{1}{r|}{F} & F & F & F & \multicolumn{1}{r|}{F} & \\
        0 & 0 & 1 & 0 & \multicolumn{1}{|r}{F} & F & T & \multicolumn{1}{r|}{T} & F & F & F & \multicolumn{1}{r|}{F} & F & F & F & \multicolumn{1}{r|}{F} & F & F & F & \multicolumn{1}{r|}{F} & \\
        0 & 0 & 1 & 1 & \multicolumn{1}{|r}{F} & T & T & \multicolumn{1}{r|}{T} & F & F & F & \multicolumn{1}{r|}{F} & F & F & F & \multicolumn{1}{r|}{F} & F & F & F & \multicolumn{1}{r|}{F} & \\
                \cline{5-20}
        0 & 1 & 0 & 0 & \multicolumn{1}{|r}{F} & F & F & \multicolumn{1}{r|}{F} & F & F & F & \multicolumn{1}{r|}{F} & F & F & F & \multicolumn{1}{r|}{F} & F & F & F & \multicolumn{1}{r|}{F} & \\
        0 & 1 & 0 & 1 & \multicolumn{1}{|r}{F} & F & F & \multicolumn{1}{r|}{F} & F & T & F & \multicolumn{1}{r|}{T} & F & F & F & \multicolumn{1}{r|}{F} & F & F & F & \multicolumn{1}{r|}{F} & \\
        0 & 1 & 1 & 0 & \multicolumn{1}{|r}{F} & F & F & \multicolumn{1}{r|}{F} & F & F & T & \multicolumn{1}{r|}{T} & F & F & F & \multicolumn{1}{r|}{F} & F & F & F & \multicolumn{1}{r|}{F} & \\
        0 & 1 & 1 & 1 & \multicolumn{1}{|r}{F} & F & F & \multicolumn{1}{r|}{F} & F & T & T & \multicolumn{1}{r|}{T} & F & F & F & \multicolumn{1}{r|}{F} & F & F & F & \multicolumn{1}{r|}{F} & \\
                \cline{5-20}
        1 & 0 & 0 & 0 & \multicolumn{1}{|r}{F} & F & F & \multicolumn{1}{r|}{F} & F & F & F & \multicolumn{1}{r|}{F} & F & F & F & \multicolumn{1}{r|}{F} & F & F & F & \multicolumn{1}{r|}{F} & \\
        1 & 0 & 0 & 1 & \multicolumn{1}{|r}{F} & F & F & \multicolumn{1}{r|}{F} & F & F & F & \multicolumn{1}{r|}{F} & F & T & F & \multicolumn{1}{r|}{T} & F & F & F & \multicolumn{1}{r|}{F} & \\
        1 & 0 & 1 & 0 & \multicolumn{1}{|r}{F} & F & F & \multicolumn{1}{r|}{F} & F & F & F & \multicolumn{1}{r|}{F} & F & F & T & \multicolumn{1}{r|}{T} & F & F & F & \multicolumn{1}{r|}{F} & \\
        1 & 0 & 1 & 1 & \multicolumn{1}{|r}{F} & F & F & \multicolumn{1}{r|}{F} & F & F & F & \multicolumn{1}{r|}{F} & F & T & T & \multicolumn{1}{r|}{T} & F & F & F & \multicolumn{1}{r|}{F} & \\
                \cline{5-20}
        1 & 1 & 0 & 0 & \multicolumn{1}{|r}{F} & F & F & \multicolumn{1}{r|}{F} & F & F & F & \multicolumn{1}{r|}{F} & F & F & F & \multicolumn{1}{r|}{F} & F & F & F & \multicolumn{1}{r|}{F} & \\
        1 & 1 & 0 & 1 & \multicolumn{1}{|r}{F} & F & F & \multicolumn{1}{r|}{F} & F & F & F & \multicolumn{1}{r|}{F} & F & F & F & \multicolumn{1}{r|}{F} & F & T & F & \multicolumn{1}{r|}{T} & \\
        1 & 1 & 1 & 0 & \multicolumn{1}{|r}{F} & F & F & \multicolumn{1}{r|}{F} & F & F & F & \multicolumn{1}{r|}{F} & F & F & F & \multicolumn{1}{r|}{F} & F & F & T & \multicolumn{1}{r|}{T} & \\
        1 & 1 & 1 & 1 & \multicolumn{1}{|r}{F} & F & F & \multicolumn{1}{r|}{F} & F & F & F & \multicolumn{1}{r|}{F} & F & F & F & \multicolumn{1}{r|}{F} & F & T & T & \multicolumn{1}{r|}{T} & \\
                \cline{5-20} 
x_0 & x_1 & x_2 & x_3 &   &   &   &   &   &   &   &   &   &   &   &   &   &   &   &   & \\
\end{array}
\end{equation*}
}
With level-$4$ CFLOBDDs, one has $n = 8$ variables for each dimension in the full-size matrix.
Thus, there are $4$ A-connection variables for each dimension of the block structure, and $4$ B-connection variables for each dimension of the sub-matrix for a block.
Consequently, a matrix of size $256 \times 256$ is decomposed into $256$ ($= 16 \times 16 = \sqrt{256} \times \sqrt{256}$) blocks, each of size $16 \times 16 = \sqrt{256} \times \sqrt{256}$.

In general, an $N \times N$ matrix is decomposed according to its $\sqrt{N} \times \sqrt{N}$ block structure, where each block is of size $\sqrt{N} \times \sqrt{N}$.
With CFLOBDDs, one hopes that many of the blocks are shared among the B-connections (and possibly some blocks are even structurally similar to the block structure itself, represented by the A-connection), so that one ends up with some---hopefully small---number of subproblems $\gamma$, each of size $\sqrt{N} \times \sqrt{N}$.

The CFLOBDD decomposition discussed above is different from (i) the natural decomposition of a matrix represented via a BDD, and (ii) the decomposition used in most divide-and-conquer algorithms on matrices.
Both (i) and (ii) use $\frac{n}{2} \times \frac{n}{2}$-decompositions (and thus decompose a matrix of size $16 \times 16$ into $4$ sub-matrices, each of size $8 \times 8$, and decompose a matrix of size $256 \times 256$ into $4$ sub-matrices, each of size $128 \times 128$).

\subsubsection{Vector Representation}
\label{Se:VectorRepresentation}

A vector can be represented via a CFLOBDD in a manner that is similar to, but simpler than, the way matrices are represented.
A vector of size $2^n \times 1$ can be represented by a CFLOBDD whose highest level is $\log{n}$.
Suppose that $V$ is a $2^n \times 1$ vector;
a CFLOBDD representing $V$ would have $n$ Boolean variables $\{x_0,x_1,\ldots,x_{n-1}\}$ with the variables $\{x_0,x_1,\ldots,x_{n-1}\}$ representing the successive bits of $x$---the index into $V$.\footnote{
  Similar to matrices, vectors of other sizes can be
  represented using CFLOBDDs;
  for instance, they can be embedded within a larger vector whose dimensions are of the form $2^n \times 1$.
}

We typically use either the increasing variable ordering or decreasing variable ordering to represent vectors.
(Similar to matrices, vectors of other sizes can be embedded within a larger vector of the form $2^n \times 1$.)
For example, if
we have a vector whose entries are defined by $\lambda x_0x_1.(x_0 \land x_1)$, the vector would be as follows:
\[
\begin{array}{ccr}
          \cline{3-3}
        0 & 0 & \multicolumn{1}{|r|}{F}\\
        0 & 1 & \multicolumn{1}{|r|}{F}\\
        1 & 0 & \multicolumn{1}{|r|}{F}\\
        1 & 1 & \multicolumn{1}{|r|}{T}\\
                \cline{3-3} 
      x_0 & x_1 & 
\end{array}
\]
This ordering helps in easily converting a vector of size $2^n \times 1$ to a matrix of size $2^n \times 2^n$ with interleaved-variable ordering (with the vector embedded into the matrix as the first column, and the rest of the entries set to zero).

\subsection{Kronecker Product}
\label{Se:KroneckerProduct}

When using CFLOBDDs to represent matrices on which Kronecker products will be performed, we typically use the interleaved-variable ordering.
In this section, we describe two variants of Kronecker product that result in different interleavings of the index variables of the argument matrices.

\subsubsection{Variant 1}
\label{Se:KroneckerProduct:VariantOne}

\begin{figure}
\centering
\begin{tabular}{c@{\hspace{.8in}}c}
  \begin{tabular}[b]{c}
    \includegraphics[height=1.7in]{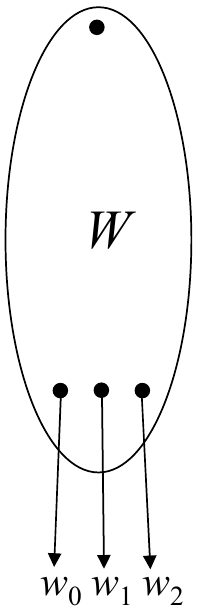}
    \\
    \\
    \includegraphics[height=1.7in]{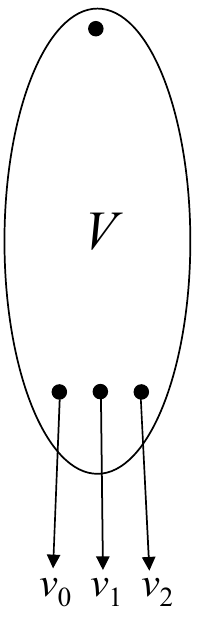}
  \end{tabular}
  &
  \begin{tabular}[b]{c}
    \includegraphics[height=3.7in]{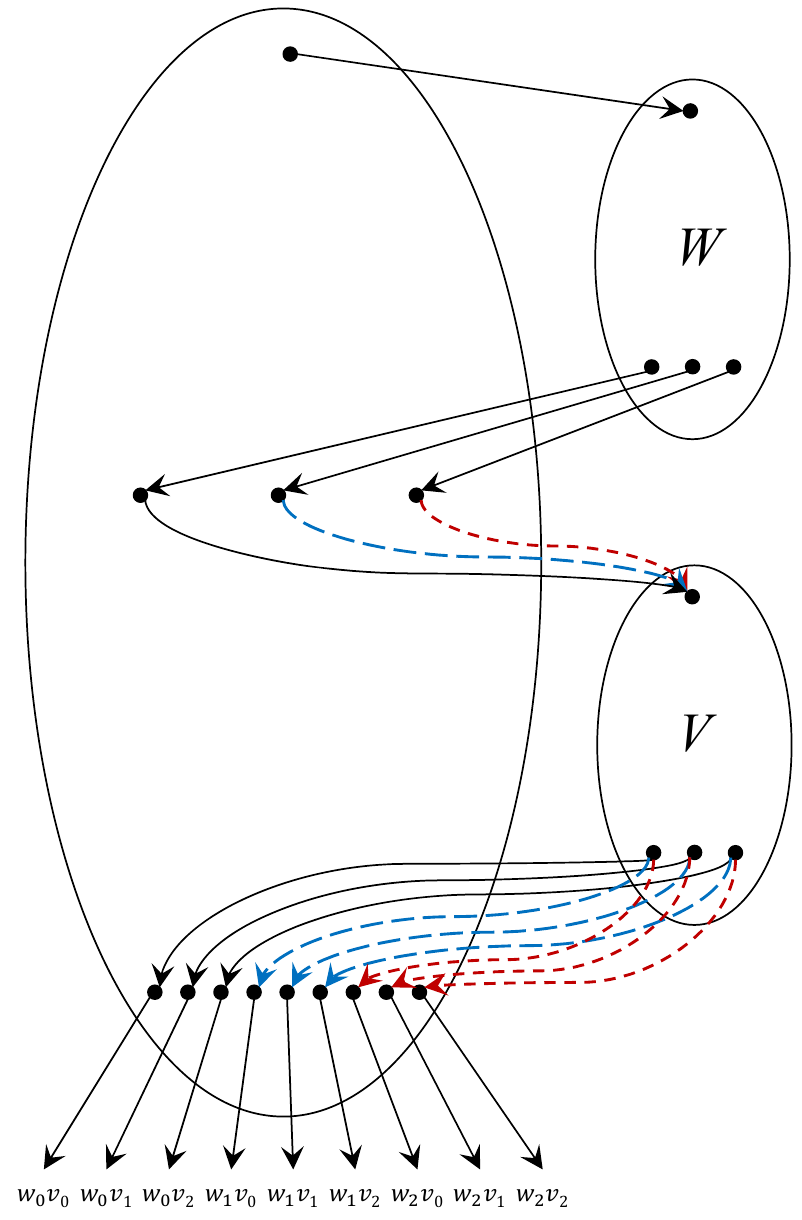}
  \end{tabular}
\end{tabular}
\caption{\protect \raggedright 
  (a) and (b) Level-$k$ CFLOBDDs for arrays $W$ and $V$, respectively;
  (c) level-$k+1$ CFLOBDD for $W \tensor V$.        
}
\label{Fi:KroneckerProduct}
\end{figure}

Suppose that matrices $W$ and $V$ are represented by level-$k$ CFLOBDDs with value tuples $[w_0,\ldots,w_m]$ and $[v_0,\ldots,v_n]$, respectively.
To create the CFLOBDD for $W \tensor V$,
\begin{enumerate}
  \item
    \label{It:ShiftA}
    Create a level $k+1$ grouping that has $m+1$ middle vertices,
    corresponding to the values $[w_0,\ldots,w_m]$, and $(m+1)(n+1)$
    exit vertices, corresponding to the terminal values
    \[
      [{w_i}{v_j} : i \in [0..m], j \in [0..n]],
    \]
    where the terminal values are ordered lexicographically by their $(i,j)$ indexes;
    i.e., $w_0 v_0$, $w_0 v_1$, $\ldots$, $w_m v_{n-1}$, $w_m v_n$.
    The grouping's $A$-connection is the proto-CFLOBDD of $W$, with return edges that map the $i^{\textit{th}}$ exit vertex to middle vertex $w_i$.
  \item
    \label{It:ShiftB}
    For each middle vertex, which corresponds to some value $w_i$,
    $0 \leq i \leq m$, create a $B$-connection to the proto-CFLOBDD of $V$, with return edges that map the $j^{\textit{th}}$ exit vertex to the exit vertex of the level $k+1$ grouping that corresponds to the value ${w_i}{v_j}$.
  \item
    \label{It:ReduceABKroneckerProduct}
    If any of the values in the sequence $[{w_i}{v_j} : i \in [0..m], j \in [0..n]]$ are duplicates, make an appropriate call on {\tt Reduce\/} to fold together the classes of exit vertices that are  associated with the same value, thereby creating a canonical multi-terminal CFLOBDD.
\end{enumerate}
The construction through step (\ref{It:ShiftB}) is illustrated in \figref{KroneckerProduct}.
Pseudo-code for the algorithm in presented in Appendix \sectref{kronecker--product}.

With this algorithm, if $x \bowtie y$ represents the variable ordering of $W$ and $w \bowtie z$ represents the variable ordering of $V$ (where $\bowtie$ denotes the operation to interleave two variable orderings), then $W \tensor V$ has the variable ordering $(x || w) \bowtie (y || z)$ (where $||$ denotes the concatenation of two sequences of variables).

\begin{algorithm}[tb!]
\Input{Grouping g with variable ordering $x \bowtie y$}
\Output{Grouping g' is equal to g with dummy variables $w'$ and $z'$ such that the variable ordering is $((x \bowtie w') \bowtie (y \bowtie z'))$}
\caption{ShiftToAConnectionAtLevelOne\label{Fi:KPA}}
\Begin{
    InternalGrouping g' = new InternalGrouping(g.level + 1)\;
    \eIf{g.level == 1}{
        g'.AConnection = g\;
        g'.AReturnTuple = [1..g.numberOfExits|]\;
        g'.numberOfBConnections = |g'.AReturnTuple|\;
        \For{$i \leftarrow 1$ \KwTo $g'.numberOfBConnections$}{
            g'.BConnection[i] = NoDistinctionProtoCFLOBDD(g.level)\;
            g'.BReturnTuples[i] = [i]\;
        }
        g'.numberOfExits = g.numberOfExits\;
    }
    {
        g'.AConnection = ShiftToAConnectionAtLevelOne(g.AConnection)\;
        g'.AReturnTuple = g.AReturnTuple\;
        g'.numberOfBConnections = |g.AReturnTuple|\;
        \For{$i \leftarrow 1$ \KwTo $g'.numberOfBConnections$}{
            g'.BConnection[i] = ShiftToAConnectionAtLevelOne(g.BConnection[i])\;
            g'.BReturnTuples[i] = g.BReturnTuples[i]\;
        }
        g'.numberOfExits = g.numberOfExits\;
    }
    \Return RepresentativeGrouping(g')\;
}
\end{algorithm}

\begin{algorithm}[tb!]
\Input{Grouping g with variable ordering $w \bowtie z$}
\Output{Grouping g' is equal to g with dummy variables $x'$ and $y'$ such that the variable ordering is $((x' \bowtie w) \bowtie (y' \bowtie z))$}
\caption{ShiftToBConnectionAtLevelOne\label{Fi:KPB}}
\Begin{
    InternalGrouping g' = new InternalGrouping(g.level + 1)\;
    \eIf{g.level == 1}{
        g'.AConnection = NoDistinctionProtoCFLOBDD(g.level)\;
        g'.AReturnTuple = [1]\;
        g'.numberOfBConnections = 1\;
        g'.BConnection[1] = ShiftToBConnectionAtLevelOne(g.level)\;
        g'.BReturnTuples[1] = g.BReturnTuples[1]\;
        g'.numberOfExits = g.numberOfExits\;
    }
    {
        g'.AConnection = ShiftToAConnectionAtLevelOne(g.AConnection)\;
        g'.AReturnTuple = g.AReturnTuple\;
        g'.numberOfBConnections = |g.AReturnTuple|\;
        \For{$i \leftarrow 1$ \KwTo $g'.numberOfBConnections$}{
            g'.BConnection[i] = ShiftToAConnectionAtLevelOne(g.BConnection[i])\;
            g'.BReturnTuples[i] = [1..g.numberOfExits]\;
        }
        g'.numberOfExits = g.numberOfExits\;
    }
    \Return RepresentativeGrouping(g')\;
}
\end{algorithm}

\subsubsection{Variant 2}
\label{Se:KroneckerProduct:VariantTwo}

There is a second way to perform a Kronecker product of $W$ and $V$ that results in a representation of $W \tensor V$ that has the variable ordering $(x \bowtie w) \bowtie (y \bowtie z)$.
(As discussed in~\sectref{simonsalgo}, this version of Kronecker product is useful in Simon's Algorithm.)
The steps for $W \tensor V$ are as follows:
\begin{itemize}
  \item
    For the CFLOBDD that represents matrix $W$, for every grouping of $W$ from the top-level grouping down to level 2, create a copy of the grouping at one level greater.
    At level 1, create a level-$2$ grouping in which (i) the A-connection is the current level-$1$ grouping, and (ii) the B-connections are all the level-$1$    no-distinction proto-CFLOBDD (\figref{NoDistinctionProtoCFLOBDD}(b)).
    In essence, this step adds dummy variables that are proxies for matrix $V$'s variables.
    \algref{KPA} shows the algorithm for this operation.
  \item
    For the CFLOBDD that represents matrix $V$, for every grouping of $V$ from the top-level grouping down to level 2, create a copy of the grouping at one level greater.
    At level 1, create a level-$2$ grouping in which (i) the A-connection is the level-$1$ no-distinction proto-CFLOBDD (\figref{NoDistinctionProtoCFLOBDD}(b)), and (ii) the B-connection is the current level-$1$ grouping.
    In essence, this step adds dummy variables that are proxies for matrix $W$'s variables.
    \algref{KPB} shows the algorithm for this operation.
  \item
    Finally, combine the two newly constructed CFLOBDDs by making a call to BinaryApplyAndReduce (\algref{BinaryApplyAndReduce}), with ``Times'' (multiplication) as the operation to apply to terminal values.
\end{itemize}
Pseudo-code for this algorithm is given as \algref{KP4Voc}.

\begin{algorithm}[tb!]
\Input{CFLOBDDs n1, n2 with variable ordering of n1: $x \bowtie y$ and n2: $w \bowtie z$}
\Output{CFLOBDD n = n1 $\otimes$ n2 with variable ordering of n: $(x || w) \bowtie (y || z)$}
\caption{Kronecker Product \label{Fi:KP4Voc}}
\Begin{
    \tcp{Add dummy variables $\langle w'_1, \cdots, w'_{n} \rangle$, $\langle z'_1, \cdots, z'_{n} \rangle$ such that variable ordering of $g1$ is $(x \bowtie w') \bowtie (y \bowtie z')$}
    CFLOBDD g1 = RepresentativeCFLOBDD(ShiftToAConnectionAtLevelOne(n1.grouping), n1.valueTuple)\;
    \tcp{Add dummy variables $\langle x'_1, \cdots, x'_{n} \rangle$, $\langle y'_1, \cdots, y'_{n} \rangle$ such that variable ordering of $g2$ is $(x' \bowtie w) \bowtie (y' \bowtie z)$}
    CFLOBDD g2 = RepresentativeCFLOBDD(ShiftToBConnectionAtLevelOne(n2.grouping), n2.valueTuple)\;
    CFLOBDD n = BinaryApplyAndReduce(g1, g2, (op)Times)\;
    \Return n\;
}
\end{algorithm}

\subsection{Vector-to-Matrix Conversion}
\label{Se:VectorToMatrixConversion}

\Omit{
In principle, we wouldn't need this specialized algorithm if we had a notion of a row vector.
Suppose column vector $V$ has variable order $v$, and that we can create the row vector $e_{0 \ldots 0}$ with variable order $w$.
Then $V \tensor e_{0 \ldots 0}$ produces a square matrix of size $2^n\times2^n$, which has $V$ in the first column and the remaining columns filled with zeros.
We would have to have the appropriate version of $\tensor$ so that the result has the variable ordering $v \bowtie w$,
}

An important operation that is repeatedly performed in the algorithms summarized in \sectref{QuantumAlgorithms} is vector-matrix multiplication.
Our approach to vector-matrix multiplication is to convert a vector $V$ of size $2^n \times 1$ into a matrix $M$ of size $2^n\times2^n$, where $V$ occupies the first column, and all other entries of $M$ are $0$.
We can then use the matrix-matrix multiplication algorithm discussed in~\sectref{matrix-mult}.\footnote{
  Matrix-vector multiplication is performed similarly.
}
Note that the CFLOBDD representation of $V$ has $n$ variables and its highest level is $\log{n}$, whereas the CFLOBDD for matrix $M$ has $2n$ variables and its highest level is $\log{n} + 1$.
\Omit{
Note that when converting the CFLOBDD that represents $V$ to the CFLOBDD that represent $M$, the number of levels increases by 1.
}

We will denote the variables in the CFLOBDD representation of $V$ as $x = \langle x_1, x_2, \cdots , x_n\rangle$.
The rows of $M$ would use the same $x$ variables. 
To represent the columns of $M$, we introduce an extra set of $n$ variables: $y = \langle y_1, y_2, \cdots, y_n\rangle$. 
As discussed in~\sectref{matrix-rep}, we will use the interleaved ordering of $x$ and $y$ ($x \bowtie y$) to represent $M$;
i.e., the decisions (at level $0$) by A-connection groupings at level 1 are the decisions for $x$ variables, and the decisions (at level $0$) by B-connection groupings at level 1 are those for $y$ variables.

At a high level, the algorithm for vector-to-matrix conversion involves the use of level-$0$ groupings (representing $x$ variables) from $V$'s CFLOBDD representation as the A-connection groupings at level 1 for representing $M$.
The detailed steps of vector-to-matrix conversion are as follows:
\begin{enumerate}
  \item
    \label{It:VectorToMatrix:StepOne}
    Create a new CFLOBDD $c1$ of level $\log{}n + 2$ from $V$'s CFLOBDD representation $c$, by using level-$0$ groupings of $c$ as the A-connection groupings of $c1$'s level-$1$ groupings and adding \emph{Don'tCareGroupings} as B-connection groupings at level $1$ similar to~\algref{KPB}.
    This step creates a CFLOBDD that represents a matrix in which every column is the vector $V$---because all of the column variables correspond to the added \emph{DontCareGroupings}.
  \item
    \label{It:VectorToMatrix:StepTwo}
    Create another CFLOBDD $c2$ of level $\log{}n + 2$ that represents a matrix $Column1Matrix_{n}$ of size $2^n\times2^n$ with only the first column filled with 1s and the rest with 0s. 
    \begin{equation*}
        Column1Matrix_{n} = 
        \begin{bmatrix}
        1 & 0 & \cdots & 0\\
        1 & 0 & \cdots & 0\\
        \vdots & \vdots & \ddots & \vdots\\
        1 & 0 & \cdots & 0\\
        \end{bmatrix}_{2^n\times2^n}\\
    \end{equation*}
    $Column1Matrix_{n}$ can be recursively defined in terms of $Column1Matrix_{n/2}$ matrices of size $2^{n/2}\times2^{n/2}$.
    This property allows $c2$ to both be created and represented efficiently.
    \[
        Column1Matrix_n = 
        \begin{cases}
            \vspace{2ex}
            \begin{bmatrix}
                Column1Matrix_{n/2} & O_{n/2} & \cdots & O_{n/2}\\
                Column1Matrix_{n/2} & O_{n/2} & \cdots & O_{n/2}\\
                \vdots & \vdots & \ddots & \vdots\\
                Column1Matrix_{n/2} & O_{n/2} & \cdots & O_{n/2}\\
            \end{bmatrix}_{2^n\times2^n} & \\
            \vspace{2ex}
             \quad \quad = Column1Matrix_{n/2}\tensor Column1Matrix_{n/2}
            & n > 1\\
            \begin{bmatrix}
                1 & 0\\
                1 & 0
            \end{bmatrix} & n = 1\\
        \end{cases}
    \]
    Here, $O_j$ represents the all-zero matrix of size $2^j \times 2^j$.
    An algorithm for the efficient construction of $Column1Matrix_n$ can be found in Appendix \sectref{Column1Matrix-construction}.
  \item
    \label{It:VectorToMatrix:StepThree}
    Finally multiply $c1$ and $c2$
    pointwise
    by calling BinaryApplyAndReduce.
\end{enumerate}

\begin{algorithm}[tb!]
\Input{CFLOBDD n representing vector $V$}
\Output{CFLOBDD n' representing matrix $M$ obtained by padding $V$ with zeros.}
\caption{Vector-to-Matrix Conversion \label{Fi:vector2matrix}}
\Begin{
    CFLOBDD g1 = RepresentativeCFLOBDD(ShiftToBConnectionAtLevelOne(n.grouping), n.valueTuple)\;
    CFLOBDD g2 = Column1Matrix(n.grouping.level + 1)\; 
    CFLOBDD n = BinaryApplyAndReduce(g1, g2, (op)Times)\;
    \Return n\;
}
\end{algorithm}

Pseudo-code for this algorithm is given as~\algref{vector2matrix}.


\begin{example}
We illustrate the steps of the algorithm using the following example:
Consider the vector $V = \left[\begin{smallmatrix} 2\\ 3\\ 5\\ 0\\ \end{smallmatrix}\right]$ and matrix
$M = \left[\begin{smallmatrix}
         2 & 0 & 0 & 0\\
         3 & 0 & 0 & 0\\
         5 & 0 & 0 & 0\\
         0 & 0 & 0 & 0
     \end{smallmatrix}\right]$.
The goal is to convert $V$ to $M$.
Steps \ref{It:VectorToMatrix:StepOne} and \ref{It:VectorToMatrix:StepTwo} construct intermediate matrices $M_1$ and $M_2$, defined as follows:
$M_1 = \left[\begin{smallmatrix}
         2 & 2 & 2 & 2\\
         3 & 3 & 3 & 3\\
         5 & 5 & 5 & 5\\
         0 & 0 & 0 & 0\\
       \end{smallmatrix}\right]
$
and
$M_2 = \left[\begin{smallmatrix}
         1 & 0 & 0 & 0\\
         1 & 0 & 0 & 0\\
         1 & 0 & 0 & 0\\
         1 & 0 & 0 & 0\\
       \end{smallmatrix}\right]
$.
Finally, the result of converting vector $V$ to a matrix is $M = M_1 \ast M_2$
$= \left[\begin{smallmatrix}
     2 & 2 & 2 & 2\\
     3 & 3 & 3 & 3\\
     5 & 5 & 5 & 5\\
     0 & 0 & 0 & 0\\
   \end{smallmatrix}\right]
   \ast
   \left[\begin{smallmatrix}
     1 & 0 & 0 & 0\\
     1 & 0 & 0 & 0\\
     1 & 0 & 0 & 0\\
     1 & 0 & 0 & 0\\
   \end{smallmatrix}\right]
= \left[\begin{smallmatrix}
     2 & 0 & 0 & 0\\
     3 & 0 & 0 & 0\\
     5 & 0 & 0 & 0\\
     0 & 0 & 0 & 0
   \end{smallmatrix}\right]
$, where ``*'' denotes pointwise matrix multiplication.
\end{example}

\subsection{Matrix Multiplication}
\label{Se:matrix-mult}

Matrix multiplication is one of the most important operations on matrices.
This section discusses how to perform matrix multiplication when the matrices are represented as CFLOBDDs using the interleaved-variable ordering.

The multiplication algorithm for CFLOBDDs presented here is similar to the standard $O(N^3)$ algorithm for multiplying two $N \times N$ matrices.
There is a potential for savings because each of the argument CFLOBDDs may have a large number of shared substructures, and function caching can be used to detect when a sub-problem has already been performed, in which case the proto-CFLOBDD for the answer can be returned immediately.

Our starting point is the observation that when the interleaved-variable ordering is used, at top level the A-connection of a CFLOBDD-represented matrix $M$ captures commonalities in the $\sqrt{N} \times \sqrt{N}$ block structure of $M$, and the B-connections represent sub-matrices of $M$ of size $\sqrt{N} \times \sqrt{N}$.
By analogy with other kinds of multi-terminal CFLOBDDs, at top-level one can think of the A-connection as a multi-terminal CFLOBDD whose value tuple is the sequence of B-connections---roughly, the A-connection is a $\sqrt{N} \times \sqrt{N}$ matrix with $\sqrt{N} \times \sqrt{N}$-matrix-valued leaves.

Suppose that $P$ and $Q$ are two $N \times N$ matrices represented by CFLOBDDs $C_P$ and $C_Q$, respectively.
The respective top-level A-connections, $A_P$ and $A_Q$, are matrices of size $\sqrt{N} \times \sqrt{N}$ with matrix-valued cells of size $\sqrt{N} \times \sqrt{N}$.
To multiply $P$ and $Q$, we first recursively multiply $A_P$ and $A_Q$.
This operation defines which cells of $A_P$ and $A_Q$ get multiplied and added---and the answer is returned as a collection of symbolic expressions (of a form that will be described shortly).
Using this information, we recursively call matrix multiplication and matrix addition on the B-connections, as appropriate.
For the base case of the recursion---namely, level $1$, which represents matrices of size $2 \times 2$---we can enumerate all the individual cases of possible matrix structures (i.e., the patterns of which cells hold equal values), and build the CFLOBDDs that result from a matrix multiplication in each case.

We now describe how the symbolic information mentioned above is organized, and how operations of addition and multiplication are performed on the data type in which the symbolic information is represented.
The challenge that we face is that at all levels below top-level, we do not have access to a \emph{value} for any cell in a matrix.
However, we can use the exit vertices as variables.

\begin{example}\label{Exa:TwoByTwoMatrixSymbolic}
  Suppose that we are multiplying two level-$1$ groupings that, when considered as $2 \times 2$ matrices over their respective exit vertices $[ \ev_1, \ev_2 ]$ and $[ \ev'_1, \ev'_2, \ev'_3 ]$, have the forms shown on the left
  \begin{equation}
    \label{Eq:TwoByTwoMatrixSymbolic}
    \begin{bmatrix}
      \ev_1 & \ev_1 \\
      \ev_2 & \ev_2 
    \end{bmatrix}
    \times
    \begin{bmatrix}
      \ev'_1 & \ev'_2 \\
      \ev'_1 & \ev'_3 
    \end{bmatrix}
    =
    \begin{bmatrix}
      \ev_1 \ev'_1 + \ev_1 \ev'_1 & \ev_1 \ev'_2 + \ev_1 \ev'_3 \\
      \ev_2 \ev'_1 + \ev_2 \ev'_1 & \ev_2 \ev'_2 + \ev_2 \ev'_3 
    \end{bmatrix}
    =
    \begin{bmatrix}
      2\ev_1 \ev'_1 & \ev_1 \ev'_2 + \ev_1 \ev'_3 \\
      2\ev_2 \ev'_1 & \ev_2 \ev'_2 + \ev_2 \ev'_3 
    \end{bmatrix}  
  \end{equation}
Each entry in the right-hand-side matrix can be represented by a set of triples, e.g.,
\[
  \begin{bmatrix}
    \{[(1,1),2]\} & \{[(1,2),1], [(1,3),1]\} \\
    \{[(2,1),2]\} & \{[(2,2),1], [(2,3),1]\} 
  \end{bmatrix}  
\]
and when listed in exit-vertex order for the interleaved-variable order, we have
\begin{equation}
  \label{Eq:MatMultTuple}
  [ \{[(1,1),2]\}, \{[(1,2),1], [(1,3),1]\}, \{[(2,1),2]\}, \{[(2,2),1], [(2,3),1]\} ].
\end{equation}
Now suppose that the two matrices are sub-matrices of level-$2$ groupings connected by ReturnTuples $\rt = [5, 2]$ and $\rt' = [6, 1, 2]$, respectively.
Then applying $\langle \rt, \rt' \rangle$ to \eqref{MatMultTuple} results in
\begin{equation}
  \label{Eq:MatMultTupleAfterMapping}
  [ \{[(5,6),2]\}, \{[(5,1),1], [(5,2),1]\}, \{[(2,6),2]\}, \{[(2,1),1], [(2,2),1]\} ].
\end{equation}
\end{example}
\Omit{
For purposes of presentation, an alternative would be to describe our symbolic expressions as being in terms of, e.g., $2\ev_2 \ev'_1$, and then describe ReturnTuples $\rt = [5, 2]$ and $\rt' = [6, 1, 2]$ from level $i$ to $i+1$ as operating on subscripts, so that $2\ev_{i,2} \ev'_{i,1}$ is mapped to $2\ev_{i+1,\textit{rt}(2)} \ev'_{i+1,\textit{rt}'(1)} = 2\ev_{i+1,2} \ev'_{i+1,6}$
}

We call the objects shown in \eqrefs{MatMultTuple}{MatMultTupleAfterMapping} \emph{MatMultTuples}.
By this device, the answer to a matrix-multiplication sub-problem (whether from A-connections or B-connections, and at any level $\geq 1$) can be treated as a multi-terminal CFLOBDD whose value tuple is a MatMultTuple.

\paragraph{Semantics of MatMultTuples.}
An alternative view of MatMultTuples comes from the right-hand matrix in \eqref{TwoByTwoMatrixSymbolic}:
a MatMultTuple is a sequence of bilinear polynomials over the exit vertices of two groupings.
We will represent a bilinear polynomial $p$ as a map from exit-vertex pairs to the corresponding coefficient.
(The pairs for which the coefficient is nonzero are called the \emph{support} of $p$.
In examples, we show only map entries that are in the support.)
In particular, suppose that $g_1$ and $g_2$ are two groupings at the same level, with exit-vertex sets $\EV$ and $\EV'$. Each entry of a MatMultTuple is of type $\BP_{\EV,\EV'} \eqdef (\EV \times \EV') \rightarrow \mathbb{N}$.
(We will drop the subscripts on $\BP$ if the exit-vertex sets are understood.)

To perform linear arithmetic on bilinear polynomials, we define
\[
  \begin{array}{r@{\hspace{0.5ex}}c@{\hspace{0.5ex}}l@{\hspace{5.0ex}}r@{\hspace{1.0ex}}c@{\hspace{1.0ex}}l}
      \ZeroBP & : &  \BP
    & \ZeroBP & \eqdef & \lambda (\ev,\ev') . 0
    \\
      + & : & \BP \times \BP \rightarrow \BP
    & \bp_1 + \bp_2 & \eqdef & \lambda (\ev,\ev') . \bp_1(\ev,\ev')  + \bp_2(\ev,\ev')
    \\
    * & : & \mathbb{N} \times \BP \rightarrow \BP
    & n * \bp & \eqdef & \lambda (\ev,\ev') . n * \bp(\ev,\ev')
  \end{array}
\]
By considering a ReturnTuple to be a map from one exit-vertex set to another, this notation allows us to give a second account of the transformation from \eqref{MatMultTuple} to \eqref{MatMultTupleAfterMapping}.
For instance, let $\rt = [1 \mapsto 5, 2 \mapsto 2]$ and $\rt' = [1 \mapsto 6, 2 \mapsto 1, 3 \mapsto 2]$.
Consider the second element of \eqref{MatMultTuple}:
$\bp = \{ [(1,2),1], [(1,3),1] \}$ $= [ (1,2) \mapsto 1, (1,3) \mapsto 1 ]$.
Then the transformation of $\bp$ induced by $\rt$ and $\rt'$ can be expressed as follows
(where \eqref{ReturnMapAppliedToMatMultTupleEntry} expresses the general case):
\begin{align}
  \langle \rt, \rt' \rangle(\bp)
      & \eqdef \{ (\rt(\ev), \rt'(\ev')) \mapsto \bp(\ev, \ev')  \mid \ev \in \EV, \ev' \in \EV' \} \label{Eq:ReturnMapAppliedToMatMultTupleEntry} \\
      & = \{ (\rt(1), \rt'(2)) \mapsto \bp(1,2), (\rt(1), \rt'(3)) \mapsto \bp(1,3) \} \notag \\
      & = \{ (5,1) \mapsto 1, (5,2) \mapsto 1 \}  \notag
\end{align}
At top level, we need a similar operation for the value induced by a pair of value tuples $\langle \vt, \vt' \rangle$ (where a value tuple is treated as a map of type $\EV \rightarrow \mathbb{V}$ for a value space $\mathbb{V}$
 that supports + and *):
\begin{equation}
  \label{Eq:ValueTupleAppliedToMatMultTupleEntry}
  \langle \vt, \vt' \rangle(\bp)
  \eqdef
  \sum \{ \bp(\ev, \ev') * \vt(\ev) * \vt'(\ev') \mid  \ev \in \EV, \ev' \in \EV' \}
\end{equation}

\begin{algorithm}[tb!]
\caption{Matrix Multiplication \label{Fi:MatrixMult}}
\Input{CFLOBDDs n1, n2}
\Output{CFLOBDD n = n1 $\times$ n2}
\Begin{
Grouping$\times$MatMultTuple [g,m] = MatrixMultOnGrouping(n1.grouping, n2.grouping)\;
ValueTuple v\_tuple = []\;   \label{Li:CallMatrixMultOnGrouping}
\For{$i \leftarrow 1$ \KwTo $|m|$ }{   \label{Li:ValueComputationStart}
    Value $v = \langle \mathrm{n1.valueTuple}, \mathrm{n2.valueTuple} \rangle(m(i))$\;  \label{Li:ConcreteWeightedDotProduct}
    v\_tuple = v\_tuple $||$ $v$;
}  \label{Li:ValueComputationEnd}
Tuple$\times$Tuple [inducedValueTuple, inducedReductionTuple] = CollapseClassesLeftmost(v\_tuple)\;
g = Reduce(g, inducedReductionTuple)\;
CFLOBDD n = RepresentativeCFLOBDD(g, inducedValueTuple)\;
\Return n\;
}
\end{algorithm}

\begin{algorithm}[tb!]
\caption{MatrixMultOnGrouping \label{Fi:MatrixMultGrouping}}
\Input{Groupings g1, g2}
\Output{Grouping$\times$MatMultTuple [g,m] such that g = g1 $\times$ g2}
\Begin{
\If(\tcp*[f]{Base Case: matrices of size $2\times 2$}){g1.level == 1}{
\tcp{Construct a level $1$ Grouping that reflects which cells of the product hold equal entries in the output MatMultTuple}
}
InternalGrouping g = new InternalGrouping(g1.level)\;
Grouping$\times$MatMultTuple [aa,ma] = MatixMultOnGrouping(g1.AConnection, g2.AConnection)\;  \label{Li:AConnectionRecursion}
g.AConnection = aa; g.AReturnTuple = [1..|ma|]; g.numberOfBConnections = |ma|\;
\tcp{Interpret ma to (symbolically) multiply and add BConnections}
MatMultTuple m = []\;
\For(\tcp*[f]{Interpret $i^{\textit{th}}~\BP$ in ma to create g.BConnections[$i$]}){$i \leftarrow 1$ \KwTo $|ma|$ }{    \label{Li:InterpretationLoopStart}
    \tcp{Set g.BConnections[$i$] to the (symbolic) weighted dot product $\sum_{((k_1,k_2),v) \in \textrm{ma}(i)} v * \textrm{g1.BConnections}[k_1] * \textrm{g2.BConnections}[k_2]$}
    CFLOBDD curr\_cflobdd = ConstantCFLOBDD(g1.level, $[\ZeroBP]$)\;  \label{Li:WeightedDotProductStart}
    \For{$((k_1,k_2),v) \in$ ma($i$)}{   
        Grouping$\times$MatMultTuple [bb,mb] = MatrixMultOnGrouping(g1.BConnections[$k_1$], g2.BConnections[$k_2$])\;  \label{Li:BConnectionRecursion}
        MatMultTuple mc = []\;
        \For{$j \leftarrow 1$ \KwTo $|mb|$ }{
            $\BP$ $\bp = \langle \mathrm{g1.BReturnTuples}[k_1], \mathrm{g2.BReturnTuples}[k_2] \rangle(\mathrm{mb}(j))$\;
            mc = mc $||$ $\bp$;
        }
        Tuple$\times$Tuple [inducedMatMultTuple, inducedReductionTuple] = CollapseClassesLeftmost(mc)\;
        bb = Reduce(bb, inducedReductionTuple)\;
        CFLOBDD n = RepresentativeCFLOBDD(bb, inducedMatMultTuple)\;
        curr\_cflobdd = curr\_cflobdd + $v$ $\ast$ n \tcp*[r]{Accumulate symbolic sum}
    }  
    g.BConnection[$i$] = curr\_cflobdd.grouping\;
    g.BReturnTuples[$i$] = curr\_cflobdd.valueTuple\;
    m = m || curr\_cflobdd.valueTuple\;  \label{Li:WeightedDotProductEnd}
} \label{Li:InterpretationLoopEnd}
g.numberOfExits = |m|\;
Tuple$\times$Tuple [inducedMatMultTuple, inducedReductionTuple] = CollapseClassesLeftmost(m) \;           
g = Reduce(g, inducedReductionTuple)\;
\Return [RepresentativeGrouping(g), m]\;
}
\end{algorithm}

\algrefs{MatrixMult}{MatrixMultGrouping} give pseudo-code for the matrix-multiplication algorithm.
\algref{MatrixMultGrouping} operates on two Groupings g1 and g2 at some level $l$, returning a (Grouping, MatMultTuple) pair (g, $m$).
\algref{MatrixMultGrouping} works symbolically, first considering the symbolic product of g1.AConnection and g2.AConnection (\lineref{AConnectionRecursion}).
The MatMultTuple ma returned from this call is really a list of directives:
each directive is a bilinear polynomial used to create one of the BConnections of g.
The workhorse computation---\lineseqref{WeightedDotProductStart}{WeightedDotProductEnd}---sets g.BConnections[$i$] to the (symbolic) weighted dot product
\[
  \sum_{((k_1,k_2),v) \in \textrm{ma}(i)} v * \textrm{g1.BConnections}[k_1] * \textrm{g2.BConnections}[k_2].
\]
To perform this computation, an auxiliary multi-terminal CFLOBDD curr\_cflobdd is used to accumulate the (symbolic) weighted dot product.
Note that the valueTuple of curr\_cflobdd is a MatMultTuple;
thus, curr\_cflobdd is essentially a matrix, each element of which is a bilinear polynomial---i.e., a directive saying how to compute the value of that element from a set of pairs of (as yet unknown) values.

\algref{MatrixMult} multiplies two matrices, n1 and n2, represented as CFLOBDDs.
It initiates the process at top level by calling MatrixMultGrouping on n1.grouping and n2.grouping (\lineref{CallMatrixMultOnGrouping}).
It then uses the returned MatMultTuple as a list of directives, similar to the way MatMultTuples are used in MatrixMultGrouping.
The difference is that at top level one has access to the actual values of the argument matrices, namely, n1.valueTuple and n2.valueTuple.
Thus, in \algref{MatrixMult} the elements of MatMultTuple m are interpreted as directives to perform a \emph{concrete} weighted dot product (\lineseqref{ValueComputationStart}{ValueComputationEnd}), using \eqref{ValueTupleAppliedToMatMultTupleEntry} in \lineref{ConcreteWeightedDotProduct}, thereby computing the values of the elements of the answer matrix .

\subsection{Path Counting and Sampling}
\label{Se:PathCountingAndSampling}

A CFLOBDD whose terminal values are non-negative numbers can be used to represent a discrete distribution over the set of assignments to the Boolean variables.
An assignment---or equivalently, the corresponding matched path in the CFLOBDD---is considered to be an elementary event.
The ``weight'' of the elementary event is the terminal value.
The probability of a matched path $p$ is the weight of $p$ divided by the total weight of the CFLOBDD---the sum of the weights obtained by following each of the CFLOBDD's matched paths.
Fortunately, it is possible to compute the aforementioned denominator by computing, for each of the terminal values, the number of matched paths that lead to that terminal value (\sectref{PathCountingInACFLOBDD}).
With those numbers in hand, it is then possible to sample an assignment/path according to the distribution that the CFLOBDD represents (\sectref{SamplingInACFLOBDD}).

The same approach can be used for CFLOBDDs whose terminal values are complex numbers, except that the weight of a matched path is the modulus of the terminal value.
This approach is used in the application of CFLOBDDs to quantum simulation (\sectrefs{quantum-algos}{ResearchQuestionTwo}).

\subsubsection{Path Counting}
\label{Se:PathCountingInACFLOBDD}

Recall that
every terminal value is connected to one exit vertex of the top-level grouping of the CFLOBDD.
Every exit vertex of a grouping is, in turn, connected to exit vertices of internal groupings.
Therefore, to compute the number of matched paths for every terminal value,
we need to compute the path-counts from the entry vertex of a grouping to every exit vertex of that grouping, for every grouping in the CFLOBDD.
For each grouping $g$, we would like to compute a vector of path-counts, in which the $i^{\textit{th}}$ element is the number of matched paths from $g$'s entry vertex to the $i^{\textit{th}}$ exit vertex of $g$.
To compute this information, we can break it down into 
(i) computing the number of matched paths from $g$'s entry vertex to $g$'s middle vertices;
(ii) computing the number of matched paths from $g$'s middle vertices to $g$'s exit vertices; and
(iii) combining this information to obtain the number of matched paths from $g$'s entry vertex to $g$'s exit vertices.


Consider a Grouping $g$ at level $l$ with $e$ exit vertices.
Suppose that $g.\textit{AConnection}$ has $p$ exit vertices, $g.\textit{BConnections}[j]$ has $k_j$ exit vertices, and let $g.\textit{BReturnTuples}[j]$ be the return edges from $g.\textit{BConnections}[j]$'s exit vertices to $g$'s exit vertices.
For step (i), we recursively compute the path-counts for $g.\textit{AConnection}$, which yields a vector of path-counts $v_A$ of size $1\times p$.
Step (ii) creates a matrix $M_B$ of size $p\times e$, in which the $j^{\textit{th}}$ row is the vector of path-counts from the $j^{\textit{th}}$ middle vertex of $g$ to $g$'s exit vertices. 
Step (iii) is the vector-matrix multiplication $v_A \times M_B$, which yields $g$'s path-count vector, of size $1\times e$.
The base-case path-count vectors are $[1,1]$ for a \emph{ForkGrouping} and $[2]$ for a \emph{DontCareGrouping}.

Because the exit vertices of $g.\textit{BConnections}[j]$ are connected to $g$'s exit vertices via $g.{\textit{BReturnTuples}}[j]$, the $j^{\textit{th}}$ row of $M_B$ is the product of the path-count vector for $g.\textit{BConnections}[j]$ (of size $1\times k_j$) and a ``permutation matrix'' $\textit{PM\,}^{g.\textit{BReturnTuples}[j]}$ (of size $k_j \times e$).
Each entry of $\textit{PM}$ is either $0$ or $1$;
each row must have exactly one $1$;
and each column must have at most one $1$.

\begin{algorithm}
\caption{CountPaths \label{Fi:PathCounting}}
\Input{Grouping g}
\Begin{
\eIf{g.level == 0}{
    \eIf{g == DontCareGrouping}{
        g.numPathsToExit = [2]\;
    }(\tcp*[h]{g == $\textit{ForkGrouping}$}){
        g.numPathsToExit = [1,1]\;
    }
}
{
    CountPaths(g.AConnection)\;
    \For{$i \leftarrow 1$ \KwTo $g.numberOfBConnections$}{
        CountPaths(g.BConnection[i]);
    }
    g.numPathsToExit = [1..|g.numberOfExits|]\;  \label{Li:path-counting-multiplication-start}
    \For{$i \leftarrow 1$ \KwTo $g.numberOfBConnections$}{
        \For{$j \leftarrow 1$ \KwTo $g.BConnection[i].numberOfExits$}{
            k = BReturnTuples[i](j)\;
            g.numPathsToExit[k] += \\
            \quad g.AConnection.numPathsToExit[i] * g.BConnection[i].numPathsToExit[j]\;
        }
    }   \label{Li:path-counting-multiplication-end}
}
}
\end{algorithm}

\algref{PathCounting} shows pseudo-code for the path-counting algorithm.
The path-counts are stored in a Tuple $\textit{numPathsToExit}$ in the Grouping data structure.
\Lineseqref{path-counting-multiplication-start}{path-counting-multiplication-end} perform the vector-matrix multiplication discussed above.
In practice, because the number of path-counts increases double exponentially in the number of levels, we only store the log-values of the path-counts.
It is also important for the implementation of the algorithm to perform function caching (\sectref{FunctionCaching}) so that each grouping at each level is visited only once.
When function caching is employed, \algref{PathCounting} visits each grouping, and hence each vertex and edge of the CFLOBDD, exactly once;
consequently, the cost of the path-counting operation is bounded by the size of the argument CFLOBDD.

This definition can also be stated equationally, in a form similar to the denotational semantics given in \sectref{ADenotationalSemantics}, where the expression in large brackets represents $M_B$.
\[
\begin{array}{@{\hspace{0ex}}l@{\hspace{0ex}}}
  \textit{numPathsToExit\,}^g_{1 \times e} = \\
  \qquad \begin{cases}
           [1,1]_{1\times2} & \text{if g = \textit{ForkGrouping}} \\
           [2]_{1\times1} & \text{if g = \textit{DontCareGrouping}} \\
           \begin{array}{@{\hspace{0ex}}l@{\hspace{0ex}}}
             \textit{numPathsToExit\,}^{g.\textit{AConnection}}_{1\times p} \, \times \\
             \qquad \begin{bmatrix}
                      \vdots\\
                      \textit{numPathsToExit\,}^{g.{\textit{BConnections}[j]}}_{1\times k_j} \times   \textit{PM\,}^{g.{\textit{BReturnTuples}}[j]}_{k_j \times e}\\
                      \vdots
                    \end{bmatrix}_{\scriptsize {\begin{array}{@{\hspace{0ex}}l@{\hspace{0ex}}} {p \times e} \\ j \in \{1..p\} \end{array}}}
           \end{array} & \text{otherwise}
         \end{cases}
\end{array}
\]

\begin{example}\label{Exa:CFLOBDDPathCounts}
For the five proto-CFLOBDDs depicted in \figref{product4RepFigure}, the vectors of path-counts are computed as follows (read top-to-bottom by level):
\[
  \begin{array}{@{\hspace{0ex}}l@{\hspace{10ex}}l@{\hspace{10ex}}l@{\hspace{0ex}}}
    \multicolumn{1}{c}{\textrm{level 2}} & \multicolumn{1}{c}{\textrm{level 1}} & \multicolumn{1}{c}{\textrm{level 0}} \\
    \hline
    \begin{bmatrix} 9 & 7 \end{bmatrix}
    = 
    \begin{bmatrix} 3 & 1 \end{bmatrix} 
    \times
    \begin{bmatrix}
      3 & 1 \\
      0 & 4
    \end{bmatrix}
    &
    \begin{bmatrix} 3 & 1 \end{bmatrix}
    =
    \begin{bmatrix} 1 & 1 \end{bmatrix}
    \times
    \begin{bmatrix}
      2 & 0 \\
      1 & 1
    \end{bmatrix}
    &
    \begin{bmatrix} 1 & 1 \end{bmatrix}
    \\
    & [4] = [2] \times [2]
    & [2]
    \\
    \hline
  \end{array}
\]
\end{example}

\subsubsection{Sampling an Assignment}
\label{Se:SamplingInACFLOBDD}

Our goal is to sample a matched path from the distribution of matched paths of a given CFLOBDD, and return the corresponding assignment.
As in \sectref{CFLOBDDOperationalSemantics}, we assume that an assignment is an array of Booleans, whose entries---starting at index-position 1---are the values of successive variables.
The concatenation of two such arrays $a_1$ and $a_2$ is denoted by
$a_1 || a_2$.

To explain how to sample from the set of assignments, we argue in terms of the structure of the corresponding matched paths.
If the distribution of matched paths was given as a vector of weights, $W = [w_1 .. w_{2^{2^l}}]$, as one would have in the corresponding decision tree, the probability of selecting the $p^{\textit{th}}$ matched path is given by
\begin{equation}
  \label{Eq:ProbabilityOfAPath}
  \textit{Prob}(p) = \dfrac{w_p}{\sum_{i=1}^{2^{2^l}}w_i}.
\end{equation}
In a CFLOBDD that represents a distribution, we do not have access to $W$ directly.
Suppose that $W' = [w'_1, \ldots, w'_K]$ is the vector of terminal values of the CFLOBDD.
The $K$ values of $W'$ are exactly the $K$ different values that appear in $W$;
however, many matched paths that start at the top-level entry vertex lead to the same terminal value, say $w'_m$.
Fortunately, the path-counting method from \sectref{PathCountingInACFLOBDD} provides us with part of what is needed via $\textit{NumPathsToExit}$ of the top-level grouping.
\[
  \sum_{i=1}^{2^{2^l}} w_i = \sum_{j=1}^{K} w'_j \times \textit{numPathsToExit}[j]
\]
Thus, while \eqref{ProbabilityOfAPath} becomes
\[
    \textit{Prob}(p) = \dfrac{w_p}{\sum_{j=1}^{K}w'_j \times \textit{numPathsToExit}[j]}
\]
this observation gives us no guidance about how to select a matched path $p$ with that probability.

Rather than selecting a single matched path immediately, what we can do instead is to select the entire set of matched paths that reach a given terminal value.
This selection can be done by sampling from the exit vertices of the top-level grouping according to the probability distribution
\begin{equation}
  \label{Eq:TopLevelSampling}
  \textit{Prob}'(\text{Path ends at terminal value } w'_t) = \dfrac{w'_t \times \textit{numPathsToExit}[t]}{\sum_{j=1}^{K}w'_j \times \textit{numPathsToExit}[j]}
\end{equation}
The result of this sampling step is the index of an exit vertex of the top-level grouping, which will be used for further sampling among the (indirectly) ``retrieved'' set of matched paths.
What remains to be done is to uniformly sample a matched path from that set, and return the assignment that corresponds to that matched path.

To achieve this goal, we take advantage of the structure of matched paths to break the assignment/path-sampling problem down to a sequence of smaller assignment/path-sampling problems that can be performed recursively.
At each grouping $g$ visited by the algorithm, the goal is to uniformly sample a matched path from the set of matched paths $P_{g,i}$
(in the proto-CFLOBDD headed by $g$) that lead from $g$'s entry vertex to a specific exit vertex $i$ of $g$.

Consider a grouping $g$ and a given exit vertex $i$.
For each middle vertex $m$ of $g$, there is some number of matched paths---possibly $0$---from the entry vertex of $g$ that pass through $m$ and eventually reach exit vertex $i$.
Those numbers of matched paths, when divided by $|P_{g,i}|$, represent a distribution $D_i$ on the set of $g$'s middle vertices.
Consequently, the first step toward uniformly sampling a matched path from the set $P_{g,i}$ is to sample the index of a middle vertex of $g$ according to distribution $D_i$.
Call the result of that sampling step $m_{\textit{index}}$.
Thus, to sample a matched path from the entry vertex of $g$ to exit vertex $i$, we
(i) sample a middle vertex of $g$ according to $D_i$ to obtain $m_{\textit{index}}$;
(ii) uniformly sample a matched path from $g.\textit{AConnection}$ with respect to the exit vertex of $g.\textit{AConnection}$ that returns to $m_{\textit{index}}$;
(iii) uniformly sample a matched path
from $g.\textit{BConnections}[m_{\textit{index}}]$ with respect to whichever of its exit vertices is connected to the $i^{\textit{th}}$ exit vertex of $g$; and
(iv) concatenate the assignments obtained from steps (ii) and (iii).

Only the B-connections of $g$ whose exit vertices are connected to $i$ (the distinguished exit vertex of $g$) can contribute to the paths leading to $i$,
and hence we need to select a middle vertex from among those for which the B-connection grouping can lead to $i$.
For such an $i$-connected B-connection grouping $k$, let $(g.\textit{BReturnTuples}[k])^{-1}[i]$ denote the exit vertex of $g.\textit{BConnections}[k]$ that leads to $i$;
i.e., $\langle j, i \rangle \in g.\textit{BReturnTuples}[k] \Leftrightarrow (g.\textit{BReturnTuples}[k])^{-1}[i] = j$.

The path-counts for the number of matched paths of $g$'s B-connections (available via the vector $\textit{NumPathsToExit}$ for each of $g$'s B-connections, denoted by, e.g., $\textit{numPathsToExit\,}^{g.\textit{BConnections}[k]}$) only considers matched paths from $g$'s middle vertices to $g$'s exit vertices.
However, to sample $m_{\textit{index}}$ correctly, we need to consider \emph{all} of the matched paths from $g$'s entry vertex to $g$'s exit vertex $i$.
Hence, we multiply the number of matched paths from $g$'s entry vertex to a middle vertex of $g$ (of interest to us because it is connected to a B-connection that is connected to $i$), denoted by, e.g., $\textit{numPathsToExit\,}^{g.\textit{AConnection}}[k]$, to the number of matched paths from that same middle vertex to $g$'s exit vertex $i$.
Thus, the probability associated with a given $m_{\textit{index}}$ is as follows
(where $g.A$ denotes $g.\textit{AConnection}$, $g.B[k]$ denotes $g.\textit{BConnections}[k]$, and $g.\textit{BRT}$ denotes $g.\textit{BReturnTuples}$):
\begin{equation}
  \label{Eq:ProbabilityOfBIndex}
  \textit{Prob}(m_{\textit{index}})
  =
  \dfrac{
        \textit{numPathsToExit\,}^{g.A}[m_{\textit{index}}] \times 
        \textit{numPathsToExit\,}^{g.B[m_{\textit{index}}]}[(g.\textit{BRT}[m_{\textit{index}}])^{-1}[i]]
       }{
         g.\textit{numPathsToExit}[i]
       }
\end{equation}

\begin{example}\label{Exa:CFLOBDDSampling}
Consider the CFLOBDD depicted in \figref{product4RepFigure}, and suppose that the goal is to sample a matched path that leads to terminal value $T$.
From \exref{CFLOBDDPathCounts}, we know that (i) the outermost grouping has $7$ matched paths that lead to $T$, and (ii) $\textit{NumPathsToExit}$ is $[3,1]$ and $[4]$ for the upper and lower level-$1$ groupings, respectively.
Both of the outermost grouping's middle vertices have return edges that lead to $T$;
thus, from \eqref{ProbabilityOfBIndex}, we should sample the middle vertices with probabilities
\begin{equation*}
  \begin{array}{@{\hspace{0ex}}r@{\hspace{1.0ex}}c@{\hspace{1.0ex}}l@{\hspace{5.0ex}}r@{\hspace{1.0ex}}c@{\hspace{1.0ex}}l@{\hspace{0ex}}}
    \textit{Prob}(m_{\textit{index}} = 1) & = & \frac{[3,1][1] \times [3,1][2]}{7} = \frac{3 \times 1}{7} = \frac{3}{7}
    &
    \textit{Prob}(m_{\textit{index}} = 2) & = & \frac{[3,1][2] \times [4][1]}{7} = \frac{1 \times 4}{7} = \frac{4}{7}
  \end{array}
\end{equation*}
\end{example}

Once $m_{\textit{index}}$ has been selected in accordance with
\eqref{ProbabilityOfBIndex}, we recursively sample a matched path---and its assignment $a_A$---from
$g.\textit{AConnection}$ with respect to exit vertex $m_{\textit{index}}$ (step (ii)).
We also recursively sample
a matched path---and its assignment $a_B$---from
$g.\textit{BConnection}[m_{\textit{index}}]$ with respect to the exit vertex $(g.\textit{BReturnTuples}[m_{\textit{index}}])^{-1}[i]$ that leads to $g$'s exit vertex $i$ (step (iii)).
Step (iv) produces the
assignment $a = a_A || a_B$.

As for the base cases of the recursion, for a $\textit{DontCareGrouping}$, we randomly choose one of the paths $0$ or $1$ with probability $0.5$,
returning the assignment ``$0$'' or ``$1$'' accordingly;
for a $\textit{ForkGrouping}$, the designated exit vertex---either $1$ or $2$---specifies a unique assignment: ``$0$'' or ``$1$,'' respectively.

\begin{algorithm}[tb!]
\caption{Sample an Assignment from a CFLOBDD\label{Fi:Sampling}}
\SetKwFunction{SampleAssignment}{SampleAssignment}
\SetKwFunction{SampleOnGroupings}{SampleOnGroupings}
\SetKwProg{myalg}{Algorithm}{}{end}
\myalg{\SampleAssignment{n}}{
\Input{CFLOBDD n}
\Output{
Assignment sampled from n according to n.valueTuple
}
\Begin{
$i \leftarrow $ Sample(n.valueTuple)\tcp*[r]{Sample terminal-value index $i$ via \eqref{TopLevelSampling}}
Assignment a = SampleOnGroupings(n.grouping, $i$)\;
\Return a\;
}
}{}
\setcounter{AlgoLine}{0}
\SetKwProg{myproc}{SubRoutine}{}{end}
\myproc{\SampleOnGroupings{g, i}}{
\Input{Grouping g, Exit index $i$}
\Output{
Assignment sampled from g, corresponding to one of the paths leading to exit $i$
}
\Begin{
\If{g.level == 0}{
\eIf{g == $\textit{DontCareGrouping}$}{
    \Return (random() \% 2) ? "1" : "0"\;
}(\tcp*[h]{g == $\textit{ForkGrouping}$ so $i \in [1,2]$}){
    \Return ($i$ == 1) ? "0" : "1"
}
}
{
Tuple PathsLeadingToI = []\;
\For(\tcp*[f]{Build path-count tuple from which to sample}){$j \leftarrow 1$ \KwTo $g.\textit{numberOfBConnections}$}{
    \If(\tcp*[f]{if $j^{th}$ B-connection leads to $i$})
    {$i \in g.\textit{BReturnTuples}[j]$}{
    PathsLeadingToI = PathsLeadingToI||(g.AConnection.numPathsToExit[$j$] * g.BConnections[$j$].numPathsToExit[$k$]), where $i$ = BReturnTuples[$j$]($k$)\;
    }
}
$m_{\textit{index}} \leftarrow$ Sample(PathsLeadingToI) \tcp*[r]{Sample middle-vertex index $m_{\textit{index}}$}
Assignment a = SampleOnGroupings(g.AConnection, $m_{\textit{index}}$) || SampleOnGroupings(g.BConnection[$m_{\textit{index}}$], $k$), where $i$ = BReturnTuples[$m_{\textit{index}}$]($k$)\;
\Return a\;
}
}
}
\end{algorithm}

\algref{Sampling} gives pseudo-code for the algorithm for sampling an assignment
from a CFLOBDD.

For a CFLOBDD at level $l$, the sampling operation involves constructing
an assignment
of size $2^l$.
Hence, the cost of sampling is at least as large as the size of the sampled
assignment.
However, the size of the argument CFLOBDD also influences the cost of sampling;
although not every grouping of the CFLOBDD is necessarily visited when sampling
an assignment,
we can say that the cost of the sampling operation is bounded by
$\mathcal{O}(\max(2^l, \text{size of argument CFLOBDD}))$.

\section{Relations efficiently represented by CFLOBDDs}
\label{Se:efficient-relations}

In this section, we prove that there exists an exponential separation between CFLOBDDs and BDDs.
We establish this result using three relations that can be efficiently represented by CFLOBDDs.
Note that we do not assume any specific variable ordering when discussing the sizes of BDDs for the functions used to prove the separation.
We use node counts in BDDs, and vertex counts and edge counts in CFLOBDDs as a proxy for memory.
(Recall from \footnoteref{NodeGroupingVertexTerminology} that we use the term ``node'' solely for BDDs, whereas ``groupings'' and ``vertices'' (depicted as the dots inside groupings) refer to CFLOBDDs.)

\paragraph{Remark}
Recently, Zhi and Reps \cite{UNPUB:ZR24} obtained a characterization of relative sizes in the opposite direction (i.e., a bound on CFLOBDD size as a function of BDD size, for all BDDs).
They showed that for every BDD of size $n$, there is a corresponding CFLOBDD, which uses the same variable ordering, of size $O(n^3)$.

\subsection{The Equality Relation $\EQ_n$}
\label{Se:Separation:EqualityRelation}

\begin{figure}[ht!]
    \centering
    \begin{subfigure}[t]{0.45\linewidth}
    \centering
    \includegraphics[scale=0.5]{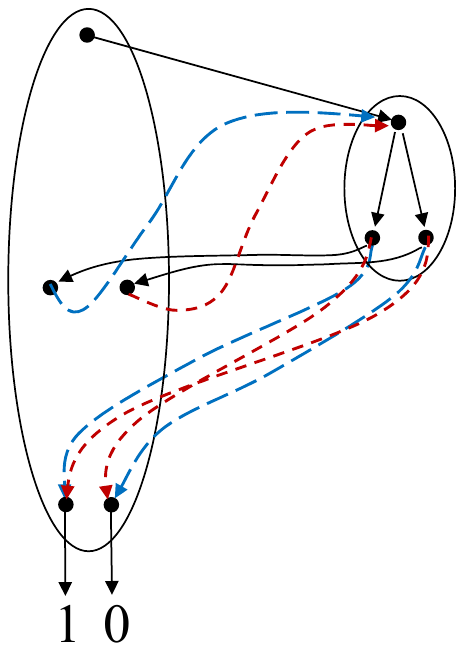}
    \caption{CFLOBDD for $\EQ_2$}
    \label{Fi:eq1}
    \end{subfigure}
    \begin{subfigure}[t]{0.45\linewidth}
    \centering
    \includegraphics[scale=0.4]{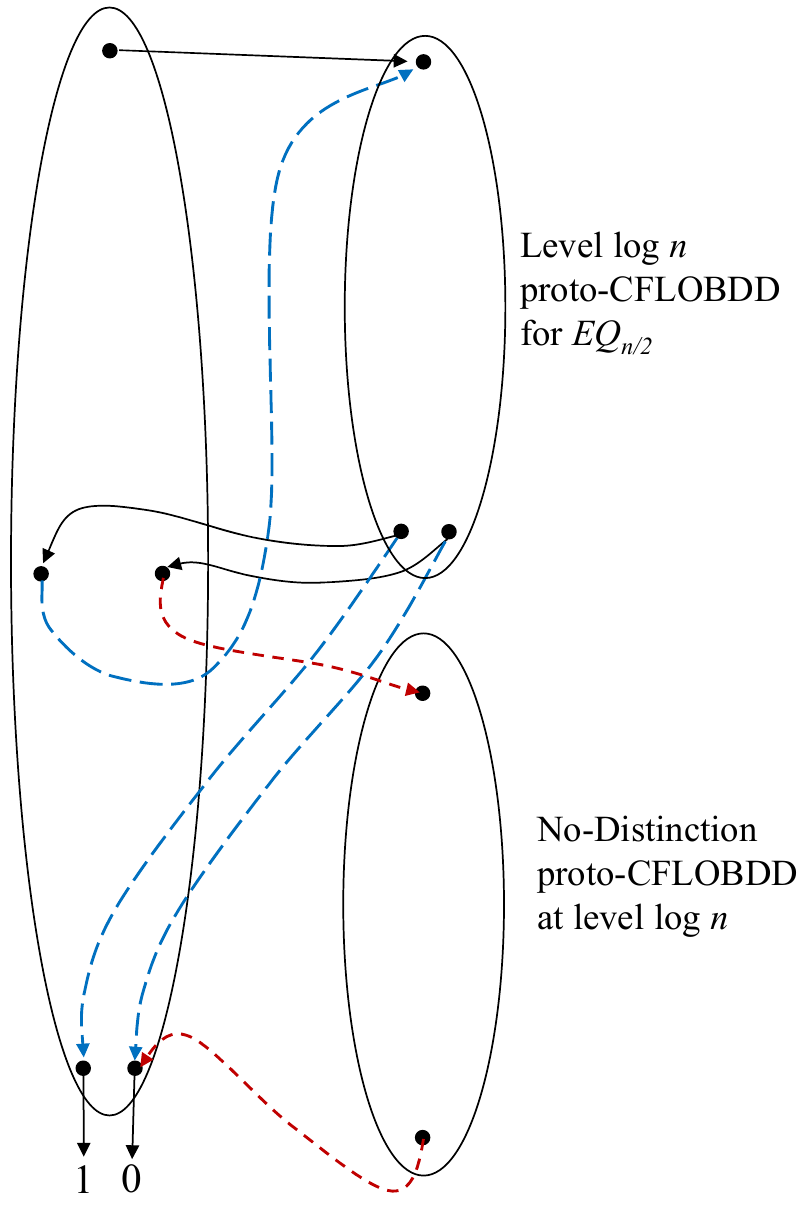}
    \caption{CFLOBDD for $\EQ_n$, $n \geq 2$
    }
    \label{Fi:eqK}
    \end{subfigure}
    \caption{}
    \label{Fi:CFLOBDDInterleavedForEQ}
\end{figure}

\begin{definition}\label{De:EqualityRelation}
The equality relation $\EQ_n : \{0,1\}^{n/2} \times \{0,1\}^{n/2} \rightarrow \{0,1\}$ on variables ($x_0\cdots x_{n/2-1}$) and ($y_0\cdots y_{n/2-1}$) is the relation
$
    \EQ_n(X,Y) \eqdef \Pi_{i=0}^{n/2-1} (x_i \Leftrightarrow y_i)
              =  \Pi_{i=0}^{n/2-1} (\bar{x_i}\bar{y_i} \lor x_i y_i)
$.
\end{definition}

\begin{theorem}[Exponential separation for the equality relation]\label{The:EQSeparation}
For $n = 2^l$, where $l \geq 1$,
$\EQ_n$ can be represented by a CFLOBDD with $\bigO(\log{n})$ vertices and edges.
In contrast, a BDD that represents $\EQ_n$ requires $\bigOmega(n)$ nodes.
\end{theorem}

\begin{proof}
\,\newline

\noindent
\textit{CFLOBDD Claim.}
We claim that with the interleaved-variable ordering $\langle x_0, y_0, \ldots, x_{n/2-1}, y_{n/2-1} \rangle$, the CFLOBDD representation of $\EQ_n$ uses $\bigO (\log{n})$ groupings, each of constant size (and hence uses $\bigO (\log{n})$ vertices and edges in total).
Diagrams that illustrate the CFLOBDD representation are shown in \figref{CFLOBDDInterleavedForEQ}.\
For $\EQ_2$ (i.e., $l = 1$),
the representation is shown in \figref{eq1}.
This CFLOBDD has two groupings, a fork grouping at level $0$ and a level-$1$ grouping.
In total, it has eight vertices and eleven edges.
Note that the ``success'' and ``failure'' exits at level $1$ are the left and right exits, respectively.

For $\EQ_n$ with $n > 2$ (i.e., $l > 1$),
the representation is defined inductively, following the pattern shown in \figref{eqK}.
For all $i$, $1 \leq i \leq l = \log n$, a level-$(i\textrm{+}1)$ grouping of the form shown at the outermost level of \figref{eqK} has an A-connection and---from the leftmost middle vertex---a B-connection to the proto-CFLOBDD for $\EQ_{2^{i-1}}$.
The leftmost exit vertex of the proto-CFLOBDD is the ``success exit'' for testing equality on $2^{i-1}$ pairs of variables.
\begin{itemize}
  \item
    When called via the level-$(i\textrm{+}1)$ grouping's A-connection, the proto-CFLOBDD tests equality on the variable pairs
    $\{ (x_0, y_0), \ldots, (x_{2^{i-1}-1}, y_{2^{i-1}-1}) \}$.
  \item
    When called from the level-$(i\textrm{+}1)$ grouping's leftmost middle vertex, the proto-CFLOBDD tests equality on the variable pairs
    $\{ (x_{2^{i-1}}, y_{2^{i-1}}), \ldots, (x_{{2^i}-1}, y_{{2^i}-1}) \}$.
\end{itemize}
The right exit vertex of the proto-CFLOBDD is the ``failure exit.''
When the proto-CFLOBDD has been called from the level-$(i\textrm{+}1)$ grouping's A-connection, the return edge is the matching solid edge to the rightmost middle vertex of the level-$(i\textrm{+}1)$ grouping.
This vertex signifies that there has been an equality mismatch on some variable pair in $\{ (x_0, y_0), \ldots, (x_{2^{i-1}-1}, y_{2^{i-1}-1}) \}$, and so its B-connection is to the no-distinction proto-CFLOBDD of level $i$, whose one exit vertex returns to the rightmost (``failure'') exit vertex of the level-$(i\textrm{+}1)$ grouping.

The level-$(i\textrm{+}1)$ grouping has five vertices and eight edges, and the level-$i$ grouping of the no-distinction proto-CFLOBDD at level $i$ has three vertices and four edges (see \figref{NoDistinctionProtoCFLOBDD}d).
Consequently, each of the $\log n + 1 $ levels of the CFLOBDD for $\EQ_n$ contributes a constant number of vertices and edges, independent of $i$, and thus the total number of vertices and edges is $\bigO (\log{n})$.

\smallskip
\noindent
\textit{BDD Claim.}
Now consider a BDD representation for $\EQ_n$.
We claim that regardless of the variable order used, a BDD requires at least $n$ nodes, one node for each argument variable. 

We prove this claim by contradiction.
Suppose that there is some BDD $B$ for $\EQ_n$ that does not need at least one node for each variable.
Let $\mathcal{T}$ denote the ``all-true'' assignment of variables;
i.e., $\mathcal{T} \eqdef \forall k \in \{0..n/2-1\}, x_k \mapsto T, y_k \mapsto T$.
There are three possible situations:
\begin{itemize}
  \item Case 1:
    $B$ does not have variable $y_k$, for some $k \in \{ 0..n/2-1 \}$.
    Now, consider two assignments of variables:
    $A_1 \eqdef \mathcal{T}$ and $A_2 \eqdef \mathcal{T}[y_k \mapsto F]$ (i.e., $A_2$ is $A_1$ with $y_k$ updated to $F$).
    Because $B$ does not depend on $y_k$, the function represented by $B$ maps both $A_1$ and $A_2$ to the same value (either $0$ or $1$), which violates the definition of the equality relation $\EQ_n$ (i.e., $\EQ_n[A_1] = 1$ and $\EQ_n[A_2] = 0$).
    Consequently, none of the $y$ variables can be dropped individually.
  \item
    Case 2: Using an argument completely analogous to Case 1, we can show that none of the $x$ variables can be dropped individually.
  \item Case 3:
    $B$ depends on neither $x_k$ nor $y_k$.
    Consider the following four assignments:
    $A_1 = \mathcal{T}$, $A_2 = \mathcal{T}[y_k \mapsto F]$, $A_3 = \mathcal{T}[x_k \mapsto F]$, and $A_4 = \mathcal{T}[x_k \mapsto F][y_k \mapsto F]$.
    Because $B$ does not depend on either $x_k$ or $y_k$, the function represented by $B$ maps all four assignments to the same value (either $0$ or $1$), which violates the definition of the equality relation $\EQ_n$ (i.e., $\EQ_n[A_1] = \EQ_n[A_4] = 1$ and $\EQ_n[A_2] = \EQ_n[A_3] = 0$).
\end{itemize}
Because (i) no $y_k$ can be dropped individually, (ii) no $x_k$ can be dropped individually, and (iii) no $(x_k, y_k)$ pair can be dropped together, $B$---and hence any  BDD representation for $\EQ_n$---requires $\bigOmega(n)$ nodes.
\end{proof}

\subsection{The Hadamard Relation $H_n$}
\label{Se:Separation:HadamardRelation}

The Hadamard Relation represents the family of Hadamard Matrices discussed in~\sectrefs{Preliminaries}{EncodingHFour}.
The Hadamard matrices play a role in many quantum algorithms, including the seven that are used in \sectref{ResearchQuestionTwo} to evaluate the effectiveness of CFLOBDDs for simulating quantum circuits (namely, GHZ, BV, DJ, Simon's algorithm, QFT, Shor's algorithm, and Grover's algorithm).
See \sectrefs{QuantumAlgorithms}{ResearchQuestionTwo}.

\begin{theorem}[Exponential separation for the Hadamard relation]\label{The:HadamardSeparation}
The Hadamard Relation $H_n:\{0,1\}^{n/2} \times \{0,1\}^{n/2} \rightarrow \{1,-1\}$ between variable sets ($x_0\cdots x_{n/2}$) and ($y_0\cdots y_{n/2}$), where $n = 2^l$, can be represented by a CFLOBDD with $\bigO(\log{n})$ vertices and edges.
In contrast, a BDD that represents $H_n$ requires $\bigOmega(n)$ nodes.
\end{theorem}

\begin{proof}
\,\newline

\smallskip
\noindent
\textit{CFLOBDD Claim.}
As shown in $\sectref{EncodingHFour}$, each matrix $H_n \in \HadamardFamily$, where $n = 2^l$ can be represented by a CFLOBDD with $\bigO (l)$ vertices and edges---i.e., with $\bigO (\log n)$ space. 

\smallskip
\noindent
\textit{BDD Claim.}
We claim that regardless of the variable ordering, the BDD representation for $H_n$ requires at least $n$ nodes, one node for each variable in the argument.
The proof strategy follows a similar structure to the 
\emph{BDD Claim} proof in \theoref{EQSeparation}. 
We prove the claim by contradiction.
Suppose that there is some BDD $B$ for $H_n$ that does not need at least one node for each variable.
In that case, the $H_n$ function represented by $B$ does not depend on that particular variable.
Let $\mathcal{T}$ denote the ``all-true'' assignment of variables;
i.e., $\mathcal{T} \eqdef \forall k \in \{0..n/2-1\}, x_k \mapsto T, y_k \mapsto T$.
There are three possible situations:

\begin{itemize}
    \item Case 1: $B$ does not depend on variable $y_k$, for some $k \in \{0..n/2-1\}$. 
    Consider two variable assignments:
    $A_1 \eqdef \mathcal{T}$ and $A_2 \eqdef \mathcal{T}[y_k \mapsto F]$ (i.e.,
    $A_2$ is $A_1$ with $y_k$ updated to $F$).

    \hspace{1.5ex}
    $A_1$ and $A_2$ yield the same value for the function represented by $B$,
    but they yield different values for the Hadamard relation. 
    That is, if $H_n[A1] = v$ (where $v$ is either 1 or -1), then $H_n[A2] = -v$.
    We prove this claim by induction on level.
    \begin{proof}
    
    \smallskip
    \noindent
    Base Case: 
    \begin{itemize}
        \item $n = 2$. $H_2[A_1]$ is the lower-right corner of \figref{walsh1_dd}, which is -1, and
        $H_2[A_2]$ is the value of the path $[x_0 \mapsto T, y_0 \mapsto F]$, which yields 1.
        \item $n = 4$. $A_1$ is the path to the rightmost ($16^{\textit{th}}$) leaf in \figref{walsh2_dd}, which yields a value of 1.
        For $k = 0$, $A_2$ ends up at the $12^{\textit{th}}$ leaf, which is -1; if $k = 1$, $A_2$ ends up at the $15^{\textit{th}}$ leaf, which is also -1.
    \end{itemize}
    Induction Step: Let us extend the notation for $A_1$ and $A_2$ by adding level information.
    $A_1^m$ denotes the ``all-true'' assignment for $2^m$ variables, and $A_2^m = A_1^m[y_k \mapsto F]$.
    Let us assume that the claim is true for $H_{2^m}$, i.e., $H_{2^m}[A_1] = v$ (could be 1 or -1) and $H_{2^m}[A_2] = -v$.
    We must show that the claim holds true for $H_{2^{m+1}}$.

    \hspace{1.5ex}
    We know that $H_{2^{m+1}} = H_{2^m} \tensor H_{2^m}$.
    Thus, $H_{2^{m+1}}[A_1^{m+1}] = H_{2^m}[A_1^{m}] \ast H_{2^m}[A_1^{m}]$, where $A_1^{m + 1} = A_1^m || A_1^m$.
    Thus, $H_{2^{m+1}}[A_1^{m+1}]$ must have the value $v^2$ ($= v \ast v$).
    A recursive relation can similarly be written for assignment $A_2$, depending on where the bit-flip for $y$ occurs. There are two possible cases:
    \begin{enumerate}
        \item $k$ occurs in the first half; $A_2^{m + 1} = A_2^m || A_1^m$ and therefore, 
        $H_{2^{m+1}}[A_2^{m+1}] = H_{2^m}[A_2^{m}] \ast H_{2^m}[A_1^{m}]$, which leads to a value of $-v^2$ ($= -v \ast v$).
        \item $k$ occurs in the second half; $A_2^{m + 1} = A_1^m || A_2^m$ and therefore, 
        $H_{2^{m+1}}[A_2^{m+1}] = H_{2^m}[A_1^{m}] \ast H_{2^m}[A_2^{m}]$ which leads to a value of $-v^2$ ($= v \ast -v$).
    \end{enumerate}
    In both cases, the values obtained with a bit-flip do not match the value for an ``all-true'' assignment.
    \end{proof}
    We conclude that none of the $y_k$ variables can be dropped individually.

    \item Case 2: None of the $x$ variables can be dropped individually, using a completely analogous argument to Case 1. 
    \item Case 3: $B$ does not depend on either $x_k$ or $y_k$.
    The assignments
    $A_1 = \mathcal{T} = [...,{x_k}\mapsto T, {y_k}\mapsto T,...]$, 
    $A_2 = \mathcal{T}[y_k \mapsto F] = [...,{x_k}\mapsto T, {y_k}\mapsto F,...]$, 
    $A_3 = \mathcal{T}[x_k \mapsto F] = [...,{x_k}\mapsto F, {y_k}\mapsto T,...]$ and 
    $A_4 = \mathcal{T}[x_k \mapsto F, y_k \mapsto F] = [...,{x_k}\mapsto F, {y_k}\mapsto F,...]$
    must be mapped to different values by the function represented by $B$, which violates the definition of the Hadamard relation $H_n$.
    (More precisely, for $n \geq 4$, $H_n[A_1] = 1$, but $H_n[A_2] = H_n[A_3] = H_n[A_4] = -1$, which can be proved using an inductive argument similar to Case 1.)
    Consequently, $(x_k, y_k)$ cannot be dropped as a pair.
\end{itemize}

Because (i) no $y_k$ can be dropped individually, (ii) no $x_k$ can be dropped individually and (iii) no ($x_k, y_k$) pair can be dropped together, $B$---and hence any BDD representation for $H_n$, requires $\Omega(n)$ nodes.
\end{proof}

\subsection{The Addition Relation $\ADD_n$}
\label{Se:Separation:AdditionRelation}

\begin{definition}\label{De:AdditionRelation}
The addition relation $\ADD_n : \{0,1\}^{n/3} \times \{0,1\}^{n/3} \times \{0,1\}^{n/3} \rightarrow \{0,1\}$ on variables ($x_0\cdots x_{n/3-1}$), ($y_0\cdots y_{n/3-1}$), and ($z_0\cdots z_{n/3-1}$) is the relation
$
  \ADD_n(X,Y,Z) \eqdef Z = (X + Y \mod 2^{n/3})
$.
\end{definition}

\begin{theorem}[Exponential separation for the addition relation]
For $n =  3 \cdot 2^l$, where $l \geq 0$,
$\ADD_n$ can be represented by a CFLOBDD with $\bigO(\log{n})$ vertices and edges.
In contrast, a BDD that represents $\ADD_n$ requires $\bigOmega(n)$ nodes.
\end{theorem}

\begin{proof}
\,\newline

\begin{figure}[tb!]
    \centering
    \includegraphics[scale=0.30]{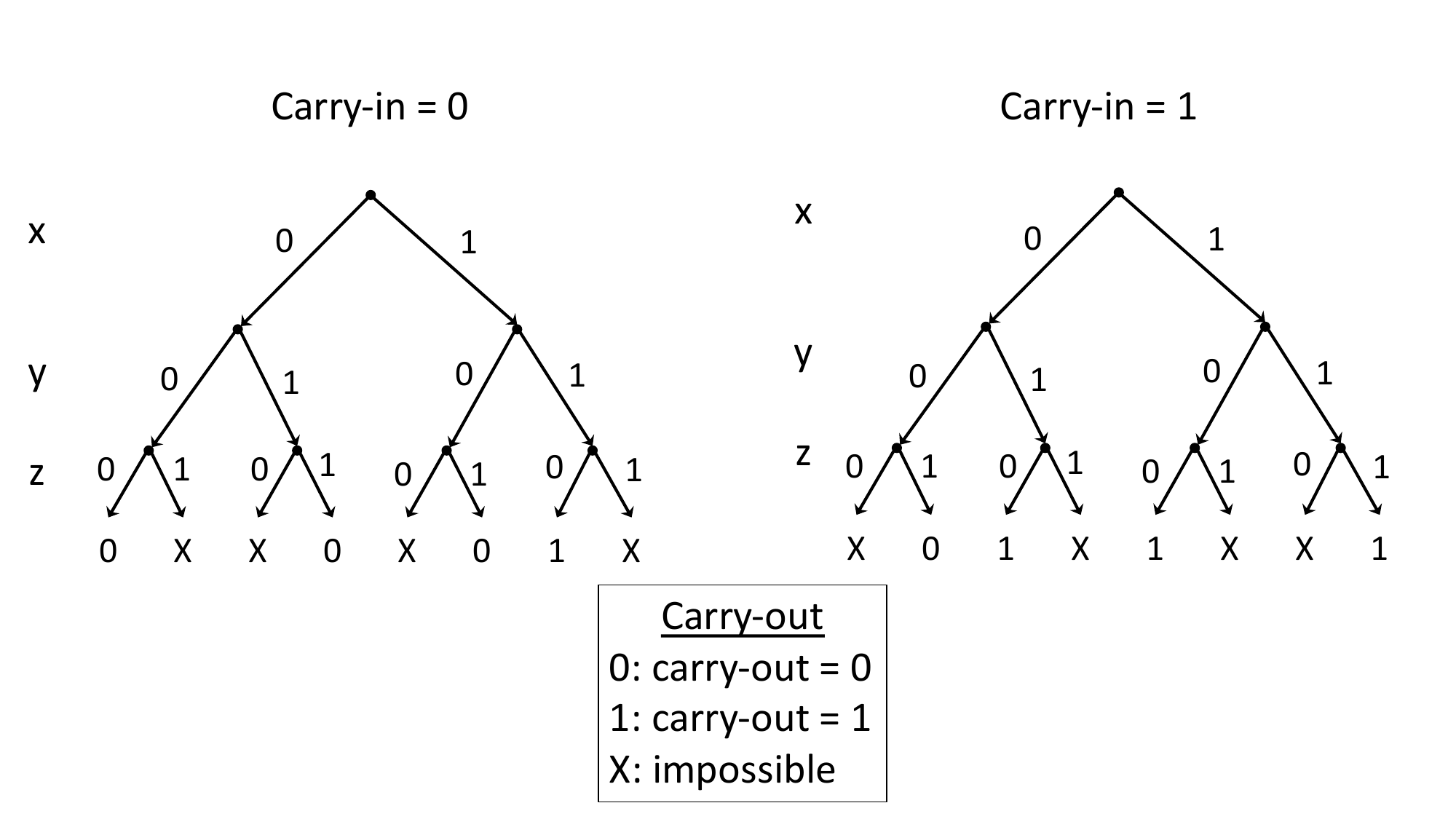}
    \caption{Decision diagrams for the $\textit{carry-in} = 0$ and $\textit{carry-in} = 1$ cases of $\ADD_n$}
    \label{Fi:add_dd}
\end{figure}

\smallskip
\noindent
\textit{CFLOBDD Claim.}
For a given bit position $i$, the representation has to distinguish between two cases for the carry-bit value coming from bit position $i-1$: $\textit{carry-in} = 0$ and $\textit{carry-in} = 1$.
\figref{add_dd} shows decision trees for the $\textit{carry-in} = 0$ and $\textit{carry-in} = 1$ cases.
The terminal values $0$ and $1$ indicate the carry-out value to be passed to bit position $i+1$. Terminal value $X$ represents a failure case: with the given carry-in value $c_i$, and the given values for $x_i$, $y_i$, and $z_i$, $z_i \neq c_i \xor x_i \xor y_1$.

\begin{figure}[tb!]
    \centering
    \includegraphics[scale=0.42]{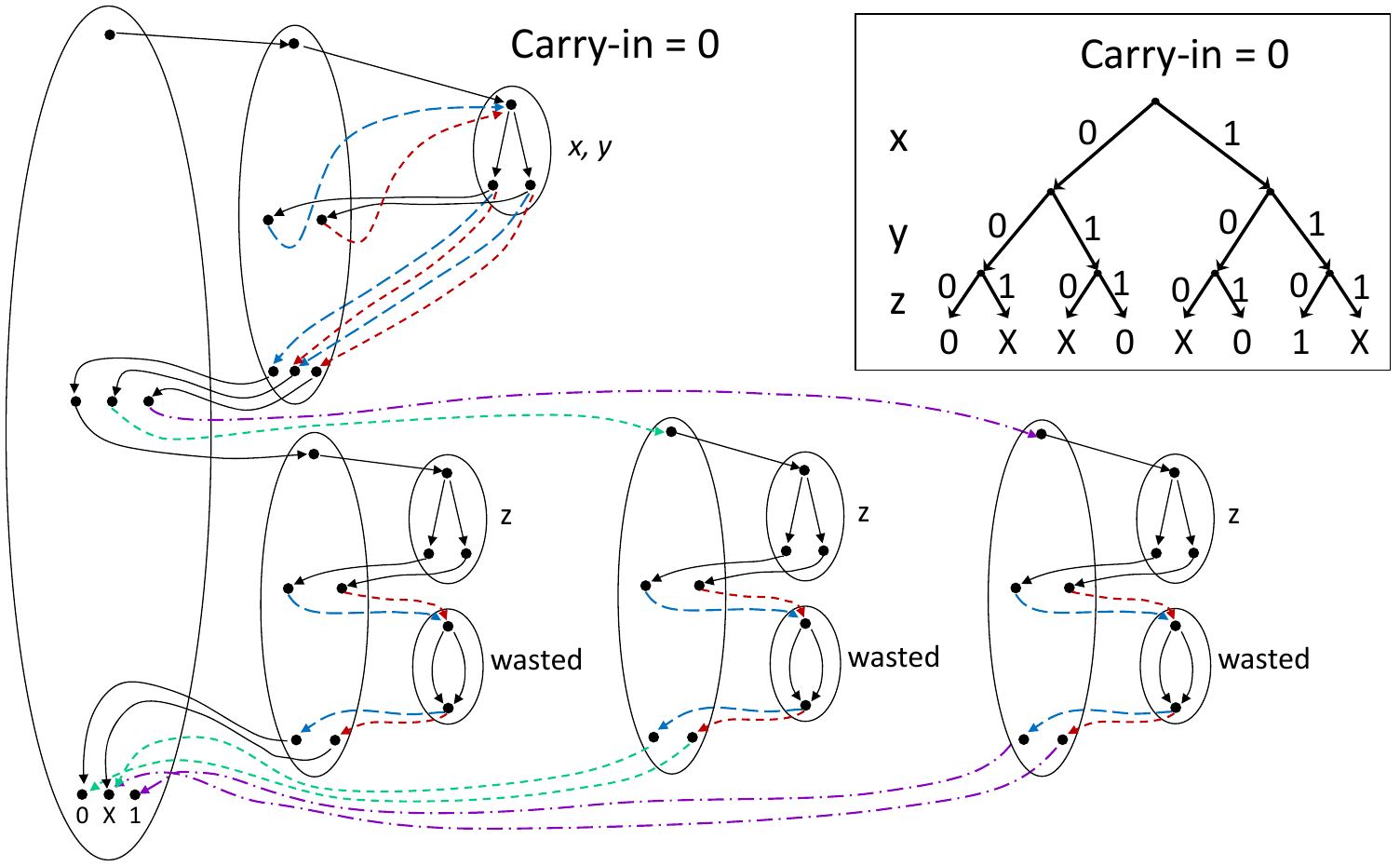}
    \caption{CFLOBDD representation for the $\textit{carry-in} = 0$ case of $\ADD_n$}
    \label{Fi:carry_in_0}
\end{figure}

\begin{figure}[tb!]
    \centering
    \includegraphics[scale=0.42]{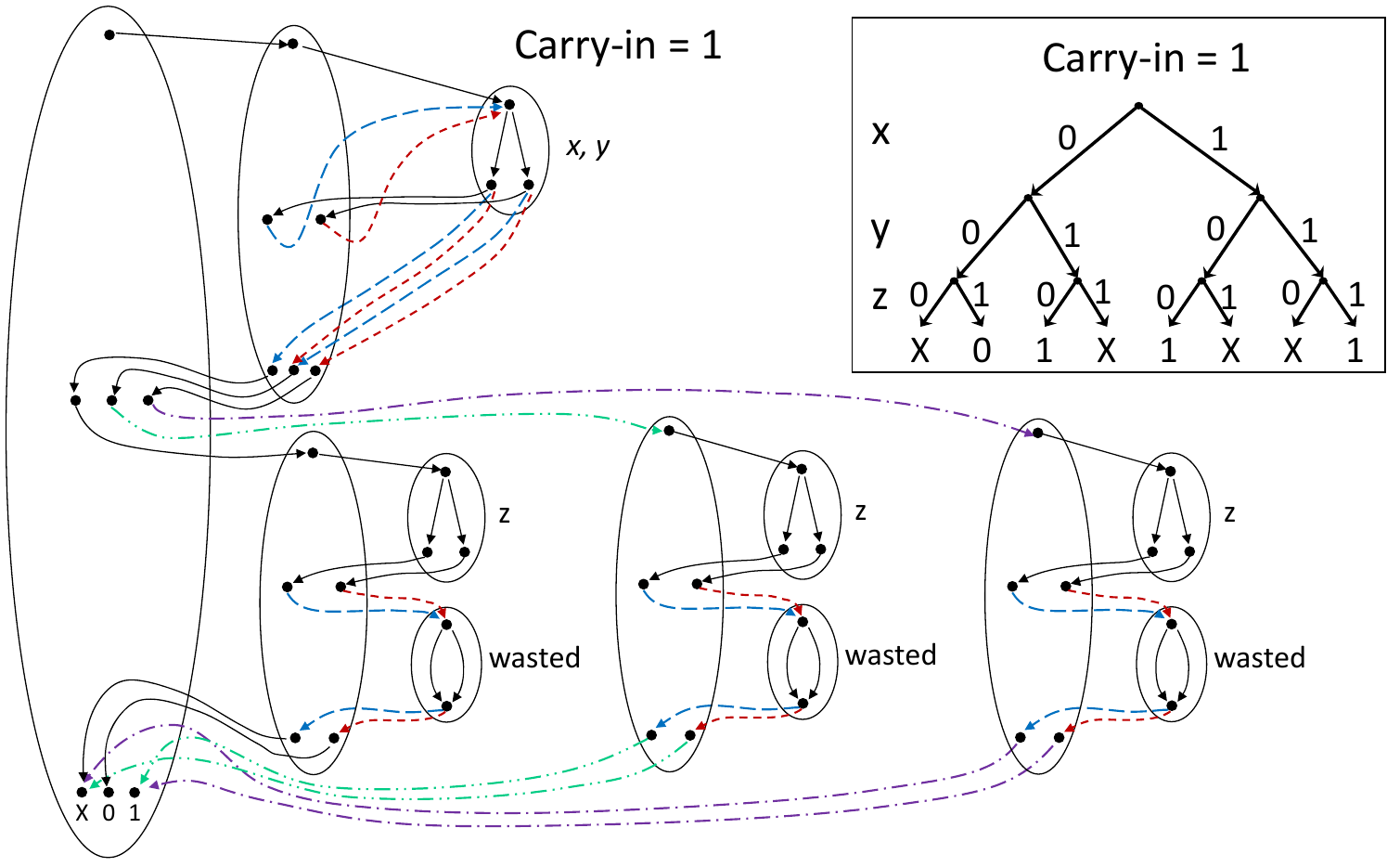}
    \caption{CFLOBDD representation for the $\textit{carry-in} = 1$ case of $\ADD_n$}
    \label{Fi:carry_in_1}
\end{figure}

\begin{figure}[tb!]
    \centering
    \begin{tabular}{cc}
      \includegraphics[scale=0.42]{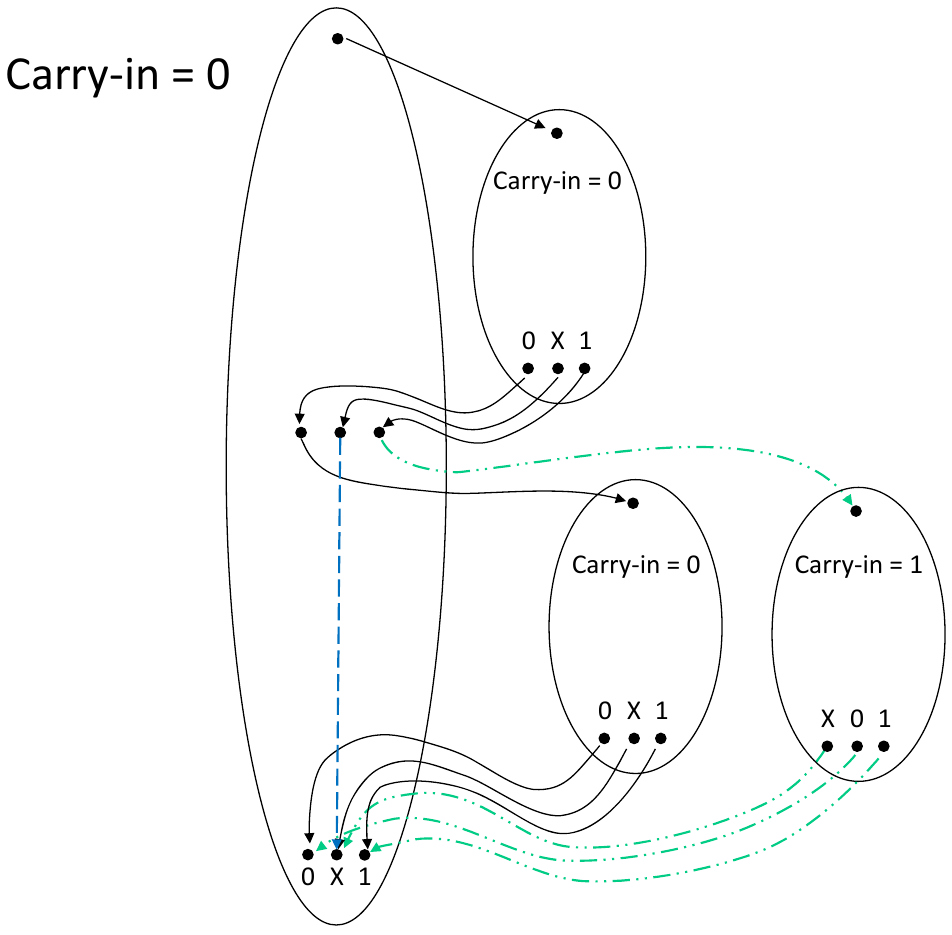}
      &
      \includegraphics[scale=0.42]{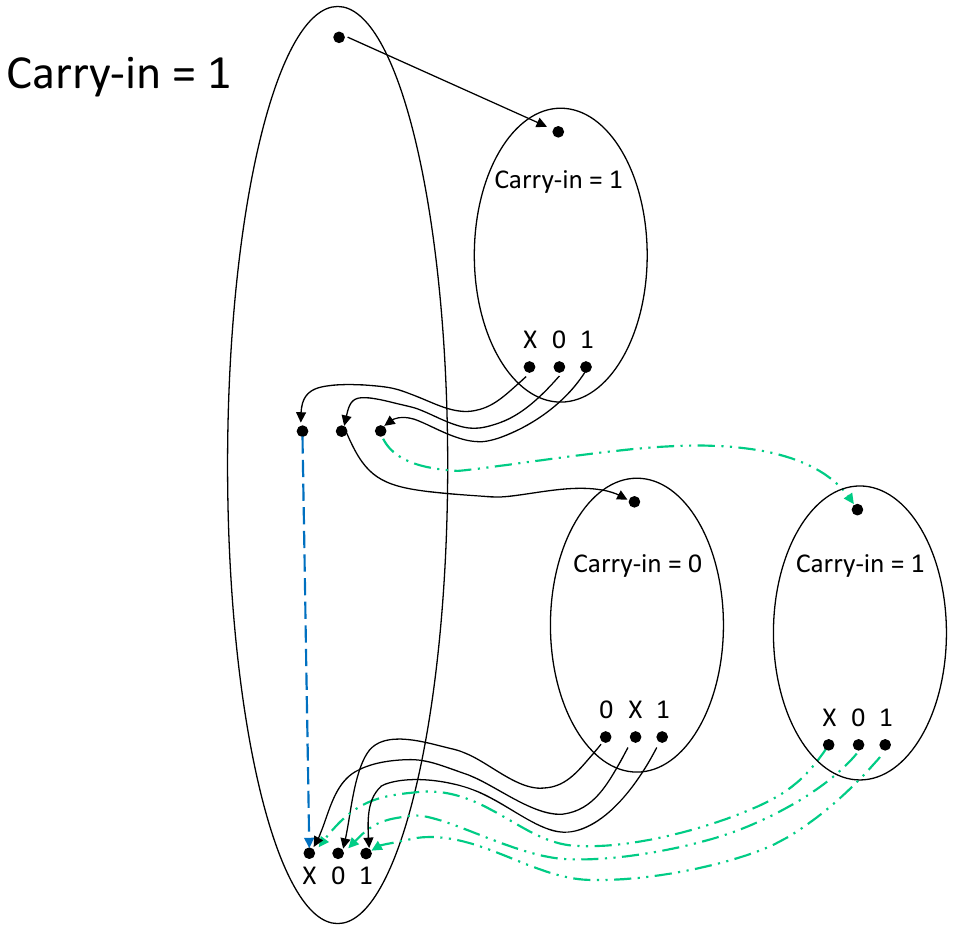}
      \\
      {\small (a)} & {\small (b)}
    \end{tabular}
    \caption{\protect \raggedright 
    Recursive CFLOBDD structures for the (a) $\textit{carry-in} = 0$ and (b) $\textit{carry-in} = 1$ cases of $\ADD_n$.
    To reduce clutter, a B-connection to a NoDistinctionProtoCFLOBDD is depicted as a straight dashed line to the $X$ exit vertex.
    }
    \label{Fi:carry_in_01_rec}
\end{figure}

\begin{figure}[tb!]
    \centering
    \begin{tabular}{cc}
      \multicolumn{2}{c}{\includegraphics[scale=0.45]{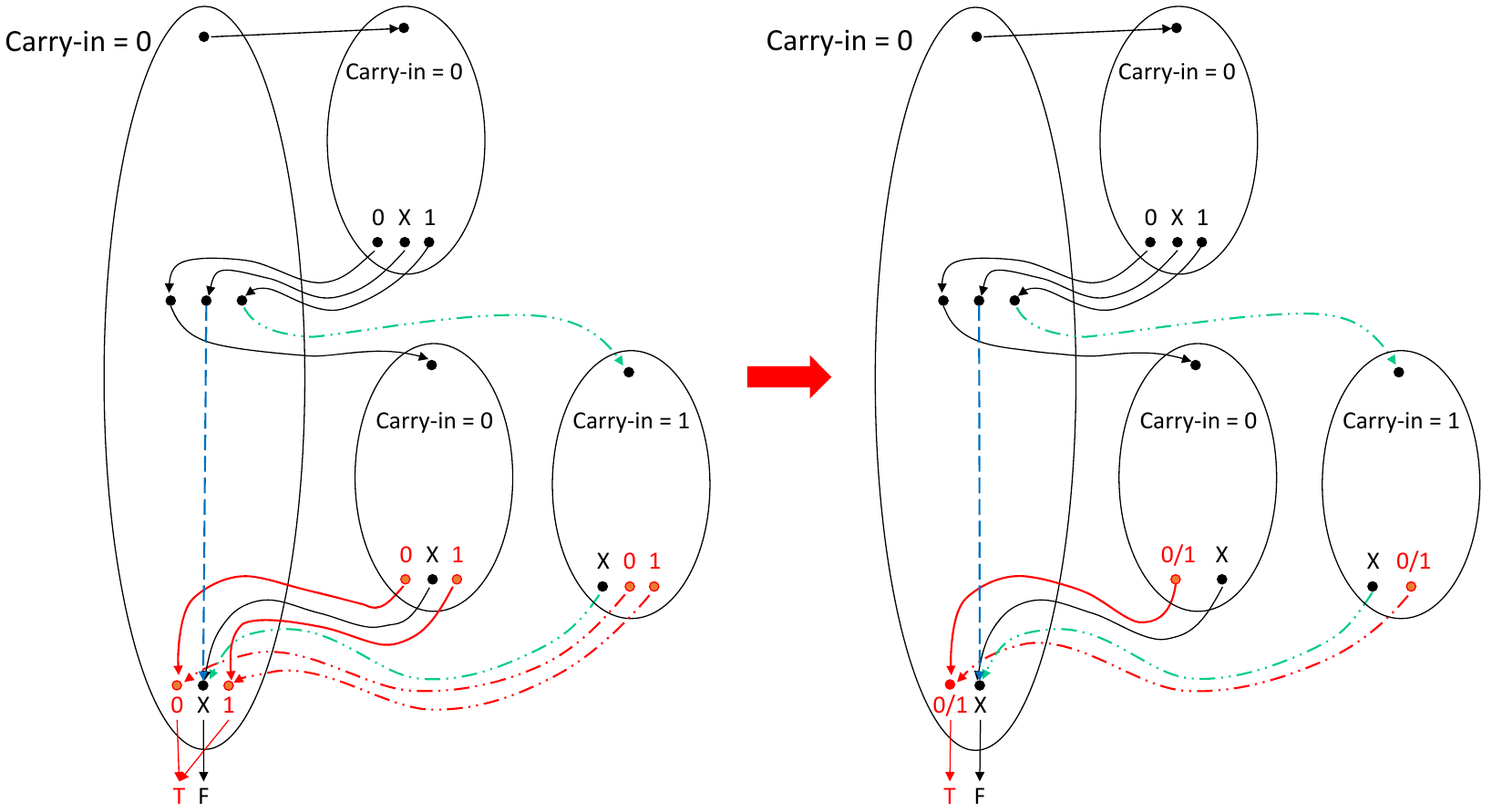}}
      \\
      \qquad\qquad\qquad {\small (a)}\label{Fi:top_level_add_without_collapse} 
      & \qquad\qquad\qquad {\small (b)}\label{Fi:top_level_add_with_collapse}
    \end{tabular}
    \caption{\protect \raggedright 
    (a) The recursive CFLOBDD structure of the $\ADD_n$ relation.
    Because the $0$ and $1$ exit vertices of the top-level grouping are associated with the single terminal $T$ (while $X$ is mapped to $F$), the $0$ and $1$ exits must be coalesced (indicated in red), which propagates to the internal groupings of the CFLOBDD.
    (b) The propagation of coalesced exit vertices.
    }
    \label{Fi:top_level_add}
\end{figure}

The CFLOBDD representation for $\ADD_n$ using the interleaved-variable ordering is shown in~\figseqref{carry_in_0}{top_level_add}.
Our claim is that $\ADD_n$ has $\bigO(\log{n})$ vertices and edges.
To keep triples of variables $x_i$, $y_i$, and $z_i$ aligned with the power-of-two nature of CFLOBDDs, our representation of $\ADD_n$ uses an additional set of $n/3$ dummy variables (which will always be ``routed'' through a DontCareGrouping, so they have no effect on the relation over $X$, $Y$, and $Z$).

\figref{carry_in_0} shows the proto-CFLOBDD at level $2$ for the $\textit{carry-in} = 0$ case for some triple of variables $x_i$, $y_i$, and $z_i$ in the $\ADD_n$ CFLOBDD.
The DontCareGroupings labeled ``wasted'' in \figref{carry_in_0} are for the associated $i^{\textit{th}}$ dummy variable.
\figref{carry_in_0} is one of two ``leaf'' building blocks that handles a triple of variables;
the other is shown in \figref{carry_in_1}, which depicts the CFLOBDD representation for the $\textit{carry-in} = 1$ case for $x_i$, $y_i$, and $z_i$.

\figref{carry_in_01_rec} shows the two recursive structures used at all higher levels (> 2) of the CFLOBDD for the $\textit{carry-in} = 0$ and  $\textit{carry-in} = 1$ cases.
Note that at all levels, including the base case of \figref{carry_in_0}, the sequence for the exit vertices of $\textit{carry-in} = 0$ proto-CFLOBDDs is $[0, X, 1]$.
Similarly, at all levels the sequence for the exit vertices of $\textit{carry-in} = 1$ proto-CFLOBDDs is $[X, 0, 1]$.
Consequently, at each level, we can construct the $\textit{carry-in} = 0$ and $\textit{carry-in} = 1$ proto-CFLOBDDs by adding just one grouping for each case, and each of the groupings contains a fixed number of vertices and edges.

At the outermost level, we use the $\textit{carry-in} = 0$ proto-CFLOBDD (which has an A-connection to the $\textit{carry-in} = 0$ proto-CFLOBDD at the next lower level, but B-connections to both the $\textit{carry-in} = 0$ and $\textit{carry-in} = 1$ variants).
The exit vertices labeled $0$ and $1$ are connected to terminal value $T$, and exit vertex $X$ to terminal value $F$.
\figref{top_level_add}a shows this structure.

To obey structural invariant \ref{Inv:6} of \defref{CFLOBDD}, it is necessary to collapse the $0$ and $1$ exit vertices at the outermost level, because both lead to the terminal value $T$.
This collapsing process propagates to the different levels of internal groupings of the CFLOBDD.
What remains to be established is that the collapsing step does not cause the CFLOBDD as a whole to blow up in size.\footnote{
  The reason a ``collapsing'' step could cause a size blow-up is because occurrences of previously shared identical substructures could turn into multiple substructures, each slightly different.
}

Fortunately, as indicated in \figref{top_level_add}b, the propagation only takes place in groupings in B-connections (and B-connections of B-connections, etc.), and does not propagate to A-Connection groupings.
As we prove below, the collapsing step only produces two more kinds of proto-CFLOBDDs at each level:
(i) a $\textit{carry-in} = 0$ variant in which the exit-vertex sequence is
$[0/1, X]$, and
(ii) a $\textit{carry-in} = 1$ variant in which the exit-vertex sequence is
$[X, 0/1]$ (where $0/1$ denotes the coalescing of the $0$ and $1$ exit vertices).
Because there are now at most four grouping patterns that arise at each level, the final CFLOBDD at most doubles in size.
The overall size of the CFLOBDD is proportional to the outermost level, and hence the CFLOBDD representation of $\ADD_n$ achieves double-exponential compression with respect to the size of the decision tree for $\ADD_n$ (i.e., $\bigO(\log{n})$ in vertices and edges, and $\bigO(l)$ in the number of levels ($l = \log{n})$).

More formally, let $R(l)$ and $P(l)$ denote the number of vertices and edges at or below level $l$ in the proto-CFLOBDD representation of $\ADD_n$ for $\textit{carry-in} = 0$ and $\textit{carry-in} = 1$, respectively, \emph{without} collapsing the $0$ and $1$ exits.
Also, let $A(l)$ and $B(l)$ denote the number of vertices and edges at or below level $l$ in the proto-CFLOBDD representation of $\ADD_n$ for $\textit{carry-in} = 0$ and $\textit{carry-in} = 1$, respectively, \emph{with} $0$ and $1$ exits collapsed.
We can give mutually recursive recurrence equations for $R(l)$ and $P(l)$.
For $R(l)$, we have
\begin{equation}
    \label{Eq:CarryIn0WithoutCollapse}
    R(l) = R(l-1) + P(l-1) + 7 + 13,
\end{equation}
which defines $R(l)$ in terms of $R(l-1)$ and $P(l-1)$ according to the pattern given in \figref{carry_in_01_rec}a.
Although \figref{carry_in_01_rec}a shows two $\textit{carry-in} = 0$ proto-CFLOBDDs at level $l-1$, they are actually shared data structures, and so $R(l-1)$ only counts once in \eqref{CarryIn0WithoutCollapse}.
We call this property the ``\emph{same-structure sharing}'' property:
stated another way, all non-zero coefficients of $R$, $P$, $A$, and $B$ terms can be replaced by $1$.
The 7 in \eqref{CarryIn0WithoutCollapse} represents the number of vertices at level $l$, and $13$ refers to the edge count between groupings at level $l$ and level $l-1$.
Similarly, the recurrence equation for $P(l)$, following the pattern in \figref{carry_in_01_rec}b, is
\begin{equation}
    \label{Eq:CarryIn1WithoutCollapse}
    P(l) = P(l-1) + R(l-1) + 7 + 13.
\end{equation}

Combining~\eqrefs{CarryIn0WithoutCollapse}{CarryIn1WithoutCollapse}, we obtain
    \begin{align}
    \label{Eq:CombiningRAndP}
        R(l) + P(l) &= (R(l-1) + P(l-1) + 7 + 13) + (P(l-1) + R(l-1) + 7 + 13) \notag\\
        &= 2R(l-1) + 2P(l-1) + 2*(7 + 13) \tag*{Collecting terms}\\
        &= R(l-1) + P(l-1) + 40 \tag*{Same-structure sharing}\\
        &= R(l-2) + P(l-2) + 2*40 \tag*{Substitution}\\
        & \hspace{12ex} \vdots \notag\\
        &= R(l-k) + P(l-k) + k*40 \tag*{For a general $k$}\\
        & \hspace{12ex} \vdots \notag\\
        &= R(2) + P(2) + (l-2)*40 \notag\\
        &= \bigO(1) + \bigO(l) \tag*{From \figrefs{carry_in_0}{carry_in_1}}\\
        &= \bigO(l)
    \end{align}
Substituting~\eqref{CombiningRAndP} in~\eqrefs{CarryIn0WithoutCollapse}{CarryIn1WithoutCollapse}, we obtain
\begin{equation}
\label{Eq:Req}
\begin{split}
    R(l) &= \bigO(l-1) + 20\\
        &= \bigO(l-1) + \bigO(1)\\
        &= \bigO(l)
\end{split}
\end{equation}
\begin{equation}
\label{Eq:Peq}
        \begin{split}
    P(l) &= \bigO(l-1) + 20\\
        &= \bigO(l-1) + \bigO(1)\\
        &= \bigO(l)
\end{split}
\end{equation}

The argument for $A(l)$ and $B(l)$ is similar.
We define $A(l)$ in terms of $R(l-1)$, $A(l-1)$, and $B(l-1)$ following the pattern in \figref{top_level_add}b.
\begin{equation}
    \label{Eq:CarryIn0WithCollapse}
    A(l) = R(l-1) + A(l-1) + B(l-1) + 6 + 12,
\end{equation}
where the numbers of vertices and edges added at each level are 6 and 12, respectively.
Similarly, the recurrence equation for $B(l)$ is
\begin{equation}
    \label{Eq:CarryIn1WithCollapse}
    B(l) = P(l-1) + A(l-1) + B(l-1) + 6 + 12
\end{equation}

Combining ~\eqrefs{CarryIn0WithCollapse}{CarryIn1WithCollapse}, we obtain
\begin{align}
    \label{Eq:CombiningAAndB}
        A(l) + B(l) &= (R(l-1) + A(l-1) + B(l-1) + 6 + 12) + (P(l-1) + A(l-1) + B(l-1) + 6 + 12) \notag\\
        &= R(l-1) + P(l-1) + 2A(l-1) + 2B(l-1) + 2*(6 + 12) \tag*{Collecting terms}\\
        &= R(l-1) + P(l-1) + A(l-1) + B(l-1) + 36 \tag*{Same-structure sharing}\\
        &= (R(l-2) + P(l-2) + 20) + (P(l-2) + R(l-2) + 20) \\
        &+ (R(l-2) + A(l-2) + B(l-2) + 6 + 12) \\
        &+ (P(l-2) + A(l-2) + B(l-2) + 6 + 12) + 36 \tag*{Substitution}\\
        &= 3R(l-2) + 3P(l-2) + 2A(l-2) + 2B(l-2) + 40 + 36 + 36 \tag*{Collecting terms}\\
        &= R(l-2) + P(l-2) + A(l-2) + B(l-2) + 40 + 2*36 \tag*{Same-structure sharing}\\
        & \hspace{12ex} \vdots \notag\\
        &= R(l-k) + P(l-k) + A(l-k) + B(l-k) + (k-1)*40 + k*36 \tag*{For a general \mbox{$k$}}\\
        & \hspace{12ex} \vdots \notag\\
        &= R(2) + P(2) + A(2) + B(2) + (l-3)*40 + (l-2)*36 \notag\\
        &= \bigO(1) + \bigO(l) \tag*{From \figrefs{carry_in_0}{carry_in_1}}\\
        &= \bigO(l)
\end{align}

Using~\eqref{CombiningAAndB}, we can rewrite $A(l)$ and $B(l)$ as follows:
\begin{equation}
\label{Eq:Aeq}
\begin{split}
    A(l) &= \bigO(l-1) + \bigO(l-1) + 18\\
        &= \bigO(l-1) + \bigO(1)\\
        &= \bigO(l)
\end{split}
\end{equation}
\begin{equation}
\label{Eq:Beq}
        \begin{split}
    B(l) &= \bigO(l-1) + \bigO(l-1) + 18\\
        &= \bigO(l-1) + \bigO(1)\\
        &= \bigO(l)
\end{split}
\end{equation}

\eqrefspp{Req}{Peq}{Aeq}{Beq} show that $R(l)$, $P(l)$, $A(l)$, and $B(l)$ are all linear in the number of levels---$\bigO(l)$---and logarithmic in the number of vertices and edges---$\bigO(\log{n})$.
Thus, for the interleaved-variable ordering for (vectors of) variables $X$, $Y$, and $Z$, the vertices and edges count for the CFLOBDD for $\ADD_n$ is $\bigO(\log{n})$.

\smallskip
\noindent
\textit{BDD Claim.}
To represent the addition relation $\ADD_n$ as a BDD, we claim that we require at least $n$ nodes, one node for each variable in the argument, regardless of the variable ordering. 

We will prove this claim by contradiction, similar to the \emph{BDD Claim} proofs for $\EQ_n$ and $H_n$.
We assume that a BDD representation $B$ of $\ADD_n$ does not need at least one node for each variable, and therefore the BDD representation of $\ADD_n$ does not depend on that particular variable. 
Let us define $\mathcal{F}$ as an ``all-false'' assignment of variables, i.e., 
$\mathcal{F} \eqdef \forall k \in \{0...n/3-1\}, x_k \mapsto F, y_k \mapsto F, z_k \mapsto F.$
There are seven possible cases:
\begin{itemize}
    \item Case 1: $B$ does not depend on variable $y_k$, for some $k \in \{ 0..n/3-1 \} $. 
    Let us consider two different assignments of variables:
    $A_1 \eqdef \mathcal{F}$ and $A_2 \eqdef \mathcal{F}[y_k \mapsto T]$;
    i.e., $A_2$ is $A_1$ with $y_k$ updated to $T$.
    (Note that $A_1 = [...,{x_k}\mapsto F, {y_k}\mapsto F, {z_k} \mapsto F,...]$ and 
    $A_2 = [...,{x_k}\mapsto F, {y_k}\mapsto T, {z_k} \mapsto F,...]$.)
    Because $B$ does not depend on $y_k$, $B[A_1]$ must equal $B[A_2]$, which violates the definition of addition relation $\ADD_n$;
    in particular, $\ADD_n[A_1] = 1$ and $\ADD_n[A_2] = 0$.
    Let us show how the violation occurs by using the example of 
    $n = 24$ and $k = 3$ $(\in \{0..7\})$.
    We have $\ADD_n[A_1] = 1$ because the following triple of numbers is a correct instance of an addition problem:
    \[
      \begin{array}{rcccccccc}
           & 0 & 0 & 0 & 0 & 0 & 0 & 0 & 0 \\
         + & 0 & 0 & 0 & 0 & 0 & 0 & 0 & 0 \\
         \hline
           & 0 & 0 & 0 & 0 & 0 & 0 & 0 & 0 \\
      \end{array}
    \]
    However, $\ADD_n[A_2] = 0$ because the following triple is an incorrect instance of addition:
    \[
      \begin{array}{rcccccccc}
           & 0 & 0 & 0 & 0 & 0 & 0 & 0 & 0 \\
         + & 0 & 0 & 0 & 0 & 1 & 0 & 0 & 0 \\
         \hline
           & 0 & 0 & 0 & 0 & 0 & 0 & 0 & 0 \\
      \end{array}
    \]
    We conclude that none of the $y_k$ variables can be dropped individually.
    The argument for arbitrary $n$ and $k$ is completely analogous.
    \item Case 2: $B$ does not depend on the pair $(x_k, y_k)$.
    We use a similar proof strategy as Case 1.
    Consider two assignments $A_1 \eqdef \mathcal{F}$ and $A_2 = \eqdef A_1[x_k \mapsto T][y_k \mapsto T]$.
    The assignments must produce the same value for the function represented by $B$, but they yield different values for the $\ADD_n$ relation.
    Again, let us show how the violation occurs by using the example of $n = 24$ and $k = 3$.
    We have $\ADD_n[A_1] = 1$ because the following triple is a correct instance of an addition problem:
    \[
      \begin{array}{rcccccccc}
           & 0 & 0 & 0 & 0 & 0 & 0 & 0 & 0 \\
         + & 0 & 0 & 0 & 0 & 0 & 0 & 0 & 0 \\
         \hline
           & 0 & 0 & 0 & 0 & 0 & 0 & 0 & 0 \\
      \end{array}
    \]
    However, $\ADD_n[A_2] = 0$ because the following triple is an incorrect instance of addition:
    \[
      \begin{array}{rcccccccc}
           & 0 & 0 & 0 & 0 & 1 & 0 & 0 & 0 \\
         + & 0 & 0 & 0 & 0 & 1 & 0 & 0 & 0 \\
         \hline
           & 0 & 0 & 0 & 0 & 0 & 0 & 0 & 0 \\
      \end{array}
    \]
    We conclude that $(x_k, y_k)$ cannot be dropped as a pair.
\end{itemize}
We can make arguments similar to the ones above for other combinations of variables, such as
(i) $B$ does not depend on variable $x_k$ individually, 
(ii) $B$ does not depend on variable $z_k$ individually, 
(iv) $B$ does not depend on either $y_k$ or $z_k$ variables as a pair,
(v) $B$ does not depend on either $x_k$ or $z_k$ variables as a pair, 
(vi) $B$ does not depend on all three variables $x_k, y_k, z_k$ together.
Consequently, we conclude that $B$---and hence any BDD representation for $\ADD_n$ requires $\Omega(n)$ nodes.
\end{proof}


\section{Applications to Quantum Algorithms}
\label{Se:quantum-algos}

For certain problems, algorithms run on quantum computers achieve polynomial to exponential speed-ups over their classical counterparts.
However, to date, there are no practical large-scale implementations of quantum computers.
Hence, simulating quantum circuits on classical computers can provide insight on how quantum algorithms perform and scale.
In this section and \sectref{evaluation}, we explore the potential of CFLOBDDs for simulating quantum circuits.\footnote{
  No knowledge of quantum algorithms is assumed.
  Everything can be understood in terms of some simple linear algebra.
  The only subtle concept is that some of the $2^n \times 2^n$ matrices are best thought of in terms of their effect on the \emph{indices} of positions in vectors of size $2^n \times 1$.
  For instance, the matrix
  $I \tensor \begin{bmatrix}
               0 & 1 \\
               1 & 0 \\
             \end{bmatrix}$
  $~= \begin{bNiceArray}{cccc}[first-col,first-row]
          & 00 & 01 & 10 & 11 \\
       00 &  0 &  1 &  0 &  0 \\
       01 &  1 &  0 &  0 &  0 \\
       10 &  0 &  0 &  0 &  1 \\
       11 &  0 &  0 &  1 &  0
     \end{bNiceArray}$
  maps the unit vectors
  $e_{00} = \begin{bNiceArray}{c}[first-col]
              00 & 1 \\
              01 & 0 \\
              10 & 0 \\
              11 & 0
           \end{bNiceArray}$,
  $e_{01} = \begin{bNiceArray}{c}[first-col]
              00 & 0 \\
              01 & 1 \\
              10 & 0 \\
              11 & 0
           \end{bNiceArray}$,
  $e_{10} = \begin{bNiceArray}{c}[first-col]
              00 & 0 \\
              01 & 0 \\
              10 & 1 \\
              11 & 0
           \end{bNiceArray}$, and
  $e_{11} = \begin{bNiceArray}{c}[first-col]
              00 & 0 \\
              01 & 0 \\
              10 & 0 \\
              11 & 1
           \end{bNiceArray}$
  to $e_{01}$, $e_{00}$, $e_{11}$, and $e_{10}$, respectively.
  In other words, on a unit vector, the effect is to negate the final bit of the index that specifies the position of the 1.
}

A single-qubit quantum state can be represented by a pair of complex numbers (i.e., a vector of size $2 \times 1$).
A quantum state of $n$ qubits can be represented by a complex-valued vector of size $2^n \times 1$;
hence, the information capacity increases exponentially with the number of qubits.

A \emph{quantum circuit} takes as input an initial quantum-state vector, and applies a sequence of \emph{quantum gates}, which are each length-preserving transformations, and can be expressed as unitary matrices.
Thus, quantum-circuit simulation requires a way to perform linear algebra with vectors of size $2^n$ and matrices of size $2^n \times 2^n$, where $n$ is the number of qubits involved.
Examples of gates that operate on single qubits are
$I = \left[\begin{smallmatrix}
        1 & 0 \\ 0 & 1  
      \end{smallmatrix}\right]$,
$H = \frac{1}{\sqrt{2}}\left[\begin{smallmatrix}
        1 & 1 \\ 1 & -1
                        \end{smallmatrix}\right]$, and
$X = \left[\begin{smallmatrix}
        0 & 1 \\ 1 & 0
      \end{smallmatrix}\right]$.
$I$ leaves a quantum state as is;
the Hadamard gate $H$ sends a basis state to a state in ``superposition'' (i.e., a state that is a non-trivial linear combination of basis states);\footnote{
  \label{Footnote:HadamardGate}
  A Hadamard gate that operates on a single qubit is the Hadamard matrix $H_2$ from \figref{HadamardMatrices} (\sectref{Preliminaries}), scaled  by $\frac{1}{\sqrt{2}}$ so that it is a unitary matrix.
}
$X$ complements the indices of a qubit's basis states, and thus flips the positions of the amplitudes, sending $\begin{bNiceArray}{cc}[first-col,small,code-for-first-col = \scriptscriptstyle]
              0 & \alpha_0 \\
              1 & \alpha_1
           \end{bNiceArray}$
to $\begin{bNiceArray}{cc}[first-col,small,code-for-first-col = \scriptscriptstyle]
              0 & \alpha_1 \\
              1 & \alpha_0
           \end{bNiceArray}$.
Let $M^{\tensor k}$ denote the $k$-fold Kronecker product of $M$ with itself: $M^{\tensor k} = \overbrace{M \tensor M \tensor \ldots \tensor M}^{k~\textrm{occurrences of}~M}$.
The quantum gate that, e.g., applies $H$ to the $j^{\textit{th}}$ qubit of an $n$-qubit quantum state is $I^{\tensor (j-1)} \tensor H \tensor I^{\tensor (n-j)}$.

A quantum state could be encoded with a decision tree of height $n$, but such a representation would be inefficient.
The potential of CFLOBDDs is for providing (up to) double-exponential compression in the sizes of the vectors and matrices that arise during quantum simulation,
using $\log{n}$ and $\log{n} + 1$ levels, respectively.
Because many quantum gates can be described using Kronecker products, there is great potential for them to have a compact representation as a CFLOBDD.

The evaluation of CFLOBDDs for quantum simulation in \sectref{ResearchQuestionTwo} uses the CFLOBDD representations of gate matrices and state vectors presented in this section,
namely, multi-terminal CFLOBDDs with a semiring value at each terminal value ({\`a} la ADDs \cite{iccad:BFGHMPS93}).
To support functions of type ${\{ 0,1 \}}^n \rightarrow \mathbb{C}$, we implemented a semiring of multi-precision-floating-point \cite{fousse2007mpfr} complex numbers.

\paragraph{Notation.}
Generally, when we wish to emphasize the dimensions of objects, a state vector with $n$ qubits (of size $2^n\times 1$) is denoted using a subscript $n$, and a gate matrix acting on $n$ qubits (of size $2^n\times 2^n$) is denoted using a subscript $2n$.
Although the subscripts differ, a vector $V_n$ and matrix $M_{2n}$ have compatible sizes in the matrix-vector product $M_{2n} V_n$.

A completely different subscript convention is used to denote basis vectors:
$e_s$ denotes the vector with a one in position $s$ (where $s$ is interpreted in binary) and zeros elsewhere.
To be concise, we sometimes use ``exponent notation'' from formal-language theory to express $s$.
For instance, $e_{0^n}$ and $e_{1^n}$ denote the basis vectors
  $e_{\underbrace{0 \ldots 0}_{n~\textit{copies}}}$ and
  $e_{\underbrace{1 \ldots 1}_{n~\textit{copies}}}$, respectively.

Sometimes both conventions come into play.
For example, $H_{2n} e_{0^n}$ involves a matrix that acts on $n$ qubits applied to a basis vector with $n$ qubits.

\paragraph{Organization.}
In the present section, we articulate some advantages of quantum simulation (\sectref{AdvantagesOfSimulation}),
and then discuss the application of CFLOBDDs to this domain.
In \sectref{SpecialMatrices}, we give constructions of CFLOBDDs for some important families of vectors and matrices used in quantum algorithms.
In \sectref{QuantumAlgorithms}, we briefly summarize some of the quantum algorithms used in the experiments in \sectref{ResearchQuestionTwo}.

To aid the reader, the following table indicates where to find details about the CFLOBDD operations needed to simulate a quantum circuit:

\medskip
\qquad\begin{tabular}{r||r|l}
  \hline
  \multicolumn{1}{l||}{\emph{State construction}}  &   &  \\
  Standard-basis vector  & \sectref{StandardBasisVectors} & \algref{StandardBasisVectorAlgo} \\
  \multicolumn{1}{l||}{\emph{Gate construction}}   &   &  \\
  Identity gate         & \sectref{IdentityMatrix} & \algref{IdAlgo} \\
  Hadamard gate         & \figref{walshKGeneralCase} and \sectref{HadamardMatrix} & \algref{HAlgo} \\
  Not gate              & \sectref{NotMatrix} & Variant of \algref{IdAlgo} \\
  $\CNOT$ gate          & \sectref{CNOTMatrix} & Appendix \sectref{CNOTConstructionAlgo} \\
  \multicolumn{1}{l||}{\emph{Operations}}         &   &  \\
  Kronecker product     & \sectref{KroneckerProduct} & \algref{KP4Voc} \\
  \qquad Matrix-matrix multiplication & \sectref{matrix-mult} & \algref{MatrixMult} \\
  Vector-matrix and matrix-vector multiplication & \sectrefs{VectorToMatrixConversion}{matrix-mult} & \algrefs{vector2matrix}{MatrixMult} \\
  Application of QFT    & \sectref{QuantumFourierTransform} & Appendix \sectref{QuantumGates} \\
  \multicolumn{1}{l||}{\emph{Measurement}} & \sectref{PathCountingAndSampling} & \algref{Sampling} \\
  \hline
\end{tabular}

\subsection{Advantages of Simulation}
\label{Se:AdvantagesOfSimulation}

Simulation of a quantum circuit can have advantages compared to actually running a circuit on a quantum computer:
\begin{enumerate}
  \item
    \label{It:DeviateFromQuantumModel}
    A simulation can deviate from certain requirements of the quantum-computation model and perform the simulation in a way that no quantum device could.
    \begin{enumerate}
      \item
        \label{It:Deviate:RepeatedSquaring}
        Some quantum algorithms perform multiple iterations of a particular quantum operator $\Op$ (e.g., $k$ iterations, where $k$ is some power of $2$).
        \emph{A simulation can operate on $\Op$ itself}
        \cite[Ch.\ 6]{Book:ZW2020},
        using repeated squaring to create the sequence of derived operators $\Op^2$, $\Op^4$, $\Op^8$, $\ldots$, $\Op^{2^{\log{k}}} = \Op^k$,
        which can be accomplished in $\log{k}$ iterations.
        The final answer is then obtained using $\Op^k$.
        A physical quantum computer can only \emph{apply $\Op$ sequentially}, and thus must perform $k$ applications of $\Op$.
        This approach is particularly useful in simulating Grover's algorithm (\sectref{groveralgo}).
      \item
        \label{It:Deviate:GoOutsideQuantumModel}
        The quantum-computation model requires the use of a limited repertoire of operations:
        every operation is a multiplication by a unitary matrix, and all results (and all intermediate values) must be produced in this way.
        In contrast, it is acceptable for a simulation to create some intermediate results by alternative pathways.
        In some cases, our simulation of a quantum algorithm directly creates a CFLOBDD that represents an intermediate value, thereby avoiding a sequence of potentially more costly computational steps that stay within the quantum model.
        (See ``A Special-Case Construction'' in \sectref{CNOTMatrix}, which is used in \sectref{groveralgo}.)
      \item
        \label{It:Deviate:RepeatedMeasurements}
        In many quantum algorithms, the final state needs to be measured multiple times.
        When running on a physical quantum computer, part or all of the quantum state is destroyed after each measurement of the state, and thus the quantum steps must be re-performed before each successive measurement.
        In contrast, in a quantum simulation the quantum steps need only be performed once, and then \emph{multiple measurements can be made because a (simulated) measurement does not cause any part of the simulated quantum state to be lost}.
    \end{enumerate}
    \emph{Quantum supremacy} refers to a computing problem and a problem size beyond which the problem can be solved efficiently on a quantum computer, but not on a classical computer.
    Quantum simulation is at one of the borders between classical computing and quantum computing:
    with quantum simulation, a classical computer performs the computation in roughly the same manner as a quantum computer, but can take advantage of shortcuts of the kind listed above.
    In principle, a more efficient quantum-simulation technique has the potential to change the threshold for quantum supremacy.
  \item
    \label{It:NoQuantumErrors}
    Current quantum computers are error-prone and can lead to incorrect results, which is not the case with a simulation (modulo bugs in the implementation of the simulation).
  \item
    \label{It:SimulationForTesting}
    Quantum simulation has a role in testing quantum computers.
    In particular, simulation can be used to create test suites for checking the correctness of the output states and measurements obtained from physical hardware.
\end{enumerate}


\subsection{Special Matrices}
\label{Se:SpecialMatrices}

In this section, we discuss how matrices and vectors used in quantum algorithms can be efficiently represented using CFLOBDDs.

\begin{algorithm}[tb!]
\caption{Hadamard Matrix \label{Fi:HAlgo}}
\SetKwFunction{HadamardMatrixCFLOBDD}{HadamardMatrixCFLOBDD}
\SetKwFunction{HadamardMatrixGrouping}{HadamardMatrixGrouping}
\SetKwProg{myalg}{Algorithm}{}{end}
\myalg{\HadamardMatrixCFLOBDD{l}}{
\Input{int $l$ -- level of the CFLOBDD to be constructed; \#variables $= 2^l$}
\Output{The CFLOBDD that represents $H_{2^l}$, of size $2^{2^{l-1}} \times 2^{2^{l-1}}$}
\Begin{
Grouping g = HadamardMatrixGrouping(l)\;
\Return RepresentativeCFLOBDD(g, [1,-1])\;  \label{Li:Hadamard:TopLevelValueTuple}
}
}{}
\setcounter{AlgoLine}{0}
\SetKwProg{myproc}{SubRoutine}{}{end}
\myproc{\HadamardMatrixGrouping{l}}{
\Input{int $l$ -- the level of the proto-CFLOBDD to be constructed}
\Output{Grouping g representing a proto-CFLOBDD for $H_{2^l}$}
\Begin{
InternalGrouping g = new InternalGrouping(l)\;
\eIf{i == 1}{
g.AConnection = ForkGrouping\;
g.AReturnTuple = [1,2]\;
g.numberOfBConnections = 2\;
g.BConnection[1] = DontCareGrouping\;
g.ReturnTuples[1] = [1]\;
g.BConnection[2] = ForkGrouping\;
g.BReturnTuples[2] = [1,2]\;
}
{
Grouping g' = HadamardMatrixGrouping(l-1)\;
g.AConnection = g'\;
g.AReturnTuple = [1,2]\;
g.numberOfBConnections = 2\;
g.BConnection[1] = g'\;
g.BReturnTuples[1] = [1,2]\;
g.BConnection[2] = g'\;
g.BReturnTuples[2] = [2,1]\;
}
g.numberOfExits = 2\;
\Return RepresentativeGrouping(g)\;
}
}
\end{algorithm}

\subsubsection{Hadamard Gate}
\label{Se:HadamardMatrix}

As we saw in \sectref{Overview}, a Hadamard matrix can be efficiently represented by a CFLOBDD.
Hadamard matrices can be recursively defined as $H_{2^{i+1}} = H_{2^i}\tensor H_{2^i}$;
thus, the Hadamard matrix at level $l+1$ can be constructed using a Kronecker product of level-$l$ Kronecker-product matrices (\sectref{KroneckerProduct}).
Alternatively, it is possible to bypass such Kronecker-product operations and directly construct the Hadamard matrix for a given level.
Pseudo-code for the latter approach is given as \algref{HAlgo}.
The algorithm takes as input the max-level $l$, and returns the CFLOBDD of level $l$ that represents the square matrix $H_{2^l}$ of size $2^{2^{l-1}} \times 2^{2^{l-1}}$.
The base case, for $l = 1$, returns the representation of
\[H_2 = \begin{bmatrix}
        1 & 1\\
        1 & -1\\
\end{bmatrix}
\]
In \algref{HAlgo}, \lineref{Hadamard:TopLevelValueTuple}, the value tuple is always $[1, -1]$ (and never $[-1, 1]$) because the $[0, 0]$ entry of every Hadamard matrix is $1$.

As noted in \footnoteref{HadamardGate}, a Hadamard gate is a Hadamard matrix scaled by a power of $\frac{1}{\sqrt{2}}$.
In particular, the scale factor for $H_{2^l}$ is $\left(\frac{1}{\sqrt{2}}\right)^l$.
In our summaries of quantum algorithms in \sectref{QuantumAlgorithms}, we omit such scale factors to reduce clutter.

\begin{algorithm}[tb!]
\caption{Identity Matrix \label{Fi:IdAlgo}}
\SetKwFunction{IdentityMatrixCFLOBDD}{IdentityMatrixCFLOBDD}
\SetKwFunction{IdentityMatrixGrouping}{IdentityMatrixGrouping}
\SetKwProg{myalg}{Algorithm}{}{end}
\myalg{\IdentityMatrixCFLOBDD{l}}{
\Input{int $l$ -- level of the CFLOBDD to be constructed; \#variables $= 2^l$}
\Output{The CFLOBDD that represents $I_{2^{2^{l-1}}}$, of size $2^{2^{l-1}} \times 2^{2^{l-1}}$}
\Begin{
Grouping g = IdentityMatrixGrouping(l)\;
\Return RepresentativeCFLOBDD(g, [1,0])\;
}
}{}
\setcounter{AlgoLine}{0}
\SetKwProg{myproc}{SubRoutine}{}{end}
\myproc{\IdentityMatrixGrouping{l}}{
\Input{int $l$ -- the level of the proto-CFLOBDD to be constructed}
\Output{Grouping g representing a proto-CFLOBDD for $I_{2^{2^{l-1}}}$}
\Begin{
InternalGrouping g = new InternalGrouping(l)\;
\eIf{i == 1}{
g.AConnection = ForkGrouping\;
g.AReturnTuple = [1,2]\;
g.numberOfBConnections = 2\;
g.BConnection[1] = ForkGrouping\;
g.ReturnTuples[1] = [1,2]\;
g.BConnection[2] = ForkGrouping\;
g.BReturnTuples[2] = [2,1]\;
}
{
Grouping g' = IdentityMatrixGrouping(l-1)\;
g.AConnection = g'\;
g.AReturnTuple = [1,2]\;
g.numberOfBConnections = 2\;
g.BConnection[1] = g'\;
g.BReturnTuples[1] = [1,2]\;
g.BConnection[2] = NoDistinctionProtoCFLOBDD(l-1)\;
g.BReturnTuples[2] = [2]\;
}
g.numberOfExits = 2\;
\Return RepresentativeGrouping(g)\;
}
}
\end{algorithm}

\subsubsection{Identity Gate}
\label{Se:IdentityMatrix}

The identity matrix, which is the same as the equality relation $EQ_N$, can be efficiently represented by a CFLOBDD, as shown in \sectref{Separation:EqualityRelation}.
The identity matrix at level $l+1$ can be recursively computed from identity matrices at level $l$ as $I_{2^{l+1}} = I_{2^l} \tensor I_{2^l}$.
The CFLOBDD for the identity matrix at level $l+1$ can either be constructed using a Kronecker product of level-$l$ identity-matrix CFLOBDDs (\sectref{KroneckerProduct}), or it can be constructed directly.
Pseudo-code for the latter approach is given as \algref{IdAlgo}.
The algorithm takes as input the max-level $l$, and returns the CFLOBDD of level $l$ that represents $I_{2^l}$ (of size $2^{2^{l-1}} \times 2^{2^{l-1}}$).
The base case, for $l = 1$, returns the representation of
\[
I_2 = \begin{bmatrix}
        1 & 0\\
        0 & 1\\
      \end{bmatrix}
\]

\Omit{
\begin{algorithm}[tb!]
\caption{Swap Matrix \label{Fi:XMatrix}}
\SetKwFunction{SwapMatrixCFLOBDD}{SwapMatrixCFLOBDD}
\SetKwFunction{SwapMatrixGrouping}{SwapMatrixGrouping}
\SetKwProg{myalg}{Algorithm}{}{end}
\myalg{\SwapMatrixCFLOBDD{}}{
\Input{int $l$ - level of the CFLOBDD = $\log{n}$, $n$ - number of bits}
\Output{The CFLOBDD that represents $X_{2^n}$}
\Begin{
Grouping g = SwapMatrixGrouping(l)\;
\Return RepresentativeCFLOBDD(g, [0,1])\;
}
}{}
\setcounter{AlgoLine}{0}
\SetKwProg{myproc}{SubRoutine}{}{end}
\myproc{\SwapMatrixGrouping{}}{
\Input{int $l$ - level of the CFLOBDD = $\log{n}$, $n$ - number of bits}
\Output{Grouping g representing proto - $X_{2^n}$}
\Begin{
InternalGrouping g = new InternalGrouping(l)\;
\eIf{i == 1}{
g.AConnection = ForkGrouping\;
g.AReturnTuple = [1,2]\;
g.numberOfBConnections = 2\;
g.BConnection[1] = ForkGrouping\;
g.ReturnTuples[1] = [1,2]\;
g.BConnection[2] = ForkGrouping\;
g.BReturnTuples[2] = [2,1]\;
}
{
Grouping g' = SwapMatrixGrouping(l-1)\;
g.AConnection = g'\;
g.AReturnTuple = [1,2]\;
g.numberOfBConnections = 2\;
g.BConnection[1] = NoDistinctionProtoCFLOBDD(l-1)\;
g.BReturnTuples[1] = [1]\;
g.BConnection[2] = g'\;
g.BReturnTuples[2] = [1,2]\;
}
g.numberOfExits = 2\;
\Return RepresentativeGrouping(g)\;
}
}
\end{algorithm}
}
\subsubsection{Not Gate}
\label{Se:NotMatrix}

The not matrix, denoted by $X_2$, flips the elements of the vector to which it is applied.
$X_2$ is defined as follows:
\Omit{
The matrix $X_{2^n}$ for $n$ bits is defined as $X_{2^n} = X_{2^{n-1}} \tensor X_{2^{n-1}}$. The pseudo-code shown in \algref{XMatrix} describes a way to efficiently build a CFLOBDD representing $X_{2^n}$.
The base case, for $l = 1$, returns the representation of
}
\[
  X_2 = \begin{bmatrix}
          0 & 1\\
          1 & 0\\
        \end{bmatrix}.
\]
$X_2$ is similar to $I_2$, but with 0 and 1 swapped.
Hence, the representation of $X_2$ uses the same level-$1$ proto-CFLOBDD used in $I_2$, but has terminal values $[0,1]$ (whereas $I_2$ has $[1,0]$).

\subsubsection{Standard-Basis Vectors}
\label{Se:StandardBasisVectors}
\begin{figure}[tb!]
    \centering
    \includegraphics[scale=0.6]{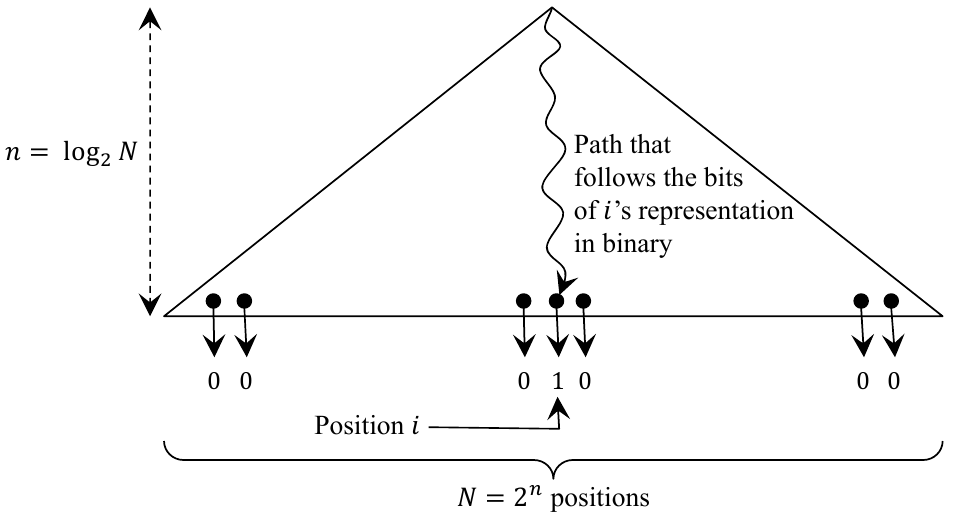}
    \caption{\protect \raggedright 
    One-hot vector as the yield of a decision tree.
    The single occurrence of $1$ is at the leaf indexed by $i$---i.e., at the end of the path from the root that follows the bits of $i$'s representation in binary.
    }
    \label{Fi:StandardBasisVectorAsDecisionTree}
\end{figure}

\begin{algorithm}[bt!]
\caption{Standard Basis Vector \label{Fi:StandardBasisVectorAlgo}}
\SetKwFunction{StandardBasisVectorCFLOBDD}{StandardBasisVectorCFLOBDD}
\SetKwFunction{StandardBasisVectorGrouping}{StandardBasisVectorGrouping}
\SetKwProg{myalg}{Algorithm}{}{end}
\myalg{\StandardBasisVectorCFLOBDD{l, i}}{
\Input{int $l$ - level of the CFLOBDD = $\log{n}$, where $n =$ number of bits; int $i$ - index}
\Output{The CFLOBDD that represents $e_x$ where $x$ = the binary representation of $i$ in $n$ bits}
\Begin{
Grouping g = StandardBasisVectorGrouping(l)\;
ValueTuple valueTuple = (i == 0) ? [1,0] : [0,1] \tcp*[f]{$1^{\textit{st}}$ elem.\ is 0 unless $i$ is 0}\;
\Return RepresentativeCFLOBDD(g, valueTuple)\;
}
}
\setcounter{AlgoLine}{0}
\SetKwProg{myproc}{SubRoutine}{}{end}
\myproc{\StandardBasisVectorGrouping{l, i}}{
\Input{int $l$ - level of the CFLOBDD = $\log{n}$, where $n =$ number of bits; int $i$ - index}
\Output{Grouping g that represents $e_x$ where $x$ = the binary representation of $i$ in $n$ bits. (Exit vertex 2 corresponds to $e_x$ unless $x = 0 \ldots 0$, in which case exit vertex 1 corresponds to $e_x$.)}
\Begin{
\If{l == 0}{
\Return RepresentativeForkGrouping\;
}
InternalGrouping g = new InternalGrouping(l)\;
int higherOrderIndex = i >> (1 << (l - 1))  \tcp*[f]{First half of $x$}\;  \label{Li:HigherOrderIndex}
g.AConnection = StandardBasisVectorGrouping(l-1, higherOrderIndex)\;  \label{Li:SBV:RecursionOne}
g.AReturnTuple = [1,2]\;
g.numberOfBConnections = 2\;
int lowerOrderIndex = i $\&$ ((1 << (1 << (l - 1))) - 1)  \tcp*[f]{Second half of $x$}\;  \label{Li:LowerOrderIndex}
Grouping g' = StandardBasisVectorGrouping(l-1, lowerOrderIndex)\;  \label{Li:SBV:RecursionTwo}
\eIf{higherOrderIndex == 0}{
g.BConnection[1] = g'  \tcp*[f]{Connection 1 = ``on path'' for $x$}\;
g.BReturnTuples[1] = [1,2]\;
g.BConnection[2] = NoDistinctionProtoCFLOBDD(l-1)\;
g.BReturnTuples[2] = (lowerOrderIndex == 0) ? [2] : [1]\;
}{
g.BConnection[1] = NoDistinctionProtoCFLOBDD(l-1)\;
g.BReturnTuples[1] = [1]\;
g.BConnection[2] = g'  \tcp*[f]{Connection 2 = ``on path'' for $x$}\;
g.BReturnTuples[2] = (lowerOrderIndex == 0) ? [2,1] : [1,2]\;
}
g.numberOfExits = 2\;
\Return RepresentativeGrouping(g)\;
}
}
\end{algorithm}

The family of standard-basis vectors are an important set of vectors in quantum algorithms because they represent the basis states.
A standard-basis vector is a one-hot vector---a vector all of whose elements are $0$, except for a single $1$.
In the terminology of formal-language theory, the vector to picture is the \emph{yield} of the decision tree obtained by unfolding the CFLOBDD (see \figref{StandardBasisVectorAsDecisionTree}).
Let a standard-basis vector of size $2^n \times 1$ with its single occurrence of $1$ at position $i$ be denoted by $e_x$, where $x$ is the binary representation of $i$ using $n$ bits.
(The vector $e_{0 \ldots 0}$ is the initial state in many quantum algorithms.)

A representation as a CFLOBDD of a standard-basis vector for one-hot position $i$ ($= x$) can be created in $n$ steps via the pseudo-code given as \algref{StandardBasisVectorAlgo}.
The code is relatively straightforward, with the small complication that $x = 0 \ldots 0$ is a special case that has to be accounted for at every level of recursion of auxiliary function StandardBasisVectorGrouping (as the bit-string $x$ becomes of half-size, quarter-size, etc.---see \lineseqref{HigherOrderIndex}{SBV:RecursionOne} and
\lineseqref{LowerOrderIndex}{SBV:RecursionTwo}).
The invariant for procedure StandardBasisVectorGrouping on input $i$ ($= x$) is that exit vertex 2 always corresponds to $e_x$---unless $x = 0 \ldots 0$, in which case exit vertex 1 corresponds to $e_x$.

The base case is for a $2\times1$ vector, either $e_0 = {[1, 0]}^t$ or $e_1 = {[0, 1]}^t$, depending on the current value of $x$.
$e_0$ would be represented by a ForkGrouping (with distinguished exit vertex $1$, because $x = 0$);
$e_1$ would be represented by a ForkGrouping (with distinguished exit vertex $2$, because $x = 1$).

\subsubsection{Controlled-NOT Gate}
\label{Se:CNOTMatrix}

\begin{figure}[tb!]
    \centering
    \begin{subfigure}[t]{0.3\linewidth}
      \centering
      \includegraphics[width=\linewidth]{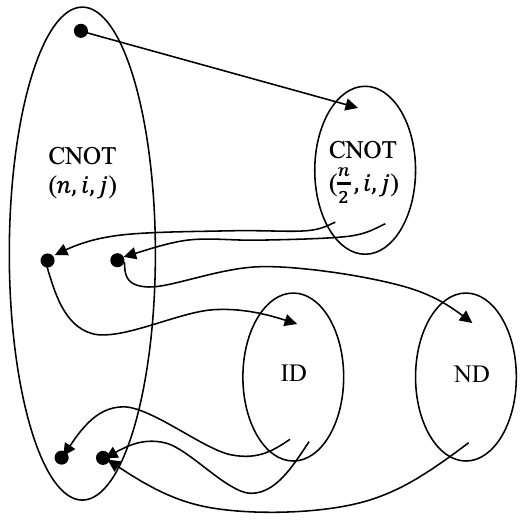}
      \caption{Case 1: Both $i$ and $j$ fall in A}
    \end{subfigure}
    \hspace{2ex}
    \begin{subfigure}[t]{0.3\linewidth}
      \centering
      \includegraphics[width=\linewidth]{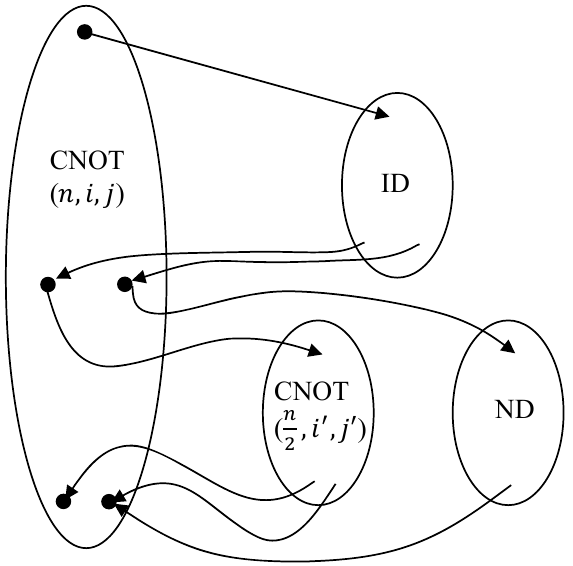}
      \caption{Case 2: Both $i$ and $j$ fall in B}
    \end{subfigure}
    \hspace{2ex}
    \begin{subfigure}[t]{0.3\linewidth}
      \centering
      \includegraphics[width=\linewidth]{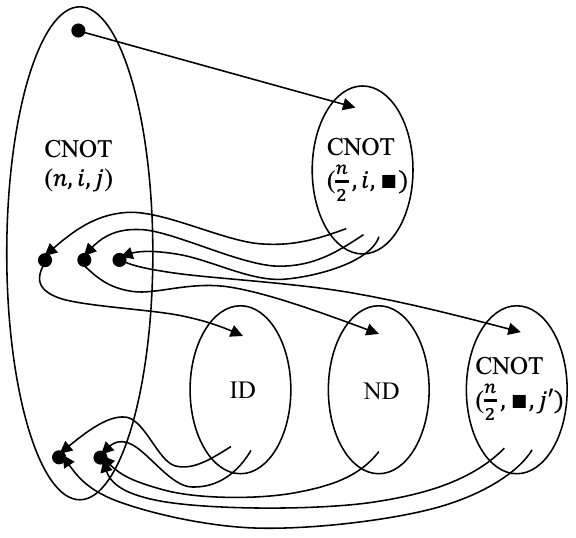}
      \caption{\protect \raggedright Case 3: $i$ in A and $j$ in B}
    \end{subfigure}
    \begin{subfigure}[t]{0.3\linewidth}
      \centering
      \includegraphics[width=\linewidth]{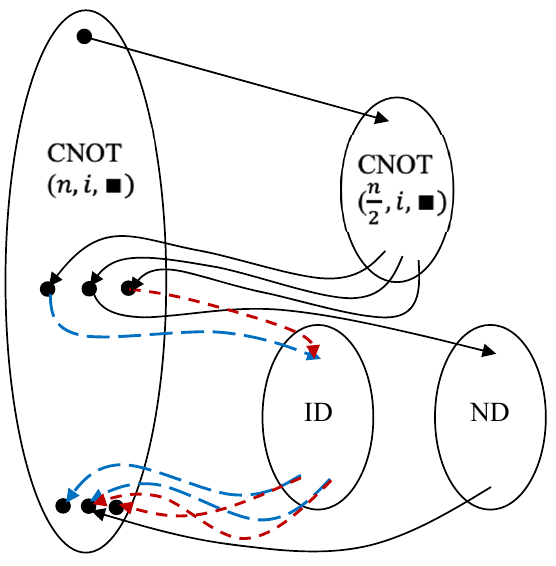}
      \caption{\protect \raggedright Case 4: $i$ in A and $j$ not in current grouping's range.}
    \end{subfigure}
    \hspace{2ex}
    \begin{subfigure}[t]{0.3\linewidth}
      \centering
      \includegraphics[width=\linewidth]{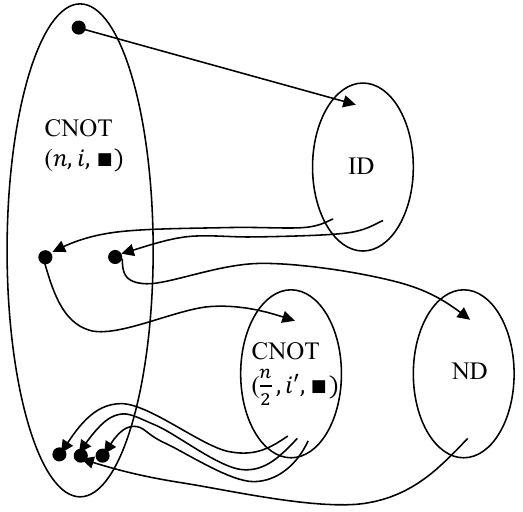}
      \caption{\protect \raggedright Case 5: $i$ in B and $j$ not in current grouping's range.}
    \end{subfigure}
    \hspace{2ex}
    \begin{subfigure}[t]{0.3\linewidth}
      \centering
      \includegraphics[width=\linewidth]{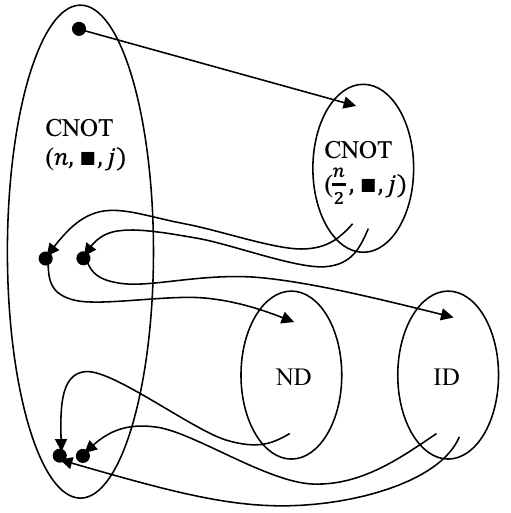}
      \caption{\protect \raggedright Case 6: $j$ in A and $i$ not in current grouping's range.}
    \end{subfigure}
    \begin{subfigure}[t]{0.3\linewidth}
      \centering
      \includegraphics[width=\linewidth]{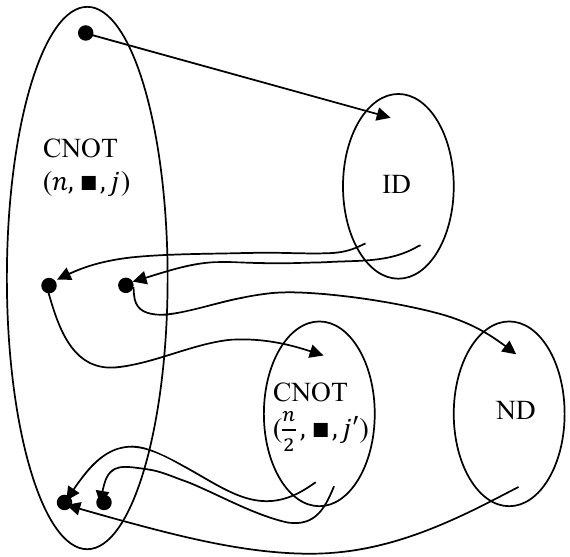}
      \caption{\protect \raggedright Case 7: $j$ in B and $i$ not in current grouping's range.}
    \end{subfigure}
    \hspace{3ex}
    \begin{subfigure}[t]{0.3\linewidth}
      \centering
      \includegraphics[width=.68\linewidth]{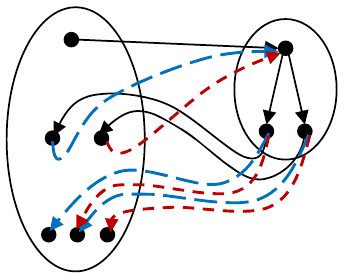}
      \caption{\protect \raggedright Base case 1: $\CNOT(n, 1, \blacksquare)$ at level $1$ interprets the control-bit.
      }
    \end{subfigure}
    \hspace{3ex}
    \begin{subfigure}[t]{0.3\linewidth}
      \centering
      \includegraphics[width=.59\linewidth]{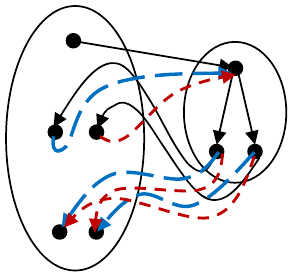}
      \caption{\protect \raggedright Base case 2: $\CNOT(n, \blacksquare, 1)$ at level $1$ interprets the controlled-bit.}
    \end{subfigure}
    \caption{\protect \raggedright 
    The different cases of the CNOT construction.
    The text in each grouping denotes the function represented by the grouping.
    ID denotes {\tt IdentityMatrixGrouping};
    ND denotes a {\tt NoDistinctionProtoCFLOBDD} (used here for an all-zero matrix).
    CNOT takes 3 arguments: $n$ for the number of bits in this proto-CFLOBDD;
    $i$ for the control-bit, and $j$ for the controlled-bit, where $0 \leq i < j < n$.
    $i'$ and $j'$ denote bit indices adjusted according to the level $l$: $i' = i - 2^{l-1}$; $j' = j - 2^{l-1}$.
    A black square indicates that a particular index is outside the grouping's index range.
    Figures (h) and (i) show the two base cases at level $1$, for $\CNOT(n, 1, \blacksquare)$ and $\CNOT(n, \blacksquare, 1)$, respectively.
    }
    \label{Fi:CNOT}
\end{figure}

A Controlled-NOT (CNOT) is an operation involving two index bits:
one bit is the control-bit and the other is the controlled-bit;
in the output, the value of the controlled-bit is flipped if the
control-bit is `1'.
The matrix that implements this behavior for two bits, denoted by $\CNOT_2$, where the first bit is the control-bit and the second bit is the controlled-bit, is as follows:
\begin{equation}
  \label{Eq:CNOT:BaseCase}
  \CNOT_2 \; = \begin{bNiceArray}{cccc}[first-col,first-row]
                 & 00 & 01 & 10 & 11 \\
              00 &  1 &  0 &  0 &  0 \\
              01 &  0 &  1 &  0 &  0 \\
              10 &  0 &  0 &  0 &  1 \\
              11 &  0 &  0 &  1 &  0 \\
            \end{bNiceArray}            
          = \begin{bmatrix}
                  I_2 & O_2 \\
                  O_2 & X_2
                \end{bmatrix}
\end{equation}
Note that when the first bit is $1$ and the second bit is flipped, the $2 \times 2$ block in the lower right is the not matrix $X_2$.
($O_2$ denotes the $2 \times 2$ matrix of zeros.)

For larger numbers of bits, the situation is more complex.
With four bits, and the second bit controlling the third, we have

\[
\CNOT(4,2,3) =
\begin{bmatrix}
I_2 & O_2 & O_2 & O_2 & O_2 & O_2 & O_2 & O_2 \\
O_2 & I_2 & O_2 & O_2 & O_2 & O_2 & O_2 & O_2 \\
O_2 & O_2 & O_2 & I_2 & O_2 & O_2 & O_2 & O_2 \\
O_2 & O_2 & I_2 & O_2 & O_2 & O_2 & O_2 & O_2 \\
O_2 & O_2 & O_2 & O_2 & I_2 & O_2 & O_2 & O_2 \\
O_2 & O_2 & O_2 & O_2 & O_2 & I_2 & O_2 & O_2 \\
O_2 & O_2 & O_2 & O_2 & O_2 & O_2 & O_2 & I_2 \\
O_2 & O_2 & O_2 & O_2 & O_2 & O_2 & I_2 & O_2 
\end{bmatrix}
= I_2 \tensor \CNOT_2 \tensor I_2.
\]
However, with four bits and the first bit controlling the third, we have

\newpage
~\newpage
~\newpage

\[
\CNOT(4,1,3) =
\begin{bmatrix}
I_2 & O_2 & O_2 & O_2 & O_2 & O_2 & O_2 & O_2 \\
O_2 & I_2 & O_2 & O_2 & O_2 & O_2 & O_2 & O_2 \\
O_2 & O_2 & I_2 & O_2 & O_2 & O_2 & O_2 & O_2 \\
O_2 & O_2 & O_2 & I_2 & O_2 & O_2 & O_2 & O_2 \\
O_2 & O_2 & O_2 & O_2 & O_2 & X_2 & O_2 & O_2 \\
O_2 & O_2 & O_2 & O_2 & X_2 & O_2 & O_2 & O_2 \\
O_2 & O_2 & O_2 & O_2 & O_2 & O_2 & O_2 & X_2 \\
O_2 & O_2 & O_2 & O_2 & O_2 & O_2 & X_2 & O_2 
\end{bmatrix}
\]
for which there is no clean expression in terms of $\tensor$.

Fortunately, we can create a double-exponentially compressed CFLOBDD representation of $\CNOT(n, i, j)$, $1 \leq i, j \leq n$, $i \neq j$.
Suppose that the control-bit is bit $i$, and the controlled-bit is bit $j$;
then for a grouping $g$, there are eight cases to consider (here we discuss the cases for $i < j$),
\footnote{
  The general case---$\CNOT(n, i, j)$, $1 \leq i, j \leq n$, $i \neq j$, where $j$ is allowed to be less than $i$---would be similar, but with some additional cases.
}
which are depicted in \figref{CNOT}.
The key to understanding \figref{CNOT} is that the construction
maintains the invariant shown in \eqref{CNOTExitVertexInvariant} on the exit vertices of the different kinds of groupings.

\begin{equation}
  \label{Eq:CNOTExitVertexInvariant}
  \begin{array}{r|c|ccc}
              & \multicolumn{1}{c|}{\multirow{2}{*}{\text{Role}}} & \multicolumn{3}{c}{\text{Significance of exit vertex}}  \\
    \cline{3-5}
              &             & 1 & 2 & 3 \\
    \hline
    \multirow{4}{*}{Proto-CFLOBDD} & \CNOT(n,i,j)            & \text{on-path}  & \text{off-path} & N/A \\
                                   & \CNOT(n,i,\blacksquare) & \text{on-path}  & \text{off-path} & \text{Controlled-bit is to be flipped} \\
                                   & \CNOT(n,\blacksquare,j) & \text{off-path} & \text{on-path}  & N/A \\
                                   & \ID                     & \text{on-path}  & \text{off-path} & N/A \\
    \hline
    CFLOBDD                        & \text{Top level}        & 1          & 0          & N/A \\
    \hline
  \end{array}
\end{equation}

Here, ``on-path'' means that the exit occurs on a matched-path that can be continued to the top-level terminal value $1$;
``off-path'' means that it will only be used to reach the top-level terminal value $0$.
For instance, in \figref{CNOT}(a)--(g), each occurrence of ND (a NoDistinctionProtoCFLOBDD) is attached to the middle vertex of a grouping $g$ that is reached from an A-connection's off-path exit.
The NoDistinctionProtoCFLOBDD has one exit vertex for which the return edge connects it to the off-path exit for $g$.
\begin{enumerate}
  \item
    Both $i$ and $j$ fall in the A-connection of $g$.
    \figref{CNOT}(a) shows the structural representation for this scenario.
    The bits in $g$'s A-connection handle the CNOT operation.
    The bits in $g$'s B-connection are not involved, and hence the ``on-path'' middle vertex is connected to Identity, and the ``off-path'' middle vertex is connected to a NoDistinctionProtoCFLOBDD (for an all-zero sub-matrix).
  \item
    Both $i$ and $j$ fall in $g$'s B-connection range.
    This scenario is depicted in \figref{CNOT}(b).
    The bits in $g$'s A-connection do not affect the CNOT operation, and hence we have an Identity matrix in the A-connection.
    The B-connection handles the relationship between the control-bit and controlled-bit through a CNOT structure of the form $\CNOT(\frac{n}{2}, i', j')$.
    $i'$ and $j'$ denote bit indices adjusted according to the level $l$: $i' = i - 2^{l-1}$; $j' = j - 2^{l-1}$.
  \item
    $i$ lies in $g$'s A-connection range and $j$ in $g$'s B-connection range.
    As shown in \figref{CNOT}(c), $g$'s A-connection handles the control-bit part, and hence $g$ has 3 middle vertices.
    The first is the ``on-path'' case when the control-bit is $0$, so no interpretation of the controlled-bit is needed;
    the second is the ``off-path'' case;
    and the third represents the information ``the control-bit was activated.''
    Consequently, the B-connection groupings correspond to the respective three cases:
    the first has the Identity matrix;
    the second a NoDistinctionProtoCFLOBDD (for an all-zero sub-matrix);
    and the third is a CNOT structure of the form $\CNOT(\frac{n}{2}, \blacksquare, j')$.
  \item
    $i$ lies in $g$'s A-connection range, but $j$ does not fall in $g$'s range (\figref{CNOT}(d)).
    This case is similar to \figref{CNOT}(c), in that $g$'s A-connection handles the control-bit part, and hence $g$ has 3 middle vertices.
    The difference comes in the B-connections, which propagate the ``states'' of the middle vertices to $g$'s 3 exit vertices (using the Identity matrix for both ``on-path'' and ``the control-bit was activated'' and a NoDistinctionProtoCFLOBDD for ``off-path'').
  \item
    $i$ lies in $g$'s B-connection range, but $j$ does not fall in $g$'s range.
    As shown in \figref{CNOT}(e), the bits in $g$'s A-connection are not involved, and hence the A-connection is the Identity matrix.
    For the ``on-path'' middle vertex, $g$'s B-connection is to a CNOT structure of the form $\CNOT(\frac{n}{2}, i', \blacksquare)$.
    For the ``off-path'' middle vertex, $g$'s B-connection is to a
    NoDistinctionProtoCFLOBDD.
  \item
    $j$ lies in $g$'s A-connection range, but $i$ does not fall in $g$'s range (\figref{CNOT}(f)).
    $g$'s A-connection handles the part of the CNOT operation for flipping the controlled-bit and has 2 exit vertices.
    The bits of $g$'s B-connection are not involved;
    hence $g$'s ``off-path'' middle vertex is a NoDistinctionProtoCFLOBDD (for an all-zero sub-matrix) and the ``on-path'' middle vertex is the Identity matrix.
  \item
    $j$ lies in $g$'s B-connection range, but $i$ does not fall in $g$'s range (\figref{CNOT}(g)).
    This case is similar to \figref{CNOT}(f), except that the part of the CNOT operation for flipping the controlled-bit is displaced to the ``on-path'' B-connection (``on-path'' from the Identity matrix in $g$'s A-connection).
    The ``off-path'' B-connection is a NoDistinctionProtoCFLOBDD (for an all-zero sub-matrix).
  \item
    There are two base cases of the recursive construction:
    \begin{itemize}
      \item
        Calls of the form $\CNOT(n, 1, \blacksquare)$ at level $1$ (\figref{CNOT}(h)).
        The grouping $g$ for this case represents a sub-matrix similar to the identity matrix $I_2$, except that the second return edge of the second B-connection of $g$ maps to a third (new) exit vertex of $g$ (instead of mapping to the first exit vertex of $g$, as is the case for $I_2$).
        The third exit represents the information ``the control-bit was activated'' (i.e., the control-bit was $1$), and hence, as indicated in \eqref{CNOTExitVertexInvariant}, for a matched-path to represent an element of the CNOT relation, there is an obligation to flip the value of the controlled-bit (elsewhere in the CFLOBDD).
      \item
        Calls of the form $\CNOT(n, \blacksquare, 1)$ at level $1$ (\figref{CNOT}(i)).
        Note from \figref{CNOT}(c) that the initial construction of the form $\CNOT(n, \blacksquare, j)$ (at some higher level) is attached to the third middle vertex of the parent grouping, which, in turn, has a return edge from the third exit of the parent grouping's A-connection $\CNOT(n, i, \blacksquare)$.
        Thus, $\CNOT(n, \blacksquare, j)$ only occurs in a path context in which it is known that the control-bit was activated.

        \hspace{1.5ex}
        The $\CNOT(n, \blacksquare, 1)$ grouping at level 1 interprets the controlled-bit.
        The proto-CFLOBDD used in this case is the same one found in both $I_2$ and $X_2$.
        Inspection of \figref{CNOT}(c), (f), and (g) reveals that such a proto-CFLOBDD is always connected to a level-$2$ grouping with return edges to middle or exit vertices that represent the ``state-tuple'' $[\text{off-path},\text{on-path}]$.
        Consequently, each occurrence of a level-$1$ $\CNOT(n, \blacksquare, 1)$ proto-CFLOBDD acts like the not matrix $X_2$.
    \end{itemize}
\end{enumerate}
Pseudo-code for the full algorithm is given in Appendix~\sectref{CNOTConstructionAlgo}.

\paragraph{Complexity}
Every CFLOBDD that represents a CNOT relation consists of only a constant number of kinds of groupings:
as shown in \figref{CNOT}, the representation of CNOT uses at most two kinds of CNOT groupings at level 1 and seven kinds of CNOT groupings at levels $\geq 2$, as well as proto-CFLOBDDs obtained from IdentityMatrixGrouping and NoDistinctionProtoCFLOBDD.
Because at each level there are a constant number of groupings, each of constant size, the CFLOBDD representation of CNOT exhibits double-exponential compression compared to the size of the decision tree.\footnote{
  \label{Footnote:CNOTGeneralCase}
  The general case---$\CNOT(n, i, j)$, $1 \leq i, j \leq n$, $i \neq j$, where $j$ is allowed to be less than $i$---would still have at each level a constant number of kinds of groupings, each of constant size.
  Consequently, the general case also exhibits double-exponential compression compared to the size of the decision tree.
}

\paragraph{A Special-Case Construction}
In some quantum algorithms, such as Simon's Algorithm, the bits are in two groups, $x_i$ and $y_i$, $1 \leq i \leq n$, and the CNOT operation is performed for each pair $(x_i,y_i)$.
The same net result can be produced by creating (and then applying) a representation for the compound operation $\Pi_{i=1}^{n} \CNOT(2n, i, n+i)$.
This expression involves $n$ matrix-multiplication operations.
Unfortunately, if the bit ordering is $\langle x_1, \cdots, x_{n}, y_1, \cdots, y_{n} \rangle$, the size of the resulting CFLOBDD is exponential in $n$.
In essence, the CFLOBDD's top-level A-connection has to ``memorize'' the exponential number of different possible combinations of $x_i$ values so that the information can be correlated with the $y_i$ values in the B-connection.

\begin{figure}[tb!]
  \centering
  \begin{tabular}{cc}
    \includegraphics[width=.55\linewidth]{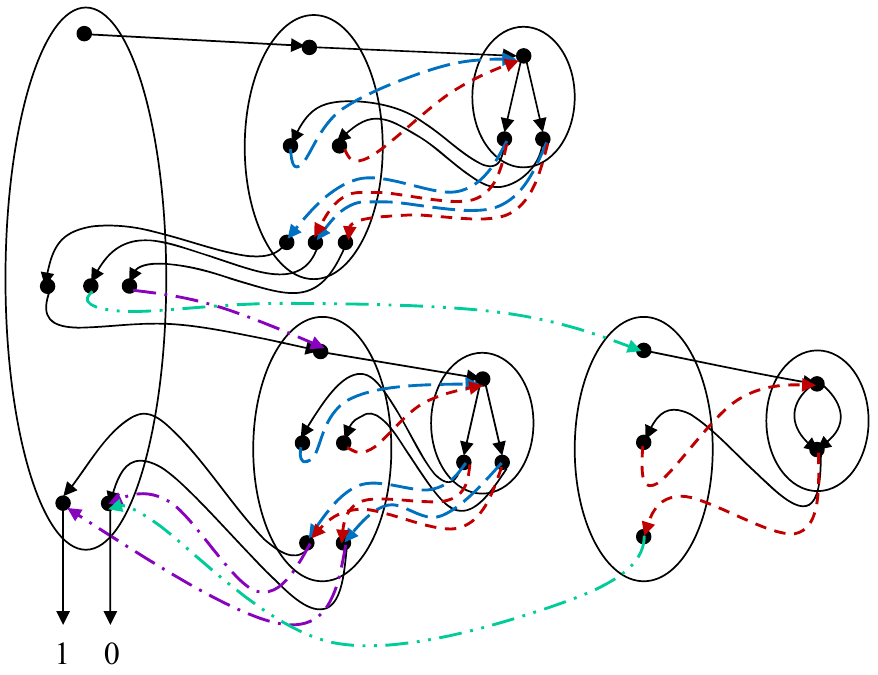}
    &
    \includegraphics[width=.45\linewidth]{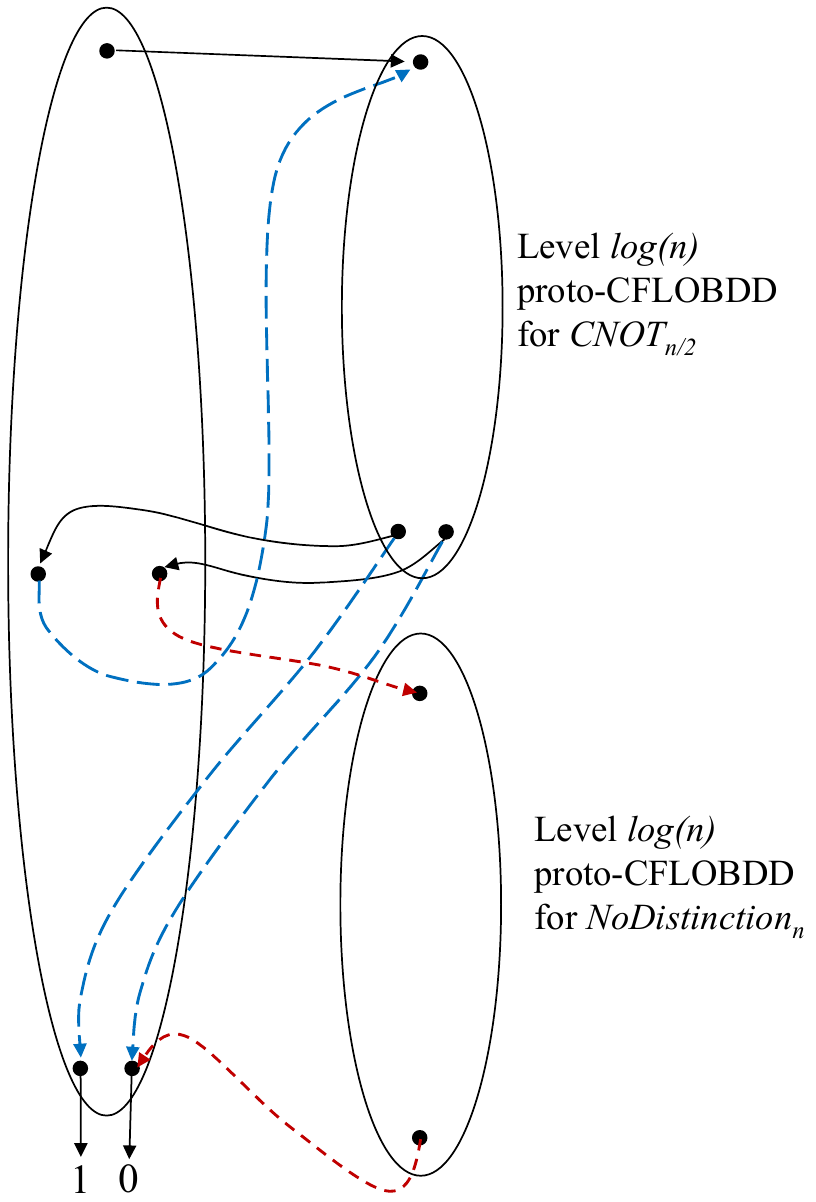}
    \\
    (a) & (b)
  \end{tabular}
  \caption{\protect \raggedright 
    CFLOBDD representation of $\CNOT_n$ = $\CNOT_2^{\tensor n}$ with variable ordering $\langle x_1, y_1, x_2, y_2, \cdots, x_{n}, y_{n} \rangle$.
    (a) Base case (level $2$): $\CNOT_2$ for 2 bits, with the first bit as the control-bit and second bit as the controlled-bit.
    (b) Grouping structure used for levels greater than $2$.  
  }
  \label{Fi:CNOT_K}
\end{figure}

Fortunately, changing to the interleaved-variable ordering $\langle x_1, y_1, x_2, y_2, \cdots, x_{n}, y_{n} \rangle$ leads to an efficient CFLOBDD representation of $\Pi_{i=1}^{n}(\CNOT(2n, i, n+i)$.
With the interleaved-variable ordering, after reassigning the index numbers $[1 \ldots 2n]$ to the variables in the new order, $\CNOT_n$ changes to the simpler expression
\begin{equation}
  \label{Eq:CNOT_n_1}
  \CNOT_n = \Pi_{\underset{i+=2}{i=1}}^{2n}\CNOT(2n, i, i+1).
\end{equation}
In \eqref{CNOT_n_1}, each CNOT operation is between \emph{adjacent} bits, and hence \eqref{CNOT_n_1} can be rewritten as
\begin{equation}
  \label{Eq:CNOT_n_2}
    \CNOT_n = \bigotimes_{{i=1}}^{n}\CNOT_2 = \CNOT_2^{\tensor n}
\end{equation}
Moreover, we do not even have to perform the $n$ explicit Kronecker-product operations of \eqref{CNOT_n_2} (or even $\log{n}$ Kronecker products) because $\CNOT_n$ can be constructed directly (cf.\ item \ref{It:Deviate:GoOutsideQuantumModel} in \sectref{AdvantagesOfSimulation}).
\figref{CNOT_K} depicts the CFLOBDD representation of $\CNOT_n$ = $\CNOT_2^{\tensor n}$, which has $\log{n} + 1$ levels and $2n$ Boolean variables, and the pseudo-code for the construction algorithm can be found in~\sectref{CNOTConstructionAlgo}.

\subsubsection{Quantum Fourier Transform}
\label{Se:QuantumFourierTransform}

The Quantum Fourier Transform (QFT) is a linear transformation that is somewhat similar to the discrete Fourier transform.
In matrix form, it looks like\footnote{
  In the description of QFT, we vary from the subscript convention explained at the beginning of \sectref{quantum-algos}.
  Here, a gate matrix acting on $n$ qubits (which is of size $2^n\times 2^n$) is denoted using a subscript $N$, where $N$ denotes $2^n$.
}
\[
  \QFT_N =
  \frac{1}{\sqrt{N}}
  \left[\begin{array}{cccccc}
    1      &   1          &   1             &   1             & \ldots &   1 \\
    1      & \omega       & \omega^2        & \omega^3        & \ldots & \omega^{N-1} \\
    1      & \omega^2     & \omega^4        & \omega^6        & \ldots & \omega^{2(N-1)} \\
    1      & \omega^3     & \omega^6        & \omega^9        & \ldots & \omega^{3(N-1)} \\
    \vdots & \vdots       & \vdots          & \vdots          &        & \vdots    \\
    1      & \omega^{N-1} & \omega^{2(N-1)} & \omega^{3(N-1)} & \ldots & \omega^{(N-1)(N-1)}
  \end{array}\right]
\]
where $\omega = e^{{2\pi i}/{N}}$.

We use the algorithm presented in~\cite{nielsen2002quantum} for the application of $\QFT$ to a state vector. More formally, we use the following steps to apply $\QFT$ for $n$ qubit vector:
\begin{enumerate}
    \item Start with a random vector $e_x$. (We use a randomly chosen basis vector as input in our experiments).
    \item Apply the Hadamard matrix $H_2$ to the $i^{th}$ qubit, where $i \rightarrow \{n \dots 1\}$.
    \item For every $i$, apply a Controlled-Phase Gate $CP$ from qubit $j \rightarrow \{1..i-1\}$ to $i$ with a phase $= \dfrac{\pi}{2^{i-j}}$.
    \item Finally, apply $n/2$ Swap Gates between $i$ and $n-i$, where $i \rightarrow \{1..n/2\}$.
\end{enumerate}

The construction of Controlled-Phase Gate and Swap Gate is given in Appendix~\sectref{QuantumGates}.

\subsection{Quantum Algorithms}
\label{Se:QuantumAlgorithms}

In this section, we briefly describe the quantum algorithms that are used in the experiments in \sectref{ResearchQuestionTwo}.
The ingredients of the algorithms are qubits, quantum gates, and measurements of qubits.
Each of the states or operations can be viewed as algebraic operations on vectors and matrices.
The mapping between the various aspects of the quantum algorithms and their corresponding algebraic operations is as follows:

\begin{itemize}
  \item
    Sequence of $n$ qubits: Unit vector of dimension $2^n \times 1$ (called a state vector)
  \item
    Quantum gate: Unitary matrix
  \item
    Augmenting a sequence of qubits with additional qubits: Kronecker product of vectors
  \item
    Application of a quantum gate to a sequence of qubits: Matrix-vector multiplication
  \item
    Measurement of qubits: Sampling from a distribution obtained from the vector
\end{itemize}

We now discuss quantum algorithms as algebraic operations.

\subsubsection{GHZ Algorithm}
\label{Se:GHZAlgorithm}

The Greenberger–Horne–Zeilinger (GHZ) state is the following (``entangled'')  state vector for $3$ qubits (i.e., a unit vector of size $8 \times 1$):
\[
  \GHZ_3 = \dfrac{e_{000} + e_{111}}{\sqrt{2}}
\]
We extend the concept to $n$ qubits by defining\footnote{
  Recall that we use
  $e_{0^n}$ and $e_{1^n}$ denote the standard-basis vectors
  $e_{\underbrace{0 \ldots 0}_{n~\textit{copies}}}$ and
  $e_{\underbrace{1 \ldots 1}_{n~\textit{copies}}}$, respectively.
}
\[
  \GHZ_n = \dfrac{e_{0^n} + e_{1^n}}{\sqrt{2}}
\]

We have used the algorithm given by Yu and Palsberg \cite{yu2021quantum} for obtaining the GHZ state for $n$ bits, which is as follows:
\begin{enumerate}
  \item
    \label{It:GHZ:StepOne}
    Prepare an initial state $e_{0^{n+1}}$ (the standard-basis vector of $n+1$ bits with a $1$ in the first position and $0$s elsewhere).
  \item
    \label{It:GHZ:StepTwo}
    Apply the Hadamard matrix $H_{2n}$ on the first $n$ bits and the Not matrix $X_2$ on the $(n+1)^{\textit{st}}$ bit.
  \item
    \label{It:GHZ:StepThree}
    Apply $n$ CNOTs of the form $\CNOT(n+1, i, n+1)$, where $i \rightarrow \{1..n\}$.
  \item
    \label{It:GHZ:StepFour}
    Apply the Hadamard matrix $H_{2(n+1)}$ ($= H_{2n} \tensor H_{2}$) on all $n+1$ bits.
  \item
    \label{It:GHZ:StepFive}
    Measure the first $n$ bits to obtain a string of $n$ $0$s or $n$ $1$s.
\end{enumerate}

Steps (\ref{It:GHZ:StepOne})--(\ref{It:GHZ:StepFour}) can be expressed in matrix-vector notation as follows:
\begin{equation}
  \label{Eq:GHZ}
  (H_{2n} \tensor H_{2}) (\Pi_{i=1}^{n} \CNOT(n+1, i, n+1)) (H_{2n} \tensor X_{2})(e_{0^{n+1}})
\end{equation}

\eqref{GHZ} is expressed in a way that uses vectors indexed by $n+1$ Boolean variables, and matrices indexed by $2n+2$ Boolean variables.
Whereas a BDD implementation can directly encode \eqref{GHZ}, CFLOBDDs are a representation of functions in which the number of Boolean variables must be a power of $2$.
For this reason, in a CFLOBDD implemenation of $\GHZ$, vectors and matrices are augmented so that vectors have $n-1$ dummy index variables and matrices have $2n-2$ dummy index variables.
Taking into account these dummy variables, what is actually computed is the following:
\begin{equation}
  \label{Eq:GHZForCFLOBDDs}
  (H_{2n} \tensor H_{2} \tensor I^{\tensor {n-1}}) (\Pi_{i=1}^{n} \CNOT(n+1, i, n+1)) (H_{2n} \tensor X_{2} \tensor I^{\tensor {n-1}})(e_{0^{2n}})
\end{equation}

By properties of Kronecker product,
\[
  (H_{2n} \tensor X_{2} \tensor I^{\tensor {n-1}}) (e_{0^{2n}})
  =
  (H_{2n} \times e_{0^n}) \tensor (X_{2} \times e_{0}) \tensor (I^{\tensor {n-1}} \times e_{0^{n-1}}).
\]
Because the matrix-vector product $H_{2^n} \times e_{0^n}$ results in a vector consisting of all ones, we can avoid performing an explicit multiplication and instead directly create a CFLOBDD whose top-level grouping is a {\tt NoDistinctionProtoCFLOBDD} and whose terminal value is $\frac{1}{\sqrt{2^n}}$.
In the implementation of these algorithms, a matrix-vector multiplication is implemented by first converting the vector $2^k\times 1$ to a matrix of size $2^k\times 2^k$ by padding with zeros (\sectref{VectorToMatrixConversion}), and then performing matrix multiplication (\sectref{matrix-mult}).

In step (\ref{It:GHZ:StepFive}), the terminal values (which represent amplitudes) are squared, and then the bit-strings of indices are sampled from the CFLOBDD based on the values of the squared amplitudes (which represent probabilities), as explained in \sectref{PathCountingAndSampling}.

\subsubsection{Bernstein-Vazirani Algorithm}
\label{Se:BVAlgorithm}

The problem that the Bernstein-Vazirani (BV) algorithm solves is as follows:
\begin{quote}
    Given an oracle that implements a function $f: \{0,1\}^n \rightarrow \{0,1\}$ in which $f(x)$ is promised to be the dot product, mod $2$, between $x$ and a secret string $s \in \{0,1\}^n$---i.e., $f(x) = x_1\cdot s_1 \oplus x_2\cdot s_2 \oplus \cdots \oplus x_n \cdot s_n$---find $s$.
\end{quote}

Given an oracle $U_{2(n+1)}$ that implements $f$, such that $U(xy) = x(f(x)\xor y)$, the BV algorithm is as follows:
\begin{enumerate}
  \item
    \label{It:BV:StepOne}
    Prepare an initial state $e_{0^{n+1}}$ (the standard-basis vector of $n+1$ bits with a $1$ in the first position and $0$s elsewhere).
  \item
    \label{It:BV:StepTwo}
    Apply the Hadamard matrix $H_{2n}$ on the first $n$ bits.
  \item
    \label{It:BV:StepThree}
    Apply the oracle $U_{2(n+1)}$ to the current state.
  \item
    \label{It:BV:StepFour}
    Finally, apply the Hadamard matrix $H_{2n}$ on the first $n$ bits.
  \item
    \label{It:BV:StepFive}
    Measure the first $n$ bits to obtain the string $s$.
\end{enumerate}

Steps (\ref{It:BV:StepOne})--(\ref{It:BV:StepFour}) of the algorithm can be expressed as
\begin{equation}
  \label{Eq:BV}
  (H_{2n} \tensor I_{2})(U_{2(n+1)})(H_{2n} \tensor I_{2})(e_{0^{n+1}})
\end{equation}

As in \sectref{GHZAlgorithm}, to implement \eqref{BV} with CFLOBDDs, we introduce dummy index variables so that the total number of index variables is a power of $2$.
Also, in the term $(H_{2n} \tensor I_{2})(e_{0^{n+1}})$ $= (H_{2n} \times e_{0^n}) \tensor (I_{2} \times e_0)$, the multiplication $(H_{2n} \times e_{0^n})$ can be avoided by directly creating a CFLOBDD whose top-level grouping is a {\tt NoDistinctionProtoCFLOBDD} and whose terminal value is $\frac{1}{\sqrt{2^n}}$.

\subsubsection{Deutsch–Jozsa algorithm}
\label{Se:DJAlgorithm}

The algorithm solves the following problem:
\begin{quote}
    Given an oracle that implements a function $f: \{0,1\}^n \rightarrow \{0,1\}$, where $f$ is promised to be either a constant function (0 on all inputs or 1 on all inputs) or a balanced function (returns 1 for half of the input domain and 0 for the other half), determine if $f$ is balanced or constant.
\end{quote}
The oracle takes the form of a matrix $U_{2(n+1)}$, for which $U(xy) = x(f(x) \xor y)$.
The steps of the Deutsch–Jozsa (DJ) algorithm are as follows:
\begin{enumerate}
    \item 
    \label{It:DJ:StepOne}
    Prepare an initial state $e_{0^n 1}$ of $n+1$ bits (the standard-basis vector of $n+1$ bits with a $1$ in the second position and $0$s elsewhere).
    \item 
    \label{It:DJ:StepTwo}
    Apply the Hadamard matrix $H_{2n}\tensor H_2$ on first $n+1$ bits.
    \item 
    \label{It:DJ:StepThree}
    Apply the oracle $U_{2(n+1)}$ to the current state.
    \item 
    \label{It:DJ:StepFour}
    Finally, apply the Hadamard matrix $H_{2n}$ to the first $n$ bits.
    \item 
    \label{It:DJ:StepFive}
    Measure the first $n$ bits to obtain the string $s$. If $s = e_{0^n}$, then $f$ is constant; otherwise, $f$ is balanced.
\end{enumerate}

Steps (\ref{It:DJ:StepOne})--(\ref{It:DJ:StepFour}) can be expressed as:
\[
   (H_{2n} \tensor I_{2})(U_{2(n+1)})(H_{2n} \tensor H_{2})(e_{0^{n}} \tensor e_{1}).
\]

\subsubsection{Simon's Algorithm}
\label{Se:simonsalgo}

The problem that Simon's algorithm addresses is as follows:

\begin{quote}
    One is given a function $f:\{0,1\}^n \rightarrow \{0,1\}^n$, where $f$ is promised to satisfy the property that there is a ``hidden vector'' $s \in \{0,1\}^n$ such that, for all $x$ and $y$, $f(x) = f(y)$ if and only if $x = y \xor s$.
    The goal is to find the hidden vector $s$.
\end{quote}

Given an oracle $U_{2{n}}$ such that $U(x) = f(x)$, satisfying the property $\exists s \in \{0,1\}^n$ such that $\forall x, y$ $U(x) = U(y)$ if and only if $x = y \xor s$, the algorithm for finding $s$ is as follows:
\begin{enumerate}
  \item
    \label{It:Simon:StepOne}
    Initialize the set of equations $E$ to the empty set
  \item
    \label{It:Simon:StepTwo}
    Prepare an initial state $e_{0^{2n}}$ of $2n$ bits (the standard-basis vector of $2n$ bits with a $1$ in the first position and $0$s elsewhere).
  \item
    \label{It:Simon:StepThree}
    Apply the Hadamard matrix $H_{2n}$ on the first $n$ bits.
  \item
    \label{It:Simon:StepFour}
    Apply the oracle $U_{2n}$ on the first $n$ bits.
  \item
    \label{It:Simon:StepFive}
    Successively apply the matrices $\CNOT(2n, i, n+i)$, for $i \rightarrow \{1..n\}$.
  \item
    \label{It:Simon:StepSix}
    Apply the oracle $U_{2n}^{*}$, which is the conjugate of $U_{2n}$, on the first $n$ bits.
  \item
    \label{It:Simon:StepSeven}
    Apply the Hadamard matrix $H_{2n}$ on the first $n$ bits.
  \item
    \label{It:Simon:StepEight}
    Measure the first $n$ bits to obtain $x$.
  \item
    \label{It:Simon:StepNine}
    Add the equation $x \cdot s = 0$ to equation set $E$.
  \item
    \label{It:Simon:StepTen}
    Repeat steps (\ref{It:Simon:StepTwo})--(\ref{It:Simon:StepNine}) to obtain $O(n)$ equations.
  \item
    \label{It:Simon:StepEleven}
    With high probability, the solution to the set of equations $E$ is $s$.
\end{enumerate}

Steps (\ref{It:Simon:StepTwo})--(\ref{It:Simon:StepTen}) operate on quantum states of $2n$ qubits.
Call the first $n$ bits the $x$ bits, and the second $n$ bits the $y$ bits.
Conventionally, one considers the bits to be ordered as follows:
$\langle x_1, \cdots, x_{n}, y_1, \cdots, y_{n} \rangle$.

Steps (\ref{It:Simon:StepTwo})--(\ref{It:Simon:StepSeven}) can be written as follows:
\[
    (H_{2n}\tensor I_{2n})(U^{*}_{2n}\tensor I_{2n})(\Pi_{i=1}^{n}(\CNOT(2n, i, n+i)))((U_{2n}\tensor I_{2n}))(H_{2n}\tensor I_{2n})(e_{0^{2n}})
\]
Using~\eqref{CNOT_n_1}, and the interleaved Kronecker product ($\tensor_i$) discussed in \sectref{KroneckerProduct:VariantTwo} (\algref{KP4Voc}), we can rewrite the above expression as follows:
\begin{equation}
\label{Eq:SimonsAlgoExpression}
    (H_{2n}\tensor_i I_{2n})(U^{*}_{2n}\tensor_i I_{2n})(\CNOT_n)(U_{2n}\tensor_i I_{2n})(H_{2n}\tensor_i I_{2n})(e_{0^{2n}})
\end{equation}
Note that in \eqref{SimonsAlgoExpression}, we have changed to the bit ordering to $\langle x_1, y_1, x_2, y_2, \cdots, x_{n}, y_{n} \rangle$.

Now consider the second, third, and fourth multiplicands in \eqref{SimonsAlgoExpression}:
\[
  (U^{*}_{2n}\tensor_i I_{2n})(\CNOT_n)(U_{2n}\tensor_i I_{2n}).
\]
Reading these terms right to left,
\begin{itemize}
  \item
    first, the oracle $U_{2n}$ is applied on the $x$ bits, which yields $f(x)$;
  \item
    second, $\CNOT_n$ copies the value in the $x$ bits to the $y$ bits, leading to the $y$ bits also representing $f(x)$;
  \item
    third, $U^*_{2n}$ is applied to the $x$ bits, to turn the values in the $x$ bits from $f(x)$ back to the original value of $x$.
\end{itemize}
However, we can achieve the same result by first copying the value of the $x$ bits to the $y$ bits by applying $\CNOT_n$, and then applying $U_{2n}$ directly on the $y$ bits, which can be expressed as
$
  (I_{2n} \tensor_i U_{2n})(\CNOT_n).
$
Thus, \eqref{SimonsAlgoExpression} can be re-written as
\begin{equation}
\label{Eq:SimonsAlgoExpressionOptimized}
    (H_{2n}\tensor_i U_{2n})(\CNOT_n)(H_{2n}\tensor_i I_{2n})(e_{0^{2n}})
\end{equation}
Formally, to show that \eqrefs{SimonsAlgoExpression}{SimonsAlgoExpressionOptimized} are equal, we use two properties of Kronecker product:
\begin{enumerate}
    \item $A \tensor (B + C) = (A \tensor B) + (A \tensor C)$
    \item $(A \tensor B)(C \tensor D) = (AC \tensor BD)$
\end{enumerate}
Starting from \eqref{SimonsAlgoExpression}, we have
\begin{align}
        &(H_{2n}\tensor_i I_{2n})(U^{*}_{2n}\tensor_i I_{2n})(\CNOT_n)(U_{2n}\tensor_i I_{2n})(H_{2n}\tensor_i I_{2n})(e_{0^{2n}}) \tag{\ref{Eq:SimonsAlgoExpression}} \\
        =~ & (H_{2n}\tensor_i I_{2n})(U^{*}_{2n}\tensor_i I_{2n})(\CNOT_n)(U_{2n}\tensor_i I_{2n})(H_{2n}\tensor_i I_{2n})(e_{0^{n}} \tensor_i e_{0^n}) \notag\\
        =~ & (H_{2n}\tensor_i I_{2n})(U^{*}_{2n}\tensor_i I_{2n})(\CNOT_n)(U_{2n}\tensor_i I_{2n})(H_{2n}e_{0^n}\tensor_i I_{2n}e_{0^n}) \notag\\
        =~ & (H_{2n}\tensor_i I_{2n})(U^{*}_{2n}\tensor_i I_{2n})(\CNOT_n)(U_{2n}\tensor_i I_{2n})(\Sigma_{x}(e_x \tensor_i e_{0^n}))\notag\\
        =~ & (H_{2n}\tensor_i I_{2n})(U^{*}_{2n}\tensor_i I_{2n})(\CNOT_n)(\Sigma_{x}(U_{2n}e_x \tensor_i I_{2n}e_{0^n}))\notag\\
        =~ & (H_{2n}\tensor_i I_{2n})(U^{*}_{2n}\tensor_i I_{2n})(\CNOT_n)(\Sigma_{x}(e_{f(x)} \tensor_i e_{0^n}))\notag\\
        =~ & (H_{2n}\tensor_i I_{2n})(U^{*}_{2n}\tensor_i I_{2n})(\Sigma_{x}(e_{f(x)} \tensor_i
        e_{f(x)}))\notag\\
        =~ & (H_{2n}\tensor_i I_{2n})(\Sigma_{x}(U^*_{2n}e_{f(x)} \tensor_i I_{2n}e_{f(x)}))\notag\\
        =~ & (H_{2n}\tensor_i I_{2n})(\Sigma_{x}(e_{x} \tensor_i e_{f(x)})) \label{Eq:SimonsExp1}
\end{align}
Starting from \eqref{SimonsAlgoExpressionOptimized}, we have
\begin{align}
    &(H_{2n}\tensor_i U_{2n})(\CNOT_n)(H_{2n}\tensor_i I_{2n})(e_{0^{2n}}) \tag{\ref{Eq:SimonsAlgoExpressionOptimized}} \\
    =~ & (H_{2n}\tensor_i U_{2n})(\CNOT_n)(H_{2n}\tensor_i I_{2n})(e_{0^n} \tensor_i e_{0^n})\notag\\
    =~ & (H_{2n}\tensor_i U_{2n})(\CNOT_n)(H_{2n}e_{0^n} \tensor_i I_{2n}e_{0^n})\notag\\
    =~ & (H_{2n}\tensor_i U_{2n})(\CNOT_n)(\Sigma_x(e_x \tensor_i e_{0^n}))\notag\\
    =~ & (H_{2n}\tensor_i U_{2n})(\Sigma_x(e_x \tensor_i e_x))\notag\\
    =~ & (H_{2n}\tensor_i I_{2n})(I_{2n} \tensor_i U_{2n})(\Sigma_x(e_x \tensor_i e_x))\notag\\
    =~ & (H_{2n}\tensor_i I_{2n})(\Sigma_x(I_{2n}e_x \tensor_i U_{2n}e_x))\notag\\
    =~ & (H_{2n}\tensor_i I_{2n})(\Sigma_x(e_x \tensor_i e_{f(x)}))      \label{Eq:SimonsExp2}
\end{align}
\eqrefs{SimonsExp1}{SimonsExp2} are identical, 
and hence~\eqref{SimonsAlgoExpressionOptimized} can be used in place of~\eqref{SimonsAlgoExpression}.

For step (\ref{It:Simon:StepEleven}), our implementation uses the algorithm presented in~\cite[\S{2.1}]{elder2014abstract} for solving the set $E$ of Boolean linear equations.

The implementation also illustrates two of the advantages of simulation discussed in \sectref{AdvantagesOfSimulation}:
\begin{itemize}
  \item
    The discussion above about choosing to interleave the bits of $x$ and $y$, and replacing $\Pi_{i=1}^{n}(\CNOT(2n, i, n+i))$ with $\CNOT_n$ is in the spirit of \itemref{Deviate:GoOutsideQuantumModel}.
  \item
    In accordance with \itemref{Deviate:RepeatedMeasurements}, the implementation does not have to perform steps (\ref{It:Simon:StepTwo})--(\ref{It:Simon:StepSeven}) for each sampling step.
    Instead, it performs steps (\ref{It:Simon:StepTwo})--(\ref{It:Simon:StepSeven}) \emph{once} to build up an appropriate CFLOBDD, and then samples the desired number of times from that structure.
\end{itemize}

\subsubsection{Shor's Algorithm}
\label{Se:ShorsAlgorithm}

Shor's Algorithm aims to find prime factors of a given integer $N$. We use the following algorithm:
\begin{enumerate}
  \item
    \label{It:Shor:StepOne}
    Pick a random number $1<a<N$.
  \item
    \label{It:Shor:StepTwo}
    Compute $K=\gcd(a,N)$. If $K \neq 1$, then $K$ is a nontrivial factor of $N$, so we are done.
  \item
    \label{It:Shor:StepThree}
    Otherwise, use the quantum period-finding algorithm (see below) to find $r$, which denotes the period of the function $f(x)=a^{x}{\bmod {N}}$.
    Equivalently, $r$ is the smallest positive integer that satisfies $a^{r}\equiv 1{\bmod {N}}$.
  \item
    \label{It:Shor:StepFour}
    If $r$ is odd, then go back to step (\ref{It:Shor:StepOne}).
  \item
    \label{It:Shor:StepFive}
    If $a^{r/2}=-1{\bmod {N}}$, then go back to step (\ref{It:Shor:StepOne}).
  \item
    \label{It:Shor:StepSix}
    Otherwise, both $\gcd(a^{r/2}+1,N)$ and $\gcd(a^{r/2}-1,N)$ are nontrivial factors of $N$, so we are done.
\end{enumerate}

The quantum period-finding algorithm is as follows:
\begin{enumerate}
    \item Start with $3n$ qubits with the first $2n$ qubits in state $e_{0^{2n}}$ and the last $n$ qubits in state $e_{0^{n-1}1}$.
    \item For each qubit $i \rightarrow \{N..1\}$, apply the controlled gate $U_a^{2^{(N-i-1)}}$ on the last $n$ qubits, where $U_a(x) = ax{ \bmod N}$.
    \item Apply Inverse QFT on the first $2n$ qubits.
    \item Measure the first $2n$ qubits to find $r$ as discussed in~\cite[\S11.6]{lipton2014quantum}.
\end{enumerate}

\subsubsection{Grover's Algorithm}
\label{Se:groveralgo}
The algorithm solves the problem of finding a needle in a haystack. More formally,
\begin{quote}
    Given a function $f:\{0,1,\cdots,N-1\} \rightarrow \{0,1\}$, the algorithm finds $x$ such that $f(x) = 1$.
\end{quote}

Here, we assume that there is only one $x$ such that $f(x)=1$.
The algorithm uses an oracle $U_{2n}$, which implements $f$ for the ``needle'' $w$ as
\[
U_w(x) = \begin{cases}
            x & \text{if } x = w\\
           -x & \text{otherwise}
         \end{cases}
\]

Given this oracle, the algorithm for finding $s$ is as follows:
\begin{enumerate}
  \item
    \label{It:Grover:StepOne}
    Prepare an initial state $e_{0^n}$ (the standard-basis vector of $n$ bits with a $1$ in the first position and $0$s elsewhere).
  \item
    \label{It:Grover:StepTwo}
    Apply the Hadamard matrix $H_{2n}$ on the initial state.
  \item
    \label{It:Grover:StepThree}
    Apply the oracle $U_w$ to the current state.
  \item
    \label{It:Grover:StepFour}
    Apply the Grover diffusion operator
    $U_s = H_{2n}(2e_{0^n}\tensor_{\textit{outer}}e_{0^n} - I_{2^n})H_{2n}$,
    where $\tensor_{\textit{outer}}$ denotes the outer product of two vectors.
  \item
    \label{It:Grover:StepFive}
    Repeat steps (\ref{It:Grover:StepThree})--(\ref{It:Grover:StepFour}) $\lceil{\frac{\pi}{4}\sqrt{N}}\rceil$ times.
  \item
    \label{It:Grover:StepSix}
    Measure the $n$ qubits to obtain the string $s$.
\end{enumerate}

Steps (\ref{It:Grover:StepOne})--(\ref{It:Grover:StepFive}) can be written concisely as follows:
\[
    \left( \Pi_{i=1}^{\lceil{\frac{\pi}{4}\sqrt{N}}\rceil}U_sU_w \right) (H_{2n}(e_{0^{n}}))
\]

As mentioned at the beginning of \sectref{quantum-algos}, one of the advantages of simulation is that we can re-associate the matrix multiplications in steps (\ref{It:Grover:StepOne})--(\ref{It:Grover:StepFour}) to compute $\left( \Pi_{i=1}^{\lceil{\frac{\pi}{4}\sqrt{N}}\rceil}U_sU_w \right)$ as an explicit quantity, which can be done using repeated squaring instead of as a sequence of  multiplications.
This approach provides a more efficient approach to steps (\ref{It:Grover:StepThree})--(\ref{It:Grover:StepFour}).
We can also use optimizations like those discussed earlier, e.g., performing (\ref{It:Grover:StepOne})--(\ref{It:Grover:StepTwo}) by means of a direct construction of the desired vector of all-$1$s.
We can also create a representation of the diffusion operator $U_s$ directly, rather than performing the computation given in step (\ref{It:Grover:StepFour}).
In particular, the construction is almost the same as the identity-matrix construction in \algref{IdAlgo}:
one would use subroutine {\tt IdentityMatrixGrouping}, but at top level the value tuple would be $[\frac{2}{N}-1, \frac{2}{N}]$ (instead of $[1,0]$).
\Omit{
\twr{This construction follows \href{https://www.scottaaronson.com/qclec/22.pdf}{https://www.scottaaronson.com/qclec/22.pdf}, page 3. Meghana: you had a $\sqrt{N}$, but Aaronson does not.}
}


\section{Evaluation}
\label{Se:evaluation}

In this section, we explain our experimental setup and describe the experiments we carried out, which were designed to address the following research questions:
\begin{description}
  \item [RQ1:]
    Do theoretical guarantees of \emph{double-exponential compression} by CFLOBDDs allow them to represent substantially larger Boolean functions than BDDs?
  \item [RQ2:]
    Do CFLOBDDs outperform BDDs when used for quantum simulation (in terms of time and space)?
\end{description}

\subsection{Experimental Setup}
\label{Se:ExperimentalSetup}

We compared our implementation of CFLOBDDs\footnote{
  The implementation is available at \url{https://github.com/trishullab/cflobdd}.
} against a widely used BDD package, CUDD~\cite{somenzi2012cudd} (version
3.0.0),using CUDD's C++ interface.
The metrics are (i) execution time, and (ii) space (node counts in the case of BDDs; vertex counts + edge counts in the case of CFLOBDDs).
We ran all experiments on AWS machines: t2.xlarge machines with 4 vCPUs, 16GB of memory, and a stack size of 8192KB, running on Ubuntu OS.

For RQ1 (\sectref{ResearchQuestionOne}), we used a collection of synthetic benchmarks, and compared the performance of CFLOBDDs against
(i) CUDD with a static variable ordering (similar to the one used in the CFLOBDDs),
(ii) CUDD with dynamic variable reordering, and
(iii) Sentential Decision Diagrams (SDDs)~\cite{darwiche2011sdd} (which can also be exponentially more succinct than BDDs), using Python package PySDD~\cite{wannes_meert_2018_1202374} (version 0.1). 

For RQ2 (\sectref{ResearchQuestionTwo}), we used a set of quantum-simulation benchmarks, and again compared the performance of CFLOBDDs against CUDD (version 3.0.0).
For the quantum benchmarks, we did not enable dynamic variable reordering for BDDs because we could not retrieve the correct order of the output bits for a sampled string.

Five of the quantum benchmarks---BV, DJ, Simon's algorithm, Shor's algorithm, and Grover's algorithm---use oracles for that either directly or indirectly incorporate the answer sought.
Our methodology is standard for quantum-simulation experiments.
Each benchmark uses a pre-processing step to create the CFLOBDD/BDD that represents the oracle.
In each run, an answer is first generated randomly, and then the CFLOBDD/BDD that represents the oracle is constructed.
Knowledge about the answer is used only during oracle construction.
Thereafter, the quantum algorithm proper is simulated;
these steps have no access to the pre-chosen answer (other than the ability to perform operations on the oracle, treated as a unitary matrix).
The final step of running the benchmark is to check that the quantum algorithm obtained the correct answer.

We could not run the quantum benchmarks with SDDs because SDDs do not support multi-terminal values.
However, we ran the quantum benchmarks using Quimb~\cite{gray2018quimb}, a quantum-simulation library that uses tensor networks.

\paragraph{Extensions to CUDD}
For the RQ2 experiments, we had to extend CUDD in two ways to be able to simulate quantum circuits using CUDD data structures:
\begin{enumerate}
  \item
    CUDD supports \emph{algebraic decision diagrams} (ADDs), which are multi-terminal BDDs with a value from a semiring at each terminal node.
    We had to extend CUDD with a semiring that was not part of the standard CUDD distribution (\sectref{ADDsWithComplexNumberLeaves}).
    For the corresponding experiments with CFLOBDDs, we used the same semiring for the terminal values of CFLOBDDs.
    \Omit{
    To support functions of type ${\{ 0,1 \}}^n \rightarrow \mathbb{C}$, we implemented a semiring of multi-precision-floating-point \cite{fousse2007mpfr} complex numbers
    (\sectref{ADDsWithComplexNumberLeaves}).
    }
  \item
    To allow quantum measurements to be carried out,
    we extended ADDs to support path sampling (i.e., selection of a path, where the probability of returning a given path is proportional to a function of the path's terminal value).
\end{enumerate}

\subsubsection{Algebraic Decision Diagrams with Complex-Number Leaves}
\label{Se:ADDsWithComplexNumberLeaves}

To use CUDD on most of the quantum benchmarks, we modified the ADD datatype and related functions to use multi-precision-floating-point complex numbers \cite{fousse2007mpfr}.

For the ``Quantum Fourier Transform'' and ``Shor's Algorithm'' benchmarks, one only needs to represent a known set of complex roots of unity, so we used a different custom ADD datatype, $\textit{dtype}$, defined as follows:

\begin{center}
\begin{tabular}{c}
\lstset{language=C++}
\begin{lstlisting}
typedef struct dtype {
    int val;
    int size;
    mpfr_t real;
    mpfr_t imag;
} dtype;
typedef dtype CUDD_VALUE_TYPE;
\end{lstlisting}
\end{tabular}
\end{center}

\noindent
Here, ``size'' represents the number of roots of unity, and a value $i$ for ``val'' represents the $i^{\textit{th}}$ root of unity with respect to the current ``size.'' 
For example, if size = 4, then val $\in$ $\{ 0,1,2,3 \}$.
Multiplying two $\textit{dtype}$s---$t_1$ and $t_2$ with size = 4, where val of $t_1$ = 2 and val of $t_2$ = 3---produces $\textit{dtype}$ $t_3$ with size = 4 and val = (2 + 3) \% 4 = 1.
This representation encodes $\omega^2 \times \omega ^ 3 = \omega$, where $\omega^k = e^{\frac{2k\pi}{4}}$ is a primitive $4^{\textit{th}}$ root of unity.
We used this representation to compute the entire quantum Fourier transform, and then filled in the values for real and imag at the last step, where real = $\cos(\frac{2\pi * val}{size})$ and imag = $\sin(\frac{2\pi * val}{size})$.

\subsubsection{Sampling in BDDs.}
\label{Se:SamplingInBDDs}

An important step in quantum simulation is measurement of the output bits,
which is equivalent to sampling a path from the BDD structure according to a probability distribution obtained from the terminal values.
In particular, the terminal values represent so-called ``amplitudes,'' and the corresponding probability distribution is based on the squares of the amplitudes.

\paragraph{Path Counting}
Suppose that the BDD has $k$ terminals.
To sample a path based on the squares of the amplitudes, we need to know, for each node $n$ in the BDD, and each terminal position $t_i$, $1 \leq i \leq k$, the number of paths that lead from $n$ to $t_i$.
This information can be computed by generalizing the method of Ball and Larus for counting the number of paths in a DAG \cite{micro:BL96} from the case of a single exit-node to $k$ exit-nodes (while also accounting for ply-skipping in a BDD):
\begin{itemize}
  \item
    Each ``\emph{path-count}'' is a k-dimensional vector $c$.
    (For convenience, we index the components of $c$ as $c_1, \ldots, c_k$.)
  \item
    Path-counts are computed bottom-up, starting from the terminal positions.
    The path-count at terminal position $i$ has a $1$ at index-position $i$ and zeros elsewhere, signifying that (i) there is exactly $1$ path from terminal node $i$ to itself, and (ii) there are no paths from terminal node $i$ to any of the other terminal nodes.
  \item
    When no plies are skipped in going from a node $n$ to each of its two children, the path-count at $n$ is the vector addition of the left-child and right-child path-counts.
  \item
    If $\Delta\textit{left}$ plies are skipped in going from $n$ to the left child (with path-count $c_l$), and $\Delta\textit{right}$ plies are skipped in going to the right child (with path-count $c_r$), the path-count at $n$ is $2^{\Delta\textit{left}}{c_l} + 2^{\Delta\textit{right}}{c_r}$.
  \item
    If $\Delta$ plies have been skipped at the apex of the DAG (with path-count $c$), the BDD's overall path-count is $2^{\Delta}{c}$.
\end{itemize}

\paragraph{Sampling}
Once path-counts are in hand, we sample a path as follows:
\begin{itemize}
  \item
    Let $n.\textit{pc}_i$ denote the path-count at a node $n$.
    Let $n.\textit{lchild}$ and $n.\textit{rchild}$ denote the two children of $n$, and let $n.\Delta\textit{left}$ and $n.\Delta\textit{right}$ denote the number of plies skipped in going to the left child and right child of $n$, respectively.
  \item
    For convenience, if $\Delta$ plies have been skipped at the apex $a$ of the DAG, we assume that a new root node $r$ is added whose left-child and right-child both point to $a$.
    $r.\textit{pc}$ is set to $2^{\Delta}{a.\textit{pc}}$, and both $r.\Delta\textit{left}$ and $r.\Delta\textit{right}$ are set to $\Delta-1$.
    Otherwise, the root $r$ is $a$.
  \item
    Let $p = \langle p_1, \ldots, p_k \rangle$, denote the vector of squared amplitudes at the $k$ terminal positions.
    The probability of choosing a path from $r$ to terminal position $i$ is $p_i * r.\textit{pc}_i$.
  \item
    Randomly choose a value $i$ based on the probabilities $\langle p_1 * r.\textit{pc}_1, \ldots, p_k * r.\textit{pc}_k \rangle$.
    Henceforth, we only consult the $i^{\textit{th}}$ components of path-counts of the BDD's nodes.
  \item
    We work down the BDD from $r$ to terminal position $i$, accumulating a path-string in $\pi$, which is initially set to $\epsilon$.
    Starting at $r$, and then at each subsequently visited node $n$ until terminal position $i$ is reached, take the following steps:
    \begin{itemize}
      \item
        If both $n.\textit{lchild}.\textit{pc}_i \neq 0$ and $n.\textit{rchild}.\textit{pc}_i \neq 0$, randomly choose one of the children in relative proportion to
        $2^{n.\Delta\textit{left}}{n.\textit{lchild}.\textit{pc}_i}$ and
        $2^{n.\Delta\textit{right}}{n.\textit{rchild}.\textit{pc}_i}$.
        If the left child is chosen, append ``0'' followed by a string chosen uniformly from $\{0,1\}^{{n.\Delta\textit{left}}}$ to $\pi$;
        otherwise, append ``1'' followed by a string chosen uniformly from $\{0,1\}^{{n.\Delta\textit{right}}}$ to $\pi$.
      \item
        If $n.\textit{lchild}.\textit{pc}_i = 0$, append ``1'' followed by a string chosen uniformly from $\{0,1\}^{{n.\Delta\textit{right}}}$ to $\pi$.
      \item
        If $n.\textit{rchild}.\textit{pc}_i = 0$, append ``0'' followed by a string chosen uniformly from $\{0,1\}^{{n.\Delta\textit{left}}}$ to $\pi$.
    \end{itemize}
    Set $n$ to the child chosen above.
\end{itemize}

\Omit{
With path information, we now present a technique for sampling a path.
Starting at the top ply, the probability of choosing a particular weight $w$ as $p_w$ = $w * path\_count$.
We randomly choose a weight based on the probability distribution $p_w$.
If both the left child and right child have paths that lead to node with value $w$, we randomly choose one of them based on l\_path\_count and r\_path\_count.
If the left child is chosen, we add ``0'' to the sampled path;
otherwise, we add ``1.''
If one of the children does not have a path to node with value $w$, we choose the other child with probability 1.
We repeat this process at every ply till we reach the terminal value node.
In case of ply-skips, we randomly choose ``0'' or ``1'' with a probability of 0.5.
}

\subsection{Benchmarks and Experimental Results}
\label{Se:BenchmarksAndExperimentalResults}

\subsubsection{RQ1:
    Do theoretical guarantees of \emph{double-exponential compression} by CFLOBDDs allow them to represent substantially larger Boolean functions than BDDs?
}
\label{Se:ResearchQuestionOne}

We used the following three benchmarks to compare the execution time and memory usage (as vertex count + edge count) of CFLOBDDs against BDDs and SDDs.

\begin{itemize}
  \item
    $\XOR_n = \bigoplus_{i=1}^{n} x_i$
  \item
    $\textit{MatMult}_n = (H_{n}I_n + X_nH_n + I_{n}X_n)$, where $H_{n}$ is the Hadamard matrix, $I_{n}$ is the Identity matrix, and $X_n$ is the NOT matrix of size $2^{{n-1}}$ x $2^{{n-1}}$.
    (The aim of the benchmark is to test the performance of the matrix-multiplication and addition operations.)
  \item
  $\ADD_n(X,Y,Z) \eqdef Z = (X + Y \mod 2^{n/4})$, where $X$, $Y$, and $Z$ are $n/4$-bit integers.
\end{itemize}

\begin{table}[tb!]
    \centering
    \begin{adjustbox}{width=1\textwidth}
    \begin{tabular}{|c|c|c|c|c|c|c|c|c|c|c|c|}
    \hline
        \multirow{3}*{Benchmark} &
        \multirow{3}{*}{
            \begin{tabular}{@{}c@{}}
                 \#Boolean  \\
                 Variables ($n$)
            \end{tabular}
        } & \multicolumn{4}{|c|}{CFLOBDD} & \multicolumn{2}{|c|}{BDD} & \multicolumn{2}{|c|}{BDD (reorder)} & \multicolumn{2}{|c|}{SDD}\\
        \cline{3-12}
        & & \multirow{2}*{\#Vertices}  & \multirow{2}*{\#Edges}  & \multirow{2}*{Total}  & Time  & \multirow{2}*{\#Nodes}  & Time  & \multirow{2}*{\#Nodes}  & Time  & \multirow{2}*{\#Nodes}  & Time  \\
        & &  &  &                       & (sec) &  & (sec) &  & (sec) &  & (sec) \\
        \hline
        \multirow{6}*{$\XOR_n$} & $2^{15}$ &  16 & 96 & 112 & \textbf{0.99} & 32769 & 551.27 & 32769 & 587.11 & 131066	& 3.91\\
        \cline{2-12}
         & $2^{16}$ &  17 & 102 & 119 & \textbf{2.18} & \multicolumn{4}{c|}{\multirow{5}{*}{Timeout (15min)}} & 262138	& 8.57\\
        \cline{2-6}\cline{11-12}
         & $2^{17}$ &  18 & 108 & 126 & \textbf{5} & \multicolumn{4}{c|}{} & 524282	& 18.71\\
        \cline{2-6}\cline{11-12}
         & $2^{18}$ &  19 & 114 & 133 & \textbf{12.75} & \multicolumn{4}{c|}{} & 1048570	& 38.63\\
        \cline{2-6}\cline{11-12}
         & $2^{19}$ &  20 & 120 & 140 & \textbf{36.06} & \multicolumn{4}{c|}{} & 2097146	& 82.03\\
        \cline{2-6}\cline{11-12}
         & $2^{20}$ &  21 & 126 & 147 & \textbf{122.97} & \multicolumn{4}{c|}{} & 4194298 & 191.68\\ 
        \hline
        \multirow{11}{*}{$\textit{MatMult}_n$} & $2^{15}$ & 84 & 1053 & 1137 & \textbf{0.002} & 294890 & 57.33 & 294890 & 156.42 & \multicolumn{2}{c|}{\multirow{11}{*}{Not Applicable}}\\
        \cline{2-10}
         & $2^{16}$  & 90 & 1125 & 1137 & \textbf{0.004} & 589802 & 186.27 & 593122 & 446.19 & \multicolumn{2}{c|}{}\\
        \cline{2-10}
         & $2^{17}$ & 96 & 1197 & 1293 & \textbf{0.007} & 1179626 & 739.66 & \multicolumn{2}{c|}{Timeout (15min)} & \multicolumn{2}{c|}{}\\
        \cline{2-10}
         & $2^{18}$ & 102 & 1269 & 1371 & \textbf{0.017} & \multicolumn{4}{c|}{\multirow{8}{*}{Timeout (15min)}} & \multicolumn{2}{c|}{}\\
        \cline{2-6}
        & $2^{19}$  & 108 & 1341 & 1449 & \textbf{0.043} & \multicolumn{4}{c|}{} & \multicolumn{2}{c|}{}\\
        \cline{2-6}
        & $2^{20}$ & 114 & 1413 & 1527 & \textbf{0.118} & \multicolumn{4}{c|}{} & \multicolumn{2}{c|}{}\\
        \cline{2-6}
        & $2^{21}$ & 120 & 1485 & 1605 & \textbf{0.343} & \multicolumn{4}{c|}{} & \multicolumn{2}{c|}{}\\
        \cline{2-6}
        & $2^{22}$ & 126 & 1557 & 1683 & \textbf{1.238} & \multicolumn{4}{c|}{} & \multicolumn{2}{c|}{}\\
        \cline{2-6}
        & $2^{23}$ & 132 & 1629 & 1761 & \textbf{4.936} & \multicolumn{4}{c|}{} & \multicolumn{2}{c|}{}\\
        \cline{2-6}
        & $2^{24}$ & 138 & 1701 & 1839 & \textbf{19.37} & \multicolumn{4}{c|}{} & \multicolumn{2}{c|}{}\\
        \cline{2-6}
        & $2^{25}$ & 144 & 1773 & 1917 & \textbf{78.98} & \multicolumn{4}{c|}{} & \multicolumn{2}{c|}{}\\
        \cline{2-6}
        & $2^{26}$ & 150 & 1845 & 1995 & \textbf{317.27} & \multicolumn{4}{c|}{} & \multicolumn{2}{c|}{}\\
        \cline{2-6}
        & $2^{27}$ & \multicolumn{4}{c|}{Timeout (15min)} & \multicolumn{4}{c|}{} & \multicolumn{2}{c|}{}\\
        \hline
        \multirow{13}{*}{$\ADD_n$} & $2^{17}$ & 80 & 574 & 654 & \textbf{<0.001} & 131073 & 0.035 & 132405 & 80.24 & 393152 & 7.72\\
        \cline{2-12}
        & $2^{18}$ & 85 & 610 & 695 & \textbf{0.001} & 262145 & 0.065 & 263477 & 280.79 & 786364 & 13.82\\
        \cline{2-12}
         & $2^{19}$ & 90 & 646 & 736 & \textbf{0.001} & 524289 & 0.148 & \multicolumn{2}{c|}{\multirow{5}{*}{Timeout (15min)}} & 1572792 & 29.72\\
        \cline{2-8}\cline{11-12}
        & $2^{20}$ & 95 & 682 & 777 & \textbf{0.001} & 1048577 & 0.293 & \multicolumn{2}{c|}{} & 3145652 & 66.26\\
        \cline{2-8}\cline{11-12}
        & $2^{21}$ & 100 & 718 & 818 & \textbf{0.001} & 2097153 & 1.368 & \multicolumn{2}{c|}{} & 6291376 & 138.40\\
        \cline{2-8}\cline{11-12}
        & $2^{22}$ & 105 & 754 & 859 & \textbf{0.001} & 4194305 & 1.155 & \multicolumn{2}{c|}{} & 12582828 & 359.26\\
        \cline{2-8}\cline{11-12}
        & $2^{23}$ & 110 & 790 & 900 & \textbf{0.002} & 8388609 & 3.316 & \multicolumn{2}{c|}{} & \multicolumn{2}{c|}{\multirow{3}{*}{Out of Memory}}\\
        \cline{2-10}
        & $2^{24}$ & 115 & 826 & 941 & \textbf{0.003} & \multicolumn{4}{c|}{\multirow{2}{*}{Out of Memory}} & \multicolumn{2}{c|}{}\\
        \cline{2-6}
        & $2^{25}$ & 120 & 862 & 982 & \textbf{0.003} & \multicolumn{4}{c|}{} & \multicolumn{2}{c|}{}\\
        \cline{2-12}
        & $\vdots$ & $\vdots$ & $\vdots$ & $\vdots$ & $\vdots$ & \multicolumn{6}{c|}{\multirow{4}{*}{Out of Memory}}\\
        \cline{2-6}
        & $2^{2^{21}}$ & 10485755 & 75497434 & 85983189 & \textbf{113.99} & \multicolumn{6}{c|}{}\\
        \cline{2-6}
        & $2^{2^{22}}$ & 20971515 & 150994906 & 171966421 & \textbf{385.75} & \multicolumn{6}{c|}{} \\
        \cline{2-6}
        & $2^{2^{23}}$ & \multicolumn{4}{c|}{Timeout (15min)} & \multicolumn{6}{c|}{} \\
        \hline
    \end{tabular}
    \end{adjustbox}
    \caption{Performance of CFLOBDDs against BDDs, BDDs with dynamic reordering, and SDDs on the synthetic benchmarks for different
    numbers of Boolean variables.
    (For the two kinds of BDD experiments and the SDD experiments, we used a stack size of 1GB.)
    }
    \label{Ta:micro-benchmarks-table}
\end{table}

\tableref{micro-benchmarks-table} shows the performance of CFLOBDDs, BDDs (with and without dynamic reordering enabled), and SDDs
within the 15-minute timeout threshold.
For the two kinds of BDD experiments and the SDD experiments, we used a stack size of 1GB.
For the $\ADD$ benchmark, BDDs (both with and without dynamic reordering) and SDDs ran out of memory within the 15-minute timeout threshold for problems with sufficiently many variables, even with such a large stack. (BDDs with dynamic reordering produced out-of-memory errors for \#variables $\geq 2^{24}$: the first step in the computation is to allocate the variables, which by itself leads to memory exhaustion for $2^{24}$ variables and beyond.)
Note that, for SDDs, benchmark
$\textit{MatMult}_n$
is not applicable because SDDs do not handle non-Boolean values.

Because of a slightly technical alignment issue, our CFLOBDD representations of $\ADD_n$ deliberately waste one-quarter of the Boolean variables (as \emph{dummy variables}).
To make a fair comparison, our BDD and SDD encodings of $\ADD_n$ use only three-quarters of the Boolean variables indicated in column two of
\tableref{micro-benchmarks-table}.

\begin{figure}[bt!]
    \centering
    \begin{tabular}{@{\hspace{0ex}}l@{\hspace{0ex}}}
      \begin{subfigure}{1.0\linewidth}
        \includegraphics[width=0.495\linewidth]{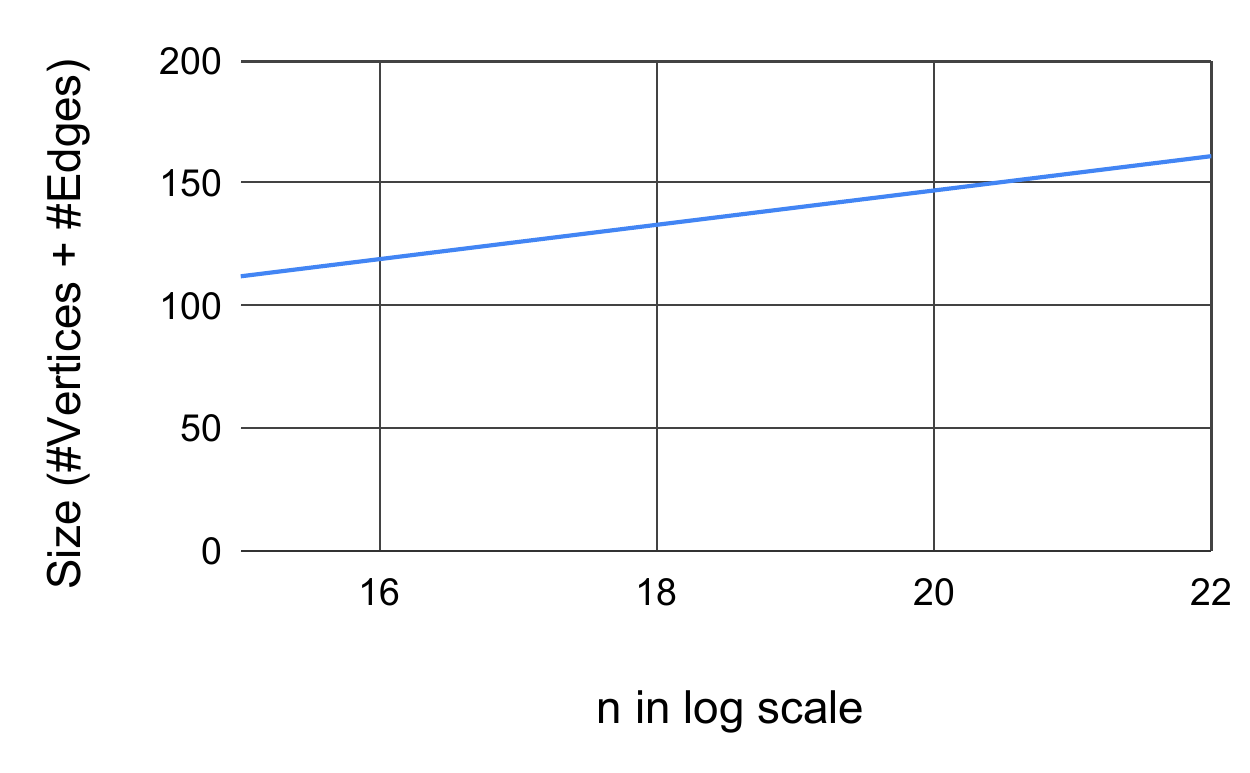}
        \includegraphics[width=0.495\linewidth]{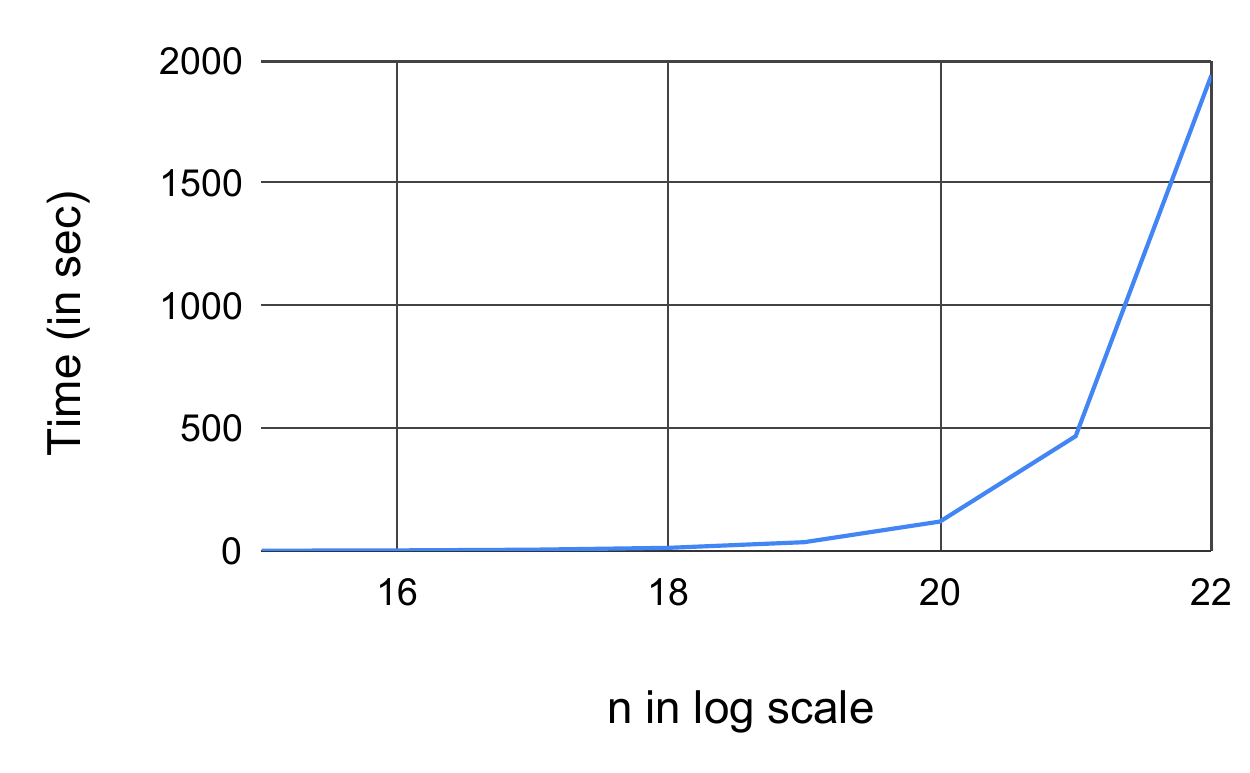}
        \caption{$\XOR_n$: Size and time vs.\ $n$ (on a $\log$ scale)}
        \vspace{3ex}
      \end{subfigure}
      \\
      \begin{subfigure}{1.0\linewidth}
        \includegraphics[width=0.495\linewidth]{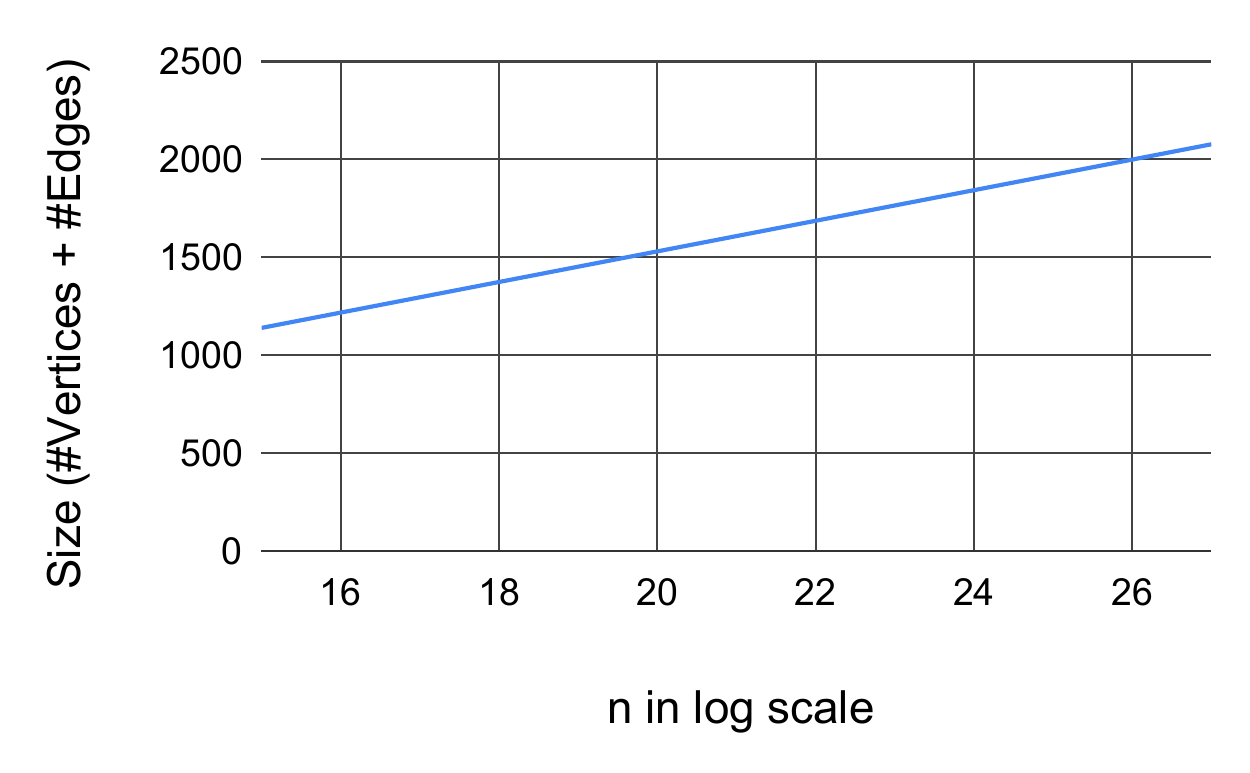}
        \includegraphics[width=0.495\linewidth]{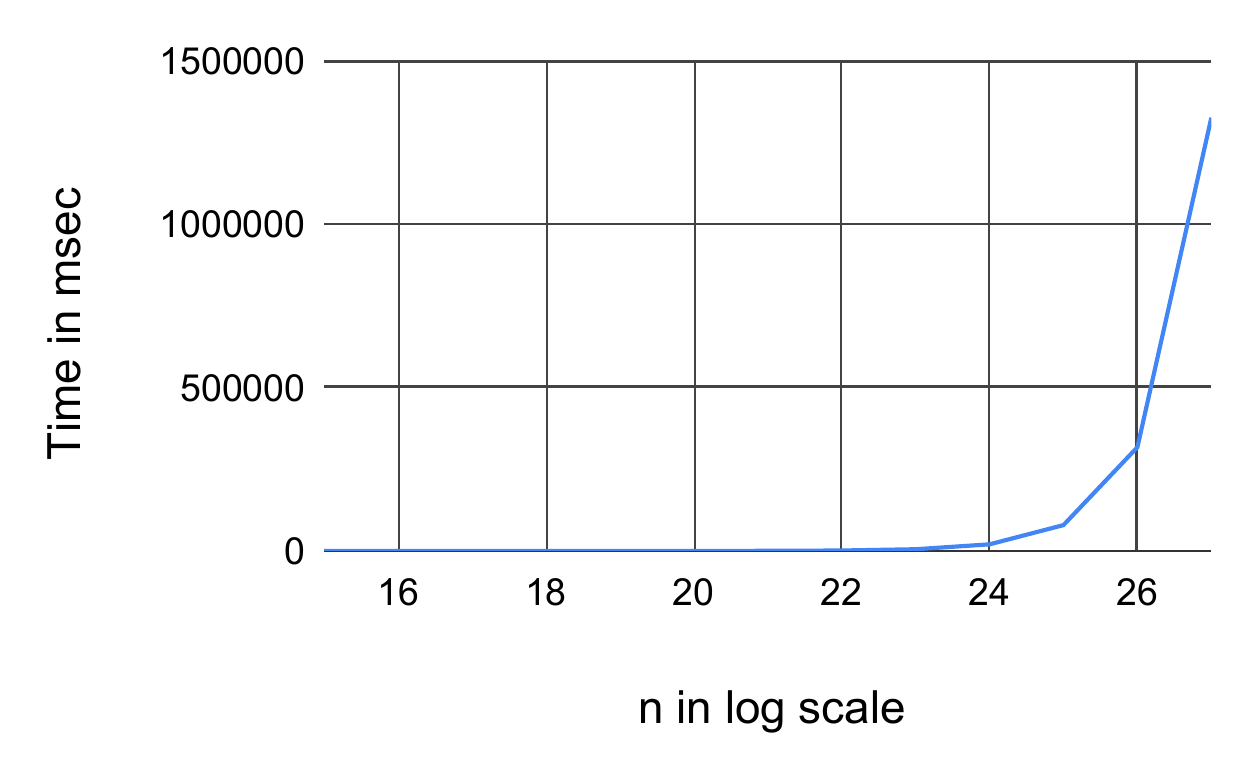}
        \caption{$\textit{MatMult}_n$: Size and time vs.\ $n$ (on a $\log$ scale)}
        \vspace{3ex}
      \end{subfigure}
      \\
      \begin{subfigure}{1.0\linewidth}
        \includegraphics[width=0.495\linewidth]{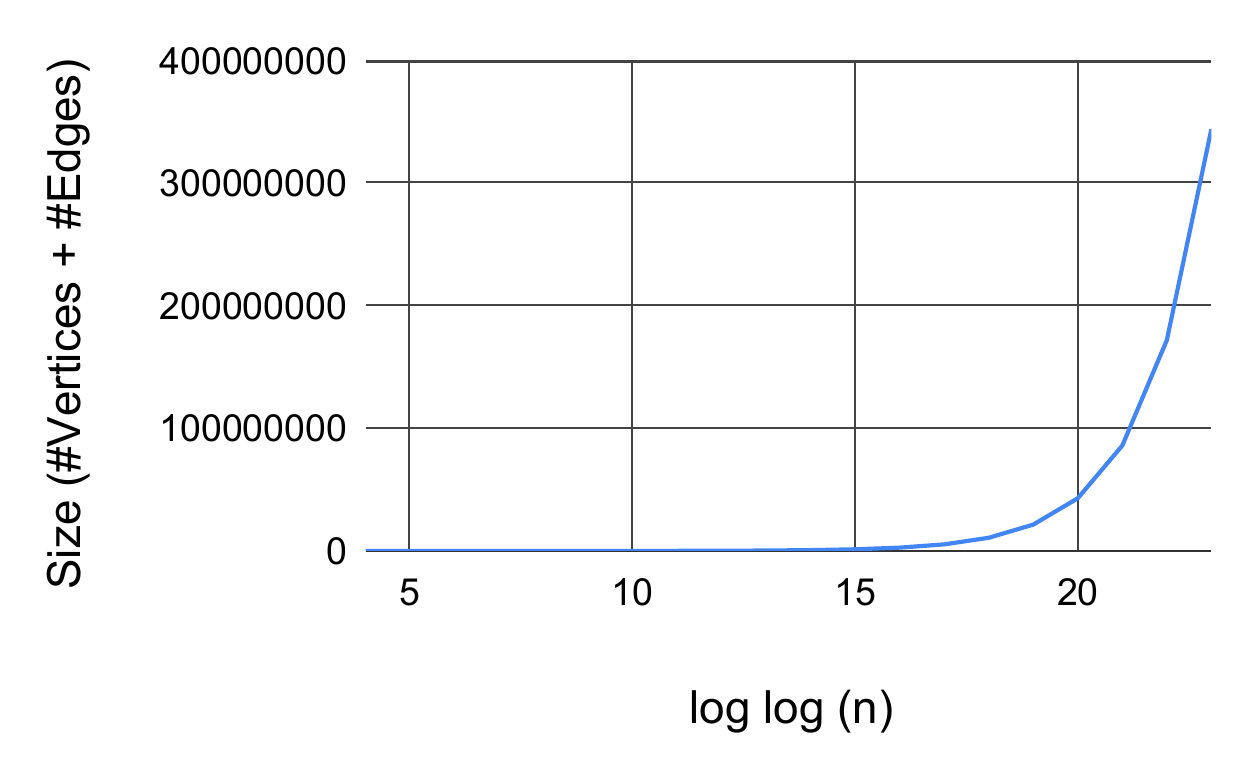}
        \includegraphics[width=0.495\linewidth]{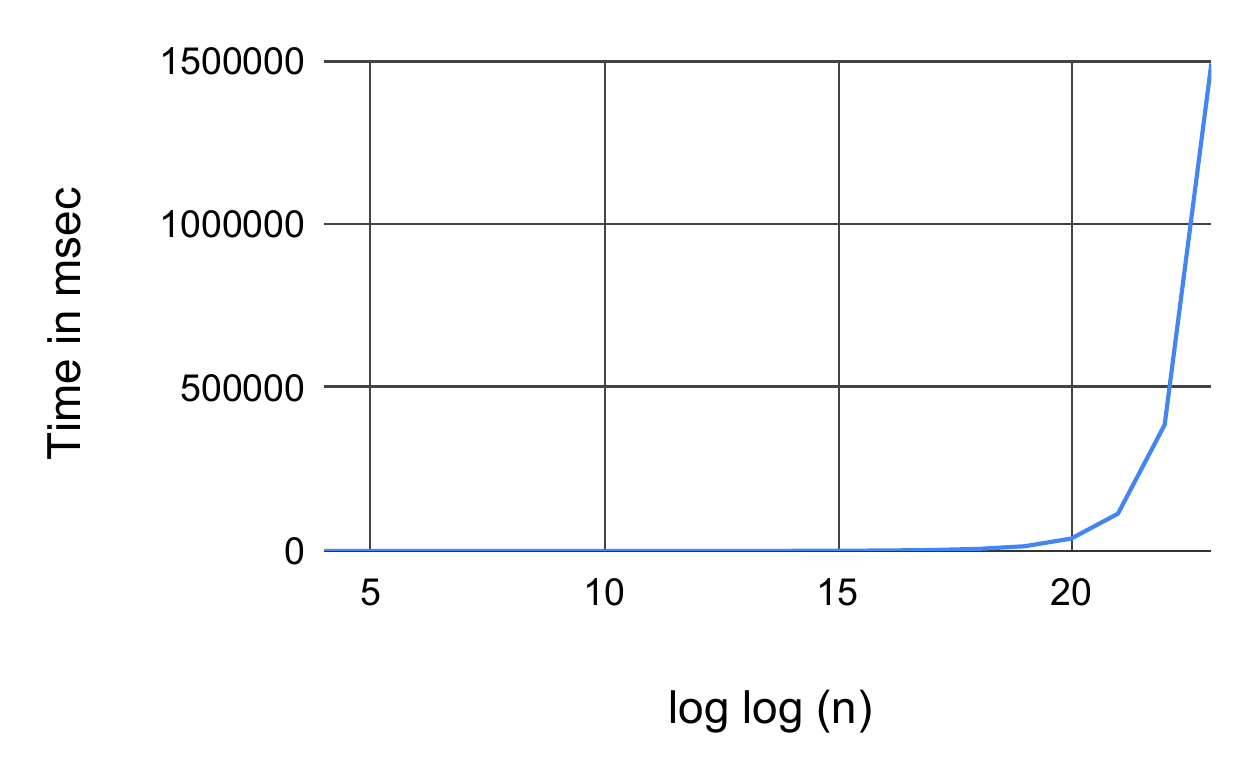}
        \caption{$\ADD_n$: Size and time vs.\ $n$ (on a $\log\log$ scale)}
      \end{subfigure}
    \end{tabular}
    \caption{
    CFLOBDD performance with a timeout of ninety minutes.
    Note that in (c)
    the number of Boolean variables is on a $\log\log$ scale.
    }
    \label{Fi:micro-benchmarks-plots}
\end{figure}

To understand how large a Boolean function could be created using CFLOBDDs (as a function of the number of Boolean variables),\footnote{
  The stack size was increased to 1GB for the runs with more than $2^{2^{15}}$ Boolean variables.
}
we also measured the performance of the CFLOBDD implementation on the micro-benchmarks using a timeout of ninety minutes.

\figref{micro-benchmarks-plots} shows graphs of size (\#vertices + \#edges) and time versus the number of Boolean variables for the three benchmarks.\footnote{
  \tableref{micro-benchmarks-table} shows the comparison of CFLOBDDs, BDDs, and SDDs for examples with a 15-minute timeout. In contrast, \figref{micro-benchmarks-plots} shows the results of the stress test that we performed, where we gave the CFLOBDD implementation a 90-minute timeout.
}
\figref{micro-benchmarks-plots}a shows the graphs for $\XOR_n$.
In these graphs, time is in seconds, and the number of Boolean variables is on a $\log$ scale.
We were able to construct $\XOR_n$ with up to $2^{22} = \textrm{4,194,304}$ variables.
\figref{micro-benchmarks-plots}b and \figref{micro-benchmarks-plots}c show the graphs for $\textit{MatMult}_n$ and $\ADD_n$, respectively.
In these graphs, time is in milliseconds, and the number of Boolean variables is on a $\log$ scale for $\textit{MatMult}_n$ and a $\log\log$ scale for $\ADD_n$.
We were able to construct $\textit{MatMult}_n$ with up to $2^{27} = \textrm{134,217,728}$ variables and $\ADD_n$ with up to $2^{2^{23}}$ $= 2^{\textrm{8,388,608}}$ $\cong 4.27 \times 10^{\textrm{2,525,222}}$ variables, which comes to $0.75 \times 2^{8,388,608} \cong 3.12 \times 10^{\textrm{2,525,222}}$ after removing dummy variables.

\smallskip
\begin{mdframed}[innerleftmargin = 3pt, innerrightmargin = 3pt, skipbelow=-0.0em]
\textbf{Findings.}
CFLOBDDs performed better than BDDs and SDDs, both in terms of time and memory.
For the benchmarks with more than $2^{18}$ Boolean variables, BDDs had memory issues.
Using CFLOBDDs, it was also possible to construct representations of the benchmark functions with astounding numbers of Boolean variables:
$2^{22} = \textrm{4,194,304}$ for $\XOR_n$;
$2^{27}$ $= \textrm{134,217,728}$ for $\textit{MatMult}_n$; and
$0.75 \times 2^{8,388,608} \cong 3.12 \times 10^{\textrm{2,525,222}}$ for $\ADD_n$.
These results support the claim that CFLOBDDs can provide substantially better compression of Boolean functions than BDDs.
\end{mdframed}

\smallskip
\noindent

\subsubsection{RQ2: Do CFLOBDDs outperform BDDs when used for quantum simulation (in terms of time and space)?}
\label{Se:ResearchQuestionTwo}

\begin{table}[!tb]
    \centering
    \resizebox{0.91\textwidth}{!}{
    \begin{tabular}{|c|c|c|c|c|c|c|c|c|}
    \hline
        \multirow{2}*{Benchmark} & \multirow{2}*{\#Qubits} & 
        \multirow{2}{*}{
        \begin{tabular}{@{}c@{}}
            \#Boolean\\
            Variables
        \end{tabular}
        } 
        & \multicolumn{4}{c|}{CFLOBDD} & \multicolumn{2}{c|}{BDD}\\
        \cline{4-9}
        & & & \#Vertices & \#Edges & Total & Time (sec) & \#Nodes & Time (sec)\\
        \hline
        \multirow{14}*{GHZ}
        & 16 & 32 & 35 & 207 & 242 & 0.005 & \textbf{36} & \textbf{0.003}\\
        \cline{2-9}
        & 32 &64& 43 & 255 & 298 & \textbf{0.007} & \textbf{68} & 0.008\\
        \cline{2-9}
        & 64 &128& 51 & 303 & 354 & \textbf{0.010} & \textbf{131} & 0.031 \\
        \cline{2-9}
        & 128 &256& 59 & 351 & 410 & \textbf{0.015} & \textbf{259} & 0.143\\
        \cline{2-9}
        & 256 &512& 67 & 399 & \textbf{466} & \textbf{0.027} & 515 & 4.9\\
        \cline{2-9}
        & 512 &1024& 75 & 447 & \textbf{522} & \textbf{0.051} & 1028 & 44\\
        \cline{2-9}
        & 1024 &2048& 83 & 495 & \textbf{578} & \textbf{0.107} & \multicolumn{2}{c|}{\multirow{8}*{Timeout (15 min)}}\\
        \cline{2-7}
        & 2048 &4096& 91 & 543 & \textbf{638} & \textbf{0.216} & \multicolumn{2}{c|}{}\\
        \cline{2-7}
        & 4096 &8192& 99 & 591 & \textbf{690} & \textbf{0.442} & \multicolumn{2}{c|}{}\\
        \cline{2-7}
        & 8192 &16384& 107 & 639 & \textbf{746} & \textbf{0.631} & \multicolumn{2}{c|}{}\\
        \cline{2-7}
        & 16384 &32768& 115 & 687 & \textbf{802} & \textbf{1.35} & \multicolumn{2}{c|}{}\\
        \cline{2-7}
        & 32768 &65536& 123 & 735 & \textbf{858} & \textbf{2.92} & \multicolumn{2}{c|}{}\\
        \cline{2-7}
        & 65536 &131072& 131 & 783 & \textbf{914} & \textbf{6.49} & \multicolumn{2}{c|}{}\\
        \cline{2-7}
        & 131072 &262144& \multicolumn{4}{c|}{Timeout (15 min)} & \multicolumn{2}{c|}{}\\
        \hline
        \multirow{17}*{BV} & 16 &32& 29 & 172 & 201 & 0.005 & \textbf{31} & \textbf{0.002}\\
        \cline{2-9}
        & 32 &64& 39 & 233 & 272 & 0.006 & \textbf{63} & \textbf{0.004}\\
        \cline{2-9}
        & 64 &128& 54 & 322 & 376 & \textbf{0.007} & \textbf{127} & 0.011\\
        \cline{2-9}
        & 128 &256& 76 & 456 & 532 & \textbf{0.010} & \textbf{255} & 0.040\\
        \cline{2-9}
        & 256 &512& 111 & 668 & \textbf{779} & \textbf{0.014} & 799 & 0.757\\
        \cline{2-9}
        & 512 &1024& 173 & 1039 & 1212 & \textbf{0.025} & \textbf{1027} & 39\\
        \cline{2-9}
        & 1024 &2048& 283 & 1701 & \textbf{1984} & \textbf{0.038} & \multicolumn{2}{c|}{\multirow{10}*{Timeout (15 min)}}\\
        \cline{2-7}
        & 2048 &4096& 476 & 2854 & \textbf{3330} & \textbf{0.067} & \multicolumn{2}{c|}{}\\
        \cline{2-7}
        & 4096 &8192& 794 & 4762 & \textbf{5556} & \textbf{0.120} & \multicolumn{2}{c|}{}\\
        \cline{2-7}
        & 8192 &16384& 1337 & 8024 & \textbf{9361} & \textbf{0.335} & \multicolumn{2}{c|}{}\\
        \cline{2-7}
        & 16384 &32768& 2363 & 14177 & \textbf{16540} & \textbf{0.673 }& \multicolumn{2}{c|}{}\\
        \cline{2-7}
        & 32768 &65536& 4391 & 26346 & \textbf{30737} & \textbf{1.42} & \multicolumn{2}{c|}{}\\
        \cline{2-7}
        & 65536 &131072& 8395 & 50372 & \textbf{58767} & \textbf{3.23} & \multicolumn{2}{c|}{}\\
        \cline{2-7}
        & 131072 &262144& 16220 & 97318 & \textbf{113538} & \textbf{8.46} & \multicolumn{2}{c|}{}\\
        \cline{2-7}
        & 262144 &524288& 31209 & 187251 & \textbf{218460} & \textbf{24.44} & \multicolumn{2}{c|}{}\\
        \cline{2-7}
        & 524288 &1048576& 58901 & 353404 & \textbf{412305} & \textbf{75.80} & \multicolumn{2}{c|}{}\\
        \cline{2-7}
        & 1048576 &2097152& \multicolumn{4}{c|}{Timeout (15 min)} & \multicolumn{2}{c|}{}\\
        \hline
        \multirow{20}*{DJ} & 16 &32& 18 & 90 & 108 & 0.006 & \textbf{18} & \textbf{0.001}\\
        \cline{2-9}
        & 32 &64& 21 & 107 & 128 & 0.008 & \textbf{34} & \textbf{0.002}\\
        \cline{2-9}
        & 64 &128& 24 & 123 & 147 & \textbf{0.008} & \textbf{66} & 0.038\\
        \cline{2-9}
        & 128 &256& 27 & 139 & 166 & \textbf{0.009} & \textbf{130} & 0.272\\
        \cline{2-9}
        & 256 &512& 30 & 154 & \textbf{184} & \textbf{0.01} & 258 & 2.1\\
        \cline{2-9}
        & 512 &1024& 33 & 170 & \textbf{203} &\textbf{ 0.011} & 516 & 795.5\\
        \cline{2-9}
        & 1024 &2048& 36 & 186 & \textbf{222} & \textbf{0.014} & \multicolumn{2}{c|}{\multirow{14}*{Timeout (15 min)}}\\
        \cline{2-7}
        & 2048 &4096& 39 & 202 & \textbf{241} & \textbf{0.019} & \multicolumn{2}{c|}{}\\
        \cline{2-7}
        & 4096 &8192& 42 & 218 & \textbf{260} & \textbf{0.028} & \multicolumn{2}{c|}{}\\
        \cline{2-7}
        & 8192 &16384& 45 & 234 & \textbf{279} & \textbf{0.048} & \multicolumn{2}{c|}{}\\
        \cline{2-7}
        & 16384 &32768& 48 & 250 & \textbf{298} & \textbf{0.09} & \multicolumn{2}{c|}{}\\
        \cline{2-7}
        & 32768 &65536& 51 & 266 & \textbf{317} & \textbf{0.182} & \multicolumn{2}{c|}{}\\
        \cline{2-7}
        & 65536 &131072& 54 & 282 & \textbf{336} & \textbf{0.418} & \multicolumn{2}{c|}{}\\
        \cline{2-7}
        & 131072 &262144& 57 & 298 & \textbf{355} & \textbf{0.956} & \multicolumn{2}{c|}{}\\
        \cline{2-7}
        & 262144 &524288& 60 & 314 & \textbf{374} & \textbf{2.57} & \multicolumn{2}{c|}{}\\
        \cline{2-7}
        & 524288 &1048576& 63 & 330 & \textbf{393} & \textbf{7.8} & \multicolumn{2}{c|}{}\\
        \cline{2-7}
        & 1048576 &2097152& 66 & 346 & \textbf{412} & \textbf{26.15} & \multicolumn{2}{c|}{}\\
        \cline{2-7}
        & 2097152 &4194304& 69 & 362 & \textbf{431} & \textbf{95.57} & \multicolumn{2}{c|}{}\\
        \cline{2-7}
        & 4194304 &8388608& 72 & 378 & \textbf{450} & \textbf{180.33} & \multicolumn{2}{c|}{}\\
        \cline{2-7}
        & 8388608 &16777216& \multicolumn{4}{c|}{Timeout (15 min)} & \multicolumn{2}{c|}{}\\
        \hline
    \end{tabular}
    }
    \caption{The performance of CFLOBDDs against BDDs for increasing numbers of qubits.
    }
    \label{Ta:quantum-table-detailed}
\end{table}

\begin{table}[!tb]
    \centering
    \begin{adjustbox}{width=1\textwidth}
    \begin{tabular}{|c|c|c|c|c|c|c|c|c|}
    \hline
        \multirow{2}*{Benchmark} & \multirow{2}*{\#Qubits} & 
        \multirow{2}{*}{
        \begin{tabular}{@{}c@{}}
            \#Boolean\\
            Variables
        \end{tabular}
        } 
        & \multicolumn{4}{c|}{CFLOBDD} & \multicolumn{2}{c|}{BDD}\\
        \cline{4-9}
        & & & \#Vertices & \#Edges & Total & Time (sec) & \#Nodes & Time (sec)\\
        \hline
        \multirow{3}*{Simon's Alg.} & 16 &64& 583 & 16335 & 16918 & 0.71 & \textbf{5512} & \textbf{0.275}\\
        \cline{2-9}
        & 32 &128& 123611 & 14096110 & 14219721 & 443.09 & \textbf{80243} & \textbf{3.31}\\
        \cline{2-9}
        & 64 &256& \multicolumn{4}{c|}{Timeout (90 min)} & \multicolumn{2}{c|}{Timeout (90 min)}\\
        \hline
        \multirow{3}*{QFT} & 4 &8& 7 & 73 & 80 & 0.001 & \textbf{31} & \textbf{0.0001}\\
        \cline{2-9}
        & 8 &16& 9  & 572 & 581 & 0.034 & \textbf{255} & \textbf{0.001}\\
        \cline{2-9}
        & 16 &32& 15 & 17868 & \textbf{17883} & 0.128 & 65535 & \textbf{0.098} \\
        \cline{2-9}
        & 32 &64& \multicolumn{4}{c|}{Timeout (15 min)} & \multicolumn{2}{c|}{Timeout (15 min)} \\
        \hline
        Shor's Alg. $(N,a) = (15,2)$ & 4 & 16 & 38 & 338 & 376 & 0.09 & \textbf{69} & \textbf{0.04}\\
        \hline
        Shor's Alg. $(N,a) = (21,2)$ & 5  &16& 72 & 877 & 949 & 2.13 & \textbf{136} & \textbf{0.72}\\
        \hline
        Shor's Alg. $(N,a) = (39,2)$ & 6  &16& 111 & 2443 & 2554 & \textbf{12.6} & \textbf{187} & 12.96\\
        \hline
        Shor's Alg. $(N,a) = (69,4)$ & 7  &16& 176 & 4331 & 4487 & 53.47 & \textbf{605} & \textbf{30.38}\\
        \hline
        Shor's Alg. $(N,a) = (95,8)$ & 7  &16& 216 & 4928 & 5144 & 53.47 & \textbf{974} & \textbf{41.47}\\
        \hline
        Shor's Alg. $(N,a) = (119,2)$ & 7  &16& 220 & 7533 & 7753 & 53.47 & \textbf{3606} & \textbf{44.95}\\
        \hline
        Shor's Alg. $(N,a) = (323,2)$ & 9  &32& \multicolumn{4}{c|}{Timeout (15min)} & \multicolumn{2}{c|}{Timeout (15min)}\\
        \hline
        \multirow{10}*{Grover's Alg.} & 16 &32& 17 & 91 & 108 & \textbf{0.009} & \textbf{47} & 0.214\\
        \cline{2-9}
        & 32 &64& 25 & 138 & 163 & \textbf{0.012} & \textbf{66} & 4.84\\
        \cline{2-9}
        & 64 &128& 38 & 212 & \textbf{250} & \textbf{0.018} & \multicolumn{2}{c|}{\multirow{8}*{Timeout (15 min)}} \\
        \cline{2-7}
        & 128 &256& 58 & 333 & \textbf{391} & \textbf{0.030} & \multicolumn{2}{c|}{}\\
        \cline{2-7}
        & 256 &512& 91 & 531 & \textbf{622} & \textbf{0.080} & \multicolumn{2}{c|}{}\\
        \cline{2-7}
        & 512 &1024& 151 & 886 & \textbf{1037} & \textbf{0.292} & \multicolumn{2}{c|}{}\\
        \cline{2-7}
        & 1024 &2048& 259 & 1535 & \textbf{1794} & \textbf{14.11} & \multicolumn{2}{c|}{}\\
        \cline{2-7}
        & 2048 &4096& 450 & 2674 & \textbf{3124} & \textbf{64.85} & \multicolumn{2}{c|}{}\\
        \cline{2-7}
        & 4096 & 8192 & 766 & 4569 & \textbf{5335} & \textbf{909.86} & \multicolumn{2}{c|}{}\\
        \cline{2-7}
        & 8192 & 16384 & \multicolumn{4}{c|}{Timeout (15 min)} & \multicolumn{2}{c|}{}\\
        \hline
    \end{tabular}
    \end{adjustbox}
    \caption{Table (cont.) of the performance of CFLOBDDs against BDDs for increasing numbers of qubits.
    }
    \label{Ta:quantum-table-detailed-contd}
\end{table}

\tablerefs{quantum-table-detailed}{quantum-table-detailed-contd} show the performance of CFLOBDDs and BDDs when simulating several well-known quantum algorithms, which are discussed in \sectref{quantum-algos}.
In each case, for both CFLOBDDs and BDDs, we used the interleaved-variable ordering.

For GHZ, the algorithms do not depend on an input;
the output is solely a function of the number of qubits used.
For BV, DJ, QFT, Simon's algorithm, Shor's algorithm, and Grover's algorithm, we ran each algorithm with 50 different randomly selected inputs, for each of the indicated number of qubits.
\tablerefs{quantum-table-detailed}{quantum-table-detailed-contd} report the average vertex and average edge counts (for CFLOBDDs), average node count (for BDDs), and average time taken.
In the case of Simon's algorithm, CFLOBDDs timed-out on 9 of the 50 test cases, whereas BDDs timed-out on 28 of the 50 test cases;
we report the average counts and average times for the test cases that did not time out.
BV, DJ, Simon's algorithm, Shor's algorithm, and Grover's algorithm make use of oracles created during a pre-processing step (see also \sectref{ExperimentalSetup});
we do not include the time for oracle construction in the execution time, but we do include it as part of the 15-minute/90-minute timeout threshold.
For the case of QFT, the input is one of the basis vectors selected randomly.
For 16 qubits and a timeout threshold of 15 minutes, QFT ran to completion in 11 of the 50 runs.
The numbers reported in \tableref{quantum-table-detailed-contd} are the averages for the 11 successful runs.
In the entries for Shor's algorithm, $N$ is the number being factored, and $a$ is the value used in the associated ``order-finding problem.''\footnote{
  Given $a$, such that $1 < a < N$, the order-finding problem is to find the smallest positive integer $r$ such that $a^r \equiv 1 (mod N)$.
    }


In several cases, the problem sizes that completed successfully using CFLOBDDs were dramatically larger than the sizes that completed successfully using BDDs.
With a 15-minute timeout, the number of qubits that CFLOBDDs can handle are 65,536 for GHZ, 524,288 for BV; 4,194,304 for DJ; and 4,096 for Grover’s Algorithm, besting BDDs by factors of
128$\times$, 1,024$\times$, 8,192$\times$, and 128$\times$, respectively.

\begin{figure}[bt!]
    \centering
    \begin{subfigure}[t]{0.495\linewidth}
    \includegraphics[width=\linewidth]{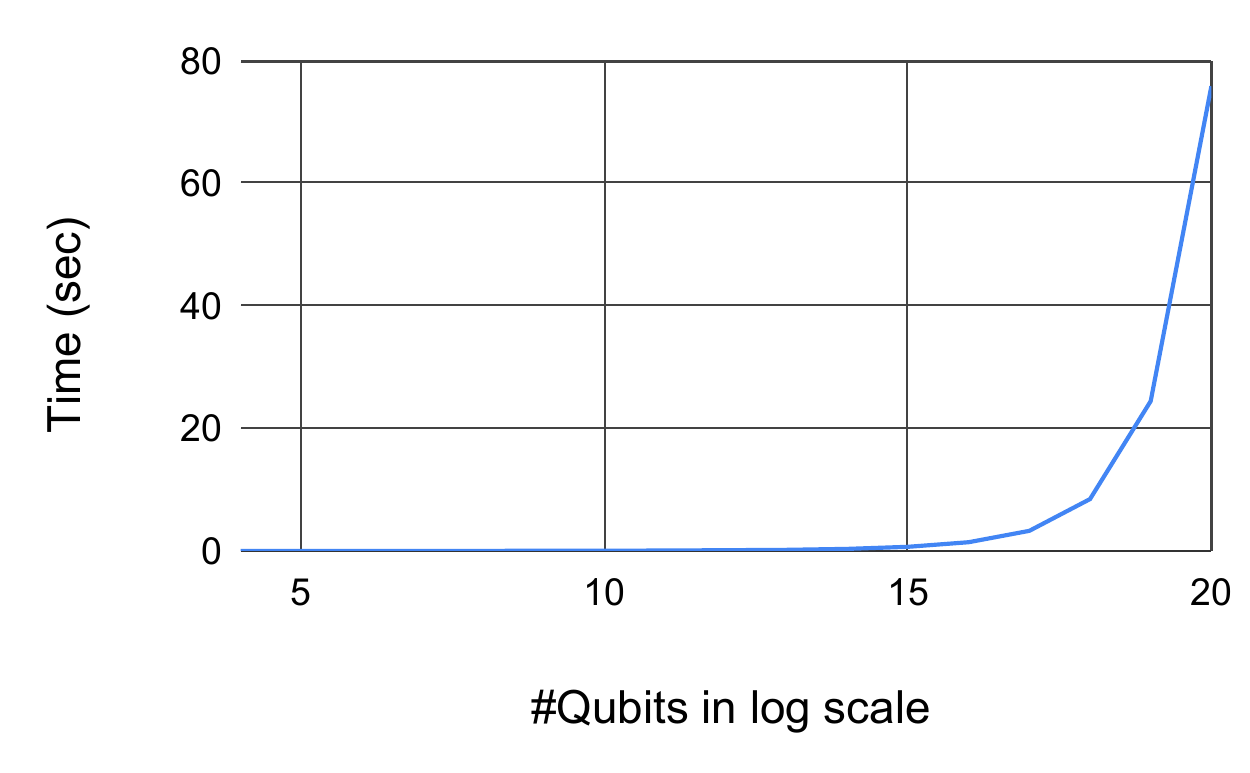}
    \caption{BV Algorithm}
    \end{subfigure}
    \begin{subfigure}[t]{0.495\linewidth}
    \includegraphics[width=\linewidth]{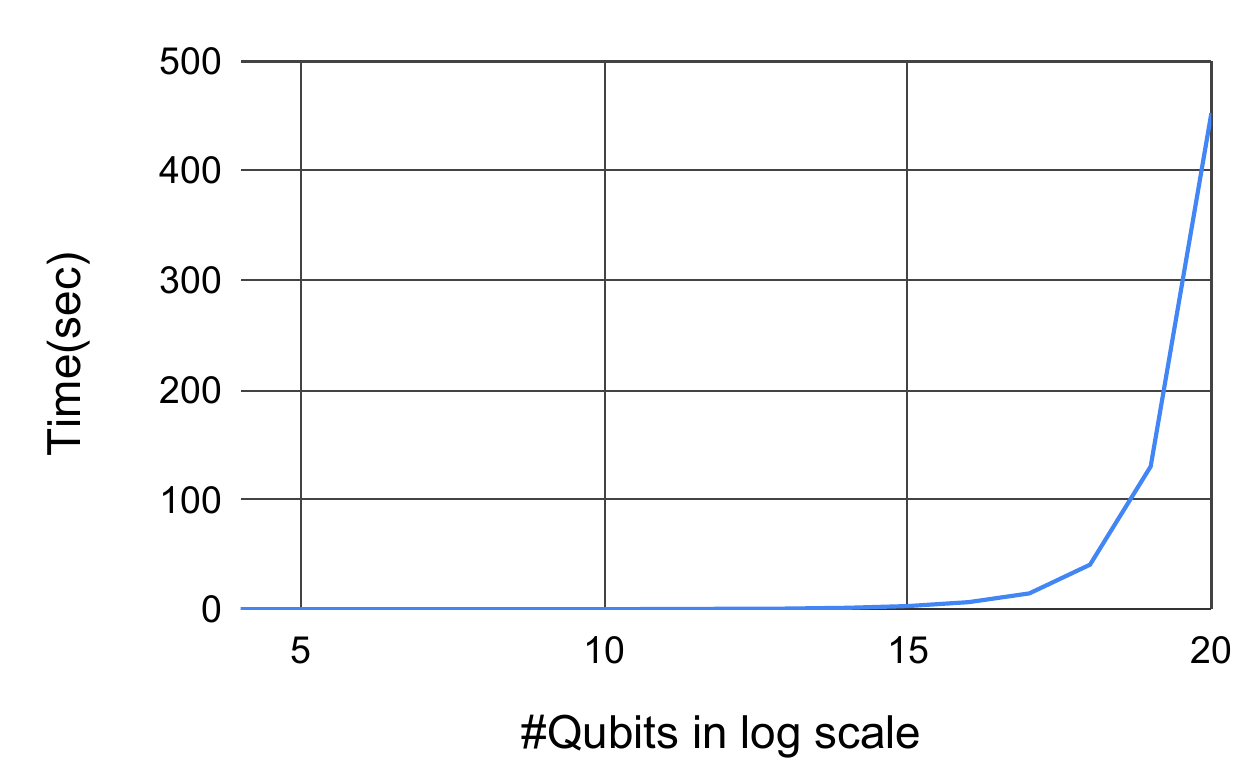}
    \caption{GHZ Algorithm}
    \end{subfigure}
    \begin{subfigure}[t]{0.495\linewidth}
    \includegraphics[width=\linewidth]{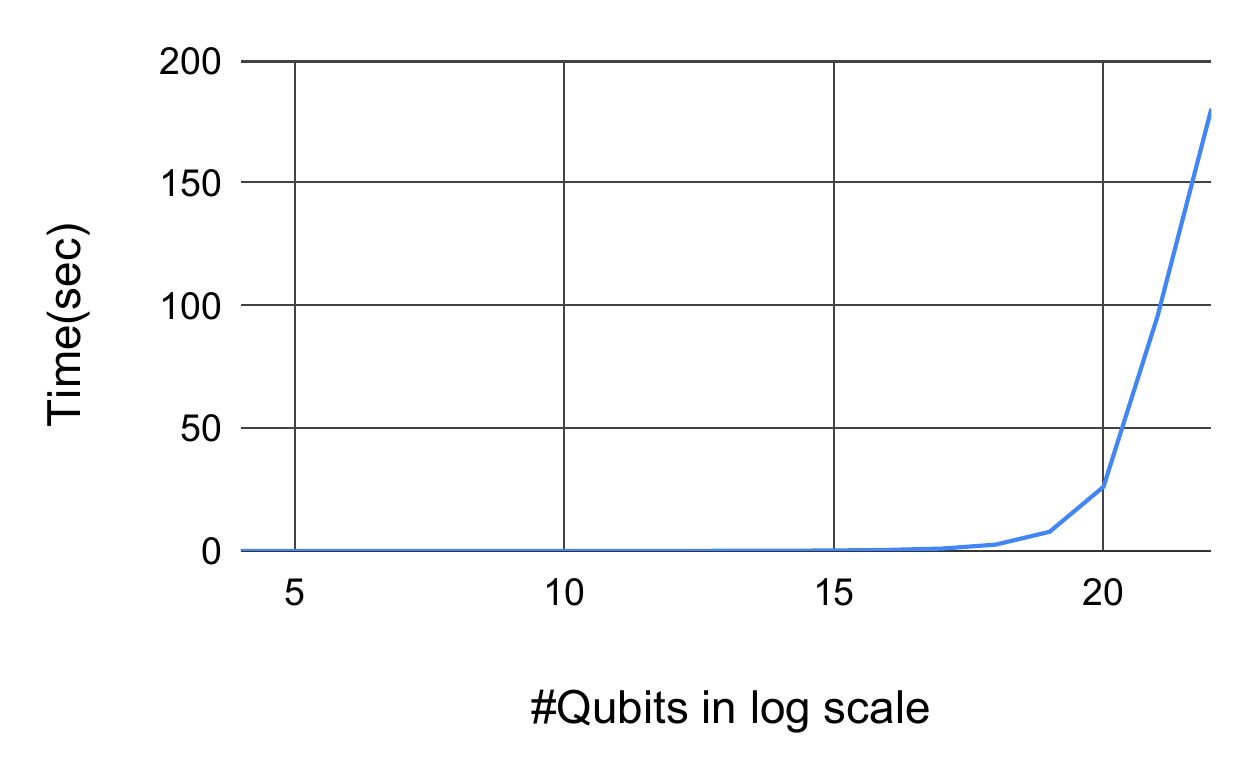}
    \caption{DJ Algorithm}
    \end{subfigure}
    \caption{\protect \raggedright
    Execution time (in seconds) vs.\ number of qubits (on a log scale) for three of the benchmarks.
    }
    \label{Fi:scale-cflobdds}
\end{figure}

We also ran the CFLOBDD simulations with a 90-minute timeout, both to understand how execution time scales, as a function of number of qubits, and to see how large a problem instance can be handled.
\figref{scale-cflobdds} shows the time taken (in seconds), with increasing numbers of qubits, for BV, GHZ, and DJ.
With a 90-minute timeout, the BV and GHZ algorithms ran to completion with $2^{20} = \textrm{1,048,576}$ qubits, and the DJ algorithm ran to completion with $2^{21} = \textrm{2,097,152}$ qubits.

For both CFLOBDDs and BDDs, the transition from a problem size that completes successfully to a problem size that fails is rather abrupt.
For all of the problems, the time reported for the CFLOBDD run with the largest number of qubits that completes successfully is well under 15 minutes.
Unfortunately, for the next larger run, oracle construction timed out after 15 minutes for the BV and DJ algorithms, and as a result we terminated the entire algorithm.
For Grover's algorithm, the number of bits for the floating-point representation is 100 for all runs, except for those with $\textrm{2,048}$, $\textrm{4,096}$, and $\textrm{8,192}$ qubits, for which we used 500, 750, and $\textrm{1,000}$ bits, respectively.
The increased cost of floating-point operations slows down matrix multiplications in Grover's algorithm, causing the $\textrm{8,192}$-qubit run to exceed 15 minutes.

\begin{figure}[bt!]
  \centering
  \begin{tabular}{@{\hspace{0ex}}l@{\hspace{0ex}}}
    \begin{subfigure}{1.0\linewidth}
      \includegraphics[width=0.495\linewidth]{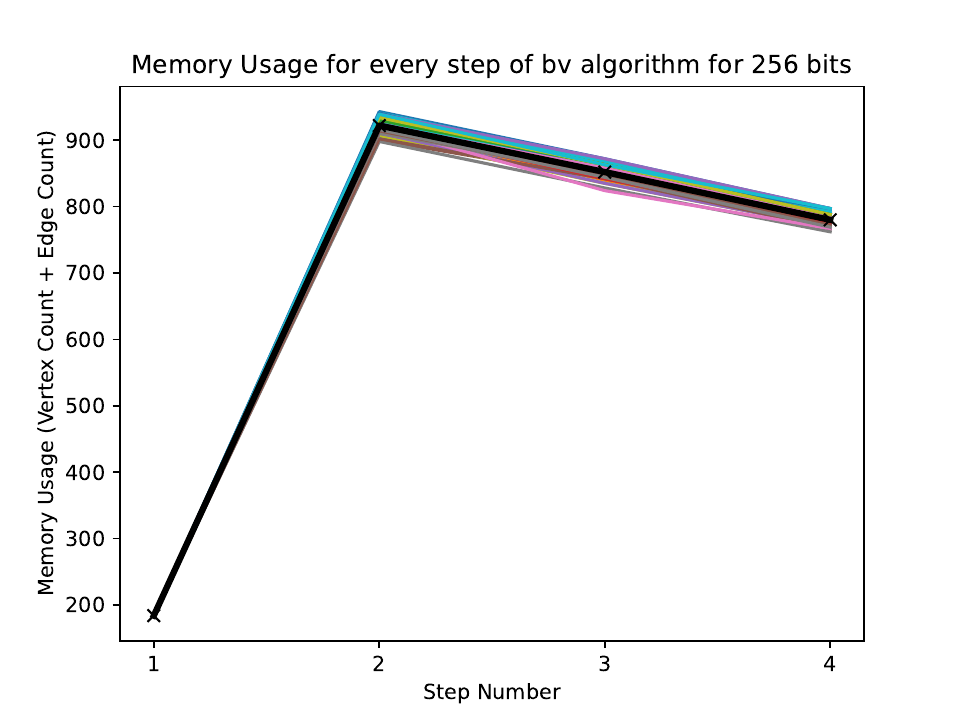}
      \includegraphics[width=0.495\linewidth]{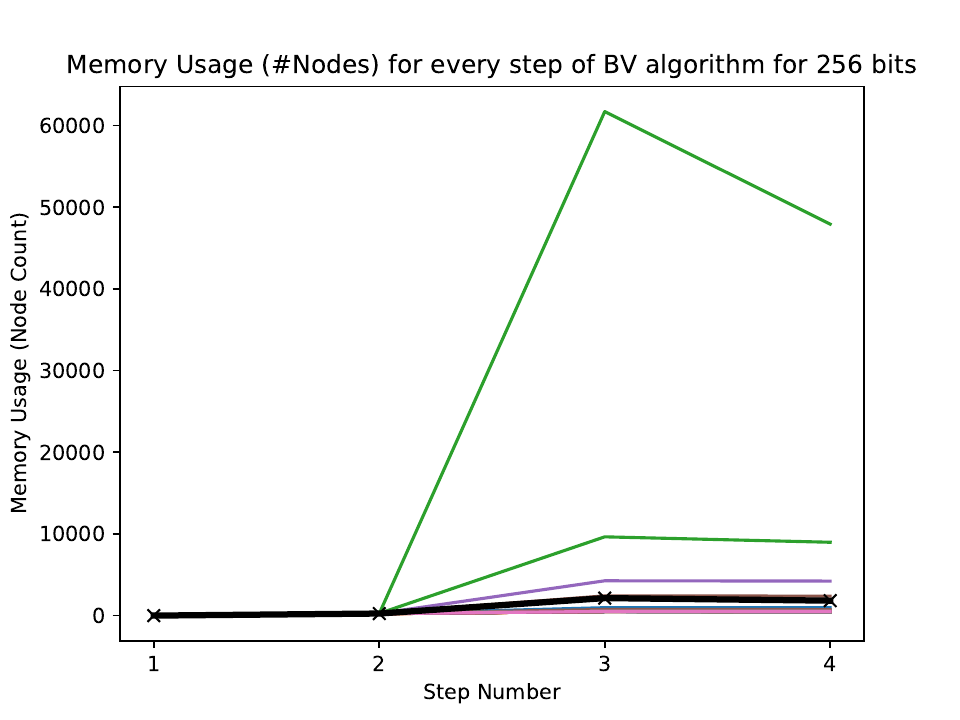}
      \caption{BV algorithm 256 qubits}
    \end{subfigure}
    \\
    \begin{subfigure}{1.0\linewidth}
      \includegraphics[width=0.495\linewidth]{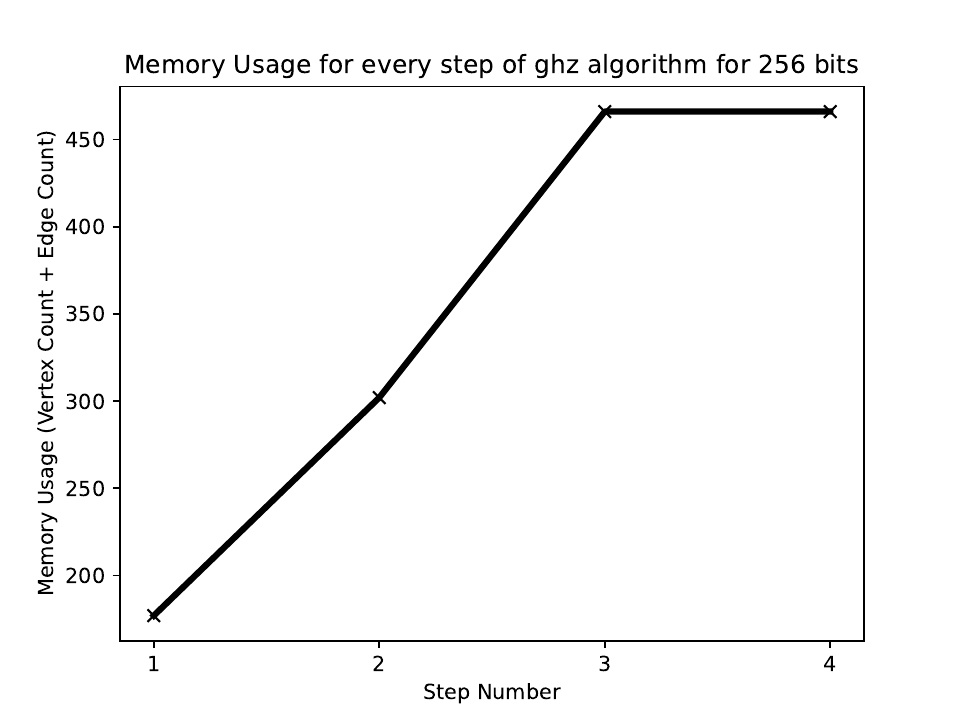}
      \includegraphics[width=0.495\linewidth]{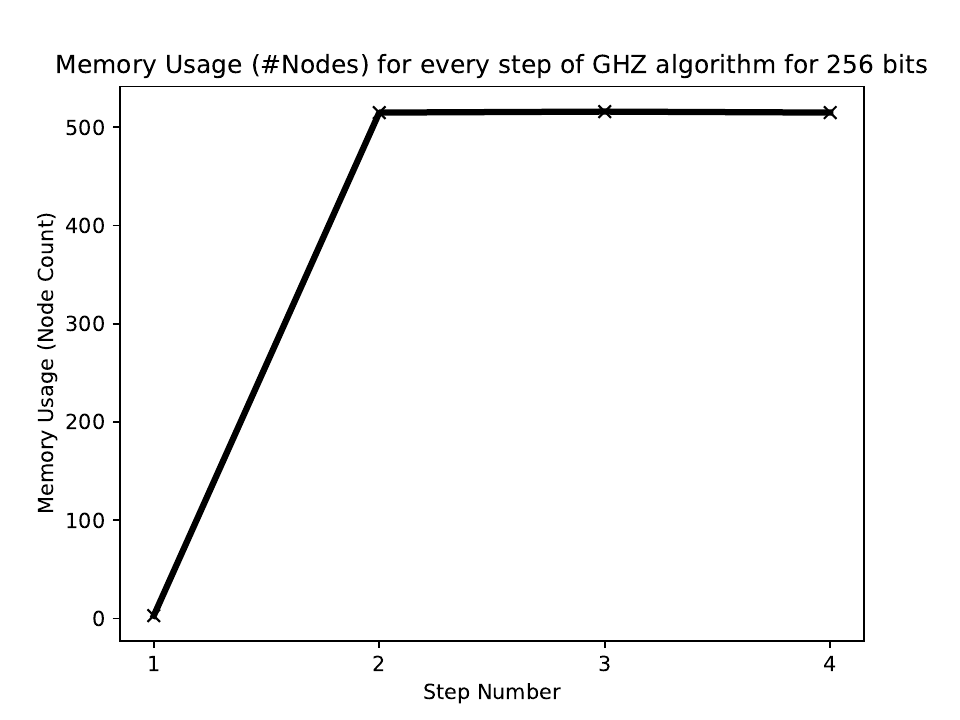}
      \caption{GHZ algorithm 256 qubits}
    \end{subfigure}
    \\
    \begin{subfigure}{1.0\linewidth}
      \includegraphics[width=0.495\linewidth]{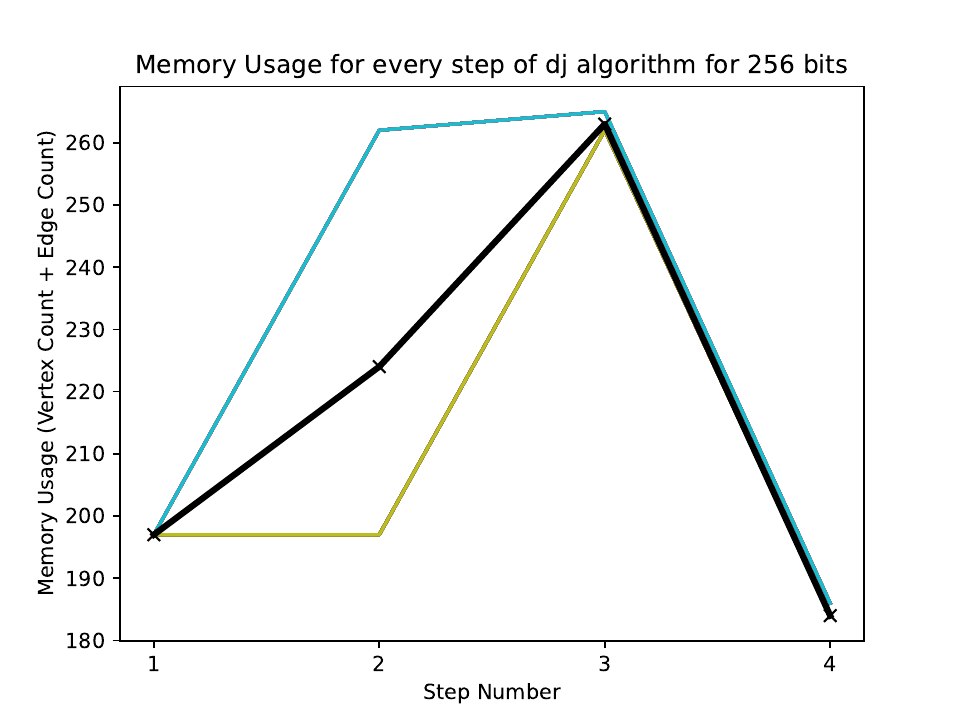}
      \includegraphics[width=0.495\linewidth]{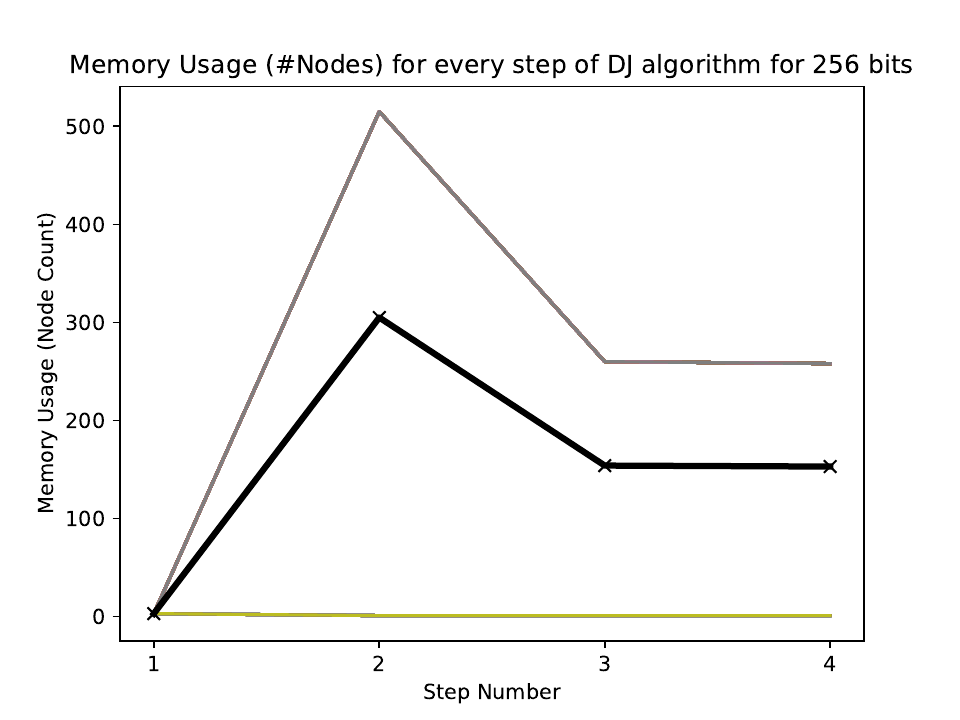}
      \caption{DJ algorithm 256 qubits}
    \end{subfigure}
    \\
  \end{tabular}
  \caption{\protect \raggedright 
    Evolution of size through the steps of the indicated algorithms. (Left: CFLOBDD-based simulation; right: BDD-based simulation.)
  }
  \label{Fi:inter-swell}
\end{figure}

\begin{figure}[bt!]
  \centering
  \begin{tabular}{@{\hspace{0ex}}l@{\hspace{0ex}}}
    \begin{subfigure}{1.0\linewidth}
      \includegraphics[width=0.495\linewidth]{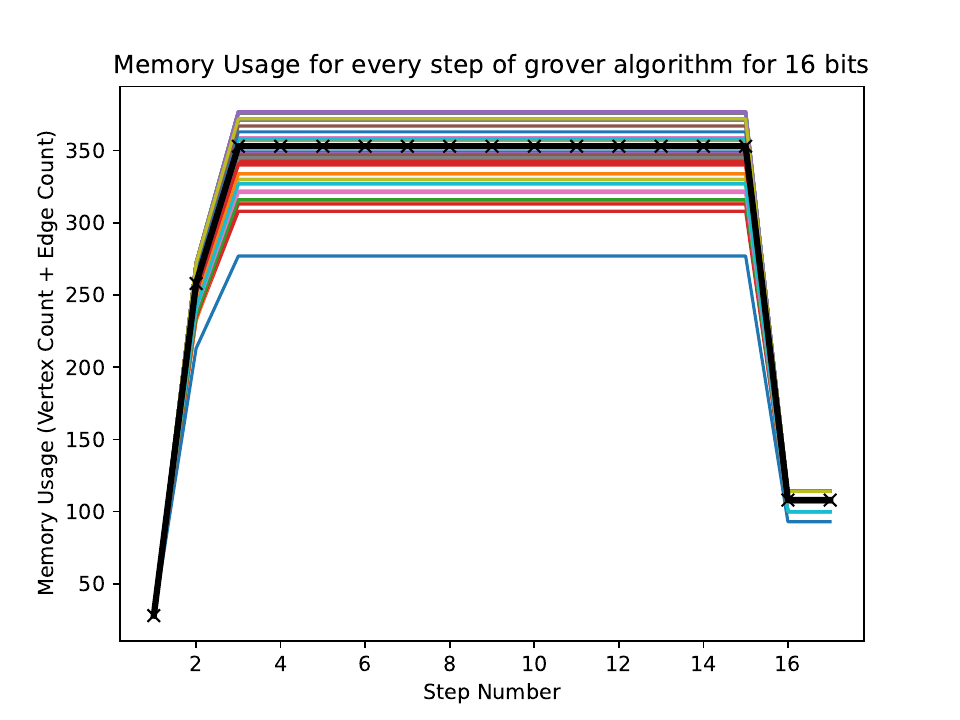}
      \includegraphics[width=0.495\linewidth]{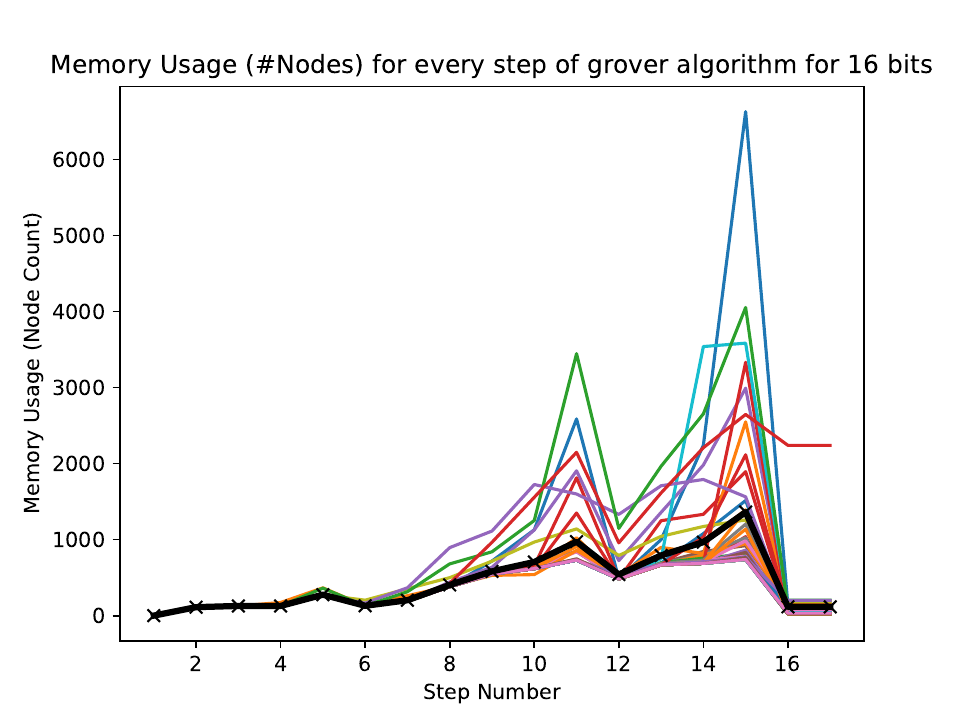}
      \caption{Grover's algorithm 16 qubits}
    \end{subfigure}
    \\
    \begin{subfigure}{1.0\linewidth}
      \centering
      \includegraphics[width=0.495\linewidth]{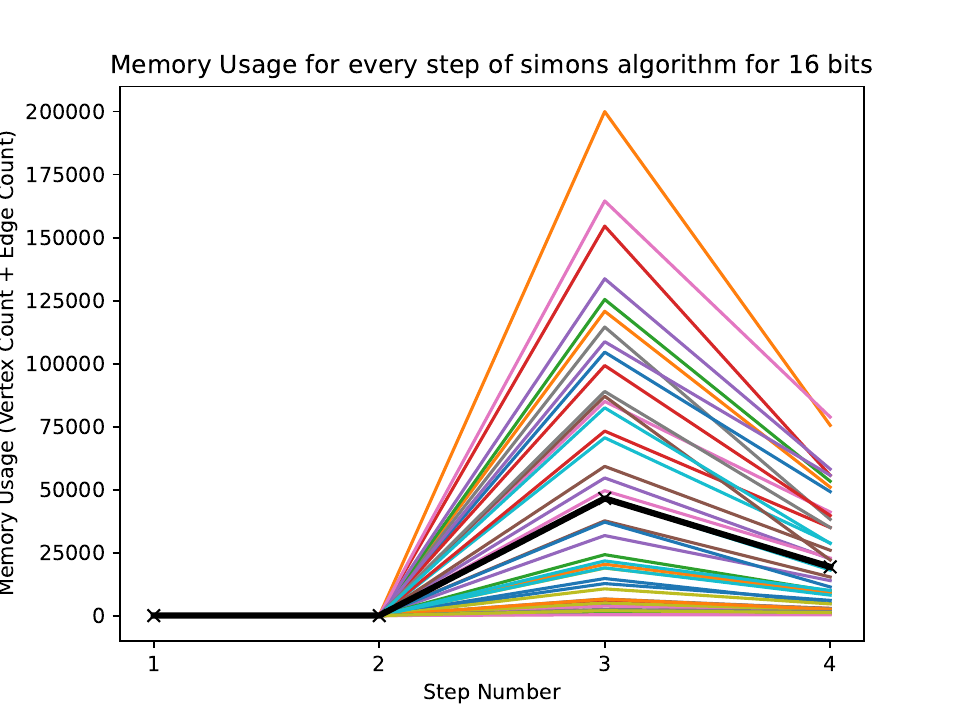}
      \includegraphics[width=0.495\linewidth]{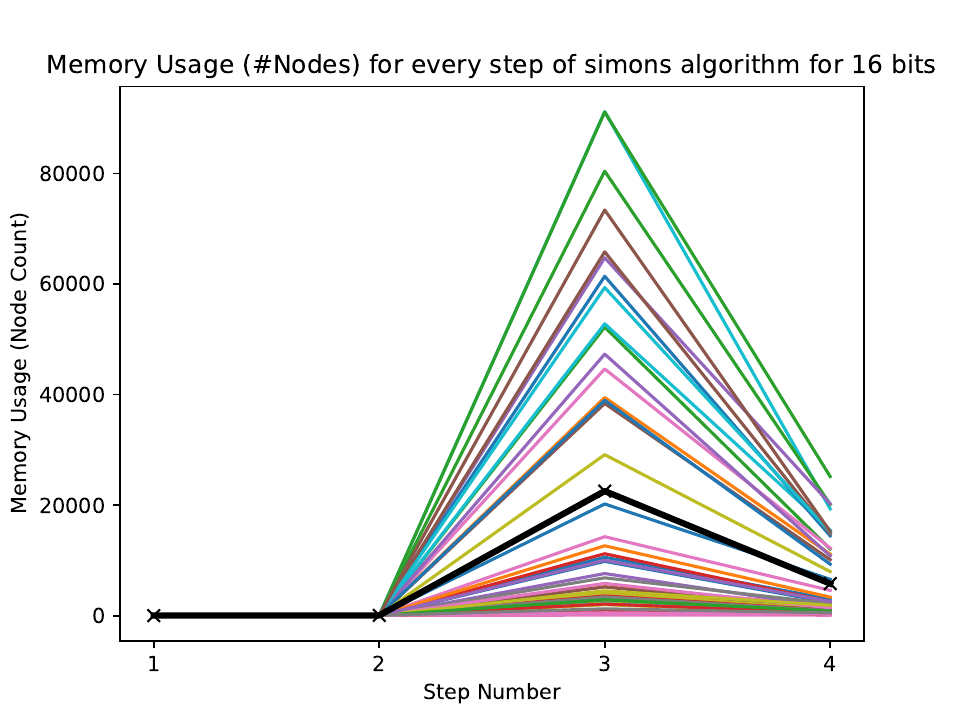}
      \caption{Simon's algorithm 16 qubits}
    \end{subfigure}
  \end{tabular}
  \caption{\protect \raggedright 
    Evolution of size through the steps of the indicated algorithms. (Left: CFLOBDD-based simulation; right: BDD-based simulation.)
  }
  \label{Fi:inter-swell-continued}
\end{figure}

\smallskip
\begin{mdframed}[innerleftmargin = 3pt, innerrightmargin = 3pt, skipbelow=-0.0em]
\textbf{Findings.}
For smaller numbers of qubits, the more-complex nature of the data structures used in CFLOBDDs resulted in slower execution times than with BDDs.
However, CFLOBDDs scaled much better than BDDs as the number of qubits increased, both in terms of memory (i.e., vertices + edges for CFLOBDDs, nodes for BDDs) and execution time.
In some cases, the problem sizes that completed successfully using CFLOBDDs were dramatically larger than the sizes that completed successfully using BDDs.
In particular, the number of qubits that could be handled using CFLOBDDs was larger---compared to BDDs---by a factor of
128$\times$ for GHZ; 1,024$\times$ for BV; 8,192$\times$ for DJ; and 128$\times$ for Grover's algorithm.
\end{mdframed}

\paragraph{Intermediate swell.}
In many of the algorithms, the initial and final CFLOBDD and BDD structures are of reasonable size, but there is an intermediate swell in size as the algorithm runs.
\figrefs{inter-swell}{inter-swell-continued} show how size evolves in the various steps of five of the algorithms during the CFLOBDD-based and BDD-based simulations.
The figures show how size evolves for all 50 runs, along with the average value at every step (highlighted in black).
\figref{inter-swell-continued}a shows that the CFLOBDD simulation of Grover's algorithm uses constant space from steps 3 to 15.
The explanation is that, although the state vector changes at each step, the size of the CFLOBDD representation of the state vector does not change.

\paragraph{Comparison with Tensor Networks.}
We also compared the performance of CFLOBDDs with
Quimb \cite{gray2018quimb}, a state-of-the-art quantum simulator.
\tableref{tensor-network-comp} shows the performance of our CFLOBDD implementation and Quimb on the previously discussed quantum benchmarks.
For the Quimb-based simulations of GHZ, BV, DJ, and Grover's algorithm, we used Matrix Product States (MPSs) \cite{vidal2003efficient,banuls2009matrix} and Matrix Product Operators (MPOs) \cite{verstraete2004matrix} in algorithms modeled after the ones described in \cite{woolfe2015matrix}.
For Simon's algorithm, we noticed that directly creating a circuit and performing contraction using Quimb led to better scalability than using MPS/MPOs.
For QFT, we tried both the standard circuit~\cite{nielsen2001quantum}
and the nearest-neighbor circuit mentioned in~\cite{fowler2004implementation}.
We found that the Quimb-based simulation results for both circuits are very similar, and only the former are reported here.
For Shor's algorithm, we use the $2n+3$ circuit from~\cite{beauregard2002circuit}, but the internal gates are created directly, as mentioned in \cite{woolfe2015matrix} (and hence the circuit only has $2n+1$ qubits).
For Grover's algorithm, we found that the maximum number of qubits that can be simulated using Quimb with a 15-minute timeout is 29 qubits.\footnote{
  With 32 qubits, Quimb takes 1496.6sec $\approx$ 25min.
}

These experiments show that, on some of the benchmarks, CFLOBDDs scale to much larger problem sizes than the Quimb tensor-network package, but on other benchmarks Quimb performs much better than CFLOBDDs.

\begin{table}[tb!]
    \centering
    \begin{adjustbox}{width=0.69\textwidth}
    \begin{tabular}{|c|c|c|c|}
    \hline
    Benchmark & \#Qubits & CFLOBDD (Time in sec) & Quimb (Time in sec)\\
    \hline
    \multirow{10}*{GHZ} & 16 & \textbf{0.005} & 0.222\\
    \cline{2-4}
    & 32 & \textbf{0.007} & 0.644\\
    \cline{2-4}
    & 64 & \textbf{0.010} & 2.29\\
    \cline{2-4}
    & 128 & \textbf{0.015} & 9.23\\
    \cline{2-4}
    & 256 & \textbf{0.027} & 40.31\\
    \cline{2-4}
    & 512 & \textbf{0.051} & 191.77\\
    \cline{2-4}
    & 1024 & \textbf{0.107} & \multirow{4}{*}{Timeout (15 min)}\\
    \cline{2-3}
    & \vdots & \vdots & \\ 
    \cline{2-3}
    & 65536 & \textbf{6.49} &\\
    \cline{2-3}
    & 131072 & Timeout (15 min) &\\
    \hline
    \multirow{10}*{BV} & 16 & \textbf{0.005} & 0.264\\
    \cline{2-4}
    & 32 & \textbf{0.006} & 0.773 \\
    \cline{2-4}
    & 64 & \textbf{0.007} & 2.75 \\
    \cline{2-4}
    & 128 & \textbf{0.010} & 11.08 \\
    \cline{2-4}
    & 256 & \textbf{0.014} & 49.49\\
    \cline{2-4}
    & 512 & \textbf{0.025} & 243.69\\
    \cline{2-4}
    & 1024 & \textbf{0.038} & \multirow{4}{*}{Timeout (15 min)}\\
    \cline{2-3}
    & \vdots & \vdots & \\
    \cline{2-3}
    & 524288 & \textbf{75.80} & \\
    \cline{2-3}
    & 1048576 & Timeout (15 min) & \\
    \hline
    \multirow{10}*{DJ} & 16 & \textbf{0.006} & 0.256\\
    \cline{2-4}
    & 32 & \textbf{0.008} & 0.761 \\
    \cline{2-4}
    & 64 & \textbf{0.008} & 2.75 \\
    \cline{2-4}
    & 128 & \textbf{0.009} & 11.18 \\
    \cline{2-4}
    & 256 & \textbf{0.010} & 49.33 \\
    \cline{2-4}
    & 512 & \textbf{0.011} & 243.01 \\
    \cline{2-4}
    & 1024 & \textbf{0.014} & \multirow{4}{*}{Timeout (15 min)} \\
    \cline{2-3}
    & \vdots & \vdots &\\
    \cline{2-3}
    & 4194304 & \textbf{180.33} & \\
    \cline{2-3}
    & 8388608 & Timeout (15 min) & \\
    \hline
    \multirow{4}*{Simon's Alg.} & 16 & \textbf{0.71} & 2.56 \\
    \cline{2-4}
    & 32 & 443.09 & \textbf{17.34} \\
    \cline{2-4}
    & 64 & \multirow{2}{*}{Timeout (15 min)} & \textbf{267} \\
    \cline{2-2}\cline{4-4}
    & 128 &  & Timeout (15min) \\
    \hline
    \multirow{8}*{QFT} & 4 & \textbf{0.001} & 0.023\\
    \cline{2-4}
    & 8 & \textbf{0.034} & 0.035 \\
    \cline{2-4}
    & 16 & 0.128 & \textbf{0.074} \\
    \cline{2-4}
    & 32 & \multirow{5}{*}{Timeout (15 min)} & \textbf{0.231} \\
    \cline{2-2}\cline{4-4}
    & 64 &  & \textbf{1.64} \\
    \cline{2-2}\cline{4-4}
    & 128 &  & \textbf{10.32} \\
    \cline{2-2}\cline{4-4}
    & 256 &  & \textbf{103.65} \\
    \cline{2-2}\cline{4-4}
    & 512 &  & Timeout (15min) \\
    \hline
    Shor's Alg. $(N,a) = (15,2)$ & 4 & 0.09 & \textbf{0.08} \\
    \hline
    Shor's Alg. $(N,a) = (21,2)$ & 5 & 2.13 & \textbf{0.1} \\
    \hline
    Shor's Alg. $(N,a) = (39,2)$ & 6 & 12.6 & \textbf{0.11} \\
    \hline
    Shor's Alg. $(N,a) = (69,4)$ & 7 & 53.47 & \textbf{0.12} \\
    \hline
    Shor's Alg. $(N,a) = (95,8)$ & 7 & 42.8 & \textbf{0.12} \\
    \hline
    Shor's Alg. $(N,a) = (119,2)$ & 7 & 64.8 & \textbf{0.12} \\
    \hline
    Shor's Alg. $(N,a) = (323,2)$ & 9 & \multirow{4}{*}{Timeout (15 min)} & \textbf{0.27} \\
    \cline{1-2}\cline{4-4}
    \vdots & \vdots &  & \vdots\\
    \cline{1-2}\cline{4-4}
    Shor's Alg. $(N,a) = (6085,8)$ & 12 &  & \textbf{107.28} \\
    \cline{1-2}\cline{4-4}
    Shor's Alg. $(N,a) = (11611,2)$ & 13 &  & Out of Memory \\
    \hline
    \multirow{5}{*}{Grover's Alg.} & 16 & \textbf{0.009} & 3.26 \\
    \cline{2-4}
    & 32 & \textbf{0.012} & \multirow{4}{*}{Timeout (15 min)}\\
    \cline{2-3}
    & \vdots & \vdots & \\
    \cline{2-3}
    & 4096 & \textbf{909.86} & \\
    \cline{2-3}
    & 8192 & Timeout (15 min) & \\
    \hline
    \end{tabular}
    \end{adjustbox}
    \caption{Performance of CFLOBDDs against Quimb on quantum benchmarks for different numbers of qubits.
    }
    \label{Ta:tensor-network-comp}
\end{table}

\paragraph{What allows CFLOBDDs to perform so well on Grover's algorithm?}
In each run of the CFLOBDD simulation of Grover's algorithm,
a random 4096-bit string $s$ is chosen, then the Grover oracle matrix is constructed, along with the Grover diffusion operation, which are then multiplied together.
A version of Grover's algorithm based on repeated squaring of the product matrix is carried out (via operations that use the cumulative-product matrix---which depends on $s$---but the operations are oblivious to the value of $s$ itself);
the algorithm’s answer $s'$ is retrieved;
and finally $s$ and $s'$ are compared to make sure that the computed result is correct.

The reason that this process is space-efficient is that the Grover oracle is basically a ``-1 hot encoding'' of $s$, and thus can be constructed by an algorithm that is a mixture of the principles used in the algorithms for constructing the representations of (i) projection functions (\sectref{ProjectionFunctions}), and (ii) the identity matrix (\sectref{IdentityMatrix}).
In the largest cases of Grover's algorithm that completed successfully within 15 minutes, the matrix has dimensions $2^{4096} \times 2^{4096}$;
all off-diagonal entries are 0;
and all diagonal entries are 1 except for the $(s,s)$ entry, which is -1.
To represent this matrix, one needs 8,192 = $2^{13}$ Boolean variables:
4,096 for the row-index and 4,096 for the column-index.
There is a CFLOBDD representation of this matrix whose highest-level grouping is at level 13---thus, there are 14 levels in total, counting level 0.
Moreover, the CFLOBDD has only a constant number of groupings at each of the 14 levels, so the matrix is one for which the CFLOBDD representation exhibits double-exponential compression.

Although multiplication of matrices represented by CFLOBDDs is not particularly efficient (see the last row of \tableref{algo-list-cflobdds}), there is little or no infill caused by the repeated-squaring operations, and so the matrix representation has only a limited amount of intermediate swell.
(See the left-hand graph in \figref{inter-swell-continued}(a) for a plot of memory usage for the CFLOBDD implementation of Grover's algorithm for 16 qubits.)

\section{Related Work}
\label{Se:RelatedWork}

\paragraph{Some History}

\Supplemental{
CFLOBDDs were devised in the late 1990s;
however, except for a rejected US patent application posted on the USPTO site in 2002 \cite{PatentApplication:Reps02}, nothing about them has ever been published.
The idea and the preliminary implementation were set aside in about 2002 due to not having found an application on which CFLOBDDs performed better than BDDs---other than some of the recursively defined spectral transforms, such as the Reed-Muller, inverse Reed-Muller, Hadamard (\sectref{Separation:AdditionRelation}), and Boolean Haar wavelet transforms \cite{Book:HMM85}.

In late 2013, Reps read Lipton's blog post from 2009, ``BDD's and factoring'' \cite{Blog:PEqNP:06:16:2009}.
In it, Lipton sketched a proposal for what he called \emph{Pushdown BDDs}, a BDD-like structure based on nondeterministic pushdown automata. 
He speculated that even if the multiplication relation were of inherently exponential complexity for Pushdown BDDs, if the constant was small enough they could represent a threat to factoring-based cryptography.
The realization that CFLOBDDs were closely related to a deterministic version of Pushdown BDDs, caused Reps to re-examine what relations could be expressed efficiently by a CFLOBDD, and to discover that $\ADD_n$ (\sectref{Separation:AdditionRelation}) was such a relation.\footnote{
  As were a few others, such as Gray codes, with the interleaved-variable order.
}
Reps presented these results to Lipton in 2017, who suggested applying CFLOBDDs to quantum simulation.
That effort was taken up by Sistla in summer 2018.
This paper combines the heretofore unpublished material from the patent application with Sistla's results on CFLOBDD-based quantum simulation.
}{
CFLOBDDs were devised in the late 1990s;
however, except for a rejected US patent application posted on the USPTO site in 2002 [citation omitted to preserve anonymity], nothing about them has ever been published.
The idea and the preliminary implementation were set aside in about 2002 due to not having found an application on which CFLOBDDs performed better than BDDs---other than some of the recursively defined spectral transforms, such as the Reed-Muller, inverse Reed-Muller, Hadamard (\sectref{Separation:AdditionRelation}), and Boolean Haar wavelet transforms \cite{Book:HMM85}.

In 2009, Lipton published a blog post \cite{Blog:PEqNP:06:16:2009} in which he sketched a proposal for what he called \emph{Pushdown BDDs}, a BDD-like structure based on nondeterministic pushdown automata.
He speculated that even if the multiplication relation were of inherently exponential complexity for Pushdown BDDs, if the constant was small enough they could represent a threat to factoring-based cryptography.
The realization that CFLOBDDs were closely related to a deterministic version of Pushdown BDDs, caused us to re-examine what relations could be expressed efficiently by a CFLOBDD, and to discover that $\ADD_n$ (\sectref{Separation:AdditionRelation}) was such a relation.\footnote{
  As were a few others, such as Gray codes, with the interleaved-variable order.
}
In a discussion of these results with Lipton in 2017, he suggested applying CFLOBDDs to quantum simulation.
That effort was taken up in summer 2018.
This paper combines the heretofore unpublished material from the patent application with our recent results on CFLOBDD-based quantum simulation.
}

In the best case, a CFLOBDD for a function $f$ can be double-exponentially smaller than the decision tree for $f$.
There is a sense in which ROBDDs are incapable of such a degree of compression.
\emph{Quasi-reduced BDDs} are the version of BDDs in which variable ordering is respected, but don't-care nodes are \emph{not} removed (i.e., plies are not skipped), and thus all paths from the root to a leaf have length $n$, where $n$ is the number of variables.
The size of a quasi-reduced BDD is at most a factor of $n+1$ larger than the size of the corresponding ROBDD \cite[Thm.\ 3.2.3]{Book:Wegener00}.
Thus, although ROBDDs can give better-than-exponential compression compared to decision trees, what one has is not double-exponential compression:
at best, it is linear compression of exponential compression.
Moreover, in \sectref{efficient-relations}, we showed that the CFLOBDD for a function $g$ can be exponentially smaller than any ROBDD for $g$.

\paragraph{Pre-History: Interprocedural Path Profiling}

The idea behind CFLOBDDs was inspired by an obstacle encountered in interprocedural path profiling.
In path profiling, a program is instrumented with code that counts the number of times particular path fragments of the program are executed.
Ball and Larus \cite{micro:BL96} devised a clever technique that allowed \emph{intra}procedural path profiling to be carried out with low run-time overhead.
Melski and Reps \cite{cc:MR99} generalized their technique to support profiling of path fragments that cross procedure boundaries.
The Melski-Reps scheme performs an interprocedural analog of a transformation used in the Ball-Larus scheme:
\begin{itemize}
  \item 
    The Ball-Larus cycle-elimination transformation is performed on the CFG of each procedure.
  \item
    The interprocedural CFG is further transformed to ``short-circuit'' paths through recursive call sites, in effect, cutting each cycle that is present because of recursive procedure calls.
\end{itemize}
Interestingly, these two transformations do not leave one with an acyclic interprocedural CFG;
the graph still contains cyclic paths.
For example, via different call-sites, a path can exit a procedure and re-enter it from a distinct call-site.
However---and this is the crucial fact on which this approach to interprocedural path profiling rests---there are only a \emph{finite} number of ''realizable paths,'' in which each exit-to-return-site edge taken to a return site matches a preceding call-to-start edge from the corresponding call site.
(Only the realizable paths represent ``observable path fragments'' that are to be reported during an execution run.)

The drawback of the Melski-Reps scheme is that there can be an enormous number of realizable paths in the transformed interprocedural CFG.
In fact, the number of realizable paths can be \emph{doubly exponential} in the size of the interprocedural CFG.
For path profiling this situation is terrible:
it means that the edge numbers that the instrumentation code has to manipulate can be enormous.
For instance, Melski \cite{PersonalComm:Melski98} found that one 20,000-line program had $2^{\textrm{400,000}}$ such paths---{which would require the path-profiling instrumentation code for the program to manipulate 400,000-bit numbers!
}
Nevertheless, these drawbacks were the genesis of the ideas used in CFLOBDDs:
\begin{itemize}
  \item 
    They showed that certain classes of cyclic structures, which contain an  \emph{infinite} number of paths, could still be interpreted as containing only a \emph{finite} collection of realizable paths (defined by a suitable ``matched-path'' condition).
  \item
    Whereas one can sometimes obtain exponential compression when a collection of paths are represented as a DAG rather than as a tree, by passing to cyclic graphs augmented with a matched-path condition, one can sometimes obtain \emph{doubly exponential compression} of a tree.
\end{itemize}
These properties suggested the basic structure of CFLOBDDs:
\begin{itemize}
  \item
    Each level-$i$ grouping in a CFLOBDD can be thought of as a loop-free procedure containing calls to level-$(i\textrm{--}1)$ groupings.
  \item
    Because a given grouping makes ``calls'' only on groupings at a strictly lower level, the ``program'' is non-recursive.
  \item
    Returns from procedure calls are somewhat non-standard in that ``control'' is not always returned to the ``call-site.''
    Instead, returns in CFLOBDDs resemble \emph{exit splitting} \cite{PLDI:BGS97}, where control can be returned to one of several return points in the caller.
  \item
    Such an interprocedural CFG is potentially cyclic, but contains only a finite number of realizable paths.
    In the \emph{best case} (which corresponds to the \emph{worst case} for path profiling), the number of realizable paths is doubly exponential in the size of the interprocedural CFG.
\end{itemize}
From here, it was only a short distance to CFLOBDDs---how to interpret such graphs as representations of Boolean functions; how to implement the standard BDD operations; how to maintain canonicalness; etc.
The feature that threatened to sink the path-profiling scheme---doubly exponential \emph{explosion}---became the linchpin of a doubly exponential \emph{compressed} representation of Boolean functions.
Metaphorically, a frog was turned into a prince.

\paragraph{Comparison with BDD variants.}
Over the years, many variants of BDDs have been proposed \cite{Collection:SF96}. 
These data structures can be broadly divided into three families:
ones that make use of weights on edges,
ones that do not use edge weights, and
ones that allow the underlying graph to have cycles.

Examples of (acyclic) edge-weighted BDD variants include 
EVBDDs \cite{dac:LS92} and FEVBDDs \cite{tafertshofer1997factored}. 
If the weights are allowed to be unboundedly large, a polynomial-sized data structure of this sort can be used to encode a decision tree that is double-exponentially larger. However, to the best of our knowledge, such double-exponential compression is impossible when the weights are required to use a constant number of bits. 

Unweighted BDD variants include Multi-Terminal BDDs \cite{dac:CMZFY93,CMU:CS-95-160}, Algebraic Decision Diagrams \cite{iccad:BFGHMPS93}, Free Binary Decision Diagrams (FBDDs) \cite[\S6]{Book:Wegener00}, Binary Moment Diagrams (BMDs) \cite{dac:BC95}, Hybrid Decision Diagrams (HDDs) \cite{iccad:CFZ95}, Differential BDDs \cite{LNCS:AMU95}, and Indexed BDDs (IBDDs) \cite{toc:JBAAF97}. 
Several of these  BDD variants offers exponential compression over classical BDDs.
However, because FBDDs, BMDs, HDDs, and IBDDs that encode, e.g., the identity function, need to examine each variable, the exponential-separation argument for CFLOBDDs from \sectref{efficient-relations} carries over for all of these variants.

Cyclic BDD variants include Linear/Exponentially Inductive Functions (LIFs/EIFs) \cite{iccad:GF93,Thesis:Gupta94} and Cyclic BDDs (CBDDs) \cite{acsc:R99}.
Because they allow cycles, CFLOBDDs are closer to the structures in this category.
The differences between CFLOBDDs and these representations can be characterized as follows:
\begin{itemize}
  \item
    The aforementioned representations all make use of numeric/arithmetic annotations on the edges of the graphs used to represent functions over Boolean arguments, rather than the matched-path principle that is the basis of CFLOBDDs.
    Matched paths can be characterized in terms of a context-free language of matched parentheses, rather than in terms of numbers and arithmetic (see \eqref{MatchedPathGrammar}).
  \item
    An essential part of the design of LIFs and EIFs is that the BDD-like subgraphs in them are connected in very restricted ways.
    In contrast, in CFLOBDDs, different groupings at the same level (or different levels) can have very different kinds of connections in them.
  \item
    Similarly, CBDDs require that there be some fixed BDD pattern that is repeated over and over in the structure;
    a given function uses only a few such patterns.  With CFLOBDDs, there can be many reused patterns (i.e., in the lower-level groupings in CFLOBDDs).
  \item
    CBDDs are not canonical representations of Boolean functions, which complicates the algorithms for performing certain operations on them, such as the operation to determine whether two CBDDs represent the same function.
  \item
    The layering in CFLOBDDs serves a different purpose than the layering found in LIFs/EIFs and CBDDs.
    In the latter representations, a connection from one layer to another serves as a jump from one BDD-like fragment to another BDD-like fragment.
    In CFLOBDDs, only the lowest layer (i.e., the collection of level-0 groupings) consists of BDD-like fragments (and just two very simple ones at that);
    it is only at level 0 that the values of variables are interpreted.
    As one follows a matched path through a CFLOBDD, the connections between the groupings at levels above level 0 serve to encode which variable is to be interpreted next.
\end{itemize}
LIFs/EIFs/CBDDs could be generalized by replacing BDD-like subgraphs in them with CFLOBDDs.
\Omit{
Similarly, other variations on BDDs \cite{Collection:SF96}, such as EVBDDs \cite{dac:LS92}, BMDs \cite{dac:BC95}, *BMDs \cite{dac:BC95}, and HDDs \cite{iccad:CFZ95}, which are all based on DAGs, could be generalized to use cyclic data structures and matched paths, along the lines of CFLOBDDs, and we believe that there could be great benefit to doing so.
}

Other data structures that generalize BDDs are representations like Sentential Decision Diagrams (SDDs)~\cite{darwiche2011sdd} and Variable Shift SDDs \cite{nakamura2020variable}. These data structures generalize BDDs by assuming a tree-shaped ordering over variables, and there are functions for which these data structures offer double-exponential compression over decision trees and an exponential compression over BDDs.
In CFLOBDDs, a grouping $g$ can have multiple middle vertices that reuse the same B-connection grouping $b$, as long as the return edges for the different invocations of $b$ use different mappings to $g$'s exit vertices.
This ``contextual rewiring'' gives CFLOBDDs greater ability to reuse substructures than SDDs and VS-SDDs.
(Moreover, $b$ can also be used as the A-connection grouping of $g$.)
SDDs and VS-SDDs (and their quantitative generalizations, such as Probabilistic SDDs~\cite{kisa2014probabilistic}) have not, so far, been used in matrix computations, and implementations of operations such as Kronecker product and matrix multiplication based on these structures are unknown, 
which meant that we could not use them in our quantum-simulation experiments.
We did compare CFLOBDDs against SDDs
for two of the micro-benchmarks, and found that CFLOBDDs were much faster (\tableref{micro-benchmarks-table}).
However, the relationship between these representations and CFLOBDDs merits future study.

\paragraph{Relationship to PDSs, NWAs, etc.}
As noted in \sectref{Introduction}, CFLOBDDs can be seen as a generalization of BDDs in which a restricted form of procedure call is permitted.
As such, CFLOBDDs are similar to various structures used in the model-checking community, namely, Hierarchical FSMs (HFSMs) \cite{TOPLAS:ABEGRY05}, Push Down Systems (PDSs) \cite{CONCUR:BEM97,ENTCS:fww97}, and Nested-Word Automata (NWA) \cite{DLT:AM2006}, which all have the same flavor of ``graph plus procedure calls.''
\begin{itemize}
  \item
    As discussed in Appendix \sectref{NWADefinition}, there is a formal sense in which a CFLOBDD is an NWA.
  \item
    HFSMs, PDSs, and NWAs have mainly been investigated for modeling infinite-state systems.
    What CFLOBDDs demonstrate is that there are advantages to considering finite restrictions of such structures.
    In particular, the advantage of introducing a procedure-call-like mechanism in a finite model is that one can obtain an exponential-factor advantage compared to the original finite models without procedure calls.
\end{itemize}
We believe that there is great potential for exploring other combinations of the idea of ``BDDs plus procedure calls'' that use ideas from HFSMs, PDSs, and NWAs in ways that are different than those found in CFLOBDDs.

CFLOBDDs can be considered to be a special case of Visibly Pushdown Automata (VPAs)~\cite{alur2004visibly} where A-connections and B-connections correspond to call transitions in a visibly pushdown language (VPL), return edges correspond to return transitions in the VPL, and edges at level-0 correspond to internal transitions of the VPL.
There are other classes of VPAs, such as $k$-module single-entry VPAs ($k$-SEVPAs) \cite{alur2005congruences} and Modular VPAs \cite{kumar2006minimization}, that have minimal canonical representations.
Such VPAs have modular components that are similar to the groupings in CFLOBDDs. However these classes of VPAs assume that the modular decomposition is fixed ahead of time.
For instance, each $k$-SEVPA has exactly $k+1$ modules.
In contrast, CFLOBDDs use a decomposition that is based on a fixed number of levels, but the specific number of groupings for a function is a by-product of the  structural invariants (\sectref{StructuralRestrictions} and Construction \ref{Constr:DecisionTreeToCFLOBDD} in Appendix~\sectref{canonical-proof}).


\paragraph{Prior Approaches to Quantum Simulation.}
Also related are prior methods for quantum simulation. Such simulation can be exact or approximate; our focus here is on exact simulation
(modulo floating-point round-off).
Decision diagrams used for such simulation include QMDDs \cite{miller2006qmdd,DBLP:journals/tcad/ZulehnerW19} and TDDs \cite{hong2020tensor}.
Both of these are weighted BDD representations, and hence cannot be compared in an apples-to-apples way with CFLOBDDs, 
which are unweighted representations (i.e., the edges of a CFLOBDD do not have associated weights).
However, to understand the potential of CFLOBDDs, we mention here how our experimental results with CFLOBDDs compare with the published data for QMDDs:
CFLOBDDs perform better than the best published numbers on some algorithms (GHZ, BV, DJ, Grover) and worse on others (QFT, Shor) \cite[Tab.\ 5.1]{Book:ZW2020}.\footnote{
  Note that the number of qubits for Shor's algorithm reported in \cite[Tab.\ 5.1]{Book:ZW2020} is the number of qubits of the circuit, whereas in \tablerefs{quantum-table-detailed-contd}{tensor-network-comp}, \#Qubits is the number of bits of the number $N$ being factored, where $\#\textit{Qubits}\text{-of-circuit} = 2 * \#\text{bits-of-}N$.
}
We also compared our approach to tensor networks, a widely used approach to quantum simulation that is not based on decision diagrams.
As shown in~\tableref{tensor-network-comp}, CFLOBDDs perform better than tensor networks on some algorithms (GHZ, BV, DJ, Grover) and worse on others (Simon, QFT, Shor).

Similar to the well-known quantum algorithms discussed in this paper, variational quantum algorithms, which include a noise channel, can also be simulated using CFLOBDDs.
Huang et al.\ \cite{huang2021logical} simulate variational quantum algorithms using knowledge-compilation techniques.
In their approach, the noise component is modeled as an additional operator whose action is represented as a matrix.
The noise matrix can be represented as a CFLOBDD, and hence CFLOBDDs can also be used for simulating variational quantum algorithms.

\paragraph{Compression of Programs and Compression Principles}
A CFLOBDD can compactly represent many finite paths.
This property is akin to a statement that the use of nonrecursive procedures in programs can enable small programs to have many execution paths, and is the essence of the aforementioned observation by Melski and Reps that an acyclic, non-recursive, interprocedural control-flow graph of size $k$ could have $2^{2^k}$ matched paths.
Although not formulated as a theorem, this observation was stated in Melski's Ph.D.\ thesis \cite[\S3.5.4]{Thesis:Melski02}.
Melski uses Yannakakis's notion of $L$-reachability \cite{kn:Yann90} (i.e., a path from node $s$ to node $t$ only counts as a valid $s$-$t$ connection if the path's labels form a word in $L$), and defines the notion of a ``finite-path graph'' with respect to some language $L$: there are only a finite number of $L$-paths from, e.g., program entry to program exit \cite[\S3.4]{Thesis:Melski02}.
He then defines an interprocedural control-flow graph, denoted by $G^*_{\textit{fin}}$.
One of the languages of interest is the language of unbalanced-left paths \cite[\S2.1]{cc:MR99}, in which each return-edge is matched with the closest preceding unmatched call-edge, but there can be zero or more umatched call-edges.
(The unbalanced-left language is typically the language of interest for context-sensitive interprocedural dataflow analysis \cite{kn:RHS95,kn:R97}.)
Melski observes, ``... the number of [unbalanced-left] paths through $G^*_{\textit{fin}}$ can be doubly exponential in the size of $G^*_{\textit{fin}}$.''\ \footnote{
  Unfortunately, the aforementioned work with the 20,000-line program that had $2^{400,000}$ paths (which is what prompted Melski and Reps to realize that they were facing double-exponential explosion) was carried out after their CC '99 paper had been published \cite{cc:MR99}.
  The latter paper states, incorrectly, ``In the worst case, the number of paths through a program is exponential in the number of branch statements $b$ $\ldots$'' \cite[\S5]{cc:MR99}.
  (This kind of mistake seems to be common among authors working with structures that are DAG-like, but are really based on acyclic hyper-graphs:
  they erroneously think that they are dealing with DAGs and conclude that there is exponential explosion/compression, whereas the true state of affairs is that they have \emph{double}-exponential explosion/compression.
  Examples are found in the literature on E-graphs \cite{OOPSLA:NWZWSASGT21,SIGPLAN:Blog:egg} and version-space algebras \cite{OOPSLA:PG15,DBLP:journals/ftpl/GulwaniPS17,DBLP:journals/corr/abs-2107-12568}.)
}

These results are tantamount to the statement (proposed by one of the referees) that ``there is a family of programs $P_n$, written with non-recursive procedures, that each would be exponentially larger if written without non-recursive procedures.''
In the 1970s, the literature on program schematology \cite{DBLP:conf/pmac/PatersonH70} explored the relative power of various programming constructs, beginning with results showing that recursive procedure calls are more expressive than iteration (in particular, there are recursive program schemes such that, for some interpretation of the function and predicate symbols, any flowchart scheme will produce results that are different from those obtained with the recursive scheme \cite{DBLP:conf/pmac/PatersonH70,DBLP:journals/mst/LynchB79}).
Thus, it would have been natural for the schematology literature to contain a result of the form stated above.
However, we were unable to find a paper with such a result;
when procedures are allowed, the main interest seems to be in recursive procedures and how such programs compare with programs written in a language without procedure calls, but other features, such as arrays, stacks, or counters \cite{DBLP:journals/siamcomp/ConstableG72,DBLP:journals/jcss/GarlandL73}.

The compression abilities of CFLOBDDs are based what might be
called ``multiplicative amplification'': calls to procedures
$P$ and $Q$, when performed in sequence, result in a structure
in which the number of $L(\textit{matched})$-paths is equal
to the product of the numbers of $L(\textit{matched})$-paths through
$P$ and $Q$.
Multiplicative amplification leads to the repeated
squaring we see in counting the number of paths from the entry
vertex to the (one) exit vertex of a no-distinction proto-CFLOBDD with $k$ levels:
\[
  P(0) = 1 \qquad\qquad P(n+1) = {P(n)}^2.
\]

A more powerful compression principle---again based on an ``amplification'' step repeated some number of times---is found in Mairson's rational reconstruction of a proof of Statman's \cite{DBLP:journals/tcs/Mairson92}.
As with CFLOBDDs, there are a finite number of stratification levels and no recursion, but instead of ``multiplicative amplification,'' Mairson uses ``powerset amplification.''
He is interested in representing all values of the stratified types defined by
\[
  \mathcal{D}_0 = \{\textrm{true}, \textrm{false}\}  \qquad\qquad   \mathcal{D}_n = \textit{powerset}(D_{n-1}).
\]
Mairson observes that one can use linked lists to represent the elements of each of the $\mathcal{D}_i$.
To represent them concisely, he defines a powerset-combinator $\texttt{powerset}$ that takes a list $l_1$ as input, and returns a list $l_2$ that contains the powerset of the elements of $l_1$ (a simple exercise in functional programming).
He can then represent $\mathcal{D}_n$ with a $\lambda$-calculus term $D_n$ that applies $\texttt{powerset}$ $n$ times to the list $\{\textrm{true}, \textrm{false}\}$.
Considered as a member of the family of terms $D_0$, $D_1 = \texttt{powerset}(D_0)$, $\ldots$, $D_n = \texttt{powerset}^n(D_0)$, $\ldots$, the size of the term $D_n$ is $\Omega(n)$.
In contrast, the size of the set that is represented by $D_n$ is described by the following recurrence relation:
\[
  S(0) = 1 \qquad\qquad S(n+1) = 2^{S(n)},
\]
whose solution is non-elementary: $S(n)$ is an exponential tower of 2s,
\stackinset{r}{-8pt}{b}{-11pt}{\tiny\rotatebox{35}{$\underbrace{\kern25pt}_n$}}{$2^{2^{2^{\rdots^2}}}$},
of height $n$.

\section{Conclusions and Future Work}
\label{Se:conclusion}

This paper described a new data structure---CFLOBDDs---for representing functions, matrices, relations, and other discrete structures.
CFLOBDDs are a plug-compatible replacement for BDDs, and can represent Boolean functions in a more compressed fashion than BDDs---exponentially smaller in the best case---and, again in the best case, double-exponentially smaller than the size of a Boolean function's decision tree.
The results presented in the paper include the following:
\begin{itemize}
  \item
    We demonstrated the exponential separation between CFLOBDDs and BDDs both theoretically and experimentally.
  \item
    Because CFLOBDDs can be used as a replacement for BDDs, we gave algorithms for various operations that can be performed on representations of Boolean functions, such as Boolean and arithmetic operations, Kronecker product, matrix multiplication, etc.
  \item
    We also gave a set of structural invariants that ensure the canonicity of CFLOBDDs.
    When CFLOBDDs are implemented via hash-consing, this property ensures that the test of whether two CFLOBDDs represent equal functions can be performed merely by comparing the values of two pointers.
  \item
    We provided both an operational semantics and a denotational semantics of CFLOBDDs.
  \item
    We investigated the use of CFLOBDDs for representing the matrices and vectors involved in quantum simulation, and gave algorithms for the required operations.
\end{itemize}

In our experiments, we compared the time and space usage of CFLOBDDs and BDDs on two types of benchmarks: (i) micro-benchmarks, and (ii) quantum-simulation benchmarks.
We found that the improvement in scalability with CFLOBDDs is quite dramatic.
\begin{itemize}
  \item
    For the micro-benchmarks, CFLOBDDs were able to represent versions of the benchmark functions with up to $1.5 \times 2^{8,388,609} \cong 1.2795 \times 10^{\textrm{2,525,223}}$ Boolean variables.
    In contrast, the largest BDD that we could construct successfully for any of the micro-benchmarks had about ${328,000} = 3.28 \times 10^5$ Boolean variables.
  \item
    For the quantum-simulation benchmarks, the number of qubits that could be handled using CFLOBDDs was larger---compared to BDDs---by a factor of
    128$\times$ for GHZ; 1,024$\times$ for BV; 8,192$\times$ for DJ; and 128$\times$ for Grover's algorithm.
\end{itemize}
These results support the conclusion that, for at least some applications, CFLOBDDs provide a much more compressed representation of discrete structures than is possible with BDDs, thereby permitting much larger problem instances to be handled than heretofore.

\paragraph{Future Work}
The work on CFLOBDDs opens up many avenues for future work.
For instance, CFLOBDDs are based on the following decomposition pattern (or ``pattern of calls''):
\[
  \begin{array}{rcl}
   \Matched_k & \rightarrow & \Matched_{k-1} \quad \Matched_{k-1} \\
   \Matched_0 & \rightarrow & \epsilon 
  \end{array}
\]
Other decomposition schemes are possible;
for instance, in \sectref{Separation:AdditionRelation}, it would have been useful to be able to define the $\ADD_n$ relation using the decomposition
\[
  \begin{array}{rcl}
   \Matched_k & \rightarrow & \Matched_{k-1} \quad \Matched_{k-1} \quad \Matched_{k-1} \\
   \Matched_0 & \rightarrow & \epsilon 
  \end{array}
\]
which would have avoided having to ``waste'' one-quarter of the variables, as in \figrefs{carry_in_0}{carry_in_1}.

Another possibility would be to permit a level-$i$ grouping to have connections to level-$i\text{-}j$ groupings, where $j \ge 1$).
For instance, one could have a Fibonacci-like decomposition
\[
  \begin{array}{rcl}
   \Matched_k & \rightarrow & \Matched_{k-1} \quad \Matched_{k-2} \\
   \Matched_1 & \rightarrow & \epsilon \\
   \Matched_0 & \rightarrow & \epsilon 
  \end{array}
\]
As long as there is a fixed structural-decomposition pattern to how the levels connect, it should still possible to keep the same properties---i.e., the ability to perform operations implicitly, without having to instantiate the decision tree; canonicalness; etc.---as enjoyed by CFLOBDDs as defined in this paper (\defrefs{MockCFLOBDD}{CFLOBDD}).
We leave the exploration of such generalized CFLOBDDs for future work, along with such questions as
\begin{itemize}
  \item
    How does one go about choosing a good structural-decomposition pattern for a given family of Boolean functions?
  \item
    Is it feasible for a CFLOBDD package to support the ability to change the structural-decomposition pattern dynamically (similar to the way that some BDD packages support the ability to change the variable ordering dynamically).
\end{itemize}

As mentioned in \sectref{RelatedWork}, there are other existing data structures \cite{dac:BC95, iccad:CFZ95, dac:LS92, tafertshofer1997factored} that are exponentially more succinct than BDDs for some functions.
However, the strategies used in the design of these data structures (for example, the use of a tree structure over variables in SDDs) are different than our strategy of exploiting the $\MPP$.
Bringing together these strategies and ours is an interesting avenue for future research.

Another natural direction is the study of \emph{weighted} CFLOBDDs
\twrchanged{
\cite{DBLP:journals/corr/abs-2305-13610}.
}
Also, while quantum simulation was the primary practical application studied in this paper, it is by no means the only one.
For example, CFLOBDDs may also have natural applications in settings like probabilistic programming, model checking, or circuit synthesis in which decision diagram representations have been successfully used in the past. 

Finally, the idea of using programs with procedure calls as succinct data structures can be applied beyond our present setting of decision diagrams.
In particular, automata-based representations of version spaces are a natural candidate for such additional compression.
The exploration of these ideas in other settings is a rich and wholly open direction.


\Supplemental{
\section{Acknowledgments}
We thank Richard Lipton for the suggestion to apply CFLOBDDs to quantum simulation,
Patrick Emonts for advice about the proper way to perform quantum-circuit simulation with tensor networks,
and the referees of Sistla et al.\ \cite{TOPLAS:SCR24} for suggestions of how to clarify several items in the presentation.
We are indebted to Alfons Laarman for helping us realize that the local-reduction property of Reduce (\lemref{SizeOfGrouping}) does not imply that a global-reduction property holds (\footnoteref{LocalGlobalReduceProperty}), and for suggesting that we look for a polynomial-time bound in terms of the sizes of Reduce's input and output proto-CFLOBDDs (\sectref{CostOfReduce}). 

The work was supported, in part,
by a gift from Rajiv and Ritu Batra;
by the John Simon Guggenheim Memorial Foundation;
by Facebook under a Probability and Programming Research Award;
by NSF under grants CCR-9986308 and CCF-2212559;
by ONR under contracts N00014-00-1-0607 and N00014-19-1-2318;
by the MDA under SBIR contract DASG60-01-P-0048 to GrammaTech, Inc.,
and by an S.N.\ Bose Scholarship to Meghana Aparna Sistla.
Thomas Reps has an ownership interest in GrammaTech, Inc., which has licensed elements of the technology reported in this publication.
}{}

\bibliography{PatentApplication/df.bib,PatentApplication/logic.bib,PatentApplication/mab.bib,refs.bib}

\begin{thebibliography}{10}

\bibitem{TOPLAS:ABEGRY05}
Rajeev Alur, Michael Benedikt, Kousha Etessami, Patrice Godefroid, Thomas Reps, and Mihalis Yannakakis.
\newblock Analysis of recursive state machines.
\newblock {\em Trans. on Prog. Lang. and Syst.}, 27(4):786--818, 2005.

\bibitem{alur2005congruences}
Rajeev Alur, Viraj Kumar, Parthasarathy Madhusudan, and Mahesh Viswanathan.
\newblock Congruences for visibly pushdown languages.
\newblock In {\em International Colloquium on Automata, Languages, and Programming}, pages 1102--1114. Springer, 2005.

\bibitem{DLT:AM2006}
Rajeev Alur and P.~Madhusudan.
\newblock Adding nesting structure to words.
\newblock In {\em Developments in Lang.\ Theory}, 2006.

\bibitem{JACM:AM09}
Rajeev Alur and P.~Madhusudan.
\newblock Adding nesting structure to words.
\newblock {\em J.\ ACM}, 56(3), 2009.

\bibitem{alur2004visibly}
Rajeev Alur and Parthasarathy Madhusudan.
\newblock Visibly pushdown languages.
\newblock In {\em Proceedings of the thirty-sixth annual ACM symposium on Theory of computing}, pages 202--211, 2004.

\bibitem{LNCS:AMU95}
Anuchit Anuchitanukul, Zohar Manna, and Tom{\'a}s~E. Uribe.
\newblock Differential {BDDs}.
\newblock In {J. van Leeuwen}, editor, {\em Computer Science Today: Recent Trends and Developments}, volume 1000 of {\em Lecture Notes in Computer Science}, pages 218--233. Springer-Verlag, 1995.

\bibitem{iccad:BFGHMPS93}
R.~Iris Bahar, Erica~A. Frohm, Charles~M. Gaona, Gary~D. Hachtel, Enrico Macii, Abelardo Pardo, and Fabio Somenzi.
\newblock Algebraic decision diagrams and their applications.
\newblock In {\em Proc. of the Int. Conf. on Computer Aided Design}, pages 188--191, November 1993.

\bibitem{micro:BL96}
Thomas Ball and James~R. Larus.
\newblock Efficient path profiling.
\newblock In {\em Proc. of {MICRO}-29}, December 1996.

\bibitem{DBLP:conf/spin/BallR01}
Thomas Ball and Sriram~K. Rajamani.
\newblock Automatically validating temporal safety properties of interfaces.
\newblock In Matthew~B. Dwyer, editor, {\em Model Checking Software, 8th International {SPIN} Workshop, Toronto, Canada, May 19-20, 2001, Proceedings}, volume 2057 of {\em Lecture Notes in Computer Science}, pages 103--122. Springer, 2001.

\bibitem{DBLP:conf/paste/BallR01}
Thomas Ball and Sriram~K. Rajamani.
\newblock Bebop: a path-sensitive interprocedural dataflow engine.
\newblock In John Field and Gregor Snelting, editors, {\em Proceedings of the 2001 {ACM} {SIGPLAN-SIGSOFT} Workshop on Program Analysis For Software Tools and Engineering, PASTE'01, Snowbird, Utah, USA, June 18-19, 2001}, pages 97--103. {ACM}, 2001.

\bibitem{banuls2009matrix}
Mari-Carmen Banuls, Matthew~B Hastings, Frank Verstraete, and J~Ignacio Cirac.
\newblock Matrix product states for dynamical simulation of infinite chains.
\newblock {\em Physical review letters}, 102(24):240603, 2009.

\bibitem{beauregard2002circuit}
Stephane Beauregard.
\newblock Circuit for {S}hor's algorithm using 2$n$+3 qubits.
\newblock {\em arXiv preprint quant-ph/0205095}, 2002.

\bibitem{DBLP:conf/icalp/BenediktGR01}
Michael Benedikt, Patrice Godefroid, and Thomas~W. Reps.
\newblock Model checking of unrestricted hierarchical state machines.
\newblock In {\em Automata, Languages and Programming, 28th International Colloquium, {ICALP} 2001, Crete, Greece, July 8-12, 2001, Proceedings}, pages 652--666, 2001.

\bibitem{PLDI:BGS97}
Ras Bodík, Rajiv Gupta, and Mary~Lou Soffa.
\newblock Interprocedural conditional branch elimination.
\newblock In {\em Prog. Lang. Design and Impl.}, pages 146--158, 1997.

\bibitem{CONCUR:BEM97}
Ahmed Bouajjani, Javier Esparza, and Oded Maler.
\newblock Reachability analysis of pushdown automata: {A}pplication to model checking.
\newblock In {\em Proc.\ {CONCUR}}, 1997.

\bibitem{dac:BRB90}
Karl~S. Brace, Richard~L. Rudell, and Randal~E. Bryant.
\newblock Efficient implementation of a {BDD} package.
\newblock In {\em Proc. of the 27th ACM/IEEE Design Automation Conf.}, pages 40--45, 1990.

\bibitem{toc:Bryant86}
Randal~E. Bryant.
\newblock Graph-based algorithms for {B}oolean function manipulation.
\newblock {\em IEEE Trans.\ on Comp.}, C-35(6):677--691, August 1986.

\bibitem{dac:BC95}
Randal~E. Bryant and Yirng-An Chen.
\newblock Verification of arithmetic circuits with binary moment diagrams.
\newblock In {\em Proc. of the 30th ACM/IEEE Design Automation Conf.}, pages 535--541, 1995.

\bibitem{CMU:CS-95-160}
Edmund~M. Clarke, Masahiro Fujita, and Xudong Zhao.
\newblock Applications of multi-terminal binary decision diagrams.
\newblock Technical Report CS-95-160, Carnegie Mellon Univ., School of Comp. Sci., April 1995.

\bibitem{iccad:CFZ95}
Edmund~M. Clarke, Masahiro Fujita, and Xudong Zhao.
\newblock Hybrid decision diagrams: Overcoming the limitations of {MTBDDs} and {BMDs}.
\newblock In {\em Proc. of the Int. Conf. on Computer Aided Design}, pages 159--163, November 1995.

\bibitem{dac:CMZFY93}
Edmund~M. Clarke, Kenneth~L. McMillan, Xudong Zhao, Masahiro Fujita, and Jerry Chih-Yuan Yang.
\newblock Spectral transforms for large {B}oolean functions with applications to technology mapping.
\newblock In {\em Proc. of the 30th ACM/IEEE Design Automation Conf.}, pages 54--60, 1993.

\bibitem{RDF:4:CFZ96}
E.M. Clarke, M.~Fujita, and X.~Zhao.
\newblock Multi-terminal binary decision diagrams and hybrid decision diagrams.
\newblock In T.~Sasao and M.~Fujita, editors, {\em Representations of Discrete Functions}, pages 93--108. Kluwer Acad., Norwell, MA, 1996.

\bibitem{DBLP:journals/siamcomp/ConstableG72}
Robert~L. Constable and David Gries.
\newblock On classes of program schemata.
\newblock {\em {SIAM} J. Comput.}, 1(1):66--118, 1972.

\bibitem{darwiche2011sdd}
Adnan Darwiche.
\newblock {SDD}: {A} new canonical representation of propositional knowledge bases.
\newblock In {\em Twenty-Second International Joint Conference on Artificial Intelligence}, 2011.

\bibitem{elder2014abstract}
Matt Elder, Junghee Lim, Tushar Sharma, Tycho Andersen, and Thomas Reps.
\newblock Abstract domains of affine relations.
\newblock {\em ACM Transactions on Programming Languages and Systems (TOPLAS)}, 36(4):1--73, 2014.

\bibitem{DBLP:conf/ml/FilliatreC06}
Jean{-}Christophe Filli{\^{a}}tre and Sylvain Conchon.
\newblock Type-safe modular hash-consing.
\newblock In Andrew Kennedy and Fran{\c{c}}ois Pottier, editors, {\em Proceedings of the {ACM} Workshop on ML, 2006, Portland, Oregon, USA, September 16, 2006}, pages 12--19. {ACM}, 2006.

\bibitem{ENTCS:fww97}
Alain Finkel, Bernard Willems, and Pierre Wolperr.
\newblock A direct symbolic approach to model checking pushdown systems.
\newblock {\em Elec. Notes in Theor. Comp. Sci.}, 9, 1997.

\bibitem{fousse2007mpfr}
Laurent Fousse, Guillaume Hanrot, Vincent Lef{\`e}vre, Patrick P{\'e}lissier, and Paul Zimmermann.
\newblock {MPFR}: {A} multiple-precision binary floating-point library with correct rounding.
\newblock {\em ACM Transactions on Mathematical Software (TOMS)}, 33(2):13--es, 2007.

\bibitem{fowler2004implementation}
Austin~G. Fowler, Simon~J. Devitt, and Lloyd~C.L. Hollenberg.
\newblock Implementation of {S}hor's algorithm on a linear nearest neighbour qubit array.
\newblock {\em arXiv preprint quant-ph/0402196}, 2004.

\bibitem{DBLP:journals/jcss/GarlandL73}
Stephen~J. Garland and David~C. Luckham.
\newblock Program schemes, recursion schemes, and formal languages.
\newblock {\em J. Comput. Syst. Sci.}, 7(2):119--160, 1973.

\bibitem{Tokyo-TR-74-03:Goto74}
Eiichi Goto.
\newblock Monocopy and associative algorithms in extended {Lisp}.
\newblock Tech. Rep. TR 74-03, Univ. of Tokyo, Tokyo, Japan, 1974.

\bibitem{gray2018quimb}
Johnnie Gray.
\newblock {quimb}: {A} python library for quantum information and many-body calculations.
\newblock {\em Journal of Open Source Software}, 3(29):819, 2018.

\bibitem{DBLP:journals/ftpl/GulwaniPS17}
Sumit Gulwani, Oleksandr Polozov, and Rishabh Singh.
\newblock Program synthesis.
\newblock {\em Found. Trends Program. Lang.}, 4(1-2):1--119, 2017.

\bibitem{Thesis:Gupta94}
Aarti Gupta.
\newblock {\em Inductive Boolean Function Manipulation: {A} Hardware Verification Methodology for Automatic Induction}.
\newblock PhD thesis, Carnegie Mellon Univ., 1994.
\newblock Tech. Rep. CMU-CS-94-208.

\bibitem{iccad:GF93}
Aarti Gupta and Allan~L. Fisher.
\newblock Representation and symbolic manipulation of linearly inductive {B}oolean functions.
\newblock In {\em Proc. of the Int. Conf. on Computer Aided Design}, pages 192--199, November 1993.

\bibitem{hong2020tensor}
Xin Hong, Xiangzhen Zhou, Sanjiang Li, Yuan Feng, and Mingsheng Ying.
\newblock A tensor network based decision diagram for representation of quantum circuits.
\newblock {\em ACM Transactions on Design Automation of Electronic Systems (TODAES)}, 2020.

\bibitem{kn:HRB90}
Susan Horwitz, Thomas Reps, and David Binkley.
\newblock Interprocedural slicing using dependence graphs.
\newblock {\em Trans.\ on Prog.\ Lang.\ and Syst.}, 12(1):26--60, January 1990.

\bibitem{huang2021logical}
Yipeng Huang, Steven Holtzen, Todd Millstein, Guy Van~den Broeck, and Margaret Martonosi.
\newblock Logical abstractions for noisy variational quantum algorithm simulation.
\newblock In {\em Proceedings of the 26th ACM International Conference on Architectural Support for Programming Languages and Operating Systems}, pages 456--472, 2021.

\bibitem{Book:HMM85}
Stanley~Leonard Hurst, D.~Michael Miller, and Jon~C. Muzio.
\newblock {\em Spectral Techniques in Digital Logic}.
\newblock Acad.\ Press, Inc., 1985.

\bibitem{toc:JBAAF97}
Jawahar Jain, James~R. Bitner, Magdy~S. Abadir, Jacob~A. Abraham, and Donald~S. Fussell.
\newblock Indexed {BDDs}: {A}lgorithmic advances in techniques to represent and verify {B}oolean functions.
\newblock {\em IEEE Trans.\ on Comp.}, C-46(11):1230--1245, November 1997.

\bibitem{kisa2014probabilistic}
Doga Kisa, Guy Van~den Broeck, Arthur Choi, and Adnan Darwiche.
\newblock Probabilistic sentential decision diagrams.
\newblock In {\em Fourteenth International Conference on the Principles of Knowledge Representation and Reasoning}, 2014.

\bibitem{DBLP:journals/corr/abs-2107-12568}
James Koppel.
\newblock Version space algebras are acyclic tree automata.
\newblock {\em CoRR}, abs/2107.12568, 2021.

\bibitem{kumar2006minimization}
Viraj Kumar, Parthasarathy Madhusudan, and Mahesh Viswanathan.
\newblock Minimization, learning, and conformance testing of {B}oolean programs.
\newblock In {\em International Conference on Concurrency Theory}, pages 203--217. Springer, 2006.

\bibitem{dac:LS92}
Yung-Te Lai and Sarma Sastry.
\newblock Edge-valued binary decision diagrams for multi-level hierarchical verification.
\newblock In {\em Proc. of the 29th Conf. on Design Automation}, pages 608--613, Los Alamitos, CA, USA, June 1992. IEEE Computer Society Press.

\bibitem{Thesis:Lhotak06}
Ondrej Lhot{\'{a}}k.
\newblock {\em Program Analysis Using Binary Decision Diagrams}.
\newblock PhD thesis, McGill University, 2006.

\bibitem{Blog:PEqNP:06:16:2009}
Richard~J. Lipton.
\newblock {BDD's} and factoring, June~16, 2009.
\newblock G\"{o}del's Lost Letter and P=NP blog, \href{https://rjlipton.wpcomstaging.com/2009/06/16/bdds-and-factoring}{https://rjlipton.wpcomstaging.com/2009/06/16/bdds-and-factoring}.

\bibitem{lipton2014quantum}
Richard~J. Lipton and Kenneth~W. Regan.
\newblock {\em Quantum Algorithms via Linear Algebra: A Primer}.
\newblock MIT Press, 2014.

\bibitem{DBLP:journals/mst/LynchB79}
Nancy~A. Lynch and Edward~K. Blum.
\newblock A difference in expressive power between flowcharts and recursion schemes.
\newblock {\em Math. Syst. Theory}, 12:205--211, 1979.

\bibitem{DBLP:journals/tcs/Mairson92}
Harry~G. Mairson.
\newblock A simple proof of a theorem of statman.
\newblock {\em Theor. Comput. Sci.}, 103(2):387--394, 1992.

\bibitem{wannes_meert_2018_1202374}
Wannes Meert and Arthur Choi.
\newblock {PySDD, v0.1}.
\newblock Zenodo, 10.5281/zenodo.1202374, March 2018.

\bibitem{PersonalComm:Melski98}
David Melski.
\newblock Personal communication, 1998.

\bibitem{Thesis:Melski02}
David Melski.
\newblock {\em Interprocedural Path Profiling and the Interprocedural Express-Lane Transformation}.
\newblock PhD thesis, Comp. Sci. Dept., Univ. of Wisconsin, Madison, WI, February 2002.
\newblock Tech.\ Rep.\ 1435.

\bibitem{cc:MR99}
David Melski and Thomas Reps.
\newblock Interprocedural path profiling.
\newblock In {\em Comp.\ Construct.}, pages 47--62, 1999.

\bibitem{Edinburgh-MIP-R-29:Michie67}
Donald Michie.
\newblock Memo functions: {A} language feature with `rote-learning' properties.
\newblock Technical Report MIP-R-29, Dept.\ of Machine Intelligence and Perception, Univ.\ of Edinburgh, Edinburgh, Scotland, November 1967.

\bibitem{miller2006qmdd}
D.~Michael Miller and Mitchell~A. Thornton.
\newblock Qmdd: A decision diagram structure for reversible and quantum circuits.
\newblock In {\em 36th International Symposium on Multiple-Valued Logic (ISMVL'06)}, pages 30--30. IEEE, 2006.

\bibitem{nakamura2020variable}
Kengo Nakamura, Shuhei Denzumi, and Masaaki Nishino.
\newblock Variable shift sdd: a more succinct sentential decision diagram.
\newblock {\em arXiv preprint arXiv:2004.02502}, 2020.

\bibitem{OOPSLA:NWZWSASGT21}
C.~Nandi, M.~Willsey, A.~Zhu, Y.R. Wang, B.~Saiki, A.~Anderson, A.~Schulz, D.~Grossman, and Z.~Tatlock.
\newblock Rewrite rule inference using equality saturation.
\newblock In {\em Object-Oriented Programming, Systems, Languages, and Applications}, 2021.

\bibitem{nielsen2002quantum}
Michael~A Nielsen and Isaac Chuang.
\newblock Quantum computation and quantum information, 2002.

\bibitem{nielsen2001quantum}
Michael~A Nielsen and Isaac~L Chuang.
\newblock Quantum computation and quantum information.
\newblock {\em Phys. Today}, 54(2):60, 2001.

\bibitem{DBLP:conf/pmac/PatersonH70}
Michael~S. Paterson and Carl~E. Hewitt.
\newblock Comparative schematology.
\newblock In Jack~B. Dennis, editor, {\em Record of the Project {MAC} Conference on Concurrent Systems and Parallel Computation, Woods Hole, Massachusetts, USA, June 2-5, 1970}, pages 119--127. {ACM}, 1970.

\bibitem{OOPSLA:PG15}
O.~Polozov and S.~Gulwani.
\newblock {FlashMeta}: {A} framework for inductive program synthesis.
\newblock In {\em Object-Oriented Programming, Systems, Languages, and Applications}, 2015.

\bibitem{acsc:R99}
Frank Reffel.
\newblock {BDD}-nodes can be more expressive.
\newblock In {\em Proc. of the Asian Computing Science Conference}, December 1999.

\bibitem{kn:R97}
T.~Reps.
\newblock Program analysis via graph reachability.
\newblock In {\em Proc. of ILPS '97: Int. Logic Programming Symposium}, pages 5--19, Cambridge, MA, 1997. M.I.T.

\bibitem{POPL:RHS95}
T.~Reps, S.~Horwitz, and M.~Sagiv.
\newblock Precise interprocedural dataflow analysis via graph reachability.
\newblock In {\em Princ. of Prog. Lang.}, pages 49--61, 1995.

\bibitem{kn:RHS95}
Thomas Reps, Susan Horwitz, and Mooly Sagiv.
\newblock Precise interprocedural dataflow analysis via graph reachability.
\newblock In {\em Princ.\ of Prog.\ Lang.}, pages 49--61, 1995.

\bibitem{PatentApplication:Reps02}
Thomas~W. Reps.
\newblock Method for representing information in a highly compressed fashion, June~20, 2002.
\newblock United States Patent Application 20020078431, filed Feb.\ 2, 2001, US Patent \& Trademark Office. (Application abandoned.) \href{https://patentimages.storage.googleapis.com/ab/f4/17/2bbd2a0fad32f6/US20020078431A1.pdf}{https://patentimages.storage.googleapis.com/ab/f4/17/2bbd2a0fad32f6/US20020078431A1.pdf}.

\bibitem{Collection:SF96}
Tsutomu Sasao and Masahira Fujita, editors.
\newblock {\em Representations of Discrete Functions}.
\newblock Kluwer Acad., 1996.

\bibitem{kn:SP81}
Micha Sharir and Amir Pnueli.
\newblock Two approaches to interprocedural data flow analysis.
\newblock In {\em Program Flow Analysis: {T}heory and Applications}. Prentice-Hall, 1981.

\bibitem{TOPLAS:SCR24}
Meghana Sistla, Swarat Chaudhuri, and Thomas Reps.
\newblock {CFLOBDDs}: {C}ontext-free-language ordered binary decision diagrams.
\newblock {\em ACM Transactions on Programming Languages and Systems (TOPLAS)}, 36(4):1--73, 2024.
\newblock To appear; DOI: 10.1145/3651157.

\bibitem{DBLP:journals/corr/abs-2305-13610}
Meghana Sistla, Swarat Chaudhuri, and Thomas~W. Reps.
\newblock Weighted context-free-language ordered binary decision diagrams.
\newblock {\em CoRR}, abs/2305.13610, 2023.

\bibitem{somenzi2012cudd}
Fabio Somenzi.
\newblock {CUDD}: {CU} decision diagram package--release 2.4.0.
\newblock {\em University of Colorado at Boulder}, 2012.

\bibitem{tafertshofer1997factored}
Paul Tafertshofer and Massoud Pedram.
\newblock Factored edge-valued binary decision diagrams.
\newblock {\em Formal Methods in System Design}, 10(2):243--270, 1997.

\bibitem{verstraete2004matrix}
Frank Verstraete, Juan~J Garcia-Ripoll, and Juan~Ignacio Cirac.
\newblock Matrix product density operators: Simulation of finite-temperature and dissipative systems.
\newblock {\em Physical review letters}, 93(20):207204, 2004.

\bibitem{vidal2003efficient}
Guifr{\'e} Vidal.
\newblock Efficient classical simulation of slightly entangled quantum computations.
\newblock {\em Physical review letters}, 91(14):147902, 2003.

\bibitem{Book:Wegener00}
Ingo Wegener.
\newblock {\em Branching Programs and Binary Decision Diagrams}.
\newblock SIAM Monographs on Disc. Math. and Appl. Society for Industrial and Applied Mathematics, 2000.

\bibitem{APLAS:WACL05}
J.~Whaley, D.~Avots, M.~Carbin, and M.S. Lam.
\newblock Using {D}atalog with {B}inary {D}ecision {D}iagrams for program analysis.
\newblock In {\em Asian Symp.\ on Prog.\ Lang.\ and Systems}, 2005.

\bibitem{PLDI:WL04}
J.~Whaley and M.~Lam.
\newblock Cloning-based context-sensitive pointer alias analyses using binary decision diagrams.
\newblock In {\em Prog.\ Lang.\ Design and Impl.}, 2004.

\bibitem{SIGPLAN:Blog:egg}
M.~Willsey.
\newblock Fast and extensible equality saturation with {egg}, 2021.

\bibitem{woolfe2015matrix}
Kieran Woolfe.
\newblock {\em Matrix product operator simulations of quantum algorithms}.
\newblock PhD thesis, University of Melbourne, School of Physics, 2015.

\bibitem{kn:Yann90}
M.~Yannakakis.
\newblock Graph-theoretic methods in database theory.
\newblock In {\em Symp. on Princ. of Database Syst.}, pages 230--242, 1990.

\bibitem{yu2021quantum}
Nengkun Yu and Jens Palsberg.
\newblock Quantum abstract interpretation.
\newblock In {\em Proceedings of the 42nd ACM SIGPLAN International Conference on Programming Language Design and Implementation}, pages 542--558, 2021.

\bibitem{UNPUB:ZR24}
Xusheng Zhi and Thomas Reps.
\newblock Polynomial bounds of {CFLOBDDs} against {BDDs}.
\newblock In preparation, May 2024.

\bibitem{DBLP:journals/tcad/ZulehnerW19}
Alwin Zulehner and Robert Wille.
\newblock Advanced simulation of quantum computations.
\newblock {\em {IEEE} Trans. Comput. Aided Des. Integr. Circuits Syst.}, 38(5):848--859, 2019.

\bibitem{Book:ZW2020}
Alwin Zulehner and Robert Wille.
\newblock {\em Introducing Design Automation for Quantum Computing}.
\newblock Springer, 2020.

\end{thebibliography}
\bibliographystyle{plain} 
\pagebreak
\appendix

\section{Details of Notation for CFLOBDDs and their Components}
\label{Se:additional-notation}
A few words are in order about the notation used in the pseudo-code:
\begin{itemize}
  \item
    {\sloppy
    A Java-like semantics is assumed.  For example, an object or
    field that is declared to be of type {\tt InternalGrouping\/} is really a
    pointer to a piece of heap-allocated storage.
    A variable of type {\tt InternalGrouping\/} is declared and initialized
    to a new {\tt InternalGrouping\/} object of level $k$ by the declaration
    \begin{center}
      {\tt InternalGrouping g = new InternalGrouping(k) }
    \end{center}
    }
  \item
    Procedures can return multiple objects by returning tuples
    of objects, where tupling is denoted by square brackets.
    For instance, if {\tt f\/} is a procedure that returns a pair
    of {\tt int\/}s---and, in particular, if {\tt f(3)\/} returns
    a pair consisting of the values 4 and 5---then {\tt int\/}
    variables {\tt a\/} and {\tt b\/} would be assigned 4 and 5
    by the following initialized declaration:
    \begin{center}
      {\tt int$\times$int [a,b] = f(3) }
    \end{center}
  \item
    The indices of array elements start at 1.
  \item
    Arrays are allocated with an initial length (which is allowed to
    be 0); however, arrays are assumed to lengthen automatically to
    accommodate assignments at index positions beyond the current
    length.
  \item
    We assume that a call on the constructor {\tt InternalGrouping(k)\/}
    returns an {\tt InternalGrouping\/} in which the
    members have been initialized as follows:
    \begin{center}
    {\tt
    \begin{minipage}{\columnwidth}
    \begin{tabbing}
      level = k \\
      AConnection = NULL \\
      AReturnTuple = NULL \\
      numberOfBConnections = 0 \\
      BConnections = new array[0] of Grouping \\
      BReturnTuples = new array[0] of ReturnTuple \\
      numberOfExits = 0
    \end{tabbing}
    \end{minipage}
    }
    \end{center}

    Similarly, we assume that a call on the constructor
    {\tt CFLOBDD(g,vt)\/} returns a {\tt CFLOBDD\/} in which the
    members have been initialized as follows:
    \begin{center}
    {\tt
    \begin{tabular}{@{\hspace{.0in}}l@{\hspace{.0in}}}
      grouping = g \\
      valueTuple = vt
    \end{tabular}
    }
    \end{center}

\end{itemize}

The class definitions of \figref{ClassDefinitions}, as well as
the algorithms for the core CFLOBDD operations make use of the following auxiliary
classes:
\begin{itemize}
  \item
    A {\tt ReturnTuple\/} is a finite tuple of positive integers.
  \item
    A {\tt PairTuple\/} is a sequence of ordered pairs.
  \item
    A {\tt TripleTuple\/} is a sequence of ordered triples.
  \item
    A {\tt ValueTuple\/} is a finite tuple of whatever values
    the multi-terminal CFLOBDD is defined over.
\end{itemize}

\section{Lexicographic-Order Proposition}
\label{Se:prop-LexicographicOrder}

\begin{SPro}{LexicographicOrder} (\textsc{Lexicographic-Order Proposition}).
Let $ex_C$ be the sequence of exit vertices of proto-CFLOBDD $C$.
Let $ex_L$ be the sequence of exit vertices reached by traversing $C$ on each possible
Boolean-variable-to-Boolean-value assignment, generated in lexicographic
order of assignments.
Let $s$ be the subsequence of $ex_L$ that retains just the leftmost occurrences
of members of $ex_L$ (arranged in order as they first appear in $ex_L$).
Then $ex_C = s$.
\newline

\begin{Proof}
We argue by induction over levels:
\newline

\begin{BaseCase}
The proposition follows immediately for level-$0$ proto-CFLOBDDs.
\end{BaseCase}
\newline

\begin{InductionStep}
The induction hypothesis is that the proposition holds for
every level-$k$ proto-CFLOBDD.

Let $C$ be an arbitrary level-$k\textrm{+}1$ proto-CFLOBDD, with $s$ and $ex_C$
as defined above.
Without loss of generality, we will refer to the exit vertices by ordinal
position; i.e., we will consider $ex_C$ to be the sequence $[1, 2, \ldots, |ex_C|]$.
Let $C_A$ denote the $A$-connection of $C$, and let $C_{B_n}$ denote $C$'s
$n^{\textit{th}}$ $B$-connection.
Note that $C_A$ and each of the $C_{B_n}$ are level-$k$ proto-CFLOBDDs,
and hence, by the induction hypothesis, the proposition holds for them.

We argue by contradiction:
Suppose, for the sake of argument, that the proposition does not hold for
$C$, and that $j$ is the leftmost exit vertex in $ex_C$ for which the proposition
is violated (i.e., $s(j) \neq j$).
Let $i$ be the exit vertex that appears in the $j^{th}$ position of $s$
(i.e., $s(j) = i$).
It must be that $j < i$.

Let $\alpha_j$ and $\alpha_i$ be the earliest assignments
in lexicographic order (denoted by $\prec$) that lead to
exit vertices $j$ and $i$, respectively.
Because $i$ comes before $j$ in $s$, it must be that $\alpha_i \prec \alpha_j$.

Let $\alpha^1_j$ and $\alpha^2_j$ denote the first and second halves
of $\alpha_j$, respectively;
let $\alpha^1_i$ and $\alpha^2_i$ denote the first and second halves
of $\alpha_i$, respectively.
Let $+$ denote the concatentation of assignments
(e.g., $\alpha_j = \alpha^1_j + \alpha^2_j$).

\begin{figure}
    \begin{subfigure}[t]{0.495\linewidth}
    \includegraphics[width=\linewidth]{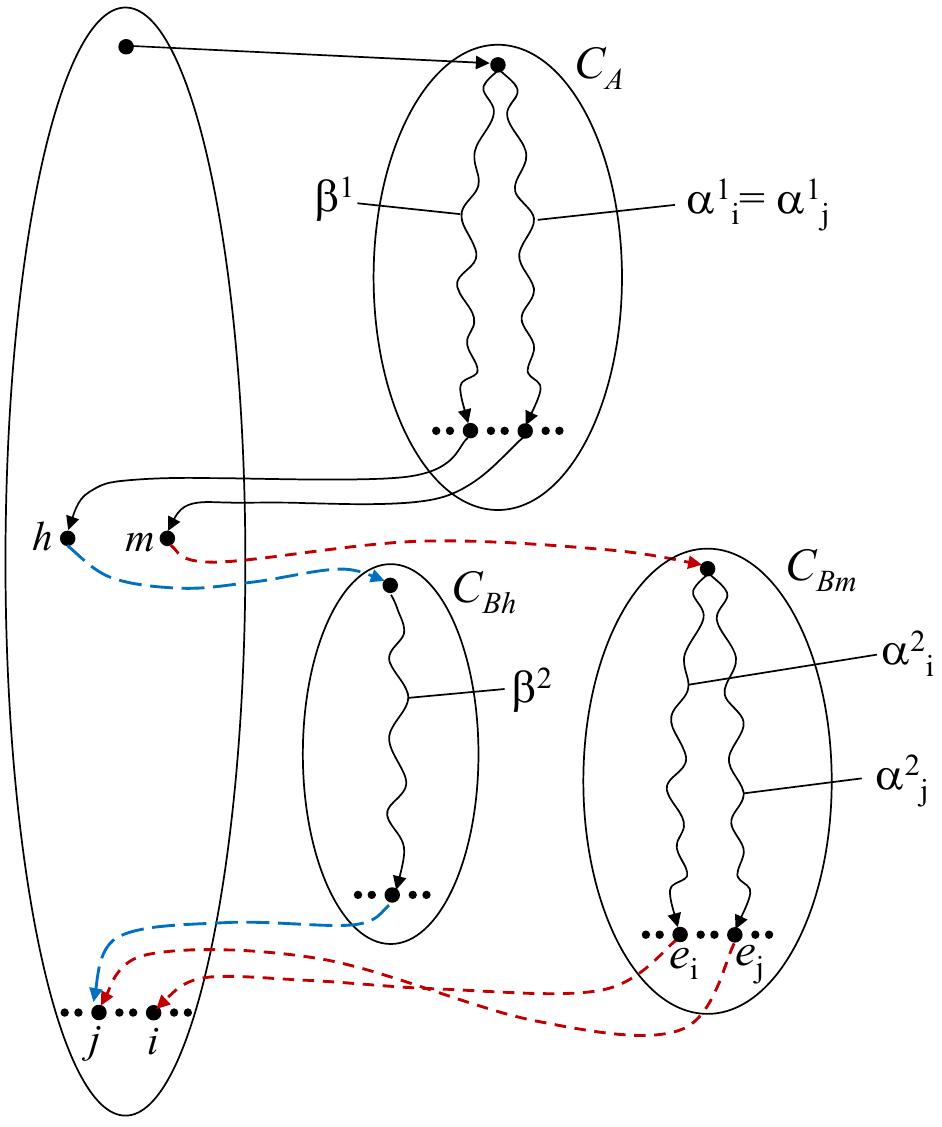}
    \caption{Case 1.A of Proposition~\ref{Prop:LexicographicOrder}}
    \vspace{2ex}
    \end{subfigure}
    \begin{subfigure}[t]{0.495\linewidth}
    \includegraphics[width=\linewidth]{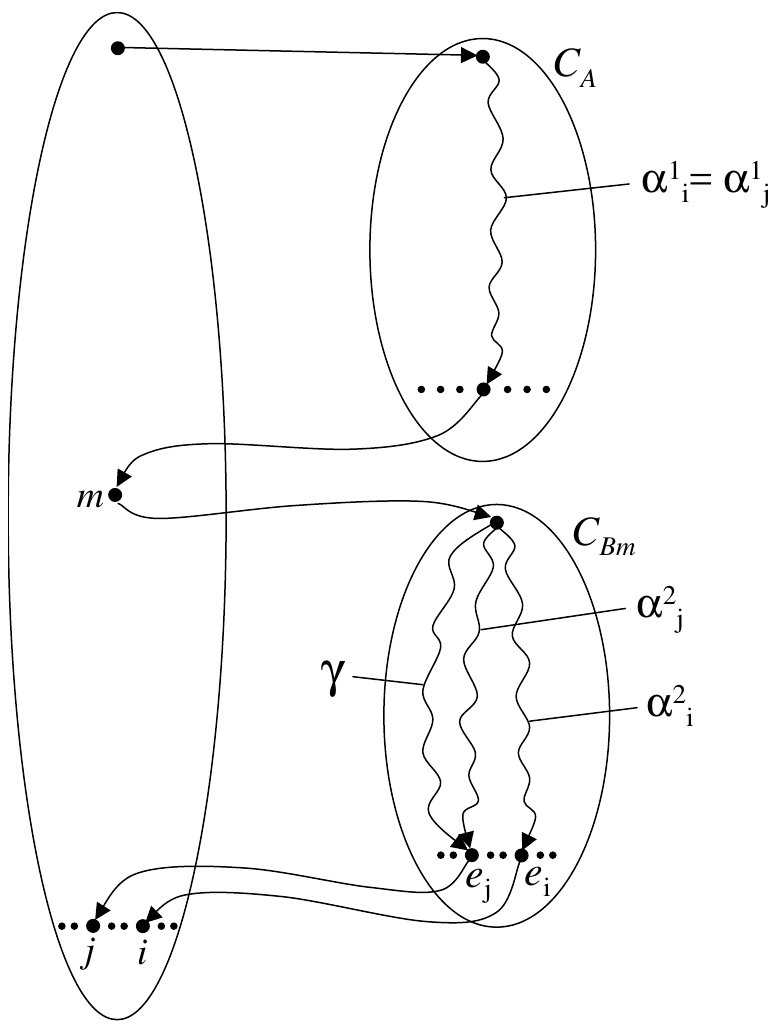}
    \caption{Case 1.B of Proposition~\ref{Prop:LexicographicOrder}.}
    \vspace{2ex}
    \end{subfigure}
    \begin{subfigure}[t]{0.495\linewidth}
    \includegraphics[width=\linewidth]{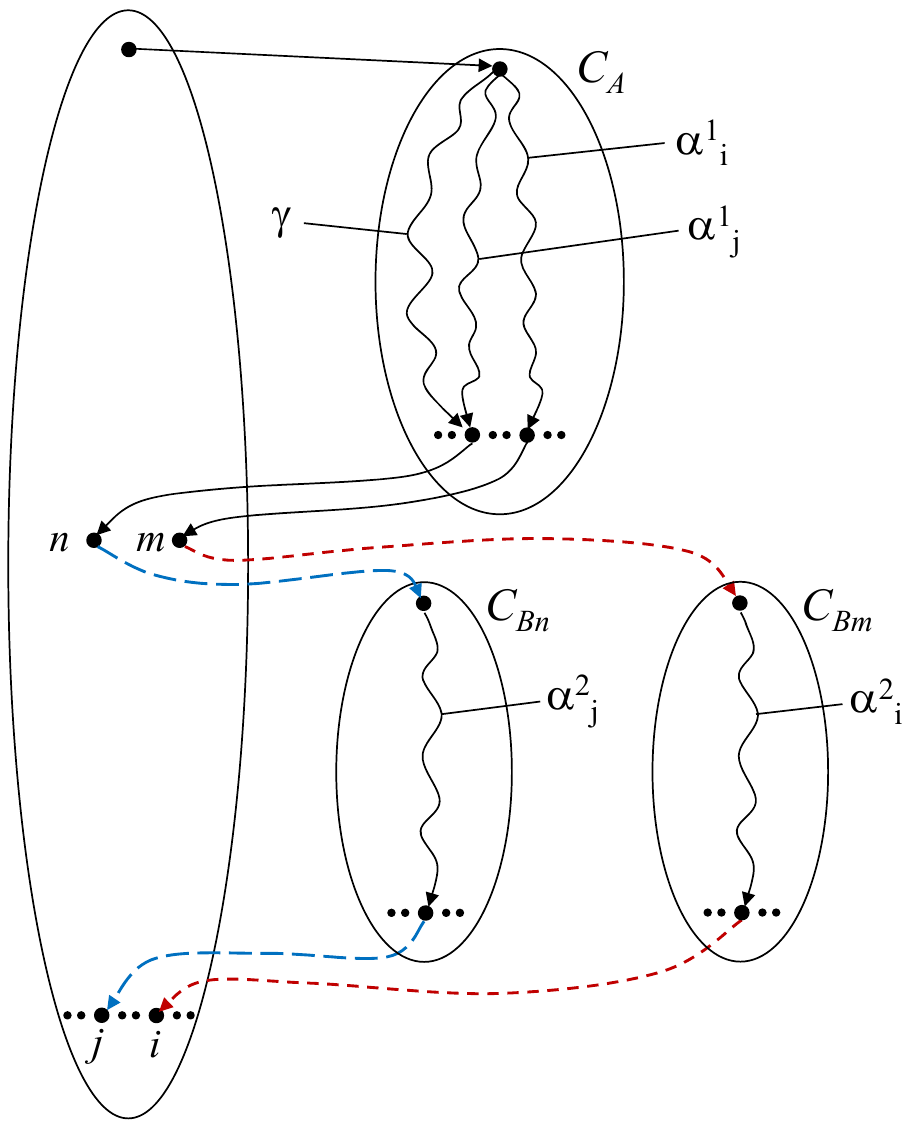}
    \caption{Case 2.A of Proposition~\ref{Prop:LexicographicOrder}.}
    \end{subfigure}
    \begin{subfigure}[t]{0.495\linewidth}
    \includegraphics[width=\linewidth]{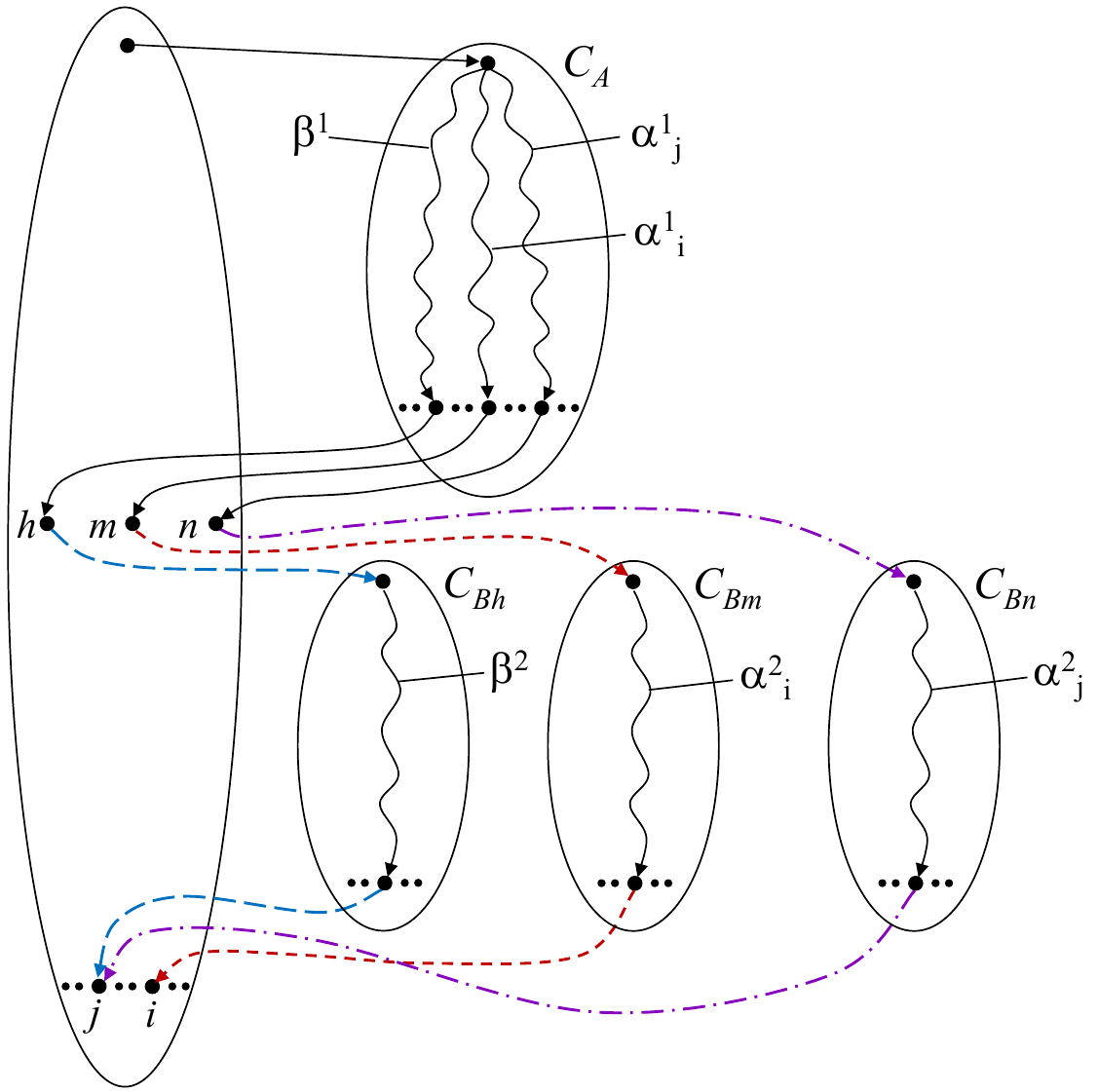}
    \caption{Case 2.B of Proposition~\ref{Prop:LexicographicOrder}.}
    \end{subfigure}
    \caption{}
    \label{Fi:LexProof}
\end{figure}


There are two cases to consider.

{\em Case 1\/}:
$\alpha^1_i = \alpha^1_j$ and $\alpha^2_i \prec \alpha^2_j$.

Because $\alpha^1_i = \alpha^1_j$, the first halves of the matched
path followed during the interpretations of assignments
$\alpha_i$ and $\alpha_j$ through $C_A$ are identical, and bring us
to some middle vertex, say $m$, of $C$;
both paths then proceed through $C_{B_m}$.
Let $e_i$ and $e_j$ be the two exit vertices of $C_{B_m}$
reached by following matched paths during the interpretations
of $\alpha^2_i$ and $\alpha^2_j$, respectively.
There are now two cases to consider:

{\em Case 1.A\/}:
Suppose that $e_i < e_j$ in $C_{B_m}$ (see \figref{LexProof}a).
In this case, the return edges $e_i \rightarrow i$
and $e_j \rightarrow j$ ``cross''.
By Structural Invariant~\ref{Inv:2b}, this can only happen if
\begin{itemize}
  \item
    There is a matched path corresponding to some assignment $\beta^1$
    through $C_A$ that leads to a middle vertex $h$, where $h < m$.
  \item
    There is a matched path from $h$ corresponding to some
    assignment $\beta^2$ through $C_{B_h}$ (where $C_{B_h}$ could be $C_{B_m}$).
  \item
    There is a return edge from the exit vertex reached by $\beta^2$ in
    $C_{B_h}$ to exit vertex $j$ of $C$.
\end{itemize}
In this case, by the induction hypothesis applied to $C_A$,
and the fact that $h < m$, it must be the case that we can choose
$\beta^1$ so that $\beta^1 \prec \alpha^1_j$.

Consequently, $\beta^1 + \beta^2 \prec \alpha^1_j + \alpha^2_j$, which
contradicts the assumption that $\alpha_j = \alpha^1_j + \alpha^2_j$
is the least assignment in lexicographic order that leads to $j$.

{\em Case 1.B\/}:
Suppose that $e_j < e_i$ in $C_{B_m}$ (see \figref{LexProof}b).
Because $\alpha^2_i \prec \alpha^2_j$, the induction hypothesis applied to
$C_{B_m}$ implies that there must exist an assignment
$\gamma \prec \alpha^2_i \prec \alpha^2_j$ that leads to $e_j$.
In this case, we have that 
$\alpha^1_j + \gamma \prec \alpha^1_j + \alpha^2_j$, which again
contradicts the assumption that $\alpha_j = \alpha^1_j + \alpha^2_j$
is the least assignment in lexicographic order that leads to $j$.

{\em Case 2\/}:
$\alpha^1_i \prec \alpha^1_j$.

Because $\alpha^1_i \prec \alpha^1_j$, the first halves of the matched
paths followed during the interpretations of assignments
$\alpha_i$ and $\alpha_j$ through $C_A$ bring us
to two different middle vertices of $C$, say $m$ and $n$, respectively.
The two paths then proceed through $C_{B_m}$ and $C_{B_n}$
(where it could be the case that $C_{B_m} = C_{B_n}$),
and return to $i$ and $j$, respectively, where $j < i$.
Again, there are two cases to consider:

{\em Case 2.A\/}:
Suppose that $n < m$ (see \figref{LexProof}c.)
The argument is similar to Case 1.B above:
By Structural Invariant~\ref{Inv:1}, $n < m$ means that the exit vertex
reached by $\alpha^1_j$ in $C_A$ comes before the exit vertex
reached by $\alpha^1_i$ in $C_A$.
By the induction hypothesis applied to $C_A$, there must exist
an assignment $\gamma \prec \alpha^1_i \prec \alpha^1_j$ that leads
to the exit vertex reached by $\alpha^1_j$ in $C_A$.
In this case, we have that
$\gamma + \alpha^2_j \prec \alpha^1_j + \alpha^2_j$, which
contradicts the assumption that $\alpha_j = \alpha^1_j + \alpha^2_j$
is the least assignment in lexicographic order that leads to $j$.

{\em Case 2.B\/}:
Suppose that $m < n$ (see \figref{LexProof}d.)
The argument is similar to Case 1.A above:
By Structural Invariant~\ref{Inv:2},
we can only have $m < n$ and $j < i$ if
\begin{itemize}
  \item
    There is a matched path corresponding to some assignment $\beta^1$
    through $C_A$ that leads to a middle vertex $h$, where $h < m$.
  \item
    There is a matched path from $h$ corresponding to some
    assignment $\beta^2$ through $C_{B_h}$ (where $C_{B_h}$ could be $C_{B_m}$
    or $C_{B_n}$).
  \item
    There is a return edge from the exit vertex reached by $\beta^2$ in
    $C_{B_h}$ to exit vertex $j$ of $C$.
\end{itemize}
In this case, by the induction hypothesis applied to $C_A$,
and the fact that $h < m < n$, it must be the case that we can choose
$\beta^1$ so that $\beta^1 \prec \alpha^1_j$.

Consequently, $\beta^1 + \beta^2 \prec \alpha^1_j + \alpha^2_j$, which
contradicts the assumption that $\alpha_j = \alpha^1_j + \alpha^2_j$
is the least assignment in lexicographic order that leads to $j$.

In each of the cases above, we are able to derive a contradiction to the
assumption that $\alpha_j$ is the least assignment in lexicographic
order that leads to $j$.
Thus, the supposition that the proposition does not hold for $C$
cannot be true.
\end{InductionStep}
\end{Proof}
\end{SPro}

\section{Proof of the Canonicalness of CFLOBDDs}
\label{Se:canonical-proof}

To show that CFLOBDDs are a canonical representation of functions over Boolean arguments, we must establish that three properties hold:
\begin{enumerate}
  \item
    \label{Obligation:1}
    Every level-$k$ CFLOBDD represents a decision tree with
    $2^{2^k}$ leaves.
  \item
    \label{Obligation:2}
    Every decision tree with $2^{2^k}$ leaves is represented
    by some level-$k$ CFLOBDD.
  \item
    \label{Obligation:3}
    No decision tree with $2^{2^k}$ leaves is represented by
    more than one level-$k$ CFLOBDD
    (up to isomorphism).
\end{enumerate}

As described earlier, following a matched path (of length $O(2^k)$)
from the level-$k$ entry vertex of a level-$k$ CFLOBDD to a final value
provides an interpretation of a Boolean assignment on $2^k$ variables.
Thus, the CFLOBDD represents a decision tree with $2^{2^k}$ leaves (and
Obligation~\ref{Obligation:1} is satisfied).

To show that Obligation~\ref{Obligation:2} holds, we describe
a recursive procedure for constructing a level-$k$ CFLOBDD from an
arbitrary decision tree with $2^{2^k}$ leaves (i.e., of height
$2^k$).
In essence, the construction shows how such a decision tree
can be folded together to form a multi-terminal CFLOBDD.

The construction makes use of a set of auxiliary tables, one for each
level, in which a unique representative for each class of equal
proto-CFLOBDDs that arises is tabulated.
We assume that the level-0 table is already seeded with a
representative fork grouping and a representative don't-care grouping.

\begin{Constr}\label{Constr:DecisionTreeToCFLOBDD}{\bf [Decision Tree to Multi-Terminal CFLOBDD]\/}
\begin{enumerate}
  \item
    \label{Construct:EquivClasses}
    The leaves of the decision tree are partitioned into some number
    of equivalence classes $e$ according to the values that label the
    leaves.  The equivalence classes are numbered 1 to $e$ according
    to the relative position of the first occurrence of a value in a
    left-to-right sweep over the leaves of the decision tree.

    \hspace{1.5ex}
    For Boolean-valued CFLOBDDs, when the procedure is applied
    at topmost level, there are at most
    two equivalence classes of leaves, for the values $F$ and $T$.
    However, in general, when the procedure is applied recursively, more
    than two equivalence classes can arise.
    
    \hspace{1.5ex}
    For the general case of multi-terminal CFLOBDDs, the number of
    equivalence classes corresponds to the number of different values
    that label leaves of the decision tree.
  \item
    \label{Construct:BaseCases}
    ({\bf Base cases\/})
    If $k = 0$ and $e = 1$, construct a CFLOBDD consisting of
    the representative don't-care grouping,
    with a value tuple that binds the exit vertex to the value
    that labels both leaves of the decision tree.

    \hspace{1.5ex}
    If $k = 0$ and $e = 2$, construct a CFLOBDD consisting of
    the representative fork grouping,
    with a value tuple that binds the two exit vertices
    to the first and second values, respectively, that label
    the leaves of the decision tree.

    \hspace{1.5ex}
    If either condition applies, return the CFLOBDD so constructed
    as the result of this invocation; otherwise, continue on to
    the next step.
  \item
    \label{Construct:Recursive1}
    Construct---via recursive applications of the procedure---$2^{2^{k-1}}$
    level-$k\textrm{--}1$ multi-terminal CFLOBDDs for the
    $2^{2^{k-1}}$ decision trees of height $2^{k-1}$ in the lower half
    of the decision tree.

    \hspace{1.5ex}
    These are then partitioned into some number $e'$ of
    equivalence classes of equal multi-terminal CFLOBDDs;
    a representative of each class is retained, and the others discarded.
    Each of the $2^{2^{k-1}}$ ``leaves'' of the upper half of the decision
    tree is labeled with the appropriate equivalence-class
    representative for the subtree of the lower half that begins
    there.  These representatives serve as the ``values'' on the
    leaves of the upper half of the decision tree when the construction
    process is applied recursively to the upper half in
    step~\ref{Construct:Recursive2}.

    \hspace{1.5ex}
    The equivalence-class representatives are also numbered 1 to $e'$
    according to the relative position of their first occurrence in a
    left-to-right sweep over the leaves of the upper half of the
    decision tree.
  \item
    \label{Construct:Recursive2}
    Construct---via a recursive application of the procedure---a
    level-$k\textrm{--}1$ multi-terminal CFLOBDD for the upper half of the decision
    tree.
  \item
    \label{Construct:Grouping}
    Construct a level-$k$ multi-terminal proto-CFLOBDD from the
    level-$k\textrm{--}1$ multi-terminal CFLOBDDs created
    in steps~\ref{Construct:Recursive1} and~\ref{Construct:Recursive2}.
    The level-$k$ grouping is constructed as follows:
    \begin{enumerate}
      \item
        \label{Construct:NewAConnection}
        The $A$-connection points to the proto-CFLOBDD of the object
        constructed in step~\ref{Construct:Recursive2}.
      \item
        \label{Construct:NewMiddleVertices}
        The middle vertices correspond to the equivalence classes
        formed in step~\ref{Construct:Recursive1}, in the order
        $1 \ldots e'$.
      \item
        \label{Construct:NewAReturnTuple}
        The $A$-connection return tuple is the identity map back to
        the middle vertices (i.e., the tuple $[1..e']$).
      \item
        \label{Construct:NewBConnections}
        The $B$-connections point to the proto-CFLOBDDs of the $e'$
        equivalence-class representatives constructed in
        step~\ref{Construct:Recursive1}, in the order $1 \ldots e'$.
      \item
        \label{Construct:NewExitVertices}
        The exit vertices correspond to the initial equivalence
        classes described in step~\ref{Construct:EquivClasses}, in the
        order $1 \ldots e$.
      \item
        \label{Construct:NewBReturnTuples}
        The $B$-connection return tuples connect the exit vertices
        of the highest-level groupings of the equivalence-class
        representatives retained from step~\ref{Construct:Recursive1}
        to the exit vertices created in step~\ref{Construct:NewExitVertices}.
        In each of the equivalence-class representatives retained from
        step~\ref{Construct:Recursive1}, the value tuple associates
        each exit vertex $x$ with some value $v$, where $1 \leq v \leq e$;
        $x$ is now connected to the exit vertex created in
        step~\ref{Construct:NewExitVertices} that is associated with
        the same value $v$.
      \item
        \label{Construct:CheckForDuplicates}
        Consult a table of all previously constructed level-$k$
        groupings to determine whether the grouping constructed by
        steps~\ref{Construct:NewAConnection}--\ref{Construct:NewBReturnTuples}
        duplicate a previously constructed grouping.  If
        so, discard the present grouping and switch to the previously
        constructed one; if not, enter the present grouping into the
        table.
    \end{enumerate}
  \item
    \label{Construct:ValueTuple}
    Return a multi-terminal CFLOBDD created from the proto-CFLOBDD constructed in step~\ref{Construct:Grouping} by attaching a value tuple that connects (in order) the exit vertices of the proto-CFLOBDD to the $e$ values from step~\ref{Construct:EquivClasses}.
\end{enumerate}
\end{Constr}

\begin{figure}
    \centering
    \begin{subfigure}[t]{0.495\linewidth}
    \includegraphics[width=\linewidth]{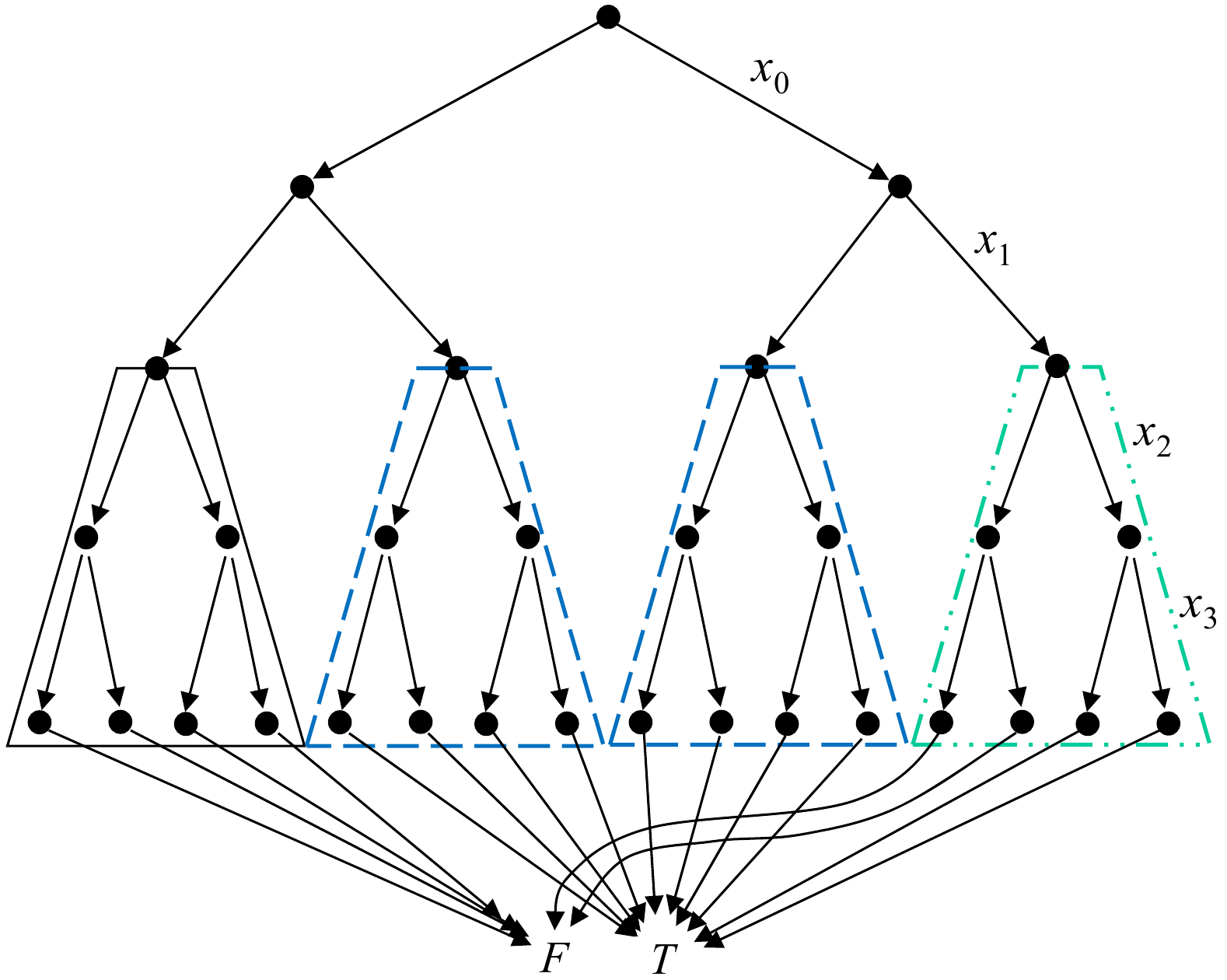}
    \caption{Decision tree}
    \vspace{4ex}
    \end{subfigure}
    \begin{subfigure}[t]{0.495\linewidth}
    \includegraphics[width=\linewidth]{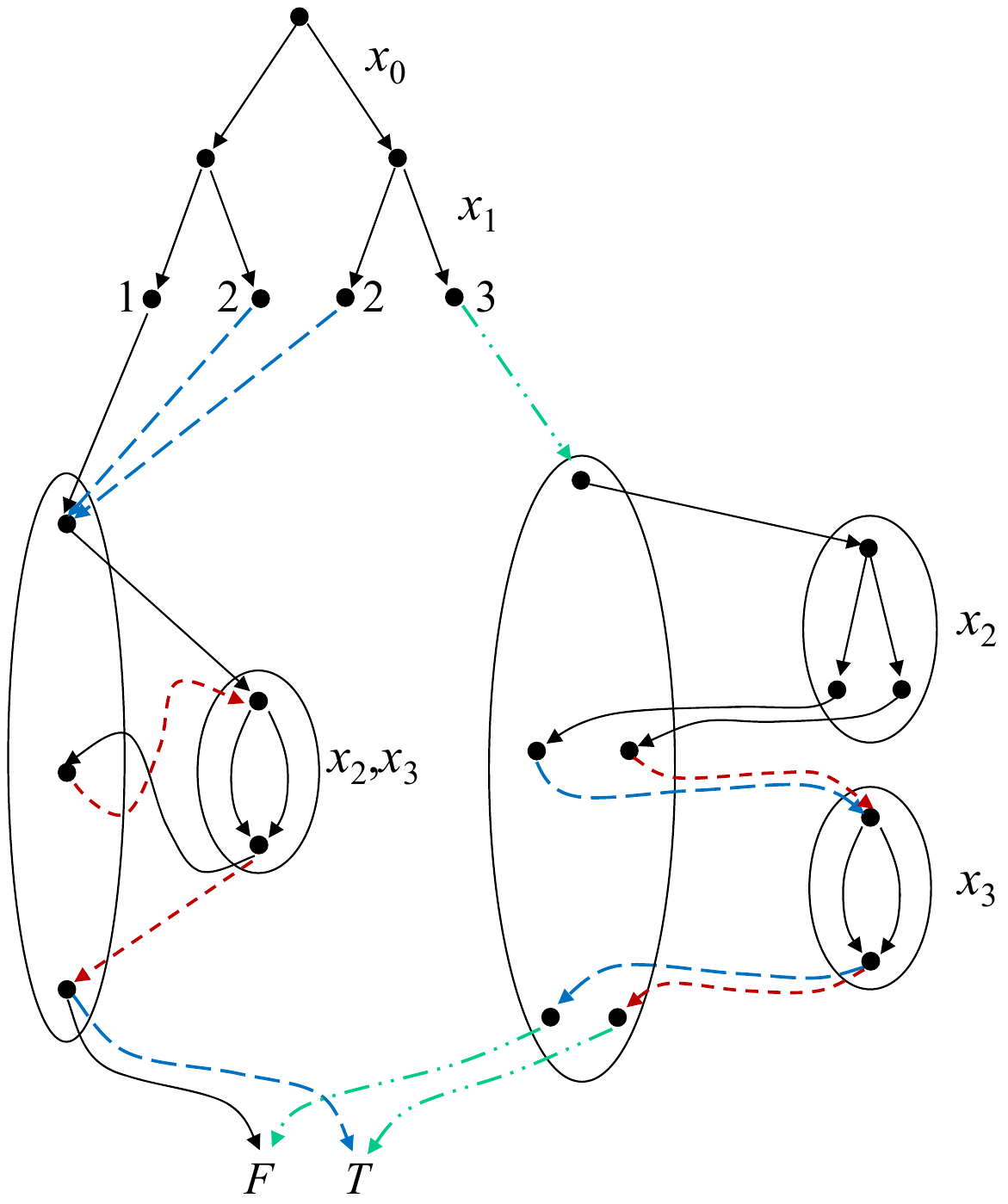}
    \caption{\protect \raggedright 
      Hybrid of decision tree for $x_0$ and $x_1$, and CFLOBDDs for $x_2$ and $x_3$.
      The solid, dashed, and dashed-double-dotted edges from the four vertices labeled 1, 2, 2, and 3, respectively, correspond to the solid, dashed, and dashed-double-dotted trapezoids in (a).
    }
    \vspace{4ex}
    \end{subfigure}
    \begin{subfigure}[t]{\linewidth}
    \centering
    \includegraphics[width=0.5\linewidth]{figures/complicated-cropped.pdf}
    \caption{CFLOBDD (repeated from \figref{MultipleMiddleVertices}).
      For clarity, some of the level-0 groupings have been
      duplicated.
    }
    \end{subfigure}
    \caption{Representations of the Boolean function
         $\lambda {x_0}{x_1}{x_2}{x_3} . (x_0 \xor x_1) \lor (x_0 \land x_1 \land x_2)$.
    }
    \label{Fi:Complicated}
\end{figure}

\figref{Complicated}a shows the decision tree for the function
$\lambda {x_0}{x_1}{x_2}{x_3} . (x_0 \xor x_1) \lor (x_0 \land x_1 \land x_2)$.
\figref{Complicated}b shows the state of things after
step~\ref{Construct:Recursive1} of Construction \ref{Constr:DecisionTreeToCFLOBDD}.
Note that even though the level-1 CFLOBDDs for the first
three leaves of the top half of the decision tree have equal
proto-CFLOBDDs,\footnote{
  The equality of the proto-CFLOBDDs is detected in
  step~\ref{Construct:CheckForDuplicates}.
}
the leftmost proto-CFLOBDD maps its exit vertex to $F$, whereas the
exit vertex is mapped to $T$ in the second and third proto-CFLOBDDs.
Thus, in this case, the recursive call for the upper half of the
decision tree (step~\ref{Construct:Recursive2}) involves three
equivalence classes of values.

It is not hard to see that the structures created by
Construction \ref{Constr:DecisionTreeToCFLOBDD} obey the structural
invariants that are required of CFLOBDDs:
\begin{itemize}
  \item
    Structural Invariant~\ref{Inv:1} holds because the $A$-connection
    return tuple created in step~\ref{Construct:NewAReturnTuple} of
    Construction \ref{Constr:DecisionTreeToCFLOBDD} is the identity map.
  \item
    Structural Invariant~\ref{Inv:2} holds because in
    steps~\ref{Construct:EquivClasses} and~\ref{Construct:Recursive1}
    of Construction \ref{Constr:DecisionTreeToCFLOBDD}, the equivalence classes
    are numbered in increasing order according to the relative position
    of a value's first occurrence in a left-to-right sweep.
    In particular, this order is preserved in the exit vertices of each grouping constructed during an invocation of Construction \ref{Constr:DecisionTreeToCFLOBDD} (cf.~step~\ref{Construct:NewBReturnTuples}), which ensures that the ``compact extension'' property of Structural Invariant~\ref{Inv:2b} holds at each level of recursion in Construction \ref{Constr:DecisionTreeToCFLOBDD}.
  \item
    Structural Invariant~\ref{Inv:3} holds because
    Construction \ref{Constr:DecisionTreeToCFLOBDD} reuses the
    representative don't-care grouping and the representative fork grouping
    in step~\ref{Construct:BaseCases}, and checks for the construction
    of duplicate groupings---and hence duplicate proto-CFLOBDDs---in
    step~\ref{Construct:CheckForDuplicates}.
  \item
    Structural Invariant~\ref{Inv:4} holds because of
    steps~\ref{Construct:Recursive1}, \ref{Construct:NewBConnections},
    and~\ref {Construct:NewBReturnTuples}.
    On recursive calls to Construction \ref{Constr:DecisionTreeToCFLOBDD},
    step~\ref{Construct:Recursive1} partitions the CFLOBDDs constructed
    for the lower half of the decision tree into equivalence classes
    of CFLOBDD values (i.e., taking into account both the proto-CFLOBDDs
    and the value tuples associated with their exit vertices).
    Therefore, in steps~\ref{Construct:NewBConnections}
    and~\ref{Construct:NewBReturnTuples}, duplicate
    $B$-connection/return-tuple pairs can never arise.
  \item
    Structural Invariant~\ref{Inv:5} holds because step~\ref{Construct:ValueTuple} uses the proto-CFLOBDD constructed in step~\ref{Construct:Grouping}.
  \item
    Structural Invariant~\ref{Inv:6}
    holds because step~\ref{Construct:EquivClasses} of Construction \ref{Constr:DecisionTreeToCFLOBDD} constructs equivalence classes of values (ordered in increasing order according to the relative position of a value's first occurrence in a left-to-right sweep over the leaves of the decision tree).
\end{itemize}

Moreover, Construction \ref{Constr:DecisionTreeToCFLOBDD} preserves interpretation under assignments:
Suppose that $C_T$ is the level-$k$ CFLOBDD constructed
by Construction \ref{Constr:DecisionTreeToCFLOBDD} for decision tree $T$;
it is easy to show by induction on $k$ that for every
assignment $\alpha$ on the $2^k$ Boolean variables $x_0, \ldots, x_{2^k-1}$,
the value obtained from $C_T$ by following the corresponding matched
path from the entry vertex of $C_T$'s highest-level grouping
is the same as the value obtained for $\alpha$ from $T$.
(The first half of $\alpha$ is used to follow a path through the
$A$-connection of $C_T$, which was constructed from the top half of
$T$.
The second half of $\alpha$ is used to follow a path through one of the
$B$-connections of $C_T$, which was constructed from an equivalence
class of bottom-half subtrees of $T$;
that equivalence class includes the subtree rooted at the vertex
of $T$ that is reached by following the first half of $\alpha$.)
Thus, every decision tree with $2^{2^k}$ leaves is represented by some
level-$k$ CFLOBDD in which meaning (interpretation under assignments)
has been preserved;
consequently, Obligation~\ref{Obligation:2} is satisfied.

We now come to Obligation~\ref{Obligation:3}
(no decision tree with $2^{2^k}$ leaves is represented by
more than one level-$k$ CFLOBDD).
The way we prove this property is to define an unfolding process,
called $\Unfold$, that starts with a multi-terminal
CFLOBDD and works in the opposite direction to
Construction \ref{Constr:DecisionTreeToCFLOBDD} to construct a decision tree;
that is, $\Unfold$ (recursively) unfolds the $A$-connection, and then
(recursively) unfolds each of the $B$-connections.
For instance, for the example shown in \figref{Complicated},
$\Unfold$ would proceed from
\figref{Complicated}c, to \figref{Complicated}b, and
then to the decision tree for the function $\lambda
{x_0}{x_1}{x_2}{x_3} . (x_0 \xor x_1) \lor (x_0 \land x_1 \land x_2)$
shown in \figref{Complicated}a.

$\Unfold$ also preserves interpretation under assignments:
Suppose that $T_C$ is the decision tree constructed
by $\Unfold$ for level-$k$ CFLOBDD $C$;
it is easy to show by induction on $k$ that for every
assignment $\alpha$ on the $2^k$ Boolean variables $x_0, \ldots, x_{2^k-1}$,
the value obtained from $C$ by following the corresponding matched
path from the entry vertex of $C$'s highest-level grouping
is the same as the value obtained for $\alpha$ from $T_C$.
(The first half of $\alpha$ is used to follow a path through the
$A$-connection of $C$, which $\Unfold$ unfolds into the top half of
$T_C$.
The second half of $\alpha$ is used to follow a path through one of the
$B$-connections of $C$, which $\Unfold$ unfolds into one or more
instances of bottom-half subtrees of $T_C$;
that set of bottom-half subtrees includes the subtree rooted at the vertex
of $T$ that is reached by following the first half of $\alpha$.)

\Omit{
\begin{figure*}[tb!]
\begin{center}
\begin{tabular}{c@{\hspace{.1in}}c@{\hspace{.1in}}c}
    \includegraphics[height=.975in]{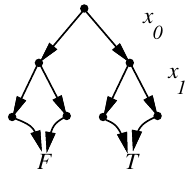}
  & \includegraphics[height=.975in]{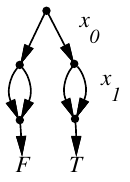}
  & \includegraphics[height=.975in]{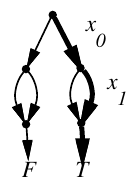}
 \\
    (a) Decision tree for $\lambda {x_0}{x_1}.{x_0}$
  & (b) Folded form (OBDD)
  & (c) Path corresponding to the
 \\ & & assignment $[{x_0} \mapsto T, {x_1} \mapsto T]$
 \\ & &
 \\
    \includegraphics[height=1.75in]{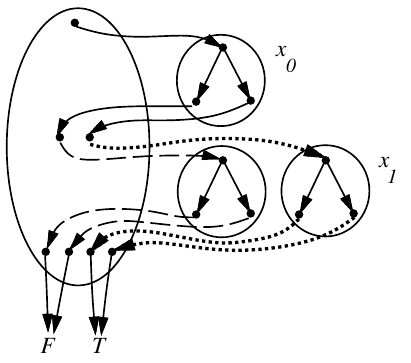}
  & \includegraphics[height=1.75in]{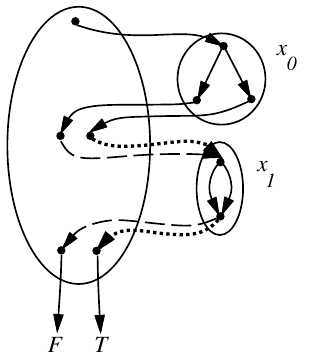}
  & \includegraphics[height=1.75in]{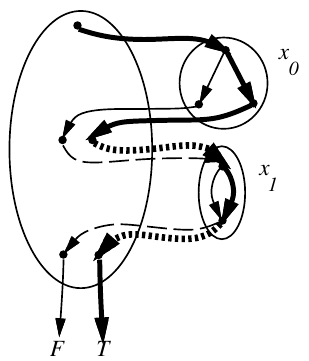}
 \\ & &
 \\
    d Fully expanded form
  & (e) Folded form (CFLOBDD)
  & (f) Path corresponding to the
 \\ & & assignment $[{x_0} \mapsto T, {x_1} \mapsto T]$
\end{tabular}
\end{center}
\caption{}
\label{Fi:projection0Figure}
\end{figure*}
}

\begin{figure*}[tb!]
\begin{center}
\begin{tabular}{c@{\hspace{.1in}}c@{\hspace{.1in}}c}
    \includegraphics[height=.975in]{PatentApplication/figures/cropped-projection0Tree.idraw.pdf}
  &  \includegraphics[height=1.75in]{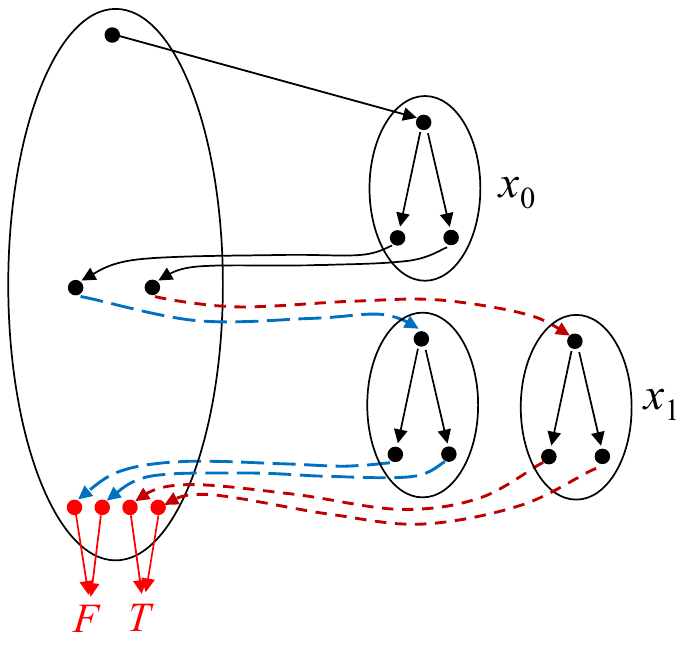}
  & \includegraphics[height=1.75in]{figures/projection0-cropped.pdf}
  \\
  {\small (a)} & {\small (b)} & {\small (c)}
\end{tabular}
\end{center}
\caption{(a) Decision tree for $\lambda {x_0}{x_1}.{x_0}$;
  (b) fully expanded form of the CFLOBDD;
  (c) CFLOBDD.
}
\label{Fi:projection0Figure}
\end{figure*}

\begin{figure*}[tb!]
\begin{center}
\begin{tabular}[t]{|c@{\hspace{.05in}}|@{\hspace{.05in}}c@{\hspace{.05in}}|@{\hspace{.05in}}c@{\hspace{.05in}}|@{\hspace{.05in}}c@{\hspace{.05in}}|@{\hspace{.05in}}c@{\hspace{.05in}}|}
    \hline
    \multicolumn{5}{|@{\hspace{.05in}}c@{\hspace{.05in}}|}{
      Trace for entire tree
    }
 \\
    \hline
    Tree & The lower-half trees & Hybrid & Upper-half tree & CFLOBDD
 \\
    \hline
    & & & &
 \\
    \includegraphics[scale=0.4,valign=m]{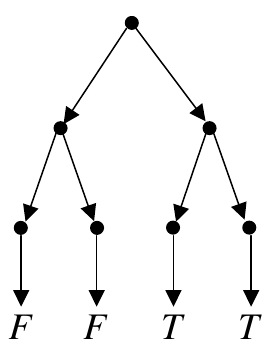}
  & \adjustbox{valign=t}{
      \begin{tabular}[t]{@{\hspace{0.0in}}c@{\hspace{.05in}}c@{\hspace{0.0in}}}
         \begin{tabular}{|@{\hspace{.02in}}c@{\hspace{.02in}}|@{\hspace{.02in}}c@{\hspace{.02in}}|}
             \hline
             \includegraphics[scale=0.4]{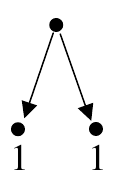}
           & \includegraphics[scale=0.4]{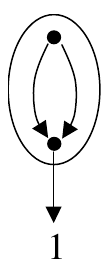}
          \\ \hline
         \end{tabular}
      & 
         \begin{tabular}{|@{\hspace{.02in}}c@{\hspace{.02in}}|@{\hspace{.02in}}c@{\hspace{.02in}}|}
             \hline
             \includegraphics[scale=0.4]{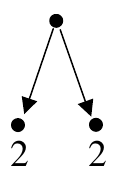}
           & \includegraphics[scale=0.4]{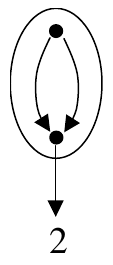}
          \\ \hline
         \end{tabular}
      \end{tabular}
    }
  & \includegraphics[scale=0.25,valign=m]{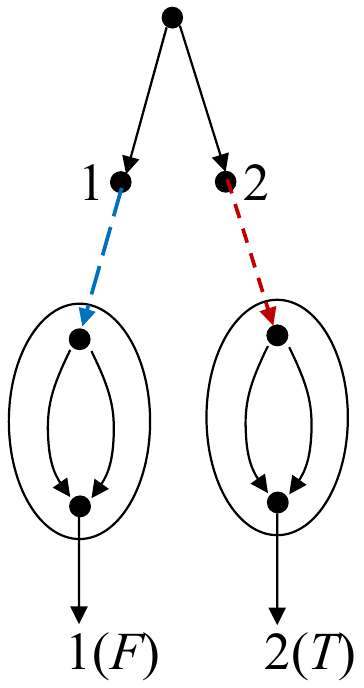}
  & \adjustbox{valign=b}{
      \begin{tabular}{|@{\hspace{.02in}}c@{\hspace{.02in}}|@{\hspace{.02in}}c@{\hspace{.02in}}|}
            \hline
            \includegraphics[scale=0.4]{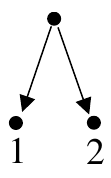}
          & \includegraphics[scale=0.4]{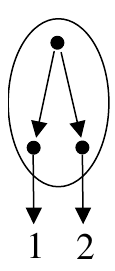}
         \\ \hline
      \end{tabular}
    }
  & \includegraphics[scale=0.36,valign=m]{figures/projection0-cropped.pdf}
  \\* [0.8in]
  \hline
\end{tabular}
\end{center}
\caption{\protect \raggedright 
The $\Fold$ trace generated by the application of Construction \ref{Constr:DecisionTreeToCFLOBDD} to the decision tree shown in \figref{projection0Figure}a to create the CFLOBDD shown in \figref{projection0Figure}c.
}
\label{Fi:FoldTrace}
\end{figure*}

\begin{figure*}[tb!]
\begin{center}
\begin{tabular}[t]{|@{\hspace{.05in}}c@{\hspace{.05in}}|@{\hspace{.05in}}c@{\hspace{.05in}}|@{\hspace{.05in}}c@{\hspace{.05in}}|@{\hspace{.05in}}c@{\hspace{.05in}}|c@{\hspace{.05in}}|}
    \hline
    \multicolumn{5}{|@{\hspace{.05in}}c@{\hspace{.05in}}|}{
      Trace for entire CFLOBDD
    }
 \\
    \hline
    CFLOBDD & $A$-connection & Hybrid & Traces for the $B$-connections & Tree
 \\
    \hline
    & & & &
 \\
    \includegraphics[scale=0.36,valign=m]{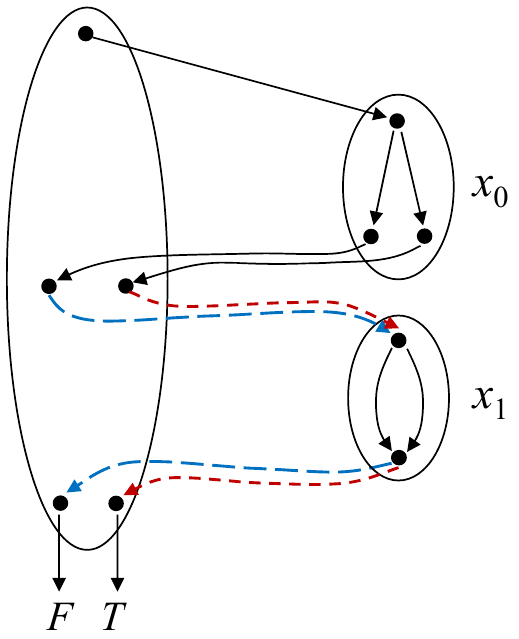}
  & \adjustbox{valign=b}{
      \begin{tabular}{|@{\hspace{.02in}}c@{\hspace{.02in}}|@{\hspace{.02in}}c@{\hspace{.02in}}|}
              \hline
              \includegraphics[scale=0.4]{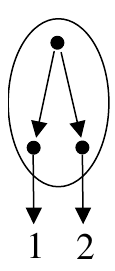}
            & \includegraphics[scale=0.4]{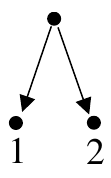}
           \\ \hline
      \end{tabular}
    }
  & \includegraphics[scale=0.25,valign=m]{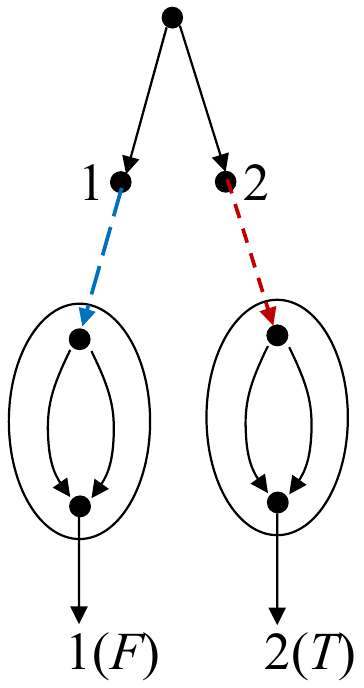}
  & \adjustbox{valign=t}{
      \begin{tabular}{@{\hspace{.02in}}c@{\hspace{.02in}}c@{\hspace{.02in}}}
         \begin{tabular}{|@{\hspace{.02in}}c@{\hspace{.02in}}|@{\hspace{.02in}}c@{\hspace{.02in}}|}
             \hline
             \includegraphics[scale=0.4]{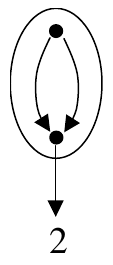}
           & \includegraphics[scale=0.4]{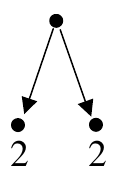}
          \\ \hline
         \end{tabular}
      & 
         \begin{tabular}{|@{\hspace{.02in}}c@{\hspace{.02in}}|@{\hspace{.02in}}c@{\hspace{.02in}}|}
             \hline
             \includegraphics[scale=0.4]{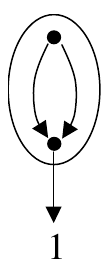}
           & \includegraphics[scale=0.4]{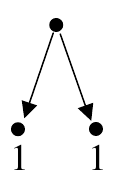}
          \\ \hline
         \end{tabular}
      \end{tabular}
    } 
  & \includegraphics[scale=0.4,valign=m]{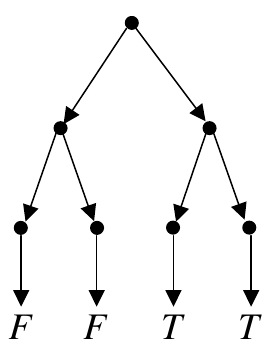}
 \\* [.4in]
    \hline
\end{tabular}
\end{center}
\caption{\protect \raggedright 
The $\Unfold$ trace generated by the application of $\Unfold$ to the CFLOBDD shown in \figref{projection0Figure}c to create the decision tree shown in \figref{projection0Figure}a.
}
\label{Fi:UnfoldTrace}
\end{figure*}

Obligation~\ref{Obligation:3} is satisfied if we can show that, for every CFLOBDD $C$, Construction \ref{Constr:DecisionTreeToCFLOBDD} applied to the decision tree produced by $\Unfold(C)$
yields a CFLOBDD that is isomorphic to $C$.
To establish that this property holds, we will define two kinds of \emph{traces}:
\begin{itemize}
  \item
    A {\em $\Fold$ trace\/} records the steps of
    Construction \ref{Constr:DecisionTreeToCFLOBDD}:
    \begin{itemize}
      \item
        At step~\ref{Construct:EquivClasses} of Construction \ref{Constr:DecisionTreeToCFLOBDD},
        the decision tree is appended to the trace.
      \item
        At the end of step~\ref{Construct:BaseCases}
        (if either of the conditions listed in step~\ref{Construct:BaseCases} holds),
        the level-$0$ CFLOBDD being returned is appended to the trace
        (and Construction \ref{Constr:DecisionTreeToCFLOBDD} returns).
      \item
        During step~\ref{Construct:Recursive1}, the trace
        is extended according to the actions carried out by
        the folding process as it is applied recursively to each of the
        lower-half decision trees.
        (For purposes of settling Obligation~\ref{Obligation:3},
        we will assume that the lower-half decision trees are
        processed by Construction \ref{Constr:DecisionTreeToCFLOBDD}
        in {\em left-to-right\/} order.)
      \item
        At the end of step~\ref{Construct:Recursive1},
        a hybrid decision-tree/CFLOBDD object (\`{a} la
        \figref{Complicated}b) is appended to the trace.
      \item
        During step~\ref{Construct:Recursive2}, the trace
        is extended according to the actions carried out by
        the folding process as it is applied recursively
        to the upper half of the decision tree.
      \item
        At the end of step~\ref{Construct:ValueTuple}, the
        CFLOBDD being returned is appended to the trace.
    \end{itemize}
    For instance, \figref{FoldTrace} shows the $\Fold$ trace generated
    by the application of Construction \ref{Constr:DecisionTreeToCFLOBDD}
    to the decision tree shown in \figref{projection0Figure}a
    to create the CFLOBDD shown in \figref{projection0Figure}c.
  \item
    An {\em $\Unfold$ trace\/} records the steps of $\Unfold(C)$:   
    \begin{itemize}
      \item
        CFLOBDD $C$ is appended to the trace.
      \item
        If $C$ is a level-$0$ CFLOBDD, then
        a binary tree of height $1$---with the leaves
        labeled according to $C$'s value tuple---is
        appended to the trace (and the $\Unfold$ algorithm returns).
      \item
        The trace is extended according to the actions carried out by
        $\Unfold$ as it is applied recursively to the $A$-connection of $C$.
      \item
        A hybrid decision-tree/CFLOBDD object (\`{a} la
        \figref{Complicated}b) is appended to the trace.
      \item
        The trace is extended according to the actions carried out by
        $\Unfold$ as it is applied recursively to instances of $B$-connections
        of $C$.
        (For purposes of settling Obligation~\ref{Obligation:3},
        we will assume that $\Unfold$ processes a separate instance
        of a $B$-connection for each leaf of the hybrid object's
        upper-half decision tree, and that the $B$-connections
        are processed in {\em right-to-left\/} order of the
        upper-half decision tree's leaves.)
      \item
        Finally, the decision tree returned by $\Unfold$ is appended to
        the trace.
    \end{itemize}
    For instance, \figref{UnfoldTrace} shows the $\Unfold$ trace generated by the application of $\Unfold$ to the CFLOBDD shown in \figref{projection0Figure}c to create the decision tree shown in \figref{projection0Figure}a.
\end{itemize}

Note how the $\Unfold$ trace shown in \figref{UnfoldTrace}
is the reversal of the $\Fold$ trace shown in \figref{FoldTrace}.
We now argue that this property holds generally.
(Technically, the argument given below in Proposition~\ref{Prop:UnfoldFoldReversability} shows that each element of an $\Unfold$ trace is
\emph{isomorphic} to the corresponding object in the $\Fold$ trace, which suffices to imply that that Obligation~\ref{Obligation:3} is satisfied, in the sense that a decision tree is represented by exactly one isomorphism class of CFLOBDDs.)

\begin{Pro}\label{Prop:UnfoldFoldReversability}
Suppose that $C$ is a multi-terminal CFLOBDD, and that $\Unfold(C)$ results in
$\Unfold$ trace $UT$ and decision tree $T_0$.
Let $C'$ be the multi-terminal CFLOBDD produced by applying
Construction \ref{Constr:DecisionTreeToCFLOBDD} to $T_0$,
and $FT$ be the $\Fold$ trace produced during this process.
Then
\begin{description}
  \item[{\it (i)}]
    $FT$ is the reversal of $UT$.
  \item[{\it (ii)}]
    $C$ and $C'$ are isomorphic.
\end{description}
\vspace{.1in}

\begin{Proof}
Because $C'$ appears at the end of $FT$, and $C$ appears at the beginning of $UT$,
clause~(i) implies~(ii).
We show clause (i) by the following inductive argument:
\newline

\begin{BaseCase}
    The proposition is trivially true of level-0 CFLOBDDs.
    Given any pair of values $v_1$ and $v_2$ (such as $F$ and $T$),
    there are exactly four possible level-0 CFLOBDDs:
    two constructed using a don't-care grouping---one in which the
    exit vertex is mapped to $v_1$, and one in which it is mapped to
    $v_2$---and two constructed using a fork grouping---one
    in which the two exit vertices are
    mapped to $v_1$ and $v_2$, respectively, and one in which
    they are mapped to $v_2$ and $v_1$, respectively.
    These unfold to the four decision trees that have $2^{2^0} = 2$
    leaves and leaf-labels drawn from $\{v_1,v_2\}$,
    and the application of Construction \ref{Constr:DecisionTreeToCFLOBDD}
    to these decision trees yields the same level-0 CFLOBDD that
    we started with.
    (See step~\ref{Construct:BaseCases} of
    Construction \ref{Constr:DecisionTreeToCFLOBDD}.)
    Consequently, the $\Fold$ trace $FT$ and the $\Unfold$ trace $UT$
    are reversals of each other.
\end{BaseCase}
\newline

\begin{InductionStep}
The induction hypothesis is that the proposition holds for every level-$k$ multi-terminal CFLOBDD.
We need to argue that the proposition extends to level-$k\textrm{+}1$ multi-terminal CFLOBDDs.

First, note that the induction hypothesis implies that each decision tree with $2^{2^k}$ leaves is represented by exactly one
level-$k$ CFLOBDD isomorphism class.
We will refer to this as the {\em corollary to the induction hypothesis\/}.

$\Unfold$ trace $UT$ can be divided into five segments:
\begin{description}
  \item[{\rm (u1)\/}]
    $C$ itself
  \item[{\rm (u2)\/}]
    the $\Unfold$ trace for $C$'s $A$-connection
  \item[{\rm (u3)\/}]
    a hybrid decision-tree/CFLOBDD object (call this object $D$)
  \item[{\rm (u4)\/}]
    the $\Unfold$ trace for $C$'s $B$-connections
  \item[{\rm (u5)\/}]
    $T_0$.
\end{description}
$\Fold$ trace $FT$ can also be divided into five segments:
\begin{description}
  \item[{\rm (f1)\/}]
    $T_0$
  \item[{\rm (f2)\/}]
    the $\Fold$ trace for $T_0$'s lower-half trees
  \item[{\rm (f3)\/}]
    a hybrid decision-tree/CFLOBDD object (call this object $D'$)
  \item[{\rm (f4)\/}]
    the $\Fold$ trace for $T_0$'s upper-half
  \item[{\rm (f5)\/}]
    $C'$.
\end{description}
Because both (f1) and (u5) are $T_0$, (u5) is obviously equal to (f1).
Our goal, therefore, is to show that
\begin{itemize}
  \item (u2) is the reversal of (f4);
  \item (u3) is equal to (f3);
  \item (u4) is the reversal of (f2); and
  \item
      (u1) is equal to (f5).
\end{itemize}

\begin{description}
  \item[{\bf (u3) is equal to (f3)\/}]
    Consider the hybrid decision-tree/CFLOBDD object $D$ obtained
    after $\Unfold$ has finished unfolding $C$'s $A$-connection.\footnote{
      The $A$-connection is actually a proto-CFLOBDD, whereas
      $\Unfold$ works on multi-terminal CFLOBDDs.  However,
      the $A$-connection return tuple (with the indices of the
      middle vertices as the value space) serves as the value tuple
      whenever we wish to consider the $A$-connection as a
      multi-terminal CFLOBDD.
    }
    The upper part of $D$ (the decision-tree part) came from
    the recursive invocation of $\Unfold$, which produced
    a decision tree for the first half of the Boolean variables,
    in which each leaf is labeled with the index of a middle vertex
    from the level-$k\textrm{+}1$ grouping of $C$ (e.g., see \figref{Complicated}b).

    \hspace{1.5ex}
    As a consequence of \propref{LexicographicOrder}, together with the fact that $\Unfold$ preserves interpretation under assignments, the relative position of the first occurrence of a label in a left-to-right sweep over the leaves of this decision tree reflects the order of the level-$k\textrm{+}1$ grouping's middle vertices.\footnote{
      This step is where the argument would break down if we attempt to apply the same argument to \figref{StructuralInvariantsIllustrated}a:
      In that case, the labels on the leaves of $D$, in left-to-right order, would be 2 and 1---whereas the sequence of middle vertices in \figref{StructuralInvariantsIllustrated}a is [1,2].
    }
    However, each middle vertex has an associated $B$-connection, and by Structural Invariants~\ref{Inv:2}, \ref{Inv:4}, and~\ref{Inv:6}, the middle vertices can be thought of as representatives for a set of pairwise non-equal CFLOBDDs (that themselves represent lower-half decision trees).

    \hspace{1.5ex}
    $\Fold$ trace $FT$ also has a hybrid decision-tree/CFLOBDD object,
    namely $D'$.  The crucial point is that
    the action of partitioning $T_0$'s lower-half CFLOBDDs that is
    carried out in step~\ref{Construct:Recursive1} of
    Construction \ref{Constr:DecisionTreeToCFLOBDD} also results in a
    labeling of each leaf of the upper-half's decision tree with a
    representative of an equivalence class of CFLOBDDs that represent
    the lower half of the decision tree starting at that point.

    \hspace{1.5ex}
    By the corollary to the induction hypothesis,
    the $2^{2^k}$ bottom-half trees of $T_0$ are represented uniquely
    (up to isomorphism)
    by the respective CFLOBDDs in $D'$.
    Similarly, by the corollary to the induction hypothesis, the $2^{2^k}$ CFLOBDDs used as labels in $D$
    represent uniquely (up to isomorphism)
    the respective bottom-half trees of $T_0$.
    Thus, the labelings on $D$ and $D'$ must be
    isomorphic.
  \item[{\bf (u2) is the reversal of (f4); (u4) is the reversal of (f2)\/}]
    Given the observation that
    $D$ and $D'$ are isomorphic, these properties
    follow in a straightforward fashion from the inductive hypothesis (applied to the $A$-connection and the $B$-connections of $C$).
  \item[{\bf (u1) is equal to (f5)\/}]
    Because (u2)~is the reversal of~(f4) and (u4)~is the reversal of~(f2), we know that the level-$k$ proto-CFLOBDDs out of which the level-$k\textrm{+}1$ grouping of $C'$ is constructed are
    isomorphic to the respective
    level-$k$ proto-CFLOBDDs that make up the $A$-connection and $B$-connections of $C$.

    \hspace{1.5ex}
    We already argued that steps~\ref{Construct:Grouping}     and~\ref{Construct:ValueTuple} of Construction \ref{Constr:DecisionTreeToCFLOBDD} lead to CFLOBDDs that obey the six structural invariants required of CFLOBDDs.
    Moreover, there is only one way for Construction \ref{Constr:DecisionTreeToCFLOBDD} to construct the level-$k\textrm{+}1$ grouping of $C'$ so that Structural Invariants~\ref{Inv:2}, \ref{Inv:3}, and~\ref{Inv:4} are satisfied.
    Therefore, $C$ is isomorphic to $C'$.
\end{description}
Consequently, $FT$ is the reversal of $UT$, as was to be shown.
\end{InductionStep}
\end{Proof}
\end{Pro}

In summary, we have now shown that Obligations~\ref{Obligation:1}, \ref{Obligation:2}, and~\ref{Obligation:3} are all satisifed.
These properties imply that, for a given ordering of Boolean variables, if two level-$k$ CFLOBDDs $C_1$ and $C_2$ represent the same decision tree with $2^{2^k}$ leaves, then $C_1$ and $C_2$ are isomorphic---i.e., CFLOBDDs are a canonical representation of functions over Boolean arguments:

\begin{theorem*}[\ref{The:Canonicity}](\textsc{Canonicity}).
If $C_1$ and $C_2$ are level-$k$ CFLOBDDs for the same Boolean function over $2^k$ Boolean variables, and $C_1$ and $C_2$ use the same variable ordering, then $C_1$ and $C_2$ are
isomorphic.
\end{theorem*}

\section{Pair Product Canonicity Proof}
\label{Se:PairProductProof}

We prove that the grouping {\tt g\/} constructed
during a call on {\tt PairProduct\/} meets
Structural Invariant~\ref{Inv:4}---and hence it is permissible
to call {\tt RepresentativeGrouping(g)\/} in line~[\ref{Li:PPTabulateAnswer}]
of \algref{PairProductContinued}.

In particular, suppose that (i) $B_1$ and $B'_1$ are $B$-connections of a grouping $g_1$ (with associated return tuples $rt_1$ and $rt'_1$, respectively),
(ii) $B_2$ and $B'_2$ are $B$-connections of a grouping $g_2$ (with associated return tuples $rt_2$ and $rt'_2$, respectively), and
(iii) at least one of the following two conditions hold:
\begin{enumerate}
  \item $\langle B_1, rt_1 \rangle \neq \langle B'_1, rt'_1 \rangle$
  \item $\langle B_2, rt_2 \rangle \neq \langle B'_2, rt'_2 \rangle$
\end{enumerate}
In addition, suppose that the recursive calls on {\tt PairProduct\/} produce
\[
  [D,pt] = \mathtt{PairProduct}(B_1,B_2)
  \qquad\mathrm{and}\qquad
  [D',pt'] = \mathtt{PairProduct}(B'_1,B'_2),
\]
Let $rt$ and $rt'$ be the return tuples that the outer
calls on {\tt PairProduct\/} in \lineref{BConnectionPairProduct} of \algref{PairProductContinued} create for $D$ and $D'$:
$pt$, $rt_1$, and $rt_2$ are used to create $rt$;
$pt'$, $rt'_1$, and $rt'_2$ are used to create $rt'$.

The question that we need to answer is whether it is ever possible for both $D = D'$ and $rt = rt'$ to hold.
This question is of concern because the hypothesized condition would violate Structural Invariant~\ref{Inv:4}:
if the condition were to hold, then the first entry of the pair returned by {\tt PairProduct\/} would not be a well-formed proto-CFLOBDD.
The following proposition shows that, in fact, this situation cannot ever occur:

\begin{Pro}
The first entry of the pair returned by {\tt PairProduct\/} is always a well-formed proto-CFLOBDD.
\newline

\begin{Proof}
We argue by induction:
\newline

\begin{BaseCase}
When $g_1$ and $g_2$ are level-0 groupings, there are four cases to consider.
In each case, it is immediate from lines~[\ref{Li:PPNoDistinctionStart}]--[\ref{Li:PPBothForkGroupings}] of \algref{PairProduct} that the first entry of the pair returned by {\tt PairProduct\/} is a well-formed proto-CFLOBDD.
\end{BaseCase}
\newline

\begin{InductionStep}
The induction hypothesis is that the first entry of the pair returned by {\tt PairProduct\/} is a well-formed proto-CFLOBDD whenever the arguments to {\tt PairProduct\/} are level-$k$ proto-CFLOBDDs.

Let $g_1$ and $g_2$ be two arbitrary well-formed level-$k\textrm{+}1$ proto-CFLOBDDs.
We argue by contradiction:
suppose, for the sake of argument, that $D$, $D'$, $rt$, and $rt'$ are as defined above, and that both $D = D'$ and $rt = rt'$ hold.
\begin{itemize}
  \item
    By the inductive hypothesis, we know that $D$ and $D'$ are each well-formed level-$k$ proto-CFLOBDDs.
    In particular, we can think of $D$ and $rt$ as corresponding to a decision tree $T_0$, labeled with the exit vertices of $g$ to which the decision tree's leaves are mapped.
    However, because of the search that is carried out in lines~[\ref{Li:PPExitVertexLoopStart}]--[\ref{Li:PPExitVertexLoopEnd}] of {\tt PairProduct\/} (\algref{PairProductContinued}), each exit vertex of $g$ corresponds to a unique pair, $\B{c_1,c_2}$, where $c_1$ and $c_2$ are exit vertices of $g_1$ and $g_2$, respectively.
    Thus, a leaf in $T_0$ can be thought of as being labeled with a pair $\B{c_1,c_2}$.

    \hspace{1.5ex}
    Furthermore, because $D = D'$ and $rt = rt'$, $D'$ and $rt'$ also 
    correspond to decision tree $T_0$.
  \item
    When $T_0$ is considered to be the decision tree associated with $D$ and $rt$, we can read off (a) the decision tree that corresponds to $B_1$ with exit vertices of $g_1$ labeling the leaves (call this tree $T_1$), and (b) the decision tree that corresponds to $B_2$ with exit vertices of $g_2$ labeling the leaves ($T_2$).
    Similarly, when $T_0$ is considered to be the decision tree associated with $D'$ and $rt'$, we can read off (c) the decision tree that corresponds to $B'_1$ with exit vertices of $g_1$ labeling the leaves ($T'_1$), and (d) the decision tree that corresponds to $B'_2$ with exit vertices of $g_2$ labeling the leaves ($T'_2$).
    (We use the first entry of each $\B{c_1,c_2}$ pair for $B_1$ and $B'_1$, and the second entry of each $\B{c_1,c_2}$ pair for $B_2$ and $B'_2$.)
    This process gives us four trees, $T_1$, $T'_1$, $T_2$, and $T'_2$, where---from the supposition that $D = D'$ and $rt = rt'$---we must have $T_1 = T'_1$ and $T_2 = T'_2$.
  \item
    By assumption, $g_1$ and $g_2$ are well-formed proto-CFLOBDDs;
    thus, by Structural Invariant~\ref{Inv:2}, all return tuples for the $B$-connections of $g_1$ and $g_2$ must represent 1-to-1 maps.
    Moreover, $B_1$, $B_2$, $B'_1$, and $B'_2$ are also well-formed proto-CFLOBDDs, which means that, in $g_1$, $B_1$ together with $rt_1$ must be the unique representative of occurrences of $T_1$ in $g_1$'s decision tree, while $B'_1$ together with $rt'_1$ must be the unique representative of occurrences of $T'_1$.
    Similarly, in $g_2$, $B_2$ together with $rt_2$ must be the unique representative of occurrences of $T_2$ in $g_2$'s decision tree, while $B'_2$ together with $rt'_2$ must be the unique representative of occurrences of $T'_2$.

    Therefore, in $g_1$, we have
    \begin{itemize}
      \item
        $B_1 = B'_1$ and $rt_1 = rt'_1$,
    \end{itemize}
    while in $g_2$, we have
    \begin{itemize}
      \item
        $B_2 = B'_2$ and $rt_2 = rt'_2$.
    \end{itemize}
    However, these conditions are a violation of Structural Invariant~\ref{Inv:4}, which, in turn, contradicts the assumption that $g_1$ and $g_2$ are well-formed level-$k\textrm{+}1$ proto-CFLOBDDs.
    Consequently, the assumption that $D = D'$ and $rt = rt'$ cannot be true.
\end{itemize}
\end{InductionStep}
\end{Proof}
\end{Pro}

\section{Additional Operations on CFLOBDDs}
\label{Se:cflobdd-additional-algos}

\subsection{Ternary Operations on CFLOBDDs}
\label{Se:ternary-op}

This section discusses how ternary operations (i.e., three-argument operations) on CFLOBDDs are performed.
\algrefs{TernaryApplyAndReduce}{TripleProduct} present the two algorithms needed to implement ternary operations on multi-terminal CFLOBDDs.
As in \sectref{BinaryOperationsOnCFLOBDDs}, we assume that the {\tt CFLOBDD\/} or {\tt Grouping\/} arguments of the operations described below are objects whose highest-level groupings are all at the same level.

\begin{algorithm}[tb!]
\caption{TernaryApplyAndReduce\label{Fi:TernaryApplyAndReduce}}
\Input{CFLOBDDs n1, n2, n3, Op op}
\Output{CFLOBDD n = op(n1, n2, n3)}
\Begin{
assert(n1.grouping.level == n2.grouping.level $\&\&$ n2.grouping.level == n3.grouping.level)\;
\tcp{Perform triple cross product}
Grouping$\times$TripleTuple [g,tt] = TripleProduct(n1.grouping, n2.grouping, n3.grouping)\;  \label{Li:TAAR:CallTripleProduct}
\tcp{Create tuple of "leaf" values}
ValueTuple deducedValueTuple = [op(n1.valueTuple[i1], n2.ValueTuple[i2], n3.ValueTuple[i3]) : [i1,i2,i3]$\in$ tt]\; \label{Li:TAAR:LeafValues}
\tcp{Collapse duplicate leaf values, folding to the left}
Tuple$\times$Tuple [inducedValueTuple,inducedReturnTuple] = CollapseClassesLeftmost(deducedValueTuple)\; \label{Li:TAAR:CollapseLeafValues}
\tcp{Perform corresponding reduction on g, folding g's exit vertices w.r.t. inducedReductionTuple}
Grouping g' = Reduce(g, inducedReductionTuple)\;  \label{Li:TAAR:Reduce}
\Return RepresentativeCFLOBDD(g', inducedValueTuple)\;
}
\end{algorithm}

\begin{algorithm}[tb!]
\SetKwBlock{Begin}{begin}{}
\caption{TripleProduct}
\label{Fi:TripleProduct}
\Input{Groupings g1, g2, g3}
\Output{Grouping g: product of g1, g2, g3; TripleTuple ttAns: tuple of triples of exit vertices}
\Begin{
\label{Li:TPNoDistinctionStart}
\If{g1, g2, g3 are all no-distinction proto-CFLOBDDs}{\Return [g1, [[1,1,1]]]\;}
\If{g1 and g2 no-distinction proto-CFLOBDDs}{\Return [g3, [[1,1,k]: k $\in$ [1..g3.numberOfExits]]\;}
\If{g1 and g3 no-distinction proto-CFLOBDDs}{\Return [g2, [[1,k,1]: k $\in$ [1..g2.numberOfExits]]\;}
\If{g2 and g3 no-distinction proto-CFLOBDDs}{\Return [g1, [[k,1,1]: k $\in$ [1..g1.numberOfExits]]\;}
\If{g1 is a no-distinction proto-CFLOBDD}{
    Grouping$\times$PairTuple [g,pt] = PairProduct(g2,g3)\;
    \Return [g,[[1,j,k]: [j,k]$\in$ pt]]\;
}
\If{g2 is a no-distinction proto-CFLOBDD}{
    Grouping$\times$PairTuple [g,pt] = PairProduct(g1,g3)\;
    \Return [g,[[j,1,k]: [j,k]$\in$ pt]]\;
}
\If{g3 is a no-distinction proto-CFLOBDD}{
    Grouping$\times$PairTuple [g,pt] = PairProduct(g1,g2)\;
    \Return [g,[[j,k,1]: [j,k]$\in$ pt]]\;
}
\label{Li:TPNoDistinctionEnd}
\If{g1,g2,g3 are all fork groupings}{\Return [g1, [[1,1,1],[2,2,2]]\;}
\tcp{Combine the A-Connections}
Grouping$\times$TripleTuple [gA, ttA] = TripleProduct(g1.AConnection, g2.AConnection, g3.AConnection)\;
InternalGrouping g = new InternalGrouping(g1.level)\;
g.AConnection = gA\;
g.AReturnTuple = [1..|ttA|]\tcp*[r]{Represents the middle vertices}
g.numberOfBConnections = |ttA|\;
}
\end{algorithm}

\begin{algorithm}[tb!]
\caption{TripleProduct Contd.}
\SetKwBlock{Begin}{}{end}
\setcounter{AlgoLine}{34}
\Begin{
\tcp{Combine the B-connections, but only for triples in ttA}
\tcp{Descriptor of triples of exit vertices}
Tuple ttAns = []\;
\tcp{Create a B-Connection for each middle vertex}
\For{$j \leftarrow 1$ \KwTo $|ttA|$}{
    Grouping$\times$TripleTuple [gB,ttB] = TripleProduct(g1.BConnections[ttA(j)(1)], g2.BConnections[ttA(j)(2)], g3.BConnections[ttA(j)(3)])]\;
    g.BConnections[j] = gB\;
    g.BReturnTuples[j] = []\;
    \For{$i \leftarrow 1$ \KwTo $|ttB|$}{
        c1 = g1.BReturnTuples[ttA(j)(1)](ttB(i)(1))\;
        c2 = g2.BReturnTuples[ttA(j)(1)](ttB(i)(2))\;
        c3 = g3.BReturnTuples[ttA(j)(1)](ttB(i)(3))\;
        \tcp{Not a new exit vertex of g}
        \eIf{[c1,c2,c3] $\in$ ttAns}{
            index = the k such that ttAns(k) == [c1,c2,c3]\;
            g.BReturnTuples[j] = g.BReturnTuples[j] || index\;
        }{
            g.numberOfExits = g.numberOfExits + 1\;
            g.BReturnTuples[j] = g.BReturnTuples[j] || g.numberOfExits\;
            ttAns = ttAns || [c1,c2,c3]\;
        }
    }
}
\Return [RepresentativeGrouping(g), ttAns]\;  \label{Li:TPTabulateAnswer}
}
\end{algorithm}

\begin{itemize}
  \item
    The operation {\tt TernaryApplyAndReduce\/} given in
    ~\algref{TernaryApplyAndReduce} is very much like the
    operation
    \linebreak
    {\tt BinaryApplyAndReduce\/} of
    ~\algref{BinaryApplyAndReduce}, except that it starts with a
    call on {\tt TripleProduct\/} instead of {\tt PairProduct} (\lineref{TAAR:CallTripleProduct}).
  \item
    The operation {\tt TripleProduct}, which is given in
    ~\algref{TripleProduct}, is very much like the operation
    {\tt PairProduct\/} of ~\algref{PairProduct}, except
    that {\tt TripleProduct\/} has a third {\tt Grouping\/} argument,
    and performs a three-way---rather than two-way---cross product of
    the three {\tt Grouping\/} arguments: {\tt g1}, {\tt g2}, and {\tt
    g3}.  {\tt TripleProduct} returns the proto-CFLOBDD {\tt g\/}
    formed in this way, as well as a discriptor of the exit vertices
    of {\tt g\/} in terms of triples of exit vertices of the highest-level
    groupings of {\tt g1}, {\tt g2}, and {\tt g3}.

    \hspace{1.5ex}
    (By an argument similar to the one given for {\tt PairProduct},
    it is possible to show that the grouping {\tt g\/} constructed
    during a call on {\tt TripleProduct\/} is always a well-formed
    proto-CFLOBDD---and hence it is permissible to call
    {\tt RepresentativeGrouping(g)\/} in line~[\ref{Li:TPTabulateAnswer}]
    of ~\algref{TripleProduct}.)

  \item
    {\tt TernaryApplyAndReduce\/} then uses the triples describing the
    exit vertices to determine the tuple of leaf values that should be
    associated with the exit vertices (i.e., a tentative value tuple) (\lineref{TAAR:LeafValues}).
  \item
    Finally, {\tt TernaryApplyAndReduce\/} proceeds in the same
    manner as {\tt BinaryApplyAndReduce\/}:
    \begin{itemize}
      \item
        Two tuples that describe the collapsing of duplicate leaf
        values---assuming folding to the left---are created via a call
        to {\tt CollapseClassesLeftmost} (\lineref{TAAR:CollapseLeafValues}).
      \item
        The corresponding reduction is performed on {\tt Grouping\/}
        {\tt g\/}, by calling {\tt Reduce\/} to fold {\tt g\/}'s exit
        vertices with respect to variable {\tt inducedReductionTuple\/}
        (one of the tuples returned by the call on {\tt CollapseClassesLeftmost}) (\lineref{TAAR:Reduce}).
    \end{itemize}
\end{itemize}

Lastly, in the case of Boolean-valued CFLOBDDs, there are
$256$ possible ternary operations, corresponding to the
$256$ possible three-argument truth tables ($2 \times 2 \times 2$ matrices with
Boolean entries).
All $256$ possible ternary operations are special cases of
{\tt TernaryApplyAndReduce}; these can be performed by
passing {\tt TernaryApplyAndReduce} an appropriate
value for argument {\tt op} (i.e., some $2 \times 2 \times 2$
Boolean matrix.

One of the $256$ ternary operations is the operation called
$\ITE$~\cite{dac:BRB90} (for ``If-Then-Else''), which is defined
as follows:
\[
  \ITE(a,b,c) = (a \land b) \lor (\neg a \land c).
\]
Appendix~\sectref{ITE} shows how the ternary $\ITE$
operation can be used to implement all $16$ of the binary operations on
Boolean-valued CFLOBDDs~\cite{dac:BRB90}.

\subsection{Restrict}
\label{Se:restrict-op}

We now discuss the restriction operation.
Given a value $v$ to which variable $x_i$ is to be bound (e.g., by giving either the assignment $[{x_i} \mapsto T]$ or the assignment $[{x_i} \mapsto F]$), the Restrict operation applies the assignment to the CFLOBDD that represents a function $f$, and returns the CFLOBDD that represents $f |_{x_i = v}$.
\algrefs{RestrictCFLOBDDAlgorithm}{RestrictGroupingAlgorithm} gives pseudo-code for the algorithm.
At each level, the algorithm checks if index $i$ belongs to the \emph{A-Connection} or \emph{B-Connections} at every level, and calls Restrict recursively, as appropriate, on lower levels with an adjusted index value $i$.
The rest of the groupings are kept as is, except that some groupings are eliminated, and the positions of the remaining ones shifts (\algref{RestrictGroupingAlgorithm}, \linerefs{RestrictGrouping:DontCareReturn}{RestrictGrouping:CompactBConnections}).

\begin{algorithm}[t]
\caption{RestrictCFLOBDD\label{Fi:RestrictCFLOBDDAlgorithm}}
\Input{CFLOBDD c representing $f$, int i -- index, bool v -- $T/F$}
\Output{CFLOBDD c' representing $f' = f |_{x_i = v}$}
\Begin{
Grouping$\times$ReturnTuple [g, rt] = RestrictGrouping(c.grouping, i, v)\;
\Return RepresentativeCFLOBDD(g, [g.valueTuple[rt[i]] | $i \in [1..|rt|]$)\;
}
\end{algorithm}

\begin{algorithm}[t]
\caption{RestrictGrouping\label{Fi:RestrictGroupingAlgorithm}}
\Input{Grouping g, int i, bool v}
\Output{Grouping$\times$ReturnTuple [g', grt]}
\Begin{
\If{g == ForkGrouping}{
\Return [DontCareGrouping, (v == False) ? [1] : [2]]\;  \label{Li:RestrictGrouping:DontCareReturn}
}
\If{g == DontCareGrouping || g == NoDistinctionProtoCFLOBDD(g.level)}{
\Return [g, [1]]\;
}
InternalGrouping g' = new InternalGrouping(g.level)\;
\eIf(\tcp*[f]{i falls in AConnection range}){i < 2$\ast\ast$(g.level-1)}{
    Grouping$\times$ReturnTuple [aa, apt] = RestrictGrouping(g.AConnection, i, v)\;
    g'.AConnection = aa\;
    g'.AReturnTuple = [1..|apt|]\;
    g'.numberOfBConnections = |apt|\;
    \For{$j \leftarrow 1$ \KwTo $|apt|$}{
        g'.BConnections[j] = g.BConnections[apt[j]]\;  \label{Li:RestrictGrouping:CompactBConnections}
        \tcp{Fill in g'.BReturnTuples[j] from g.BReturnTuples[apt[j]] appropriately, along with populating grt   
        }
        \tcp{Keep track of number of exits of g'}
    }
}(\tcp*[f]{i falls in BConnections range}){
\For{$j \leftarrow 1$ \KwTo $g.numberOfBConnections$}{
Grouping$\times$ReturnTuple [bb, bpt] = RestrictGrouping(g.BConnections[j], i-(1 << (g.level-1)), v)\;
g'.BConnections[j] = bb\;
\tcp{Fill in g'.BReturnTuples[j] from bpt appropriately, along with populating grt  
}
\tcp{Keep track of number of exits of g'}
}
g'.AConnection = g.AConnection\;
g'.AReturnTuples = [1..|distinct g'.BConnections|]\;
g'.numberOfBConnections = |distinct g'.BConnections|\;
}
g'.numberOfExits = number of exits tracked so far\;
\Return [RepresentativeGrouping(g'), grt]\;
}
\end{algorithm}

\subsection{Existential Quantification}
\label{Se:existential-op}

For a CFLOBDD that represents a Boolean function $f$, existential quantification with respect to the Boolean variable at index $i$ yields $f |_{x_i == T} \lor f |_{x_i == F}$.
This operation can be implemented using two calls to Restrict, followed by a final call on {\tt BinaryApplyAndReduce} to perform the ``or.''

\begin{figure}
\begin{center}
\begin{tabular}{|c|c|c|}
  \hline
  Truth Table & Defining Expression & Definition Using $\ITE$ \\
  \hline
  {\scriptsize $\BoolOp{a}{b}{F}{F}{F}{F}$} & $\lambda a,b. F$ & $\lambda a,b. F$ \\* [3ex]
  \hline
  {\scriptsize $\BoolOp{a}{b}{F}{F}{F}{T}$} & $\lambda a,b. a \land b$ & $\lambda a,b. \ITE(a,b,F)$ \\* [3ex]
  \hline
  {\scriptsize $\BoolOp{a}{b}{F}{F}{T}{F}$} & $\lambda a,b. a \land \neg b$ & $\lambda a,b. \ITE(a,\neg b,F)$ \\* [3ex] 
  \hline
  {\scriptsize $\BoolOp{a}{b}{F}{F}{T}{T}$} & $\lambda a,b. a$ & $\lambda a,b. a$ \\* [3ex] 
  \hline
  {\scriptsize $\BoolOp{a}{b}{F}{T}{F}{F}$} & $\lambda a,b. \neg a \land b$ & $\lambda a,b. \ITE(a,F,b)$ \\* [3ex] 
  \hline
  {\scriptsize $\BoolOp{a}{b}{F}{T}{F}{T}$} & $\lambda a,b. b$ & $\lambda a,b. b$ \\* [3ex]
  \hline
  {\scriptsize $\BoolOp{a}{b}{F}{T}{T}{F}$} & $\lambda a,b. a \xor b$ & $\lambda a,b. \ITE(a,\neg b,b)$ \\* [3ex] 
  \hline
  {\scriptsize $\BoolOp{a}{b}{F}{T}{T}{T}$} & $\lambda a,b. a \lor b$ & $\lambda a,b. \ITE(a,T,b)$ \\* [3ex]
  \hline
  {\scriptsize $\BoolOp{a}{b}{T}{F}{F}{F}$} & $\lambda a,b. \neg(a \lor b)$ & $\lambda a,b. \ITE(a,F,\neg b)$ \\* [3ex] 
  \hline
  {\scriptsize $\BoolOp{a}{b}{T}{F}{F}{T}$} & $\lambda a,b. \neg(a \xor b)$ & $\lambda a,b. \ITE(a,b,\neg b)$ \\* [3ex] 
  \hline
  {\scriptsize $\BoolOp{a}{b}{T}{F}{T}{F}$} & $\lambda a,b. \neg b$ & $\lambda a,b. \ITE(b,F,T)$ \\* [3ex] 
  \hline
  {\scriptsize $\BoolOp{a}{b}{T}{F}{T}{T}$} & $\lambda a,b. a \lor \neg b$ & $\lambda a,b. \ITE(a,T,\neg b)$ \\* [3ex] 
  \hline
  {\scriptsize $\BoolOp{a}{b}{T}{T}{F}{F}$} & $\lambda a,b. \neg a$ & $\lambda a,b. \ITE(a,F,T)$ \\* [3ex] 
  \hline
  {\scriptsize $\BoolOp{a}{b}{T}{T}{F}{T}$} & $\lambda a,b. \neg a \lor b$ & $\lambda a,b. \ITE(a,b,T)$ \\* [3ex] 
  \hline
  {\scriptsize $\BoolOp{a}{b}{T}{T}{T}{F}$} & $\lambda a,b. \neg(a \land b)$ & $\lambda a,b. \ITE(a,\neg b,T)$ \\* [3ex]  \hline
  {\scriptsize $\BoolOp{a}{b}{T}{T}{T}{T}$} & $\lambda a,b. T$ & $\lambda a,b. T$ \\* [3ex] 
  \hline
\end{tabular}
\end{center}
\caption{}
\label{Fi:BinaryOpsViaITE}
\end{figure}

\section{Boolean Operations via ITE}
\label{Se:ITE}

The ternary operation $\ITE$~\cite{dac:BRB90} (for ``If-Then-Else''), is defined
as follows:
\[
  \ITE(a,b,c) = (a \land b) \lor (\neg a \land c).
\]
\figref{BinaryOpsViaITE} shows how the ternary $\ITE$ operation can be used to implement all $16$ of the binary operations on Boolean-valued CFLOBDDs~\cite{dac:BRB90}.

\section{Kronecker Product}
\label{Se:kronecker--product}

This section presents and discusses pseudo-code for the variant of Kronecker product sketched in \sectref{KroneckerProduct:VariantOne}.
Given CFLOBDDs for matrices $W$ and $V$, with the interleaved-variable orderings $x \bowtie y$ and $w \bowtie z$, respectively, the goal is to create a CFLOBDD for $W \tensor V$ with variable ordering $(x || w) \bowtie (y || z)$.\footnote{
  $\bowtie$ denotes the interleaving of two sequences of variables;
  $||$ denotes concatenation.
}

\begin{algorithm}
\Input{CFLOBDDs n1, n2 with variable ordering of n1: $x \bowtie y$ and n2: $w \bowtie z$}
\Output{CFLOBDD n = n1 $\otimes$ n2 with variable ordering of n: $(x || w) \bowtie (y || z)$}
\caption{Kronecker Product \label{Fi:KroneckerProductAlgo}}
\Begin{
\tcp{Create a CFLOBDD of size level(n1) + 1 with n1 as the AConnection}
    CFLOBDD g1 = RepresentativeCFLOBDD(ShiftToAConnection(n1.grouping), n1.valueTuple)\;
\tcp{Create a CFLOBDD of size level(n2) + 1 with n2 as the BConnection}
    CFLOBDD g2 = RepresentativeCFLOBDD(ShiftToBConnection(n2.grouping), n2.valueTuple)\;
    CFLOBDD n = BinaryApplyAndReduce(g1, g2, (op)Times)\;  \label{Li:ReduceKroneckerProduct}
    \Return n\;
}
\end{algorithm}

\begin{algorithm}
\Input{Grouping g}
\Output{Grouping g' such that AConnection of g' = g}
\caption{ShiftToAConnection\label{Fi:ShiftToAConnection}}
\Begin{
    InternalGrouping g' = new InternalGrouping(g.level + 1)\;
    g'.AConnection = g\;
    g'.AReturnTuple = [1..|g.numberOfExits|]\;
    g'.numberOfBConnections = |g.numberOfExits|\;
    \For{$j \leftarrow 1$ \KwTo $g.numberOfExits$}{
        g'.BConnection[j] = NoDistinctionProtoCFLOBDD(g.level)\;
        g'.BReturnTuples[j] = [j]\;
}
    g'.numberOfExits = |g.numberOfExits|\;
    \Return RepresentativeGrouping(g')\;
}
\end{algorithm}

\begin{algorithm}
\caption{ShiftToBConnection\label{Fi:ShiftToBConnection}}
\Input{Grouping g}
\Output{Grouping g' such that BConnection of g' = g}
\BlankLine
\Begin{
    InternalGrouping g' = new InternalGrouping(g.level + 1)\;
    g'.AConnection = NoDistinctionProtoCFLOBDD(g.level)\;
    g'.AReturnTuple = [1]\;
    g'.numberOfBConnections = 1\;
    g'.BConnection[1] = g\;
    g'.numberOfExits = |g.numberOfExits|\;
    \Return RepresentativeGrouping(g')\;
}
\end{algorithm}

\figref{KroneckerProduct}a shows a level-$k$ CFLOBDD for some array $W$, where $W$'s value tuple is $[w_0, w_1, w_2]$;
\figref{KroneckerProduct}b shows a level-$k$ CFLOBDD for some array $V$, where $V$'s value tuple is $[v_0, v_1, v_2]$.
($W$ and $V$ could have been embedded into level-$k\textrm{+}1$ CFLOBDDs;
for the sake of clarity, we have not depicted such structures.)
\figref{KroneckerProduct}c shows the level-$k\textrm{+}1$ CFLOBDD that would be constructed by \algref{KroneckerProductAlgo} before any collapsing of the value tuple via {\tt Reduce} in the call to {\tt BinaryApplyAndReduce} in \lineref{ReduceKroneckerProduct}.

Under the variable order $(x || w) \bowtie (y || z)$, as we work through the CFLOBDD shown in \figref{KroneckerProduct}c for a given assignment, the values of the first $2^k$ Boolean variables lead us to a middle vertex of the level-$k\textrm{+}1$ grouping.
This path will be continued according to the values of the next $2^k$ variables.
Call these two paths $p_A$ and $p_B$, respectively.
Under the interleaved-variable ordering, $p_A$ takes us to a particular block of the matrix that \figref{KroneckerProduct}c represents, and $p_B$ takes us to a particular element of that block.

The  path $p_A$ uses the $2^k$ variables in $x \bowtie y$, and thus can also be thought of as taking us to an exit vertex $e_w$ that in matrix $W$ is associated with some terminal value $w_i$.
The path $p_B$ uses the $2^k$ variables in $w \bowtie z$, and thus can be thought of as taking us to an exit vertex $e_v$ that in matrix $V$ is associated with some terminal value $v_j$.
In the CFLOBDD shown in \figref{KroneckerProduct}c, the terminal value at the end of the path $p_A || p_B$ is $w_i v_j$.
This value is exactly what is required of the matrix $W \tensor V$.
After {\tt Reduce} is called on \figref{KroneckerProduct}c, by the canonicity property, the multi-terminal CFLOBDD that results must be the unique representation of $W \tensor V$ under the variable ordering $(x || w) \bowtie (y || z)$.

\section{Efficient Construction of $\textit{Column1Matrix}_n$}
\label{Se:Column1Matrix-construction}

\begin{figure}[bt!]
    \centering
    \begin{subfigure}[t]{0.495\linewidth}
    \centering
    \includegraphics[align=c,width=.5\linewidth]{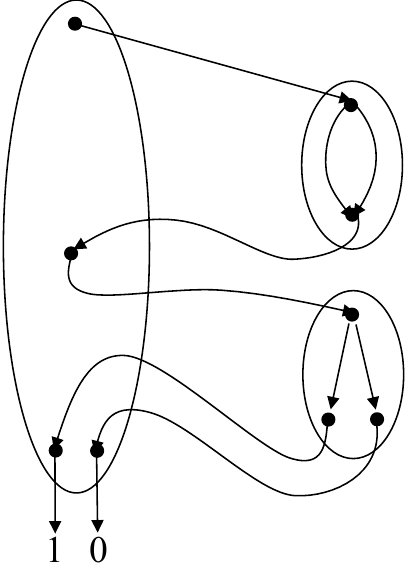}
    \caption{CFLOBDD representation of $\textit{Column1Matrix}_2$}
    \end{subfigure}
    \begin{subfigure}[t]{0.495\linewidth}
    \includegraphics[align=c,width=\linewidth]{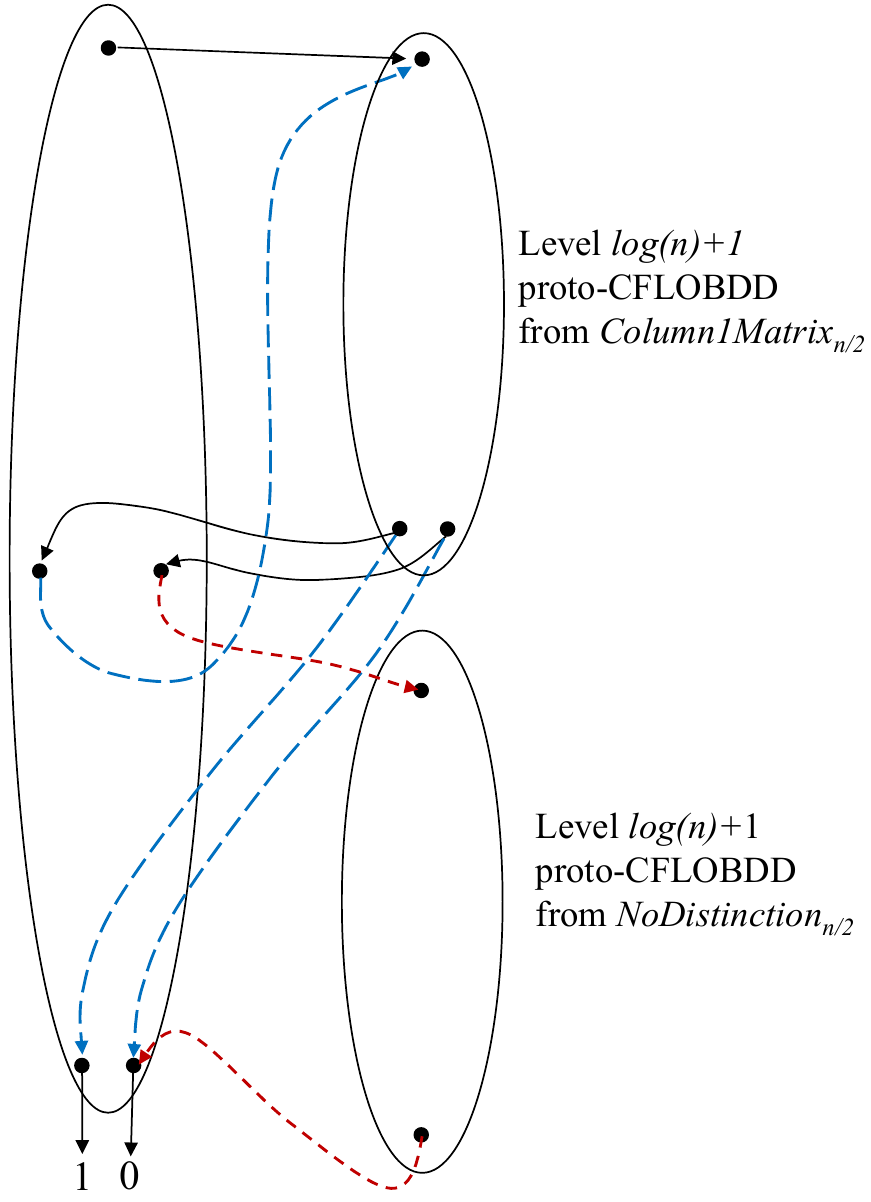}
    \caption{CFLOBDD representation of $\textit{Column1Matrix}_n$}
    \end{subfigure}
    \caption{(a) Base case; (b) the recursive structure for general $n$.
    }
    \label{Fi:column1K}
\end{figure}

\begin{algorithm}
\caption{Column1Matrix \label{Fi:column1Algo}}
\SetKwFunction{ColumnMatrixCFLOBDD}{ColumnMatrixCFLOBDD}
\SetKwFunction{ColumnMatrixGrouping}{ColumnMatrixGrouping}
\SetKwProg{myalg}{Algorithm}{}{end}
\myalg{\ColumnMatrixCFLOBDD{l}}{
\Input{int $l$ - level of the CFLOBDD = $(\log{n}) + 1$, where $2n =$ the number of variables}
\Output{CFLOBDD c representing $\textit{Column1Matrix}_{n}$}
\Begin{
Grouping g = Column1MatrixGrouping(l)\;
\Return RepresentativeCFLOBDD(g, [1,0])\;
}
}{}
\setcounter{AlgoLine}{0}
\SetKwProg{myproc}{SubRoutine}{}{end}
\myproc{\ColumnMatrixGrouping{l}}{
\Input{int $l$ - level of the CFLOBDD = $(\log{n}) + 1$, where $2n =$ the number of variables}
\Output{Grouping g representing proto - $\textit{Column1Matrix}_{n}$}
\Begin{
InternalGrouping g = new InternalGrouping(l)\;
\eIf{l == 1}{
g.AConnection = ForkGrouping\;  \label{Li:BaseCaseColumnMatrixStart}
g.AReturnTuples = [1,2]\;
g.numberOfBConnections = 2\;
g.BConnection[1] = DontCareGrouping\;
g.ReturnTuples[1] = [1]\;
g.BConnection[2] = DontCareGrouping\;
g.BReturnTuples[2] = [2]\;   \label{Li:BaseCaseColumnMatrixEnd}
}
{
Grouping g' = ColumnMatrixGrouping(l-1)\;  \label{Li:GeneralCaseColumnMatrixStart}
g.AConnection = g'\;     \label{Li:FirstRecursiveReuse}
g.AReturnTuples = [1,2]\;
g.numberOfBConnections = 2\;
g.BConnection[1] = g'\;    \label{Li:SecondRecursiveReuse}
g.BReturnTuples[1] = [1,2]\;
g.BConnection[2] = NoDistinctionProtoCFLOBDD(l-1)\;
g.BReturnTuples[2] = [2]\;  \label{Li:GeneralCaseColumnMatrixEnd}
}
g.numberOfExits = 2\;
\Return RepresentativeGrouping(g)\;
}
}
\end{algorithm}

$\textit{Column1Matrix}_n$ is a square matrix of size $2^n\times2^n$ in which the first column is filled with 1s, and all other entries are 0. 
\begin{equation*}
    \textit{Column1Matrix}_{n} = 
    \begin{bmatrix}
    1 & 0 & \cdots & 0\\
    1 & 0 & \cdots & 0\\
    \vdots & \vdots & \ddots & \vdots\\
    1 & 0 & \cdots & 0\\
    \end{bmatrix}_{2^n\times2^n}\\
\end{equation*}
$\textit{Column1Matrix}_{n}$ can be recursively defined in terms of the matrices $\textit{Column1Matrix}_{n/2}$ (of size $2^{n/2}\times2^{n/2}$) and $O_{n/2}$ (the all-zero matrix of size $2^n\times 2^n$).
\[
    \textit{Column1Matrix}_n = 
    \begin{cases}
        \vspace{2ex}
        \begin{bmatrix}
            \textit{Column1Matrix}_{n/2} & O_{n/2} & \cdots & O_{n/2}\\
            \textit{Column1Matrix}_{n/2} & O_{n/2} & \cdots & O_{n/2}\\
            \vdots & \vdots & \ddots & \vdots\\
            \textit{Column1Matrix}_{n/2} & O_{n/2} & \cdots & O_{n/2}\\
        \end{bmatrix}_{2^n\times2^n} & \\
        \vspace{2ex}
         \quad \quad = \textit{Column1Matrix}_{n/2}\tensor \textit{Column1Matrix}_{n/2}
            & n > 1\\
        \begin{bmatrix}
            1 & 0\\
            1 & 0
        \end{bmatrix} & n = 1\\
    \end{cases}
\]

\paragraph{Base Case.}
The CFLOBDD representation for the base case of $n = 1$, $\textit{Column1Matrix}_1$, is shown in \figref{column1K}a.
The base-case matrix requires two Boolean variables $x_0$ and $y_0$:
$x_0$ specifies the row, and $y_0$ specifies the column;
hence, the CFLOBDD that represents $\textit{Column1Matrix}_1$ has two levels, $0$ and $1$.
The rows of $\textit{Column1Matrix}_1$ are identical, so the \emph{A-Connection} grouping at level $1$ (for $x_0$) is a \emph{DontCareGrouping}.
In contrast, the columns of $\textit{Column1Matrix}_1$ are not identical, so the \emph{B-Connection} grouping at level $1$ (for $y_0$) is a \emph{ForkGrouping}.
See \figref{column1K}a and \lineseqref{BaseCaseColumnMatrixStart}{BaseCaseColumnMatrixEnd} of {\tt ColumnMatrixGrouping} in \algref{column1Algo}.
Note that the left exit vertex of the level-$1$ proto-CFLOBDD represents the first-column entries of the matrix (which are to have the value $1$), and the right exit vertex represents the matrix entries that are to have the value $0$.

\paragraph{General Case ($n > 1$).}
The steps to create $\textit{Column1Matrix}_n$, for $n > 1$, are shown in \lineseqref{GeneralCaseColumnMatrixStart}{GeneralCaseColumnMatrixEnd} of {\tt ColumnMatrixGrouping} in \algref{column1Algo}.
The number of levels in the CFLOBDD that represents $\textit{Column1Matrix}_n$ is $\log{n}$ + 2, to provide for the needed $2n$ Boolean variables.
(The leaves are at level $0$, so the outermost level is $l = \log{n} + 1$.)
The \emph{A-Connection} of the level-$l$ grouping represents a function involving the most-significant $\frac{n}{2}$ row and $\frac{n}{2}$ column variables.
From the recursive definition of $\textit{Column1Matrix}_n = \textit{Column1Matrix}_{n/2} \tensor \textit{Column1Matrix}_{n/2}$, the function for the first $\frac{n}{2}$ variables is obtained via a recursive call on $\textit{Column1Matrix}_n$ with half of the variables. 
Hence, the \emph{A-Connection} grouping at level-$l$ is a proto-CFLOBDD that represents $\textit{Column1Matrix}_{n/2}$ (\figref{column1K}b and \lineref{FirstRecursiveReuse} in \algref{column1Algo}).
The number of exits of this proto-CFLOBDD is two (for which the invariant is maintained that the left exit vertex of the level-$l\textrm{--}1$ proto-CFLOBDD is associated with $1$, and the right exit vertex is associated with $0$).
The grouping at level $l$ therefore has two \emph{B-Connection} groupings, the first one being a second use of the proto-CFLOBDD for $\textit{Column1Matrix}_{n/2}$
(\figref{column1K}b and \lineref{SecondRecursiveReuse} in \algref{column1Algo}).
The second \emph{B-Connection} grouping at level-$l$ is a proto-CFLOBDD for NoDistinction$_{n/2}$, representing $O_{n/2}$.
This recursive structure is shown in \figref{column1K}b.

At top level, the level-$l$ grouping has two exit vertices---the first maps to the value $1$ and the second maps to $0$.

\section{Algorithm for Constructing the CNOT Matrix}
\label{Se:CNOTConstructionAlgo}

Pseudo-code for the algorithm for constructing the CFLOBDD that represents the Controlled-NOT matrix is given in \algseqref{cnot2voc}{cnot2vocFinalPart}.
The algorithm for the special-case construction of CNOT$_n$ discussed in \sectref{CNOTMatrix} is given in~\algref{cnot4voc}.

\begin{algorithm}[bt]
\caption{Controlled-NOT matrix \label{Fi:cnot2voc}}
\SetKwFunction{CNOTCFLOBDD}{CNOTCFLOBDD}
\SetKwFunction{CNOTGrouping}{CNOTGrouping}
\SetKwProg{myalg}{Algorithm}{}{end}
\myalg{\CNOTCFLOBDD{n, i, j}}{
\Input{int $n$, where $2n =$ \#Variables of the CFLOBDD, int $i$ - control-bit; int $j$ - controlled-bit}
\Output{CFLOBDD c representing CNOT($n$, $i$, $j$)}
\Begin{
Grouping g = CNOTGrouping($\log{}n + 1$, $i$, $j$)\;
\Return RepresentativeCFLOBDD(g, [1,0])\;
}
}
\setcounter{AlgoLine}{0}
\SetKwProg{myproc}{SubRoutine}{}{}
\myproc{\CNOTGrouping{l, i, j}}{
\Input{int $l$ - grouping level, int $i$ - control-bit; int $j$ - controlled-bit}
\Output{Grouping $g$ representing CNOT($l$, $i$, $j$) with $n = 2^l$ bits}
\SetKwBlock{Begin}{begin}{}
\Begin{
\If(\tcp*[f]{Base Case}){l == 2}{
InternalGrouping g = new InternalGrouping($2$)\;
InternalGrouping g' = new InternalGrouping(1)\tcp*[r]{Level 1}
g'.AConnection = ForkGrouping\;
g'.AReturnTuple = [1,2]\;
g'.numberOfBConnections = 2\;
g'.BConnection[1] = ForkGrouping\;
g'.BReturnTuples[1] = [1,2]\;
g'.BConnection[2] = ForkGrouping\;
g'.BReturnTuples[2] = [2,3]\;
g'.numberOfExits = 3\;

g.AConnection = g'\;
g.AReturnTuple = [1,2,3]\;
g.numberOfBConnections = 3\;
g.BConnection[1] = IdentityMatrixGrouping(1)\;
g.BReturnTuples[1] = [1,2]\;
g.BConnection[2] = NoDistinctionProtoCFLOBDD(1)\;
g.BReturnTuples[2] = [2]\;
g.BConnection[3] = IdentityMatrixGrouping(1)\;
g.BReturnTuples[3] = [2,1]\;
g.numberOfExits = 2\;
\Return RepresentativeGrouping(g)\;
}
}
}
\end{algorithm}

\begin{algorithm}[tb]
\setcounter{AlgoLine}{26}
\SetKwProg{myproc}{}{}{}
\SetKwBlock{Begin}{}{}
\myproc{}{
\Begin{
\If(\tcp*[f]{Case 1}){$i$ and $j$ fall in A-connection range}{
InternalGrouping g = new InternalGrouping($l + 1$)\;
g.AConnection = CNOTGrouping(l-1, i, j)\;
g.AReturnTuple = [1,2]\;
g.numberOfBConnections = 2\;
g.BConnection[1] = IdentityMatrixGrouping(g.level - 1)\;
g.BReturnTuples[1] = [1,2]\;
g.BConnection[2] = NoDistinctionProtoCFLOBDD(g.level - 1)\;
g.BReturnTuples[2] = [2]\;
g.numberOfExits = 2\;
\Return RepresentativeGrouping(g)\;
}
\If(\tcp*[f]{Case 2}){$i$ and $j$ fall in B-connection range}{
InternalGrouping g = new InternalGrouping($l + 1$)\;
g.AConnection = IdentityMatrixGrouping(g.level - 1)\;
g.AReturnTuple = [1,2]\;
g.numberOfBConnections = 2\;
g.BConnection[1] = CNOTGrouping(l-1, i', j')\tcp*[r]{$i' = i - 2^{l-1}$, $j' = j - 2^{l-1}$}
g.BReturnTuples[1] = [1,2]\;
g.BConnection[2] = NoDistinctionProtoCFLOBDD(g.level - 1)\;
g.BReturnTuples[2] = [2]\;
g.numberOfExits = 2\;
\Return RepresentativeGrouping(g)\;
}
\If(\tcp*[f]{Case 3}){$i$ in A-connection range and $j$ in B-connection range}{
InternalGrouping g = new InternalGrouping($l + 1$)\;
g.AConnection = CNOTGrouping(l-1, i, -1)\;
g.AReturnTuple = [1,2,3]\;
g.numberOfBConnections = 3\;
g.BConnection[1] = IdentityMatrixGrouping(g.level-1)\;
g.BReturnTuples[1] = [1,2]\;
g.BConnection[2] = NoDistinctionProtoCFLOBDD(g.level - 1)\;
g.BReturnTuples[2] = [2]\;
g.BConnection[3] = CNOTGrouping(l-1, -1, j')\tcp*[r]{$j' = j - 2^{l-1}$}
g.BReturnTuples[3] = [2,1]\;
g.numberOfExits = 2\;
\Return RepresentativeGrouping(g)\;
}
}
}
\caption{CNOT contd.}
\end{algorithm}

\begin{algorithm}[tb]
\setcounter{AlgoLine}{67}
\SetKwProg{myproc}{}{}{}
\SetKwBlock{Begin}{}{}
\myproc{}{
\Begin{
\If(\tcp*[f]{Case 4}){$i$ in A-connection range but $j$ is not in this range}{
InternalGrouping g = new InternalGrouping($l + 1$)\;
g.AConnection = CNOTGrouping(l-1, i, -1)\;
g.AReturnTuple = [1,2,3]\;
g.numberOfBConnections = 3\;
g.BConnection[1] = IdentityMatrixGrouping(g.level-1)\;
g.BReturnTuples[1] = [1,2]\;
g.BConnection[2] = NoDistinctionProtoCFLOBDD(g.level - 1)\;
g.BReturnTuples[2] = [2]\;
g.BConnection[3] = IdentityMatrixGrouping(g.level-1)\;
g.BReturnTuples[3] = [3,2] \;
g.numberOfExits = 3\;
\Return RepresentativeGrouping(g)\;
}
\If(\tcp*[f]{Case 5}){$i$ in B-connection range but $j$ is not in this range}{
InternalGrouping g = new InternalGrouping($l + 1$)\;
g.AConnection = IdentityMatrixGrouping(g.level-1)\;
g.AReturnTuple = [1,2]\;
g.numberOfBConnections = 2\;
g.BConnection[1] = CNOTGrouping(l-1, i', -1)\tcp*[r]{$i' = i - 2^{l-1}$}
g.BReturnTuples[1] = [1,2,3]\;
g.BConnection[2] = NoDistinctionProtoCFLOBDD(g.level - 1)\;
g.BReturnTuples[2] = [2]\;
g.numberOfExits = 3\;
\Return RepresentativeGrouping(g)\;
}
\If(\tcp*[f]{Case 6}){$j$ in A-connection range but $i$ is not in this range}{
InternalGrouping g = new InternalGrouping($l + 1$)\;
g.AConnection = CNOTGrouping(l-1, -1, j)\;
g.AReturnTuple = [1,2]\;
g.numberOfBConnections = 2\;
g.BConnection[1] = NoDistinctionProtoCFLOBDD(g.level-1)\;
g.BReturnTuples[1] = [1,2]\;
g.BConnection[2] = IdentityMatrixGrouping(g.level - 1)\;
g.BReturnTuples[2] = [1]\;
g.numberOfExits = 2\;
\Return RepresentativeGrouping(g)\;
}
}
}
\caption{CNOT contd.}
\end{algorithm}

\begin{algorithm}[tb]
\setcounter{AlgoLine}{108}
\SetKwProg{myproc}{}{}{end}
\SetKwBlock{Begin}{}{end}
\myproc{}{
\Begin{
\If(\tcp*[f]{Case 7}){$j$ in B-connection range but $i$ is not in this range}{
InternalGrouping g = new InternalGrouping($l + 1$)\;
g.AConnection = IdentityMatrixGrouping(g.level-1)\;
g.AReturnTuple = [1,2]\;
g.numberOfBConnections = 2\;
g.BConnection[1] = CNOTGrouping(l-1, -1, j')\tcp*[r]{$j' = j - 2^{l-1}$}
g.BReturnTuples[1] = [1,2]\;
g.BConnection[2] = NoDistinctionProtoCFLOBDD(g.level - 1)\;
g.BReturnTuples[2] = [1]\;
g.numberOfExits = 2\;
\Return RepresentativeGrouping(g)\;
}
}
}
\caption{CNOT contd. \label{Fi:cnot2vocFinalPart}}
\end{algorithm}

\begin{algorithm}[tb]
\caption{Algorithm for constructing $\CNOT_n$ \label{Fi:cnot4voc}}
\SetKwFunction{CNOTInterleavedMatrixCFLOBDD}{CNOTInterleavedMatrixCFLOBDD}
\SetKwFunction{CNOTInterleavedMatrixGrouping}{CNOTInterleavedMatrixGrouping}
\SetKwProg{myalg}{Algorithm}{}{end}
\myalg{\CNOTInterleavedMatrixCFLOBDD{l}}{
\Input{int $l$ - level of the CFLOBDD = $\log{2n} + 1$, where $2n=$ number of bits}
\Output{The CFLOBDD that represents $\CNOT_n$}
\Begin{
Grouping g = CNOTInterleavedMatrixGrouping(l)\;
\Return RepresentativeCFLOBDD(g, [1,0])\;
}
}{}
\setcounter{AlgoLine}{0}
\SetKwProg{myproc}{SubRoutine}{}{end}
\myproc{\CNOTInterleavedMatrixGrouping{l}}{
\Input{int $l$ - level of the CFLOBDD = $\log{n}$, where $2n =$ number of bits}
\Output{Grouping g representing the proto-CFLOBDD for $\CNOT_n$}
\Begin{
InternalGrouping g = new InternalGrouping(l)\;
\eIf{l == 2}{
\Return CNOTGrouping(2, 1, 2)\tcp*[r]{The base case is $\CNOT_2$}
}
{
Grouping g' = CNOTInterleavedMatrixGrouping(l-1)\;
g.AConnection = g'\;
g.AReturnTuple = [1,2]\;
g.numberOfBConnections = 2\;
g.BConnection[1] = g'\;
g.BReturnTuples[1] = [1,2]\;
g.BConnection[2] = NoDistinctionProtoCFLOBDD(l-1)\;
g.BReturnTuples[2] = [2]\;
}
g.numberOfExits = 2\;
\Return RepresentativeGrouping(g)\;
}
}
\end{algorithm}

\section{Construction of Additional Quantum Gates}
\label{Se:QuantumGates}

In this section, we discuss the construction of two quantum gates:
the Controlled Phase gate and the Swap gate.

\subsection{Controlled-Phase Gate}
\label{Se:CPMatrix}

A Controlled-Phase gate ($\textit{CP}$) for two qubits with angle $\theta$ has the following matrix:
\[
\textit{CP}(\theta)~~= 
\begin{bNiceArray}{cccc}[first-col,first-row]
& 00 & 01 & 10 & 11\\
00 & 1 & 0 & 0 & 0\\
01 & 0 & 1 & 0 & 0\\
10 & 0 & 0 & 1 & 0\\
11 & 0 & 0 & 0 & e^{i \theta}\\
\end{bNiceArray}
\]
More generally, the controlled-phase gate with control qubit $i$ and controlled qubit $j$ adds the phase angle $\theta$ to the $j^{\textit{th}}$ qubit when both qubits have the value 1.
The construction of the CFLOBDD that represents a $\textit{CP}$ gate's matrix for $n$ qubits is similar to the construction of the CFLOBDD for the $\textit{CNOT}$ matrix.
We do not give pseudo-code, but \figref{CP} depicts the different cases of the construction of the proto-CFLOBDD for $\textit{CP}$.
The CFLOBDD for a $\textit{CP}$ gate attaches the three terminal values $[1, 0, e^{i\theta}]$.

\begin{figure}[tb!]
    \centering
    \begin{subfigure}[t]{0.3\linewidth}
      \centering
      \includegraphics[width=\linewidth]{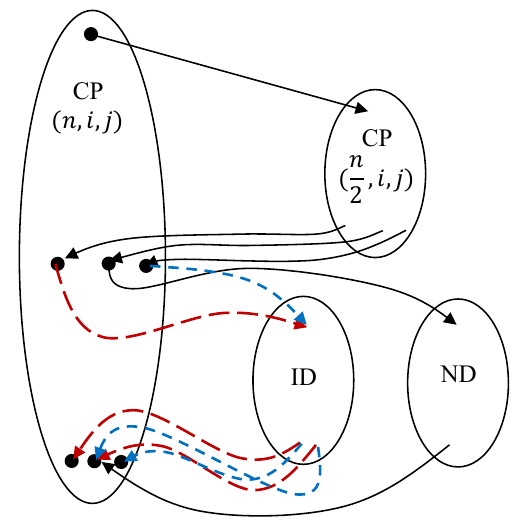}
      \caption{Case 1: Both $i$ and $j$ fall in A.}
    \end{subfigure}
    \hspace{2ex}
    \begin{subfigure}[t]{0.3\linewidth}
      \centering
      \includegraphics[width=\linewidth]{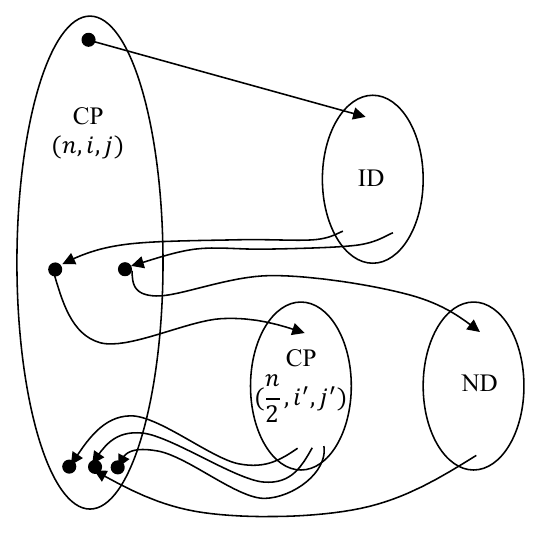}
      \caption{Case 2: Both $i$ and $j$ fall in B.}
    \end{subfigure}
    \hspace{2ex}
    \begin{subfigure}[t]{0.31\linewidth}
      \centering
      \includegraphics[width=\linewidth]{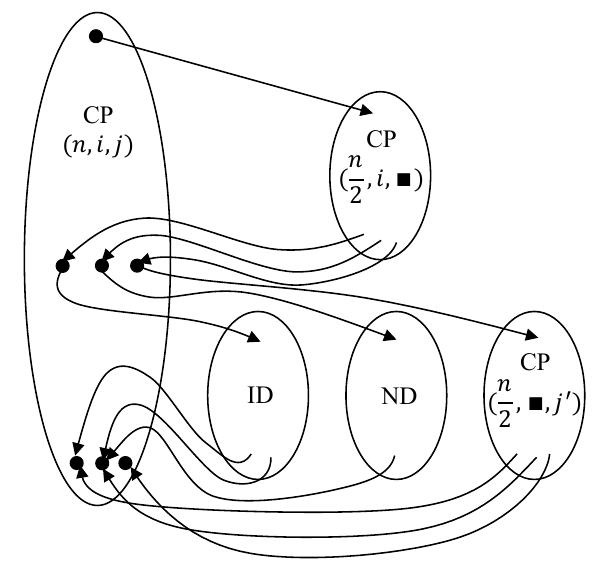}
      \caption{\protect \raggedright Case 3: $i$ in A and $j$ in B.}
    \end{subfigure}
    \begin{subfigure}[t]{0.3\linewidth}
      \centering
      \includegraphics[width=\linewidth]{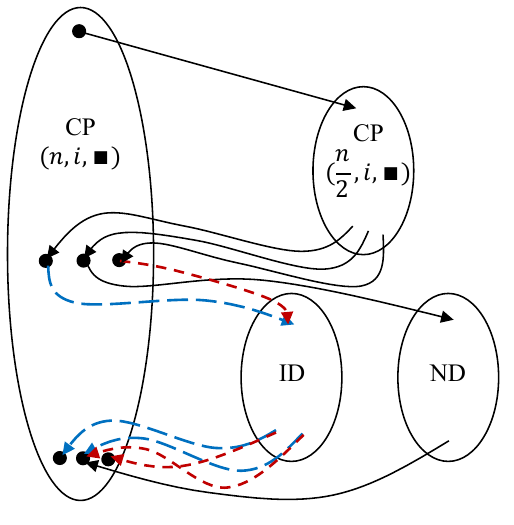}
      \caption{\protect \raggedright Case 4: $i$ in A and $j$ not in current grouping's range.}
    \end{subfigure}
    \hspace{2ex}
    \begin{subfigure}[t]{0.31\linewidth}
      \centering
      \includegraphics[width=\linewidth]{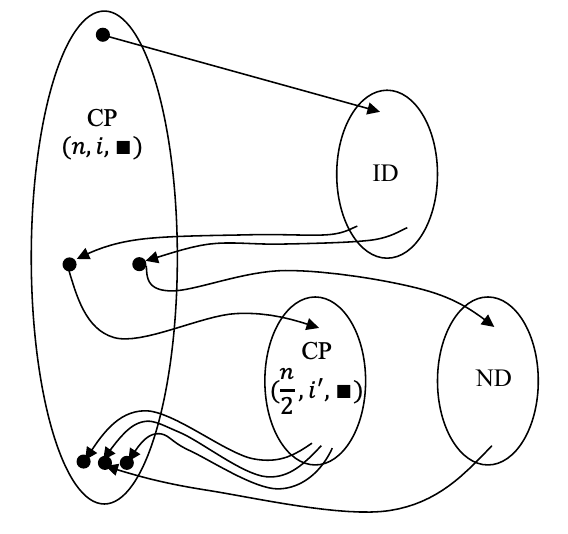}
      \caption{\protect \raggedright Case 5: $i$ in B and $j$ not in current grouping's range.}
    \end{subfigure}
    \hspace{2ex}
    \begin{subfigure}[t]{0.3\linewidth}
      \centering
      \includegraphics[width=\linewidth]{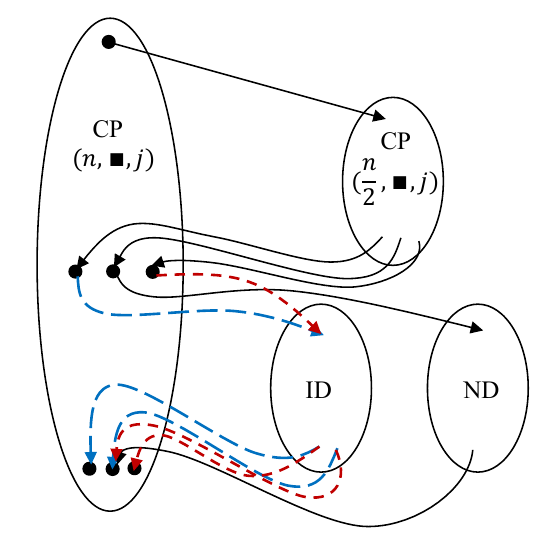}
      \caption{\protect \raggedright Case 6: $j$ in A and $i$ not in current grouping's range.}
    \end{subfigure}
    \begin{subfigure}[t]{0.33\linewidth}
      \centering
      \includegraphics[width=\linewidth]{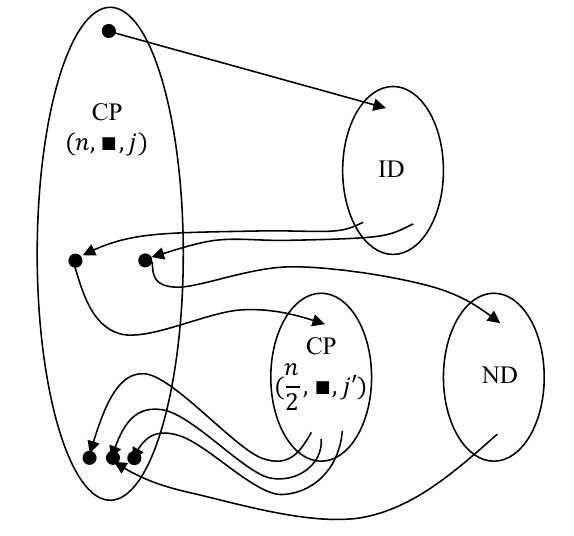}
      \caption{\protect \raggedright Case 7: $j$ in B and $i$ not in current grouping's range.}
    \end{subfigure}
    \hspace{3ex}
    \begin{subfigure}[t]{0.45\linewidth}
      \centering
      \includegraphics[width=.45\linewidth]{figures/CNOT_1_i_square-cropped.pdf}
      \caption{\protect \raggedright
      Base case: the same proto-CFLOBDD is used for both $CP(n, 1, \blacksquare)$ (to interpret the control-bit) and $CP(n, \blacksquare, 1)$ (to interpret the controlled-bit).
      }
    \end{subfigure}
    \caption{\protect \raggedright 
    The different cases of the $\textit{CP}$ construction.
    The text in each grouping denotes the function represented by the grouping.
    ID denotes {\tt IdentityMatrixGrouping}, and
    ND denotes a {\tt NoDistinctionProtoCFLOBDD} (used here for an all-zero matrix).
    $\textit{CP}$ takes 4 arguments:
    $n$ for the number of bits in this proto-CFLOBDD;
    $i$ for the control-bit, $j$ for the controlled-bit, where $0 \leq i < j < n$; and $\theta$ (which is only used at top level in the CFLOBDD's value tuple).
    $i'$ and $j'$ denote bit indices adjusted to the index range of the current level: $i' = i - n/2$; $j' = j - n/2$.
    A black square indicates that a particular index is outside the grouping's index range.
    Figure (h) shows the base case at level $1$;
    the same proto-CFLOBDD is used for both $CP(n, 1, \blacksquare)$ and $CP(n, \blacksquare, 1)$.
    }
    \label{Fi:CP}
\end{figure}

$\textit{CP}$ takes 4 arguments:
$n$ for the number of bits in this proto-CFLOBDD;
$i$ for the control-bit, $j$ for the controlled-bit, where $0 \leq i < j < n$; and the phase $\theta$.
The key to understanding \figref{CP} is that the construction
maintains the invariant shown in \eqref{CPExitVertexInvariant} on the exit vertices of the different kinds of groupings.

\begin{equation}
  \label{Eq:CPExitVertexInvariant}
  \begin{array}{r|c|ccc}
              & \multicolumn{1}{c|}{\multirow{2}{*}{\text{Role}}} & \multicolumn{3}{c}{\text{Significance of exit vertex}}  \\
    \cline{3-5}
              &             & 1 & 2 & 3 \\
    \hline
    \multirow{4}{*}{Proto-CFLOBDD} & \textit{CP}(n,i,j)                      & \text{on-path}    & \text{off-path} & \text{phase-path} \\
                                   & \textit{CP}(n,i,\blacksquare)           & \text{on-path}    & \text{off-path} & \text{phase-path} \\
                                   & \textit{CP}(n,\blacksquare,j)           & \text{on-path}   & \text{off-path}  & \text{phase-path} \\
                                   & \ID~\textrm{called from middle vertex 1} & \text{on-path}    & \text{off-path} & N/A \\
                                   & \ID~\textrm{called from middle vertex 3} & \text{phase-path} & \text{off-path} & N/A \\
    \hline
    CFLOBDD                        & \text{Top level}        & 1          & 0          & e^{i\theta} \\
    \hline
  \end{array}
\end{equation}
Here, ``on-path'' means that the exit occurs on a matched-path that can be continued to the top-level terminal value $1$;
``off-path'' means that it will only be used to reach the top-level terminal value $0$; and
``phase-path'' means that the exit occurs on a matched-path that can be continued by \texttt{IdentityMatrixGrouping}s to the top-level terminal values $\theta$ and $0$.

\subsection{Swap Gate}
\label{Se:SwapMatrix}

A Swap gate is a matrix that swaps two bits (or variables).
The Swap matrix for 2 bits (i.e., for 2 input variables and 2 output variables) is

\newpage

\begin{equation}
\label{Eq:SwapGateTwoByTwo}
\textit{SwapGate}~~= 
\begin{bNiceArray}{cccc}[first-row,first-col]
& 00 & 01 & 10 & 11\\
00 & 1 & 0 & 0 & 0\\
01 & 0 & 0 & 1 & 0\\
10 & 0 & 1 & 0 & 0\\
11 & 0 & 0 & 0 & 1\\
\end{bNiceArray}
\end{equation}
A matrix entry is 1 in exactly the positions where swapping the bits of the row index yields the bits of the column index.
The matrix can be divided into four quadrants that have different values:
\begin{equation*}
\begin{bmatrix}
1 & 0\\
0 & 0\\
\end{bmatrix}, 
\begin{bmatrix}
0 & 0\\
1 & 0\\
\end{bmatrix},
\begin{bmatrix}
0 & 1\\
0 & 0\\
\end{bmatrix},
~\textrm{and}~
\begin{bmatrix}
0 & 0\\
0 & 1\\
\end{bmatrix}.
\end{equation*}
This observation comes in handy in the construction of the generalized Swap matrix for higher numbers of bits $n$.

The construction of the $\textit{Swap}$ matrix is similar to the construction of the $\textit{CNOT}$ matrix.
Again, we do not give pseudo-code, but give a graphical depiction of the different cases of the construction.
The CFLOBDD for a $\textit{Swap}$ gate always has $[1,0]$ as the
value tuple for the terminal values.
\figrefsp{Swap}{SwapContd}{SwapContd2} depict the different cases involved in the construction of the proto-CFLOBDD for $\textit{Swap}$, given $n$ and the bits $i$ and $j$ to swap (where $i < j$).
The construction also uses an additional parameter, called ``$\textit{State}$,'' which takes the values $\{0..4\}$.
$\textit{State}$ is initially $0$, and remains $0$ until the Boolean variable corresponding to bit $i$ is encountered.
At this point, the construction creates representations of sub-matrices, where $\textit{State}$ takes the values 1, 2, 3, and 4, corresponding to the different $2\times2$ sub-matrices discussed above.
This situation is depicted in the base case shown in~\figref{SwapContd2}d.
These states are propagated further through the proto-CFLOBDD (see \figref{Swap}e, \figref{Swap}f, \figref{SwapContd}a, \figref{SwapContd}b, \figref{SwapContd2}b, and \figref{SwapContd2}c) until the base cases that interpret the controlled-bit are encountered (see \figref{SwapContd2}e, \figref{SwapContd2}f, \figref{SwapContd2}g, and \figref{SwapContd2}h).

The key to understanding \figrefsp{Swap}{SwapContd}{SwapContd2} is that the construction maintains the invariant shown in \eqref{SwapExitVertexInvariant} on the exit vertices of the different kinds of groupings.
(``On-path'' means that the exit occurs on a matched-path that can be continued to the top-level terminal value $1$;
``off-path'' means that it will only be used to reach the top-level terminal value $0$; and
``$\textit{cd-bit}$'' abbreviates ``controlled-bit.'')

\begin{equation}
  \label{Eq:SwapExitVertexInvariant}
  \begin{array}{@{\hspace{0ex}}r@{\hspace{0.95ex}}|@{\hspace{0.5ex}}c@{\hspace{0.95ex}}|@{\hspace{0.55ex}}c@{\hspace{0.95ex}}c@{\hspace{0.95ex}}c@{\hspace{0.95ex}}c@{\hspace{0.95ex}}c@{\hspace{0ex}}}
              & \multicolumn{1}{c@{\hspace{0.95ex}}|@{\hspace{0.55ex}}}{\multirow{2}{*}{\text{Role}}} & \multicolumn{5}{c}{\text{Significance of exit vertex}}  \\
    \cline{3-7}
              &             & 1 & 2 & 3 & 4 & 5\\
    \hline
    \multirow{4}{*}{Proto-CFLOBDD}
        & \textit{SWAP}(n,i,j)                      & \text{on-path}    & \text{off-path} & N/A & N/A & N/A \\
        & \textit{SWAP}(n,i,\blacksquare), \textit{cd-bit $=$ $\frac{n}{2}$-1} & \textit{State}\text{\,$=$\,}1    & \textit{State}\text{\,$=$\,}2 & \textit{State}\text{\,$=$\,}3 & \textit{State}\text{\,$=$\,}4 & \text{off-path} \\
        & \textit{SWAP}(n,i,\blacksquare), \textit{cd-bit $\neq \frac{n}{2}$-1} & \textit{State}\text{\,$=$\,}1    & \text{off-path}  & \textit{State}\text{\,$=$\,}2 & \textit{State}\text{\,$=$\,}3 & \textit{State}\text{\,$=$\,}4  \\
        & \textit{SWAP}(n,\blacksquare,j),~\textit{State}\text{\,$=$\,}1 & \text{on-path}   & \text{off-path}  & N/A & N/A & N/A \\
        & \textit{SWAP}(n,\blacksquare,j),~\textit{State}\text{\,$=$\,}2 & \text{off-path}   & \text{on-path}  & N/A & N/A & N/A \\
        & \textit{SWAP}(n,\blacksquare,j),~\textit{State}\text{\,$=$\,}3 & \text{off-path}   & \text{on-path}  & N/A & N/A & N/A \\
        & \textit{SWAP}(n,\blacksquare,j),~\textit{State}\text{\,$=$\,}4 & \text{off-path}   & \text{on-path}  & N/A & N/A & N/A \\
        & \ID,~\text{called from}~\textit{State}\text{\,$=$\,}1 & \text{\textit{State}\text{\,$=$\,}1}   & \text{off-path}  & N/A & N/A & N/A \\
        & \ID,~\text{called from}~\textit{State}\text{\,$=$\,}2 & \text{\textit{State}\text{\,$=$\,}2}   & \text{off-path}  & N/A & N/A & N/A\\
        & \ID,~\text{called from}~\textit{State}\text{\,$=$\,}3 & \text{\textit{State}\text{\,$=$\,}3}   & \text{off-path}  & N/A & N/A & N/A\\
        & \ID,~\text{called from}~\textit{State}\text{\,$=$\,}4 & \text{\textit{State}\text{\,$=$\,}4}   & \text{off-path}  & N/A & N/A & N/A \\
    \hline
    CFLOBDD & \text{Top level}        & 1          & 0          & N/A & N/A & N/A \\
    \hline
  \end{array}
\end{equation}

\begin{figure}[tb!]
    \centering
    \begin{subfigure}[t]{0.4\linewidth}
      \centering
      \includegraphics[width=0.75\linewidth]{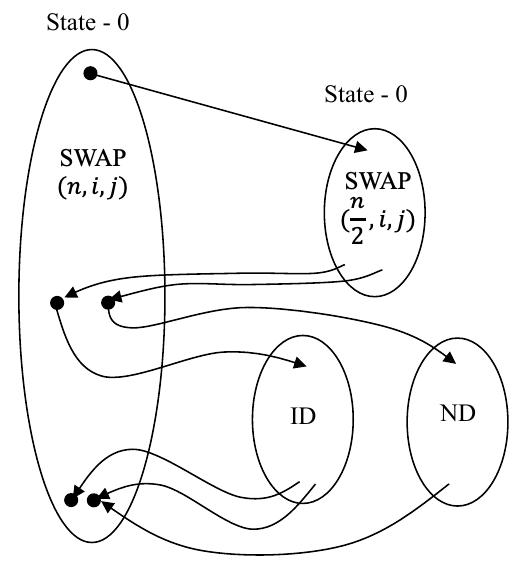}
      \vspace{-3pt}
      \caption{Case 1: Both $i$ and $j$ fall in A}
    \end{subfigure}
    \hspace{2ex}
    \begin{subfigure}[t]{0.4\linewidth}
      \centering
      \includegraphics[width=0.75\linewidth]{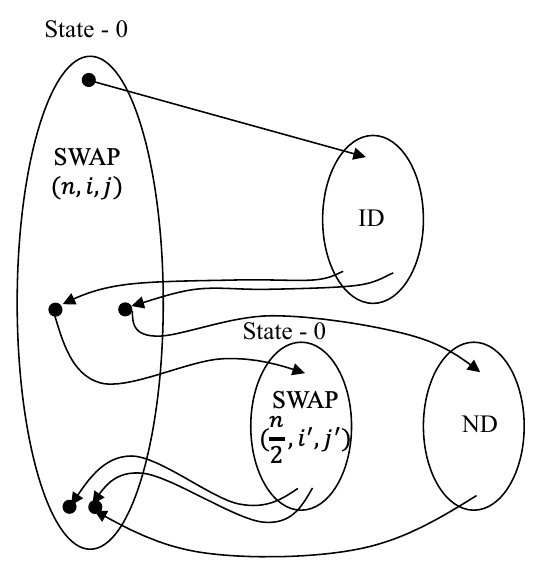}
      \vspace{-3pt}
      \caption{Case 2: Both $i$ and $j$ fall in B}
    \end{subfigure}
    \begin{subfigure}[t]{0.44\linewidth}
      \centering
      \includegraphics[width=\linewidth]{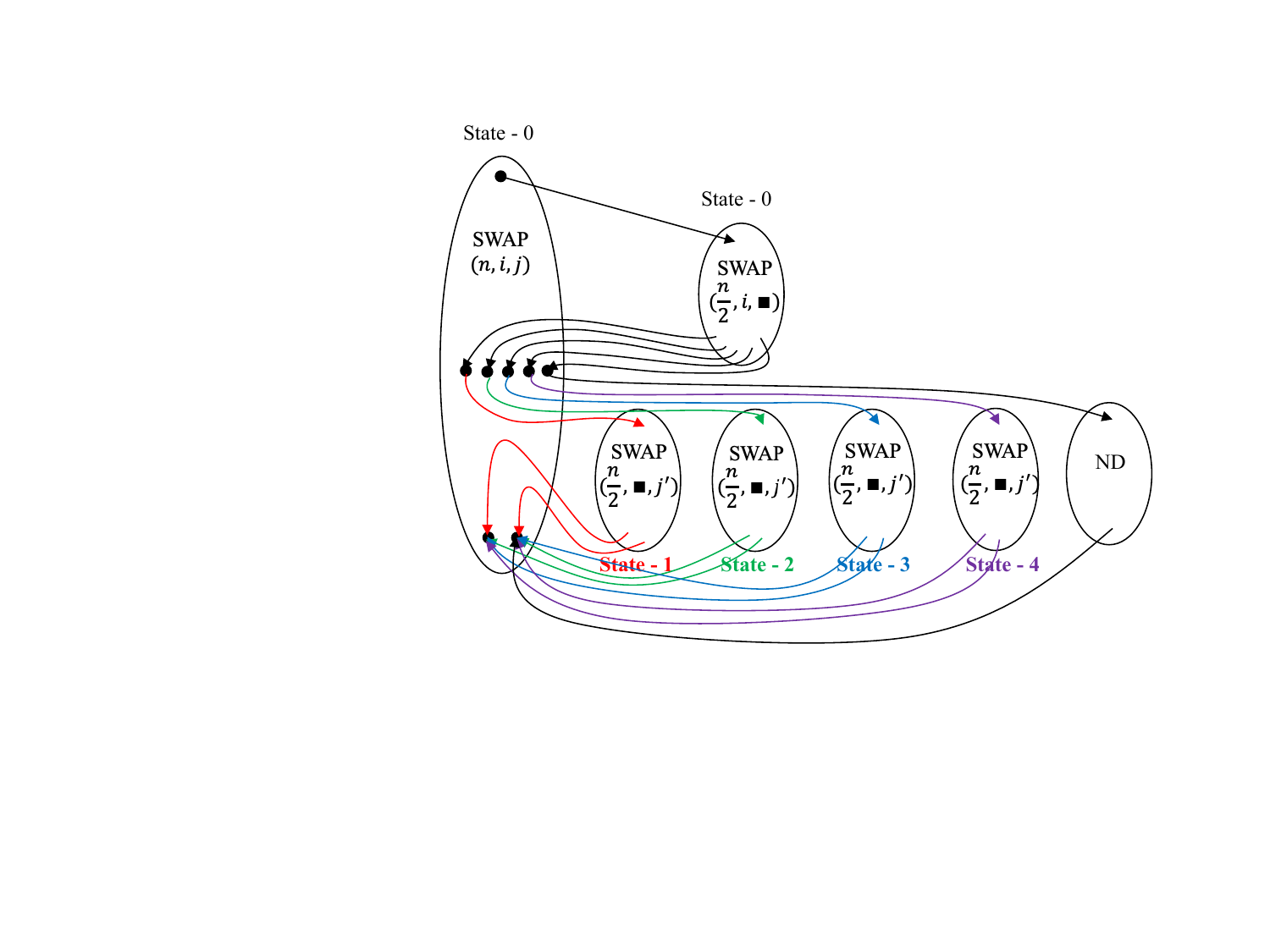}
      \vspace{-16pt}
      \caption[Case 3(a): $i$ in A and $j$ in B; controlled-bit == n/2-1]
        {\tabular[t]{@{}l@{}}Case 3(a): $i$ in A and $j$ in B; \\ controlled-bit = $n/2$-1
        \endtabular
      }
    \end{subfigure}
    \hspace{2ex}
    \begin{subfigure}[t]{0.44\linewidth}
      \centering
      \includegraphics[width=\linewidth]{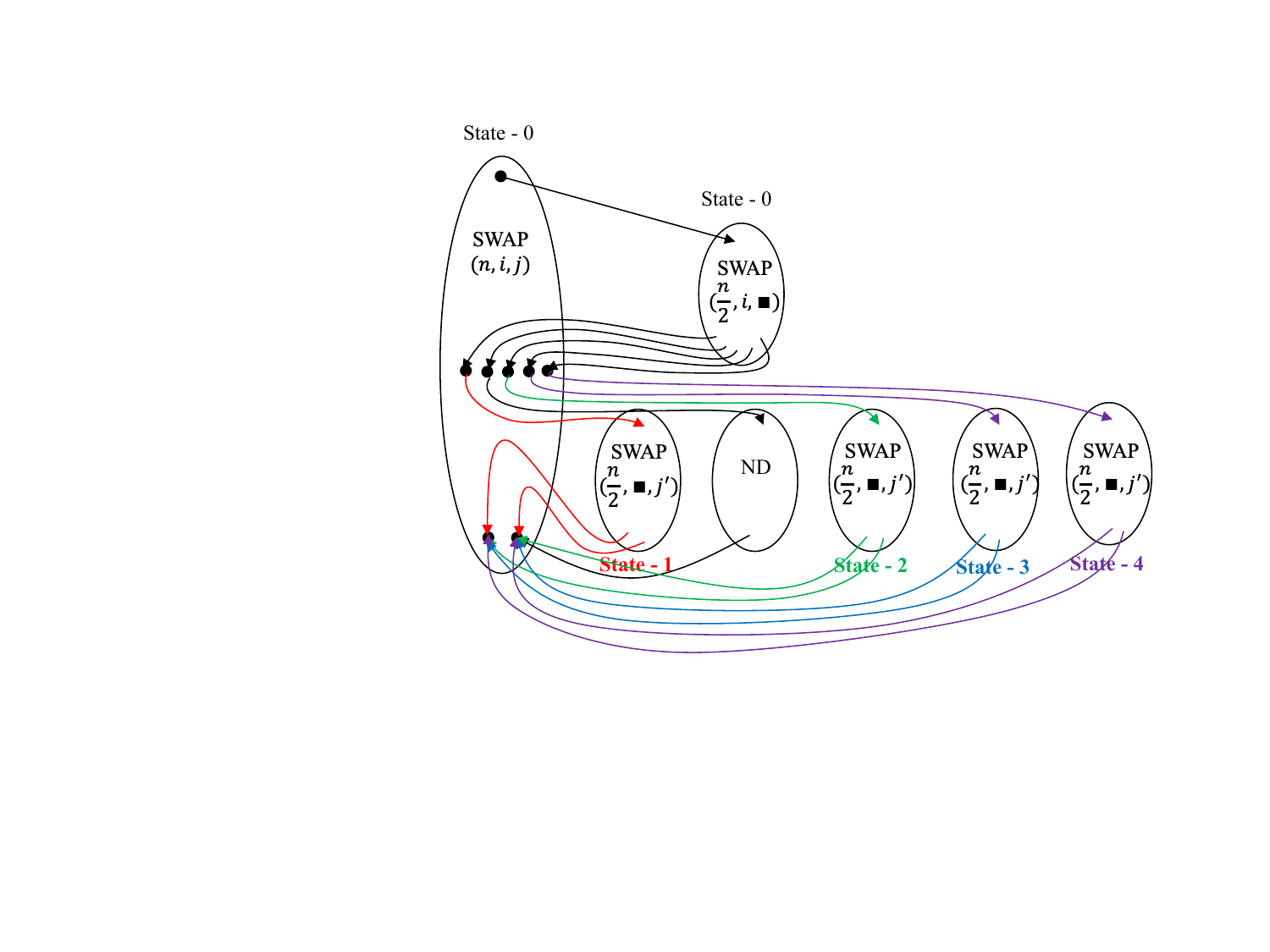}
      \vspace{-16pt}
      \caption[Case 3(b): $i$ in A and $j$ in B; controlled-bit != n/2-1]
        {\tabular[t]{@{}l@{}}Case 3(b): $i$ in A and $j$ in B; \\ controlled-bit $\neq$ $n/2$-1\endtabular}
    \end{subfigure}
    \begin{subfigure}[t]{0.4\linewidth}
      \centering
      \includegraphics[width=0.65\linewidth]{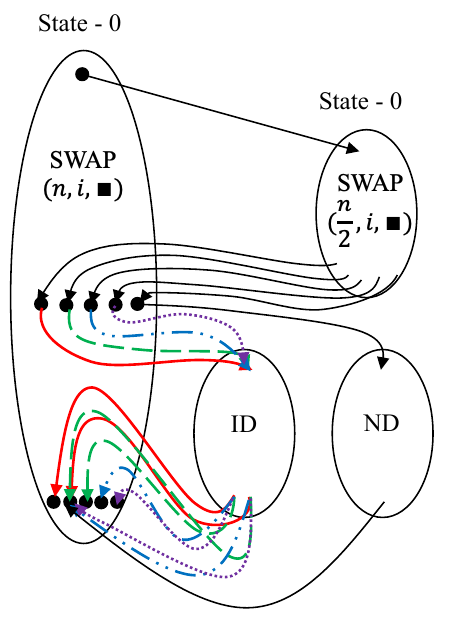}
      \vspace{-3pt}    
      \caption[Case 4(a): $i$ in A and $j$ not in current grouping's range; controlled-bit == n/2-1]
        {\tabular[t]{@{}l@{}}Case 4(a): $i$ in A and $j$ not in current range; \\ controlled-bit = $n/2$-1
        \endtabular
        }
    \end{subfigure}
    \hspace{4ex}
    \begin{subfigure}[t]{0.4\linewidth}
      \centering
      \includegraphics[width=0.65\linewidth]{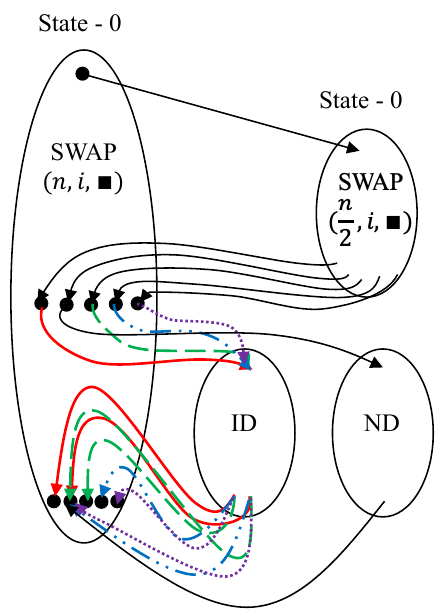}
      \vspace{-3pt}
      \caption[Case 4(b): $i$ in A and $j$ not in current grouping's range; controlled-bit != n/2-1]
        {\tabular[t]{@{}l@{}}Case 4(b): $i$ in A and $j$ not in current range; \\ controlled-bit $\neq$ $n/2$-1
        \endtabular}
    \end{subfigure}
    \caption{\protect \raggedright 
    The different cases of the \textit{SWAP} matrix construction.
    The text in each grouping denotes the function represented by the grouping.
    ID denotes {\tt IdentityMatrixGrouping}, and ND denotes a {\tt NoDistinctionProtoCFLOBDD} (used here for an all-zero matrix).
    \textit{SWAP} takes 4 arguments: $n$ for the number of bits (number of variables will be $2n$) in this proto-CFLOBDD;
    $i$ for the control-bit, $j$ for the controlled-bit, and $\textit{State} \in \{ 0,1,2,3,4 \}$ to indicate the current mode of the construction.
    $i'$ and $j'$ denote bit indices adjusted to the index range of the current level:
    $i' = i - n/2$; $j' = j - n/2$.
    A black square indicates that a particular index is outside the grouping's index range.
    }
    \label{Fi:Swap}
\end{figure}

\begin{figure}[tb!]
    \begin{subfigure}[t]{0.45\linewidth}
      \centering
      \includegraphics[width=0.8\linewidth]{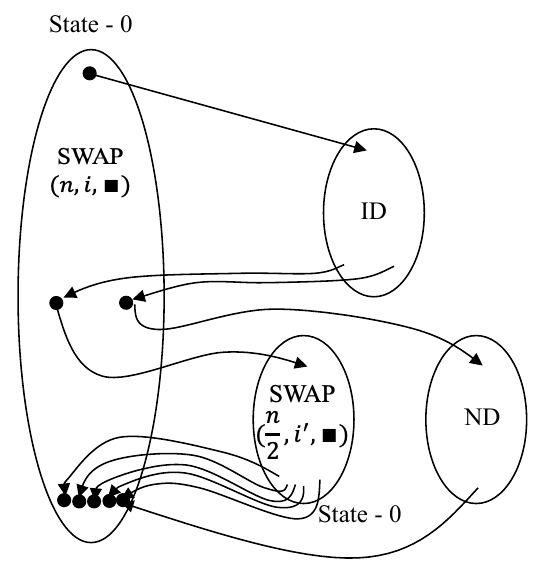}
      \vspace{-8pt}
      \caption[Case 5(a): $i$ in B and $j$ not in current range; controlled-bit = $n$-1]
        {\tabular[t]{@{}l@{}}Case 5(a): $i$ in B and $j$ not in current range; \\ control-bit = $n$-1
        \endtabular}
    \end{subfigure}
    \begin{subfigure}[t]{0.45\linewidth}
      \centering
      \includegraphics[width=0.8\linewidth]{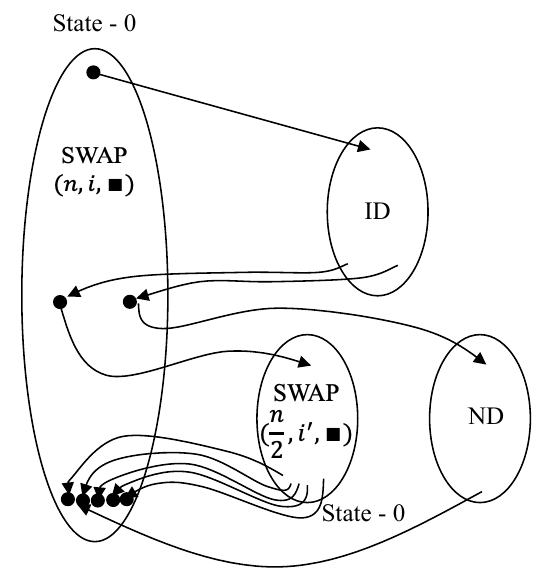}
      \vspace{-8pt}
      \caption[Case 5(b): $i$ in B and $j$ not in current range; controlled-bit $\neq$ $n$-1]
        {\tabular[t]{@{}l@{}}Case 5(b): $i$ in B and $j$ not in current range; \\ control-bit $\ne$ $n$-1\endtabular}
    \end{subfigure}
    \begin{subfigure}[t]{0.45\linewidth}
      \centering
      \includegraphics[width=0.8\linewidth]{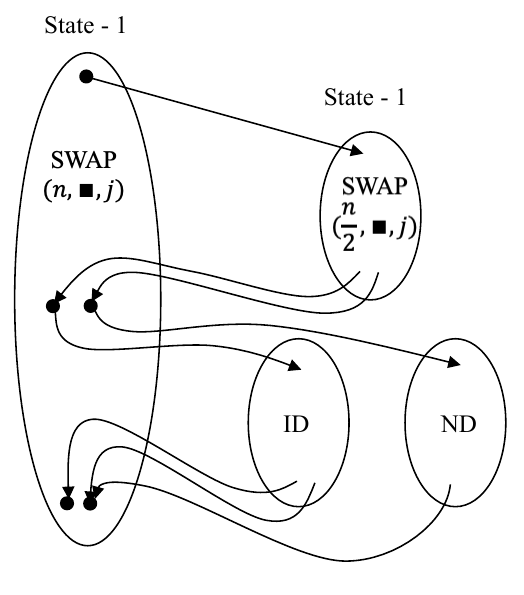}
      \vspace{-8pt}
      \caption[Case 6(a): $j$ in A and $i$ not in current range; $\textit{State}$ = 1]
        {\tabular[t]{@{}l@{}}Case 6(a): $j$ in A and $i$ not in current range;\\ $\textit{State}$ = 1\endtabular}
    \end{subfigure}
    \begin{subfigure}[t]{0.45\linewidth}
      \centering
      \includegraphics[width=0.8\linewidth]{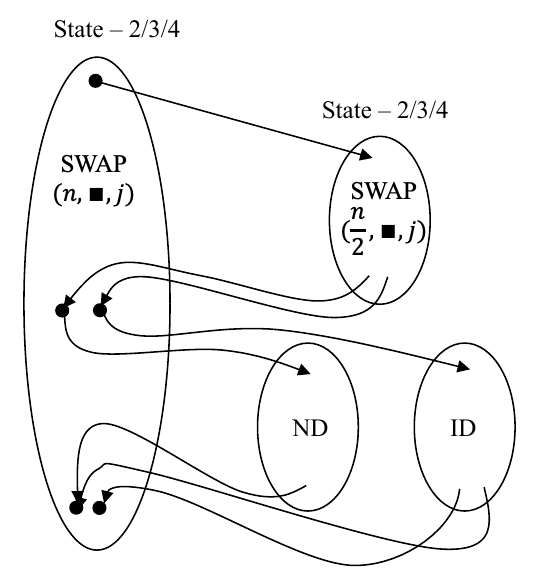}
      \vspace{-8pt}
      \caption[Case 6(b): $j$ in A and $i$ not in current range; $\textit{State}$ = 2,3, or 4]
        {\tabular[t]{@{}l@{}}Case 6(b): $j$ in A and $i$ not in current range;\\ $\textit{State}$ = 2, 3, or 4\endtabular}
    \end{subfigure}
    \begin{subfigure}[t]{0.45\linewidth}
      \centering
      \includegraphics[width=0.8\linewidth]{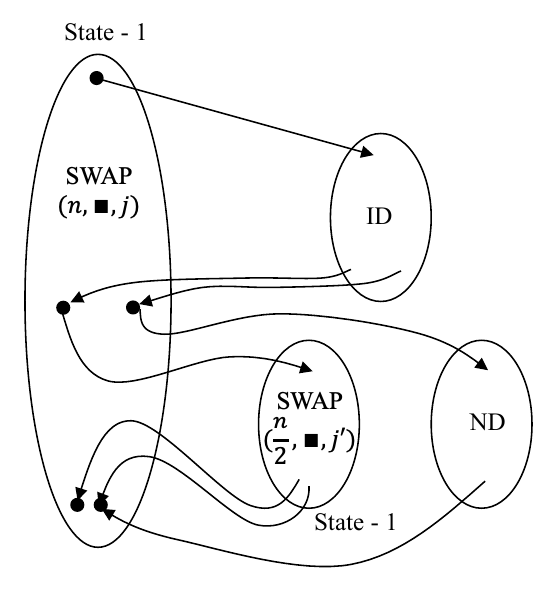}
      \vspace{-8pt}
      \caption[Case 7(a): $j$ in B and $i$ not in current range; $\textit{State}$ = 1]
        {\tabular[t]{@{}l@{}}Case 7(a): $j$ in B and $i$ not in current range;\\ $\textit{State}$ = 1\endtabular}
    \end{subfigure}
    \begin{subfigure}[t]{0.45\linewidth}
      \centering
      \includegraphics[width=0.8\linewidth]{figures/swap_Xb_0.pdf}
      \vspace{-8pt}
      \caption[Case 7(b): $j$ in B and $i$ not in current range; $\textit{State}$ = 2,3, or 4]
        {\tabular[t]{@{}l@{}}Case 7(b): $j$ in B and $i$ not in current range;\\ $\textit{State}$ = 2, 3, or 4\endtabular}
    \end{subfigure}
    \caption{The different cases of the \textit{SWAP} matrix construction, continued.}
    \label{Fi:SwapContd}
\end{figure}

\begin{figure}[tb!]
        \begin{subfigure}[t]{0.51\linewidth}
      \centering
      \includegraphics[width=\linewidth]{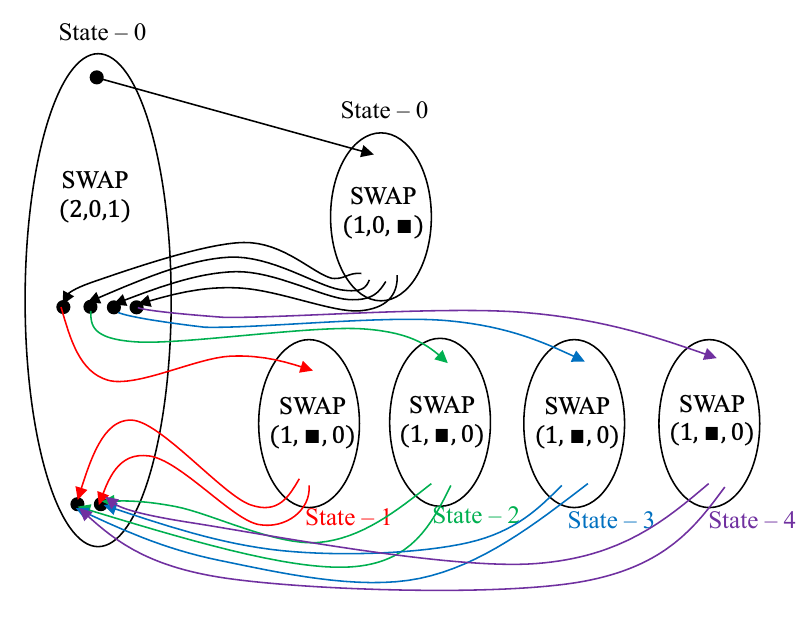}
      \vspace{-18pt}
      \caption{Base Case: $\textit{SWAP}(2,0,1)$ at level 2 (see \eqref{SwapGateTwoByTwo})}
    \end{subfigure}
    \begin{subfigure}[t]{0.45\linewidth}
      \centering
      \includegraphics[width=0.7\linewidth]{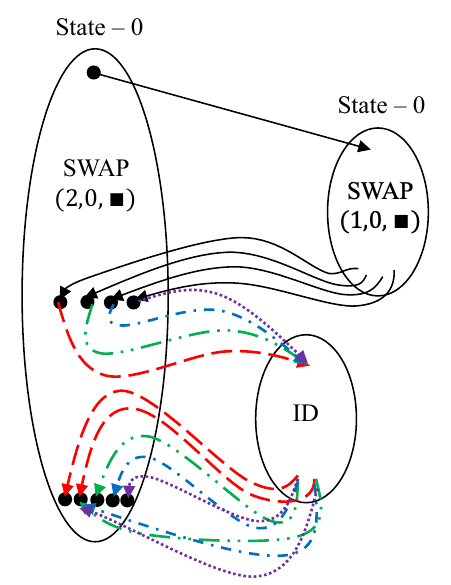}
      \vspace{-3pt}
      \caption{Base Case: $\textit{SWAP}(2,0,\blacksquare)$ at level 2
      }
    \end{subfigure}
    \begin{subfigure}[t]{0.45\linewidth}
      \centering
      \includegraphics[width=0.7\linewidth]{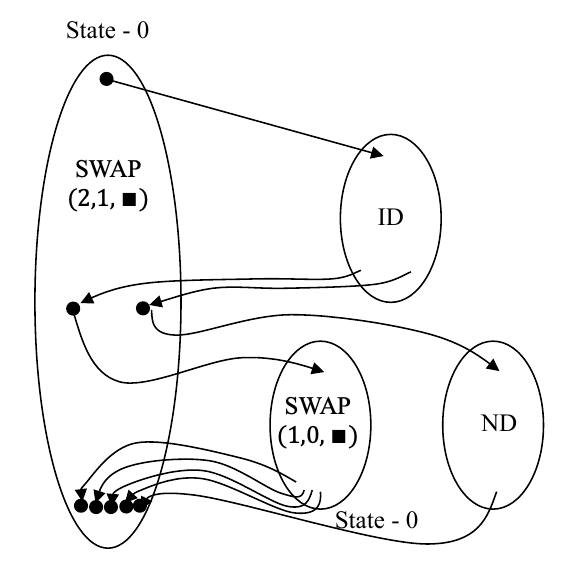}
      \vspace{-8pt}
      \caption{Base Case: $\textit{SWAP}(2,1,\blacksquare)$ at level 2}
    \end{subfigure}
    \begin{subfigure}[t]{0.45\linewidth}
      \centering
      \includegraphics[width=0.6\linewidth]{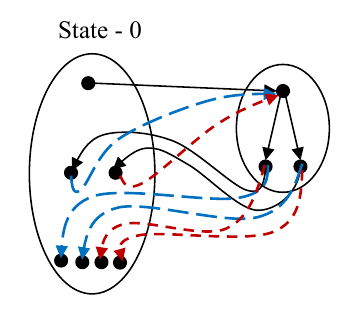}
      \vspace{-8pt}
      \caption{Base Case: $\textit{SWAP}(1,0,\blacksquare)$ at level 1; interprets control-bit ($\textit{State} = 0$)}
    \end{subfigure}
    \begin{subfigure}[t]{0.4\linewidth}
      \centering
      \includegraphics[width=0.8\linewidth]{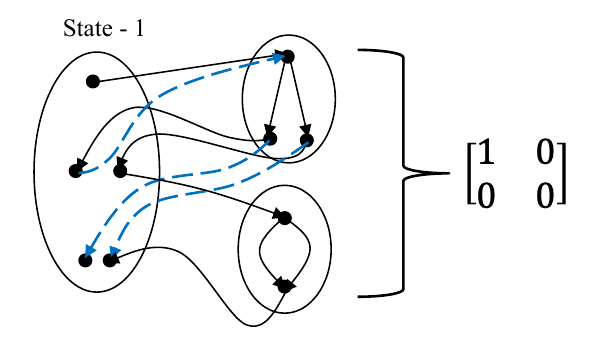}
      \vspace{-8pt}
      \caption{Base Case: $\textit{SWAP}(1,\blacksquare, 0)$ at level 1; interprets controlled-bit ($\textit{State} = 1$)}
    \end{subfigure}
    \hspace{3ex}
    \begin{subfigure}[t]{0.4\linewidth}
      \centering
      \includegraphics[width=0.8\linewidth]{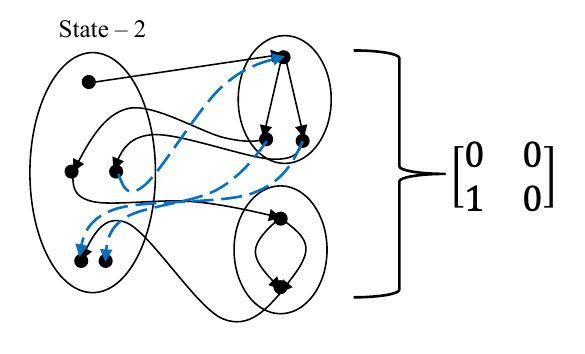}
      \vspace{-8pt}
      \caption{Base Case: $\textit{SWAP}(1,\blacksquare, 0)$ at level 1; interprets controlled-bit ($\textit{State} = 2$)}
    \end{subfigure}
    \begin{subfigure}[t]{0.4\linewidth}
      \centering
      \includegraphics[width=0.8\linewidth]{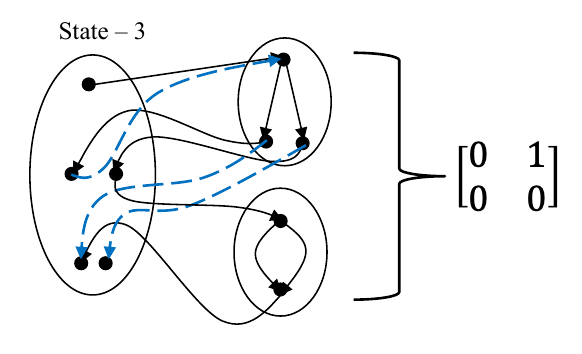}
      \vspace{-8pt}
      \caption{Base Case: $\textit{SWAP}(1,\blacksquare, 0)$ at level 1; interprets controlled-bit ($\textit{State} = 3$)}
    \end{subfigure}
    \hspace{3ex}
    \begin{subfigure}[t]{0.4\linewidth}
      \centering
      \includegraphics[width=0.8\linewidth]{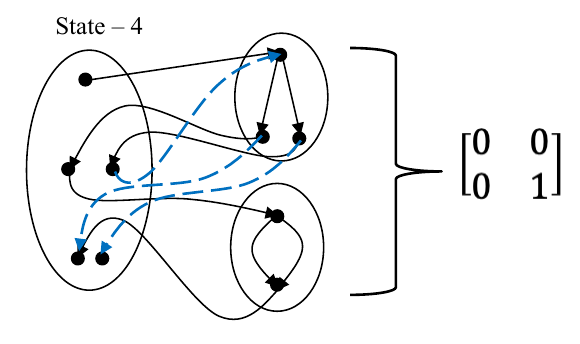}
      \vspace{-8pt}
      \caption{Base Case: $\textit{SWAP}(1,\blacksquare, 0)$ at level 1; interprets controlled-bit ($\textit{State} = 4$)}
    \end{subfigure}
    \caption{\protect \raggedright
    Base cases for the construction of the \textit{SWAP} matrix. There are three base cases at level 2 with 2 bits (i.e., 4 Boolean variables) and five base cases at level 1 with 1 bit (i.e., 2 Boolean variables).
    }
    \label{Fi:SwapContd2}
\end{figure}

\begin{figure}[tb!]
  \centering
  \begin{tabular}{c@{\hspace{8.0ex}}c}
    \includegraphics[align=c,scale=0.45]{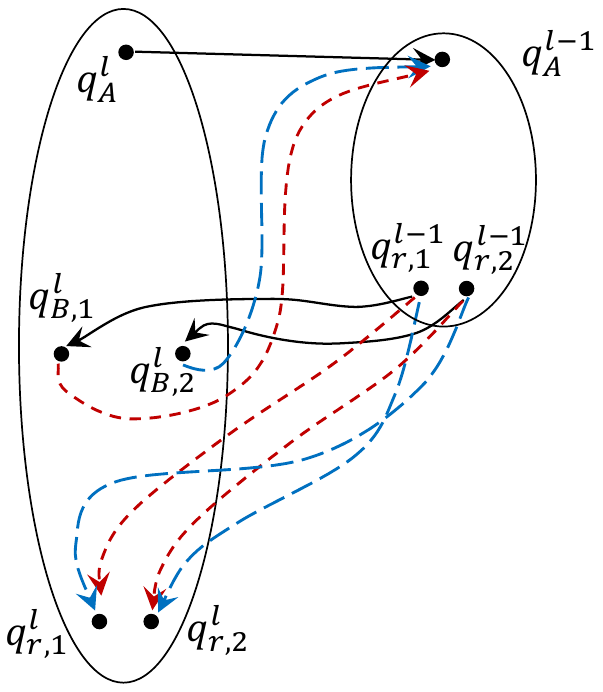}
    &
    \begin{tabular}{c}
      \includegraphics[scale=0.45]{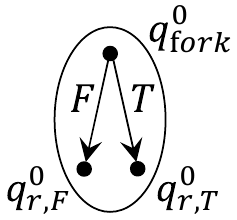}
      \\
      (b)
      \\
      \\
      \\
      \includegraphics[scale=0.45]{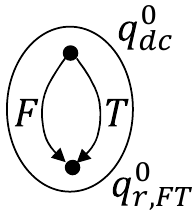}
      \\
      (c)
    \end{tabular}
    \\
    (a) &
    \\
    \\
    \multicolumn{2}{c}{
    \begin{tabular}{l|ll}
      \hline\hline
      \multirow{9}{*}{(a)} & \multirow{3}{*}{call transitions}     & $(q_A^l,\epsilon,q_A^{l-1}) \in \delta_c$ \\
                           &                                       & $(q_{B,1}^l,\epsilon,q_A^{l-1}) \in \delta_c$ \\
                           &                                       & $(q_{B,2}^l,\epsilon,q_A^{l-1}) \in \delta_c$ \\
                           \cline{2-3}
                           & \multirow{6}{*}{return transitions}   & $(q_{r,1}^{l-1},q_A^l,\epsilon,q_{B,1}^l) \in \delta_r$ \\
                           &                                       & $(q_{r,2}^{l-1},q_A^l,\epsilon,q_{B,2}^l) \in \delta_r$ \\
                           &                                       & $(q_{r,1}^{l-1},q_{B,1}^l,\epsilon,q_{r,1}^l) \in \delta_r$ \\
                           &                                       & $(q_{r,2}^{l-1},q_{B,1}^l,\epsilon,q_{r,2}^l) \in \delta_r$ \\
                           &                                       & $(q_{r,1}^{l-1},q_{B,2}^l,\epsilon,q_{r,2}^l) \in \delta_r$ \\
                           &                                       & $(q_{r,2}^{l-1},q_{B,2}^l,\epsilon,q_{r,1}^l) \in \delta_r$ \\
      \hline
      \multirow{2}{*}{(b)} & \multirow{2}{*}{internal transitions} & $(q^0_{\textrm{fork}}, F, q^0_{r,F}) \in \delta_i$ \\
                           &                      & $(q^0_{\textrm{fork}}, T, q^0_{r,T}) \in \delta_i$ \\
      \hline
      \multirow{2}{*}{(c)} & \multirow{2}{*}{internal transitions} & $(q^0_{\textrm{dc}}, F, q^0_{r,FT}) \in \delta_i$  \\
                           &                      & $(q^0_{\textrm{dc}}, T, q^0_{r,FT}) \in \delta_i$  \\
      \hline\hline
    \end{tabular}
    }
  \end{tabular}
  \caption{\protect \raggedright 
  (a) Encoding of a grouping's A-connection and B-connections as call transitions, and its return edges as return transitions of an NWA.
  The grouping is the one used to encode the family of Hadamard matrices $\HadamardFamily$.
  (b) and (c) Encoding of the two kinds of level-$0$ groupings as internal transitions of an NWA.}
  \label{Fi:CFLOBDDAsNWA}
\end{figure}

\section{Nested Words and Nested Word Automata}
\label{Se:NWADefinition}

\begin{definition}[\cite{JACM:AM09}]
\label{De:NestedWord}
A \textbf{nested word} $(w,\rightsquigarrow)$ over alphabet $\Sigma$
is an ordinary word $w \in \Sigma^*$, together with a \textbf{nesting relation}
$\rightsquigarrow$ of length $|w|$. $\rightsquigarrow$ is a collection of edges
(over the positions in $w$) that do not cross.
A nesting relation of length $l \geq 0$ is a subset of
$\{-\infty, 1, 2, \ldots, l \} \times \{1, 2, \ldots, l, +\infty\}$ such that
\begin{itemize}
  \item
    Nesting edges only go forwards: if $i \rightsquigarrow j$ then $i < j$.
  \item
    No two edges share a position: for $1 \leq i \leq l$,
    $|\{ j \mid i \rightsquigarrow j \}| \leq 1$
    and $|\{ j \mid j \rightsquigarrow i \}| \leq 1$.
  \item
    Edges do not cross: if $i \rightsquigarrow j$ and $i' \rightsquigarrow j'$,
    then one cannot have $i < i' \leq j < j'$.
\end{itemize}
When $i \rightsquigarrow j$ holds, for $1 \leq i \leq l$,
$i$ is called a \textbf{call} position; if $i \rightsquigarrow +\infty$,
then $i$ is a \textbf{pending call}; otherwise $i$ is a \textbf{matched call},
and the unique position $j$ such that $i \rightsquigarrow j$ is called
its \textbf{return successor}.
Similarly, when $i \rightsquigarrow j$ holds, for $1 \leq j \leq l$,
$j$ is a \textbf{return} position; if $-\infty \rightsquigarrow j$,
then $j$ is a \textbf{pending return}, otherwise $j$ is a \textbf{matched return},
and the unique position $i$ such that $i \rightsquigarrow j$ is called
its \textbf{call predecessor}.
A position $1 \leq i \leq l$ that is neither a call nor a return
is an \textbf{internal} position.

$\textbf{MatchedNW}$ denotes the set of nested words that have no
pending calls or returns.
$\textbf{NWPrefix}$ denotes the set of nested words that have no
pending returns.

A \textbf{nested word automaton} (NWA) $A$ is a 5-tuple $(Q, \Sigma,
q_0, \delta, F)$,
where $Q$ is a finite set of states, $\Sigma$ is a
finite alphabet, $q_0 \in Q$ is the initial state, $F \subseteq Q$ is a set
of final states, and $\delta$ is a transition relation. The transition
relation $\delta$ consists of three components, $(\delta_c, \delta_i,
\delta_r)$, where
\begin{itemize}
  \item
    $\delta_i \subseteq Q \times \Sigma \times Q$ is the transition
    relation for internal positions.
  \item
    $\delta_c \subseteq Q \times \Sigma \times Q$ is the transition
    relation for call positions.
  \item
    $\delta_r \subseteq Q \times Q \times \Sigma \times Q$ is the
    transition relation for return positions.
\end{itemize}

Starting from $q_0$, an NWA $A$ reads a nested word
$\textit{nw} = (w,\rightsquigarrow)$ from left to
right, and performs transitions (possibly non-deterministically)
according to the input symbol and $\rightsquigarrow$.
If $A$ is in state $q$ when reading input symbol $\sigma$ at position $i$ in $w$,
and $i$ is an internal or call position, $A$ makes a transition to $q'$
using $(q,\sigma,q') \in \delta_i$ or $(q,\sigma,q') \in \delta_c$, respectively.
If $i$ is a return position, let $k$ be the call predecessor of $i$,
and $q_c$ be the state $A$ was in just before the transition it made on the
$k^{\text{th}}$ symbol;
$A$ uses $(q,q_c,\sigma,q') \in \delta_r$ to make a transition to $q'$.
If, after reading $\textit{nw}$, $A$ is in a state $q \in F$,
then $A$ \textbf{accepts} $\textit{nw}$.
\end{definition}

\figref{CFLOBDDAsNWA} illustrates a schema by which a CFLOBDD can be translated to an NWA $M$.
Each matched path through the CFLOBDD corresponds to a nested word in MatchedNW for $M$. 
The matched-path principle is obeyed because of the ability of an NWA to ``peek'' at the state of the most-recent ``call'' to match a return edge with the appropriate preceding A-connection or B-connection.
All transitions taken at a level $\geq 1$ are $\epsilon$-transitions (\figref{CFLOBDDAsNWA}a).
The only transitions that consume an alphabet symbol are the $F$ and $T$ transitions of the level-$0$ fork grouping (\figref{CFLOBDDAsNWA}b) and the $F$ and $T$ transitions of the level-$0$ don't-care grouping (\figref{CFLOBDDAsNWA}c).

\begin{figure}[t!]
    \centering
    \begin{tabular}{c}
     \includegraphics[scale=0.4]{figures/walsh2_path-cropped.pdf} \\
    (a) \\
    \\
    \includegraphics[scale=0.46]{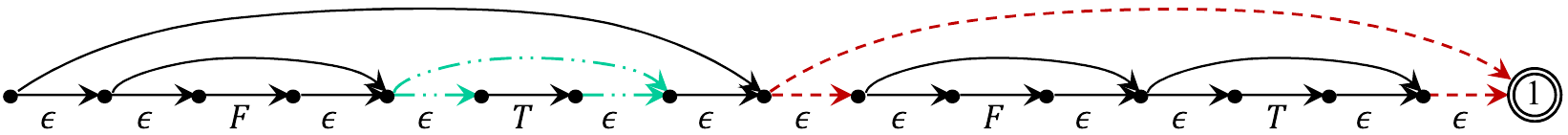} \\
    (b)
    \end{tabular}
    \vspace{-2ex}
    \caption{\protect \raggedright 
    (a) The CFLOBDD for the Hadamard matrix $H_4$, with the variable ordering $\langle x_0, y_0, x_1, y_1 \rangle$ (repeated from \figref{walsh4_0_3_path}).
    The matched path for $[x_0 \mapsto F, y_0 \mapsto T, x_1 \mapsto F, y_1 \mapsto T]$, which corresponds to $H_4[0,3]$ (with value $1$), is shown in bold.
    (b) The nested word for the path for  $[x_0 \mapsto F, y_0 \mapsto T, x_1 \mapsto F, y_1 \mapsto T]$.
    }
    \label{Fi:NestedWordExample}
\end{figure}

\figref{NestedWordExample} shows the nested word that corresponds to the path for the assignment  $[x_0 \mapsto F, y_0 \mapsto T, x_1 \mapsto F, y_1 \mapsto T]$ in the CFLOBDD for the Hadamard matrix $H_4$, with the variable ordering $\langle x_0, y_0, x_1, y_1 \rangle$.
\section{Time Complexity of Reduce}
\label{Se:CostOfReduce}

In this section, we give a bound on the time complexity of the call on Reduce (\algref{Reduce}) in \lineref{BAAR:CallReduce} of BinaryApplyAndReduce (\algref{BinaryApplyAndReduce}).
Let $C$ be the level-$l$ proto-CFLOBDD on which Reduce is invoked, and $C'$ be the level-$l$ proto-CFLOBDD that is returned.
The accounting is somewhat subtle because of three factors
\begin{itemize}
  \item
    hash-consing of groupings
  \item
    function caching of calls to Reduce and other functions
  \item
    for $C' = \textit{Reduce}(C, \red)$ (where $\red$ is some reduction tuple), for their respective top-level groupings, $g'$ and $g$, it is always the case that $|g'| \leq |g|$, yet $|C'|$ and $|C|$ have no fixed relationship:
    $|C'| < |C|$, $|C'| = |C|$ and $|C'| > |C|$ are all possible.\footnote{
\changed{
      \label{Footnote:LocalGlobalReduceProperty}
      We refer to the property that $|g'| \leq |g|$ as the \emph{local-reduction property}, in contradistinction to the absence of a \emph{global-reduction property} for $|C'|$ and $|C|$.
}
    }
\end{itemize}
The size measure $|\cdot|$ counts vertices and edges (and, for proto-CFLOBDDs, groupings---with no double-counting of shared groupings due to hash-consing).

In this section, we show that the time complexity of Reduce is bounded by $O(|C| \times {|C'|})$, where when counting the time for operations, we consider the cost of function-caching operations (lookup and update) to be O(1).

We illustrate the point about there being no fixed relationship between $C'$ and $C$ with the following example, which shows that when Reduce is called on a proto-CFLOBDD, it can lead to both
(i) less sharing of proto-CFLOBDDs in the resultant proto-CFLOBDD, and
(ii) more sharing of proto-CFLOBDDs than in the input proto-CFLOBDD.

\begin{figure}[tb!]
    \centering
    \begin{subfigure}[t]{0.35\linewidth}
    \centering
        \includegraphics[width=0.9\linewidth]{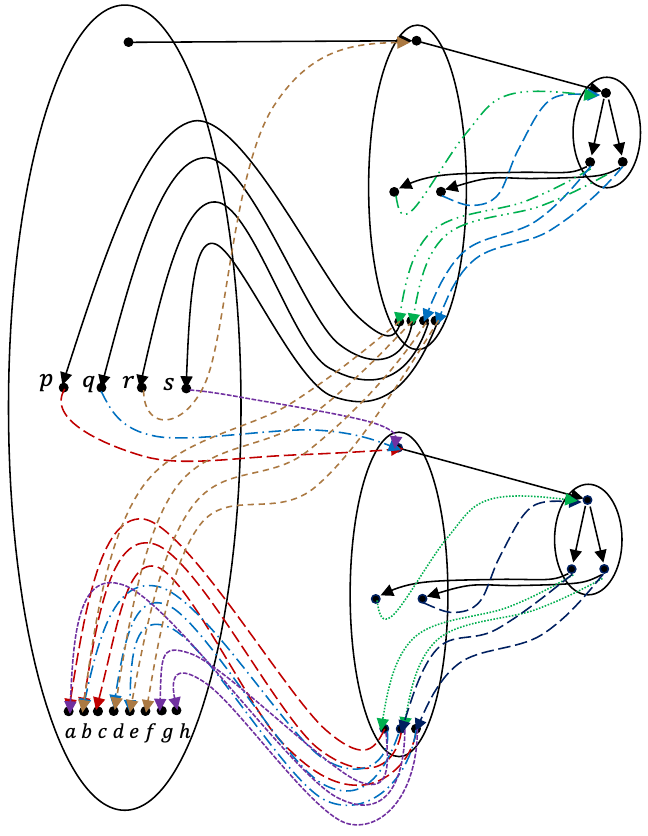}
    \caption{Example $C$}
    \label{Fi:reduce-c}
    \end{subfigure}
    \begin{subfigure}[t]{0.55\linewidth}
    \centering
        \includegraphics[width=0.9\linewidth]{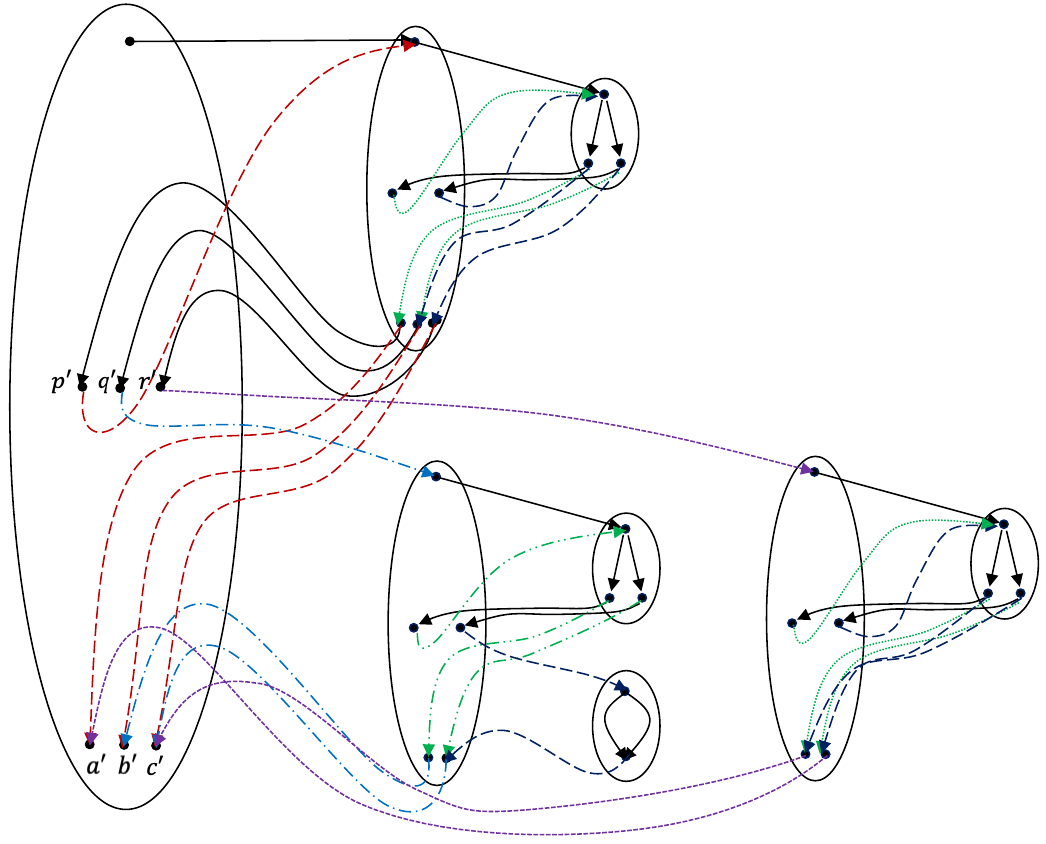}
    \caption{Example $C'$}
    \label{Fi:reduce-c-prime}
    \end{subfigure}
    \caption{
    $C' = \text{Reduce}(C, [1,2,3,3,3,3,3,3])$.
    The colors of the edges to proto-CFLOBDDs in $C'$ correspond to the edges to the originating proto-CFLOBDDs in $C$.
    }
    \label{Fi:ReduceTwoLevel}
\end{figure}

\begin{Exa}\label{Exa:ReduceTwoLevel}
Consider the level-$2$ proto-CFLOBDD $C$ shown in~\figref{reduce-c}, which has four middle vertices ($p$, $q$, $r$, $s$) and eight exit vertices ($a$, $b$, $c$, $d$, $e$, $f$, $g$, $h$).
The A-connection of $C$ ($C_A$) is a proto-CFLOBDD at level $1$.
$C_A$ partitions the strings $\{0,1\}^2$ into $P1$ = $[\{00\}, \{01\},\{10\},\{11\}]$,
i.e., $C_A$ has four exit vertices and thus $C$ has four middle vertices and four B-connections ($C_{B1}$, $C_{B2}$, $C_{B3}$, $C_{B4}$).
$C_{B1}$ partitions the strings $\{0,1\}^2$ into $P2$ = $[\{00\}, \{01, 10\}, \{11\}]$ and its exit vertices are connected to the exit vertices of $C$ in the order ($a, b, c$).
$C_{B2}$ and $C_{B4}$ both equal $C_{B1}$, but for $C_{B2}$ and $C_{B4}$ the three exit vertices are connected to the exit vertices of $C$ in the orders ($b, d, e$), and ($g, a, h$), respectively.
Finally, $C_{B3}$ equals $C_A$, but for $C_{B3}$ the four exit vertices are connected to $C$'s exit vertices in the order ($b, d, e, f$).
Consequently, $C$ partitions the strings $\{0,1\}^4$ into $[\{0000,1101,1110\}$, $\{0001,0010,0100,1000\}$,$\{0011\}$,$\{0101,0110,1001\}$,$\{0111,1010\}$,$\{1011\}$,$\{1100\}$,$\{1111\}]$.

Let $C' = \text{Reduce}(C, [1,2,3,3,3,3,3,3])$, i.e., the vertices $c,d,e,f,g,h$ are all mapped to exit vertex $c$.
$C'$ is shown in~\figref{reduce-c-prime}.
$C'$ has only three exit vertices ($a'$, $b'$, $c'$).
Consider how $C$ is ``reduced'' to $C'$, which partitions the strings $\{0,1\}^4$ into $[\{0000$,$1101$,$1110\}$, $\{0001$,$0010$,$0100$,$1000\}$,$\{0011$,$0101$,$0110$,$1001$,$0111$,$1010$,$1011$,$1100$,$1111\}]$.
\begin{itemize}
  \item
    $C_{B1}$'s exit vertices are mapped to ($a,b,c$), which leads to the call $\text{Reduce}(C_{B1},[1,2,3])$ and thus $C_{B1}$ does not change;
    that is, the first B-connection of $C'$, $C'_{B1}$, is equal to $C_{B1}$.
    Its exit vertices are connected to the exit vertices ($a'$, $b'$, $c'$) of $C'$.
    (As we will see below, $C'_{B1}$ is also equal to $C'_{A}$, the A-connection of $C'$.)
  \item
    $C_{B2}$'s exit vertices are mapped to ($b, c, c$), which leads to the call $\text{Reduce}(C_{B2},[1,2,2])$.
    Therefore, the second and third exit vertices are folded together, and this collapse affects the structure of the level-$0$ groupings as well, thereby creating a new proto-CFLOBDD, $C'_{B2}$, which partitions the strings $\{0,1\}^2$ into $[\{00\},\{01,10,11\}]$.
    The exit vertices of $C'_{B2}$ are mapped to exit vertices ($b',c'$) of $C'$.
  \item
    $C_{B3}$'s exit vertices are mapped to ($b,c,c,c$), which leads to the call $\text{Reduce}(C_{B3},[1,2,2,2])$.
    Thus, the exit vertices of $C_{B3}$ are collapsed to only two exit vertices, and the resulting proto-CFLOBDD partitions the strings $\{0,1\}^2$ into $[\{00\},\{01,10,11\}]$, which are mapped to the exit vertices $(b',c')$.
    This result is identical to the result from $\text{Reduce}(C_{B2},[1,2,2])$, and thus $C'$ has only one copy of $C'_{B2}$ with its exit vertices mapped to exit vertices ($b',c'$) of $C'$.
  \item
    $C_{B4}$'s exit vertices are mapped to ($c,a,c$), which leads to the call $\text{Reduce}(C_{B4},[1,2,1])$---folding together the first and third exit vertices.
    This call creates yet another new proto-CFLOBDD, $C'_{B3}$, which partitions the strings $\{0,1\}^2$ into $[\{00,11\},\{01,10\}]$.
    The exit vertices of $C'_{B3}$ are mapped to exit vertices ($c', a'$) of $C'$.
  \item
    Because the calls $\text{Reduce}(C_{B2},[1,2,2])$ and $\text{Reduce}(C_{B3},[1,2,2,2])$ produce the same proto-CFLOBDD with the same return edges in $C'$---and because the calls on $\text{Reduce}$ arose in the B-connection of the same grouping in $C$---middle vertices ($q$, $r$) of $C$ are folded together.
    This collapsing is propagated to the A-connection of $C$ by the call $\text{Reduce}(C_A,[1,2,2,3])$.
    The resulting proto-CFLOBDD has three exit vertices that partition the strings $\{0,1\}^2$ into $[\{00\},\{01,10\},\{11\}]$.
    This proto-CFLOBDD is identical to $C'_{B1}$---although their exit vertices are mapped to different vertices of $C'$:
    the exit vertices of $C'_{B1}$ are connected to exit vertices ($a'$, $b'$, $c'$) of $C'$, whereas the exit vertices of $C'_{A}$
    are connected to middle vertices ($p'$, $q'$, $r'$) of $C'$.
\end{itemize}

We see from this example that a call $C' = \text{Reduce}(C,\red)$ can cause entirely new proto-CFLOBDDs to be created in $C'$;
proto-CFLOBDDs that occur in $C$ to occur in entirely different places in $C'$;
proto-CFLOBDDs that occur in $C$ to not occur in $C'$; and
two or more proto-CFLOBDDs with identical sets of return edges to be combined into just a single occurrence when they arise in the B-connection of the same enclosing grouping.
This example highlights the challenges for establishing a bound on the time complexity of Reduce---namely, both expansion and compaction of proto-CFLOBDDs can occur.
\end{Exa}

Because of the effects illustrated in \exref{ReduceTwoLevel}, the cost-bound argument we give is slightly indirect.
At a high-level, it is structured as follows:
we establish a relationship between $\text{Reduce}(C,\red)$ and that of a certain call on $\text{PairProduct}$ (\theoref{LanguagePartitions}).
This approach is beneficial because we already know a time bound on $\text{PairProduct}$ in terms of the product of the sizes of $\text{PairProduct}$'s arguments (which is expressed more precisely in \footnoteref{ProductConstructionCost}).
\theoref{ReduceCostBound} uses that bound to give an asymptotic bound on the time to perform $\text{Reduce}(C,\red)$ in terms of the product of the sizes of its input and output CFLOBDDs.

\begin{theorem}\label{The:LanguagePartitions}
  Let $C$ and $C'$ be two proto-CFLOBDDs such that $C' = \textit{Reduce}(C,\red)$ for some reduction tuple $\red$.
  Then $C = \text{PairProduct}(C, C')$.\footnote{
    To reduce clutter, we ignore the tuple of pairs of exit vertices that is returned by $\text{PairProduct}$ (\algref{PairProduct}), except for two places in \theoref{ReduceCostBound}.
  }
\end{theorem}
\begin{Proof}
We know that each proto-CFLOBDD at level $k$ with $m$ exit vertices partitions the space of strings $\{0,1\}^{2^k}$ into $m$ groups (see \sectref{ADenotationalSemantics}).
We make use of the properties of Reduce and PairProduct with respect to such partitions:
\begin{enumerate}
  \item 
    \label{It:ReduceCoarsensPartitions}
    For every proto-CFLOBDD $X$ and reduction tuple $\red$, Reduce(X, \red) produces a coarser partition of the exit languages of $X$ defined by the mapping of $\red$ to $X$'s exit vertices.
  \item
    \label{It:PairProductAsLeastCoarsestRefinement}
    For every pair of proto-CFLOBDDs $X$ and $Y$, $\PairProduct(X,Y)$ produces the coarsest partition that refines both of the partitions corresponding to $X$ and $Y$.
\end{enumerate}
In particular, we consider the two-statement sequence
\begin{align}
  C' & = \text{Reduce}(C,\red);    \label{Eq:TwoStatementReduce} \\
  \tilde{C} & = \text{PairProduct}(C, C');  \label{Eq:TwoStatementPairProduct}
\end{align}
$C'$ created in \eqref{TwoStatementReduce} represents a coarser partition of the strings in $\{ 0, 1 \}^n$ than $C$'s partition.
Because $C'$ represents a coarser partition than $C$, the proto-CFLOBDD $\tilde{C}$ created in \eqref{TwoStatementReduce} represents the same partition as $C$, and thus $\tilde{C}$ and $C$ are equal by canonicity.  $~~\QED$
\end{Proof}

\medskip
In essence, \theoref{LanguagePartitions} shows that $\text{PairProduct}(C, C')$ ``undoes'' all of the actions taken during $\text{Reduce}(C,\red)$.

\begin{Exa}\label{Exa:PairProductOfReduce}
Consider the result $\tilde{C}$ = PairProduct($C$, $C'$) for $C$ and $C'$ from \exref{ReduceTwoLevel}.
PairProduct($C$, $C'$) is first called on the A-connections of the respective outermost groupings, followed by calls on B-connections.
\begin{itemize}
  \item
    PairProduct($C_A$, $C'_A$) produces a proto-CFLOBDD whose exit vertices represent the coarsest partition of $\{0,1\}^2$ that refines both of the partitions corresponding to the exit vertices of $C_A$ and $C'_A$ (i.e., $[\{00\},\{01\},\{10\},\{11\}]$ and $[\{00\},\{01,10\},\{11\}]$, respectively). Hence, the new proto-CFLOBDD $\tilde{C}_A$ is constructed such that the exit vertices of $\tilde{C}_A$ represent the partition $[\{00\},\{01\},\{10\},\{11\}]$. PairProduct also returns a tuple of index-pairs indicating the B-connections on which PairProduct needs to be called. In this case, the returned tuple is $[[1,1],[2,2],[3,2],[4,3]]$. Mapping this result to the middle vertices of $C$ and $C'$, we obtain $[[p,p'],[q,q'],[r,q'],[s,r']]$. These pairs are processed left-to-right, generating calls to PairProduct on B-connections.
  \item
    PairProduct($C_{B1}$, $C'_{B1}$) (corresponding to the pair $[p,p']$) creates proto-CFLOBDD $\tilde{C}_{B1}$ with three exit vertices corresponding to the partition $[\{00\},\{01,10\},\{11\}]$, returning the tuple $[[1,1],[2,2],[3,3]]$. Mapping this result to the exit vertices of $C$ and $C'$, the initial (as-yet incomplete) sequence of exit vertices of $\tilde{C}$ would be $[[a,a'],[b,b'],[c,c']]$.
  \item
    PairProduct($C_{B2}$, $C'_{B2}$) (corresponding to the pair $[q,q']$) creates proto-CFLOBDD $\tilde{C}_{B2}$ with three exit vertices corresponding to the partition $[\{00\},\{01,10\},\{11\}]$ (the same as $\tilde{C}_{B1}$), returning the tuple $[[1,1],[2,2],[3,2]]$.
    Mapping this result to the exit vertices of $C$ and $C'$,
    the exit vertices of $\tilde{C}$ would be extended to be $[[a,a'], [b,b'], [c,c'],[d,c'],[e,c']]$, and the exit vertices of $\tilde{C}_{B2}$ would be connected to $[b,b']$, $[d,c']$, and $[e,c']$. 
  \item
    PairProduct($C_{B3}$, $C'_{B2}$) (corresponding to the pair $[r,q']$) creates proto-CFLOBDD $\tilde{C}_{B3}$ with four exit vertices corresponding to the partition $[\{00\},\{01\},\{10\},\{11\}]$ (the same as $\tilde{C}_A$), returning the tuple $[[1,1],[2,2],[3,2],[4,2]]$.
    Mapping this result to the exit vertices of $C$ and $C'$, the exit vertices of $\tilde{C}$ would be extended to be $[[a,a'], [b,b'], [c,c'],[d,c'],[e,c'],[f,c']]$, and the exit vertices of $\tilde{C}_{B3}$ would be connected to $[b,b']$, $[d,c']$, $[e,c']$, and $[f,c']$.
  \item
    PairProduct($C_{B4}$, $C'_{B3}$) (corresponding to the pair $[s,r']$) creates proto-CFLOBDD $\tilde{C}_{B4}$ with three exit vertices corresponding to the partition $[\{00\},\{01, 10\},\{11\}]$ (again, the same as $\tilde{C}_{B1}$), returning the tuple $[[1,1],[2,2],[3,1]]$. Mapping this result to the exit vertices of $C$ and $C'$, the final sequence of exit vertices of $\tilde{C}$ would be set to $[[a,a'], [b,b'], [c,c'],[d,c'],[e,c'],[f,c'],[g,c'],[h,c']]$, and the exit vertices of $\tilde{C}_{B4}$ would be connected to $[g,c']$, $[a,a']$, and $[h,c']$.
\end{itemize}
$\tilde{C}$ has eight exit vertices, four middle vertices, and each of the A-connections and B-connections of $\tilde{C}$ and $C$ are connected to isomorphic proto-CFLOBDDs.
Consequently, $\tilde{C}$ = $C$ up to isomorphism.
Because hash-consing enforces that the members of each isomorphism class have a unique representation in memory, $\text{PairProduct}(C,C')$ would return a pointer to $C$.
%
\end{Exa}

\begin{lemma}\label{Lem:SizeOfGrouping}(Local-Reduction Property).
    Let $C$ and $C'$ be two proto-CFLOBDDs such that $C' = \textit{Reduce}(C,\red)$ for some reduction tuple $\red$, and let $g$ and $g'$ be their respective outermost groupings.
    Then $|g'| \leq |g|$.
\end{lemma}
\begin{Proof}
    The size of a grouping is equal to the number of entry, middle, and exit vertices, plus the number of A-connection and B-connection edges and return edges.
    Because $g'$ is obtained by reducing $g$ with respect to $\red$, the number of exit vertices in $g'$ can be no more than the number in $g$.
    Moreover, $\textit{Reduce}$ can never cause there to be more B-connections in $g'$ than in $g$, but it can cause some B-connections of $g$ to be folded together in $g'$;
    thus, the number of middle vertices in $g'$ can be no more than the number in $g$.
    Similarly, for the A-connection of $g'$ and all the B-connections of $g'$, the number of return edges can be no more than the number of return edges in the corresponding A-/B-connections in $g$.
    Consequently, $|g'| \leq |g|$.   $~~\QED$
\end{Proof}

\begin{Exa}\label{Exa:LocalReduceTwoLevel}
  Consider the proto-CFLOBDDs $C$ and $C'$ from \figref{ReduceTwoLevel}.
  The size of the level-$2$ grouping $g$ equals 1 (entry-vertex) + 4 (middle vertices) + (1 + 3) ($1^{\textit{st}}$ B-connection) + (1 + 3)($2^{\textit{nd}}$ B-connection) + (1 + 4) ($3^{\textit{rd}}$ B-connection) + (1 + 3)($4^{\textit{th}}$ B-connection) + 8 (exit vertices) = 30.

  The size of $g'$ equals 1 (entry-vertex) + 3 (middle vertices) + (1 + 3) ($1^{\textit{st}}$ B-connection) + (1 + 2)($2^{\textit{nd}}$ B-connection) + (1 + 2) ($3^{\textit{rd}}$ B-connection) + 3 (exit vertices) = 17.

  Thus, $|g'| \leq |g|$, whereas $68 = |C'| > |C| = 66$.
\end{Exa}

We now turn to the question of bounding the time complexity of $\text{Reduce}$.
Whereas \theoref{LanguagePartitions} showed that $\text{PairProduct}(C, C')$ ``undoes'' all of the actions taken during $\text{Reduce}(C,\red)$, \theoref{ReduceCostBound} shows that 
for every action in $\text{Reduce}(C,\red)$, there is an action of at least the same cost in $\text{PairProduct}(C, C')$.
Consequently, the time to perform $\text{Reduce}(C)$ is bounded by the time that it would take to perform $\text{PairProduct}(C, C')$, which is $O(|C| \times |C'|)$.

\begin{theorem}\label{The:ReduceCostBound}
  Let $C$ and $C'$ be two proto-CFLOBDDs such that $C' = \textit{Reduce}(C,\red)$ for some reduction tuple $\red$. 
  Let $\Cost(\text{Reduce}(C))$ and $\Cost(PP(C, C'))$ denote the costs of $\text{Reduce}(C,\red)$ and $\PairProduct(C,C')$, respectively.
  Then $\Cost(Reduce(C)) \leq \Cost(PP(C, C'))$.
\end{theorem}
\begin{Proof}
    The proof is by induction on the level $k$ of proto-CFLOBDDs $C$ and $C'$.

\begin{BaseCase} ($k = 0$)
        Consider the following table,
        \[
        \begin{array}{c|c|c|c}
             C & \red & C' = \text{Reduce}(C,\red) & \text{PairProduct}(C, C')\\
             \hline
             \text{ForkGouping} & [1,1] & \text{DontCareGrouping} & [\text{ForkGrouping}, ([1, 1], [2, 1])]\\
             \text{DontCareGrouping} & [1]  & \text{DontCareGrouping} & [\text{DontCareGrouping}, ([1, 1])]\\
             \text{ForkGouping} & [1,2]  & \text{ForkGouping} & [\text{ForkGouping}, ([1, 1], [2, 2])]\\
             \text{DontCareGrouping} & --  & \text{ForkGouping} & \text{Not Applicable}\\
        \end{array}
        \]
        The last line in the table cannot arise because there is no reduction tuple that can be used to reduce a DontCareGrouping to a Fork Grouping.        
        In each of the other three cases in the table, PairProduct($C, C')$ returns a tuple that has $C$ as the first component.

        \hspace{1.5ex}
        Moreover, the results produced by Reduce($C$) and PairProduct($C, C'$) are of constant size, and could be implemented by table lookup.
        The return value from PairProduct($C, C'$) is larger than the return value from Reduce($C,\red$), which justifies saying that $\Cost(Reduce(C)) \leq \Cost(PP(C, C'))$.
\end{BaseCase}

\begin{InductionStep}  

\textit{Induction Hypothesis:}
Assume that for all level-$k$ proto-CFLOBDDs $C'_k$ and $C_k$ for which $C'_k = \textit{Reduce}(C_k,\red)$, for some reduction tuple $\red$,
$\Cost(Reduce(C_k)) \leq \Cost(PP(C_k, C'_k))$.

\hspace{1.5ex}
Consider two level-$k\text{+}1$ proto-CFLOBDDs, $C'_{k+1}$ and $C_{k+1}$, such that  $C'_{k+1} = \text{Reduce}(C_{k+1}, \red)$.
The proof breaks down into the following three cases:
\begin{description}

  \item [(i) A-connections.]
    PairProduct is first called recursively on $C_{k+1}.A$ and $C'_{k+1}.A$---i.e., the level-$k$ A-connections of $C_{k+1}$ and $C'_{k+1}$, respectively.
    By the construction of $C'_{k+1}$ from $C_{k+1}$, we know that $C'_{k+1}.A = \text{Reduce}(C_{k+1}.A, \red_A)$ for some reduction tuple $\red_A$.
    Thus, by the induction hypothesis,
    \begin{equation}
      \label{Eq:CostInequalityA}
      \Cost(Reduce(C_{k+1}.A)) \leq \Cost(PP(C_{k+1}.A, C'_{k+1}.A)).
    \end{equation}

  \item [(ii) B-connections.]
    The return value from the call on $\text{PairProduct}(C_{k+1}.A, C'_{k+1}.A)$ considered in the previous case is actually a tuple $[\tilde{C}_{k+1}.A, \textit{midVertexPairs}]$.
    By the construction of $C'_{k+1}$ from $C_{k+1}$, we know that $C'_{k+1}.A = \text{Reduce}(C_{k+1}.A, \red_A)$ for some reduction tuple $\red_A$, and thus by \theoref{LanguagePartitions}, $\tilde{C}_{k+1}.A = C_{k+1}.A$.

    \hspace{1.5ex}
    For every $(i,j) \in \textit{midVertexPairs}$, PairProduct is called recursively on $C_{k+1}.B[i]$ and $C'_{k+1}.B[j]$, which are level-$k$ proto-CFLOBDDs.
    To be able to invoke the induction hypothesis, we must establish that $C'_{k+1}.B[j] = \textit{Reduce}(C_{k+1}.B[i], \red_{B[i]})$, for some $\red_{B[i]}$.

    \hspace{1.5ex}
    We now consider the meaning of the pairs $(i,j) \in \textit{midVertexPairs}$ from the standpoint of the language partitions used in \theoref{LanguagePartitions}.
    Because $C_{k+1}$ represents a finer partition of the strings in $\{0, 1\}^{2^{k+1}}$ than $C'_{k+1}$, the $i^{\textit{th}}$ exit vertex of $C_{k+1}.A$ represents a finer partition of the strings in $\{0, 1\}^{2^k}$ than the $j^{\textit{th}}$ exit vertex of $C'_{k+1}.A$.
    Thus, in general, there can be multiple exit vertices $i_1, i_2, \ldots, i_p$ of $C_{k+1}.A$ whose language partitions were combined to create the language partition of the $j^{\textit{th}}$ exit vertex of $C'_{k+1}.A$.

    \hspace{1.5ex}
    Because these vertices are exit vertices of A-connections, we can equivalently refer to the set $\{ i_1, i_2, \ldots, i_p \}$ of middle vertices of $C_{k+1}$ and the $j^{\textit{th}}$ middle vertex of $C'_{k+1}$.
    The reason this combining of languages took place during $\text{Reduce}(C_{k+1},\red)$ can only be because there were calls on $\text{Reduce}(C_{k+1}.B[i_1],\red_1)$, $\text{Reduce}(C_{k+1}.B[i_2],\red_2)$, $\ldots$, $\text{Reduce}(C_{k+1}.B[i_p],\red_p)$, for which the results were all equal to $C'_{k+1}.B[j]$.
    (The fifth bullet point of \exref{ReduceTwoLevel} illustrates how calls to $\text{Reduce}$ on two different B-connections in the same grouping yield the same result, which folds together two middle vertices of the grouping---thereby unioning their language partitions in the proto-CFLOBDD returned by $\text{Reduce}$.)
    Consequently, by the induction hypothesis,
    \begin{equation}
      \label{Eq:CostInequalities}
      \begin{split}      
      \Cost(Reduce(C_{k+1}.B[i_1])) & \leq \Cost(PP(C_{k+1}.B[i_1], C'_{k+1}.B[j])) \\
      \Cost(Reduce(C_{k+1}.B[i_2])) & \leq \Cost(PP(C_{k+1}.B[i_2], C'_{k+1}.B[j])) \\
      \ldots \\
      \Cost(Reduce(C_{k+1}.B[i_p])) & \leq \Cost(PP(C_{k+1}.B[i_p], C'_{k+1}.B[j]))
      \end{split}
    \end{equation}
    Let $e_A$ denote the number of exit vertices of $C_{k+1}.A$ (which is also the number of middle vertices of $C_{k+1}$).
    These inequalities can be expressed more succinctly by observing that for each index $i$, $1 \leq i \leq e_A$ on the left-hand side (corresponding to an A-connection language-partition of $C_{k+1}.A$), there is a unique $j$ to use on the right-hand side of the inequality.
    (Index $j$ corresponds to the coarsened A-connection language-partition of $C'_{k+1}.A$.)
    Let \emph{$\text{reductum}$} denote this index map: i.e., $j = \text{reductum}(i)$.
    We can now rewrite \eqref{CostInequalities} as
    \begin{equation}
      \label{Eq:CostInequalityB}
      \Cost(Reduce(C_{k+1}.B[i])) \leq \Cost(PP(C_{k+1}.B[i], C'_{k+1}.B[\text{reductum}(i)]))
    \end{equation}

  \item [(iii) Overall cost.]
    Let $g'$ and $g$ denote the outermost groupings (at level $k+1$) of $C'_{k+1}$ and $C_{k+1}$, respectively.
    Reduce and PairProduct each make a call on \texttt{RepresentativeGrouping} at the end of their computations to hash-cons the outermost grouping that has been constructed.
    The time complexity of a call on \texttt{RepresentativeGrouping} is dominated by the cost of computing the grouping's hash value, and thus the costs in Reduce and PairProduct are linear in $|g'|$ and $|g|$, respectively.
    By \lemref{SizeOfGrouping}, we know that $|g'| \leq |g|$, and thus the cost of the call on \texttt{RepresentativeGrouping} in Reduce is no more than the cost of the call in PairProduct.

    \hspace{1.5ex}
    Finally, using \lemref{SizeOfGrouping} and \eqrefs{CostInequalityA}{CostInequalityB}, we obtain the desired result:
    \begin{align*}
      \Cost(Reduce(C_{k+1}))
          & = |g'| + \sum_{i=1}^{e_A} \Cost(Reduce(C_{k+1}.B[i])) + \Cost(Reduce(C_{k+1}.A)) \\
          & \leq |g| + \sum_{i=1}^{e_A} \Cost(Reduce(C_{k+1}.B[i])) + \Cost(Reduce(C_{k+1}.A)) \\
          & = |g| + \sum_{i=1}^{e_A} \Cost(PP(C_{k+1}.B[i], C'_{k+1}.B[\text{reductum}(i)])) + \Cost(PP(C_{k+1}.A, C'_{k+1}.A)) \\
          & = \Cost(PP(C_{k+1}, C'_{k+1})).
    \end{align*}
\end{description}
$~~\QED$
\end{InductionStep}
\end{Proof}

\end{document}